\documentclass[10pt]{article}

\usepackage{color}
\definecolor{dred}{rgb}{.8,0.2,.2}
\definecolor{ddred}{rgb}{.8,0.5,.5}
\definecolor{dblue}{rgb}{.2,0.2,.5}

\definecolor{myurlcolor}{rgb}{0.6,0,0}
\definecolor{mycitecolor}{rgb}{0,0,0.8}
\definecolor{myrefcolor}{rgb}{0,0,0.8}
\usepackage[pagebackref]{hyperref}
\hypersetup{colorlinks,
linkcolor=myrefcolor,
citecolor=mycitecolor,
urlcolor=myurlcolor}

\usepackage{cite,amsthm,amsfonts,amssymb,amsmath,exscale,bbm}
\usepackage{graphicx}
\usepackage[utf8]{inputenc}
\usepackage{verbatim}
\usepackage{enumerate}

\usepackage{fancyhdr}
\usepackage{makeidx}
\makeindex
\graphicspath{{pics/}}

\usepackage[all]{xy}

\theoremstyle{plain}
\newtheorem{theorem}{Theorem}

\newtheorem{conjecture}[theorem]{Conjecture}
\newtheorem*{thesis}{Thesis}
\newtheorem{lemma}[theorem]{Lemma}

\newtheorem{definition}[theorem]{Definition}

\theoremstyle{definition}
\newtheorem{problem}{Problem}

\newtheorem*{answer}{Answer}

\newcommand{\C}{{\mathbb C}}  % complex numbers

\newcommand{\be}{\begin{equation}}
\newcommand{\ee}{\end{equation}}

\begin{document}

\centerline{\Large \bf Quantum Techniques for Stochastic Mechanics}

\bigskip

\centerline{\large John C.\ Baez$^{1,2}$ and Jacob D.\ Biamonte$^{3}$}

\bigskip \bigskip
\centerline{${}^1$ Department of Mathematics}
\centerline{University of California}
\centerline{Riverside, CA 92521, USA}
\bigskip
\centerline{${}^2$ Centre for Quantum Technologies}
\centerline{National University of Singapore}
\centerline{Singapore 117543}
\bigskip
\centerline{${}^3$ Quantum Complexity Science Initiative}
\centerline{Department of Physics}
\centerline{University of Malta}
\centerline{MSD 2080, Malta}

\bigskip 
\centerline{\small email: baez@math.ucr.edu, \; jacob.biamonte@qubit.org} 

\bigskip \bigskip

\centerline{ \bf Abstract} 

\bigskip 

\noindent 
{\small 
Some ideas from quantum theory are just beginning to percolate back to classical probability theory. For example, there is a widely used and successful theory of `chemical reaction networks', which describes the interaction of molecules 
in a stochastic rather than quantum way.  Computer scientists use a different
but equivalent formalism called `stochastic Petri nets' to describe collections
of randomly interacting entities.  These two equivalent formalisms also underlie many models in population biology and epidemiology.  Surprisingly, the mathematics underlying these formalisms is very much like that used in the \emph{quantum} theory of interacting particles---but modified to use probabilities instead of complex amplitudes.

In this text, we explain this fact as part of a detailed analogy between quantum mechanics and the theory of random processes.  To heighten the analogy, we call the latter `stochastic mechanics'.  

We use this analogy to explain two major results in the theory of chemical reaction networks.  First, the `deficiency zero theorem' gives conditions for the existence of equilibria in the approximation where the number of molecules of each kind is treated as varying continuously in a deterministic way.  Our proof uses tools borrowed from quantum mechanics, including a stochastic analogue of Noether's theorem relating symmetries and conservation laws.  Second, the `Anderson--Craciun--Kurtz theorem' gives conditions under which these equilibria continue to exist when we treat the number of molecules of each kind as discrete, varying in a random way.   We prove this using another tool borrowed from quantum mechanics: coherent states.

We also investigate the overlap of quantum and stochastic mechanics.  Some Hamiltonians can describe either quantum-mechanical or stochastic processes.  
These are called `Dirichlet operators', and they have an intriguing connection to
the theory of electrical circuits.

In a section on further directions for research, we describe how the stochastic Noether theorem simplifies for Hamiltonians that are Dirichlet operators, and explain some connections between stochastic Petri nets and computation.

}

%\index{text this@abc}

%%%%%%%%%%%%%%%%%%%%%%%%%%%%%%%%%%%%%%%%%%%%%%%%%%%%%%%%%%%%%%%%%%%%%%%%%%%%%%%%%%
%%%% CROSS REFERENCE LISTS 
%%%% These won't produce page numbers so can go anywhere 
\index{conductance|see{electrical circuits}} 
\index{resistance|see{electrical circuits}} 
\index{mass-action kinetics|see{chemical reaction network}}
\index{chemical reaction network|see{reaction network}} 
\index{classical random walk|see{random walk}} 
\index{quantum walk|see{random walk}} 
\index{open quantum system|see{quantum mechanics}} 
\index{linkage class|see{connected component}} 
\index{generating function|see{power series}} 
%we've improved the language a bit (I hope) and 
%here's our chance to list the terms we've changed 
\index{intensity matrix|see{infinitesimal stochastic operator}} 
\index{transition rate matrix|see{infinitesimal stochastic operator}} 
\index{multigraph|see{directed multigraph}} 
\index{graph Laplacian|see{Laplacian}} 
\index{conserved quantity|see{Noether's theorem}} 

%%%%%%%%%%%%%%%%%%%%%%%%%%%%%%%%%%%%%%%%%%%%%%%%%%%%%%%%%%%%%%%%%%%%%%%%%%%%%%%%%%
%%%% included as cross references in already defined terms 

\index{Markov process|seealso{random walk, stochastic mechanics, master equation}} 
\index{category theory|seealso{diagrammatic language}}
\index{vector|seealso{vector space}} 
\index{linear operator|see{operator}} 
\index{Hilbert space|seealso{vector space}}  
\index{quantum field theory|seealso{Feynman diagram}} 
%%%%%%%%%%%%%%%%%%%%%%%%%%%%%%%%%%%%%%%%%%%%%%%%%%%%%%%%%%%%%%%%%%%%%%%%%%%%%%%%%%

%%% some comments and possible changes/things I don't like RE the index 

\vspace{\fill}

%%%%%%%%%%%% Foreword %%%%%%%%%%%%

\section*{Foreword}
\label{foreword}

This book is about a curious relation between two ways of describing situations that change randomly with the passage of time.  The old way is \emph{probability theory} and the new way is \emph{quantum theory}.  

Quantum theory is based, not on probabilities, but on amplitudes.  We can use amplitudes to compute probabilities.\index{nonlinearity!quantum probabilities}  However, the relation between them is nonlinear: we take the absolute value of an amplitude and square it to get a probability.  It thus seems odd to treat amplitudes as directly analogous to probabilities.  Nonetheless, if we do this, some good things happen.  In particular, we can take techniques devised in quantum theory and apply them to probability theory.  This gives new insights into old problems.   

There is, in fact, a subject eager to be born, which is mathematically very much like quantum mechanics, but which features probabilities in the same equations where quantum mechanics features amplitudes.  We call this subject {\bf stochastic mechanics}.

\subsubsection*{Plan of the book}

In Section \ref{sec:1} we introduce the basic object of study here: a `stochastic Petri net'.  A stochastic Petri net describes in a very general way how collections of things of different kinds can randomly interact and turn into other things.  If we consider large numbers of things, we obtain a simplified deterministic model called the `rate equation', discussed in Section \ref{sec:2}.  More fundamental, however, is the `master equation', introduced in Section \ref{sec:3}.  This describes how the probability of having various numbers of things of various kinds changes with time.  

In Section \ref{sec:4} we consider a very simple stochastic Petri net and notice that in this case, we can solve the master equation using techniques taken from quantum mechanics.   In Section \ref{sec:5} we sketch how to generalize this: for any stochastic Petri net, we can write down an operator called a `Hamiltonian' built from `creation and annihilation operators', which describes the rate of change of the probability of having various numbers of things.  In Section \ref{sec:6} we illustrate this with an example taken from population biology.  In this example the rate equation is just the logistic equation, one of the simplest models in population biology.  The master equation describes reproduction and competition of organisms in a stochastic way.

In Section \ref{sec:7} we sketch how time evolution as described by the master equation can be written as a sum over Feynman diagrams.  We do not develop this in detail, but illustrate it with a predator--prey model from population biology.  In the process, we give a slicker way of writing down the Hamiltonian for any stochastic Petri net.

In Section \ref{sec:8} we enter into a main theme of this course: the study of \emph{equilibrium solutions} of the master and rate equations.  We present the Anderson--Craciun--Kurtz theorem, which shows how to get equilibrium solutions of the master equation from equilibrium solutions of the rate equation, at least if a certain technical condition holds.  Brendan Fong\index{Fong, Brendan} has translated Anderson, Craciun and Kurtz's original proof into the language of annihilation and creation operators, and we give Fong's proof here.  In this language, it turns out that the equilibrium solutions are mathematically just like `coherent states' in quantum mechanics.   

In Section \ref{sec:9} we give an example of the Anderson--Craciun--Kurtz theorem coming from a simple reversible reaction in chemistry.  This example leads to a puzzle that is resolved by discovering that the presence of `conserved quantities'---quantitites that do not change with time---let us construct many equilibrium solutions of the rate equation other than those given by the Anderson--Craciun--Kurtz theorem.   

Conserved quantities are very important in quantum mechanics, and they are related to symmetries by a result called Noether's theorem.  In Section \ref{sec:10} we describe a version of Noether's theorem for stochastic mechanics, which was proved with the help of Brendan Fong\index{Fong, Brendan}.   This applies, not just to systems described by stochastic Petri nets, but a much more general class of processes called `Markov processes'.  In the analogy to quantum mechanics, Markov processes are analogous to arbitrary quantum systems whose time evolution is given by a Hamiltonian.   Stochastic Petri nets are analogous to a special case of these: the case where the Hamiltonian is built from annihilation and creation operators.  

In Section \ref{sec:11} we state the analogy between quantum mechanics and stochastic mechanics more precisely, and with more attention to mathematical rigor.  This allows us to set the quantum and stochastic versions of Noether's theorem side by side and compare them in Section \ref{sec:12}.

In Section \ref{sec:13} we take a break from the heavy abstractions and look at a fun example from chemistry, in which a highly symmetrical molecule randomly hops between states.  These states can be seen as vertices of a graph, with the transitions as edges.  In this particular example we get a famous graph with 20 vertices and 30 edges, called the `Desargues graph'.  

In Section \ref{sec:14} we note that the Hamiltonian in this example is a `graph Laplacian', and, following a computation done by Greg Egan, we work out the eigenvectors and eigenvalues of this Hamiltonian explicitly.  One reason graph Laplacians are interesting is that we can use them as Hamiltonians to describe time evolution in \emph{both} stochastic \emph{and} quantum mechanics.  Operators with this special property are called `Dirichlet operators', and we discuss them in Section \ref{sec:15}.   As we explain, they also describe electrical circuits made of resistors.    Thus, in a peculiar way, the intersection of quantum mechanics and stochastic mechanics is the study of electrical circuits made of resistors!

In Section \ref{sec:16}, we study the eigenvectors and eigenvalues of an arbitrary Dirichlet operator.  We introduce a famous result called the Perron--Frobenius theorem for this purpose.   However, we also see that the Perron--Frobenius theorem is important for understanding the equilibria of Markov processes.   This becomes important later when we prove the `deficiency zero theorem'. 

We introduce the deficiency zero theorem in Section \ref{sec:17}.  This result, proved by the chemists Feinberg, Horn and Jackson, gives equilibrium solutions for the rate equation for a large class of stochastic Petri nets.  Moreover, these equilibria obey the extra condition that lets us apply the Anderson--Craciun--Kurtz theorem and obtain equilibrium solutions of the master equations as well.  However, the deficiency zero theorem is best stated, not in terms of stochastic Petri nets, but in terms of another, equivalent, formalism: `chemical reaction networks'.  So, we explain chemical reaction networks here, and use them heavily throughout the rest of the course.  However, because they are applicable to such a large range of problems, we call them simply `reaction networks'.    Like stochastic Petri nets, they describe  how collections of things of different kinds randomly interact and turn into other things.

In Section \ref{sec:18} we consider a simple example of the deficiency zero theorem taken from chemistry: a diatomic gas.   In Section \ref{sec:19} we apply the Anderson--Craciun--Kurtz theorem to the same example.

In Section \ref{sec:20} we begin working toward a proof of the deficiency zero
theorem, or at least a part of this theorem.  In this section we discuss the concept 
of `deficiency', which had been introduced before, but not really explained: the
definition that makes the deficiency easy to compute is not the one that says what this
concept really means.  In Section \ref{sec:21} we show how to rewrite the rate
equation of a stochastic Petri net---or equivalently, of a reaction network---in terms of a Markov process.  This is surprising because the rate equation is nonlinear,\index{nonlinearity!rate equation} 
while the equation describing a Markov process is linear in the probabilities involved.  The trick is to use a nonlinear operation called `matrix exponentiation'.\index{nonlinearity!matrix exponentiation}\index{matrix!exponentiation} 
In Section \ref{sec:22} we study equilibria for Markov processes.  In Section \ref{sec:23}, we use these equilibria to obtain equilibrium solutions of the rate equation, completing our treatment of the deficiency zero theorem.

We conclude the book by sketching some optional further topics:
Noether's theorem for Markov processes generated by Dirichlet operators
in Section \ref{sec:24}, and connections between computation
and Petri nets in Section \ref{sec:25}.

\subsubsection*{Acknowledgements}

These course notes are based on a series of articles on the Azimuth blog\index{Azimuth Project!blog}.  The original articles are available via this webpage:

\begin{itemize}
\item \href{http://math.ucr.edu/home/baez/networks/}{Network Theory, 
http:/$\!$/math.ucr.edu/home/baez/networks/}.\index{Azimuth Project!network theory series} 
\end{itemize}

\noindent
On the blog you can read discussion of these articles, and also make your own comments or ask questions.

We thank the readers of Azimuth for many helpful online discussions, including David Corfield, Manoj Gopalkrishnan, Greg Egan, Arjun Jain and Blake Stacey, but also many others we apologize for not listing here.  We especially thank Brendan Fong\index{Fong, Brendan} for his invaluable help in finding a `quantum proof' of the Anderson--Craciun--Kurtz theorem and proving the stochastic version of Noether's theorem.  

We thank David Doty, David Soloveichik, Erik Winfree, and others at the 
for inviting the first author to \textit{Programming with Chemical Reaction Networks: Mathematical Foundations}, a workshop at the Banff International Research Station held in June 2014, and for educating him about chemical reaction networks and their 
computational power.  In particular, Section \ref{sec:25_3} owes its existence to a talk
David Soloveichik gave in this workshop \cite{Sol14}.  \index{Soloveichik, David}

We thank Jim Stuttard of the Azimuth Project \index{Azimuth Project} \index{Stuttard, Jim} for patiently sifting through a lot of Petri net software.  This made Section 
\ref{sec:25_5} possible, though we used only a fraction of his work.

{We thank Federica Ferraris for drawing the `Baez-magician', the `rabbit archer' and the `rabbit melee' as well as helping format several of the other pictures appearing in this book.}\index{Ferraris, Federica}  We also thank Wikimedia Commons for the use of many pictures. 

Finally, we thank the Centre for Quantum Technologies for its hospitality and support, not only for ourselves but also for Brendan Fong.

%%%%%%%%%%%% TABLE OF CONTENTS %%%%%%%%%%%%

\newpage

\tableofcontents

%%%%%%%%%%%% SECTION 1 %%%%%%%%%%%%

\setcounter{section}{0}

\newpage
\section[Stochastic Petri nets]{Stochastic Petri nets} 
\index{Petri net|(}
\label{sec:1}

Stochastic Petri nets are one of many different diagrammatic languages\index{diagrammatic language} people have evolved to study complex systems.  We'll see how they're used in chemistry, molecular biology, population biology and queuing theory, which is roughly the science of waiting in line.  Here's an example of a Petri net taken from chemistry:

\begin{center}
 \includegraphics[width=119.0625mm]{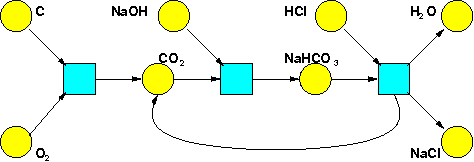}
\end{center}

\noindent It shows some chemicals and some reactions involving these chemicals.   To make it into a stochastic Petri net, we'd just label each reaction by a positive real number: the reaction rate constant, or {\bf Petri net}\index{Petri net!rate constant} for short.  
 
Chemists often call different kinds of chemicals `species'.  In general, a Petri net will have a set of {\bf species}\index{Petri net!species}, which we'll draw as yellow circles, and a set of {\bf transitions}\index{Petri net!transition}, which we'll draw as blue rectangles. Here's a Petri net from population biology: \index{species} \index{transition}
\begin{center}
 \includegraphics[width=114mm]{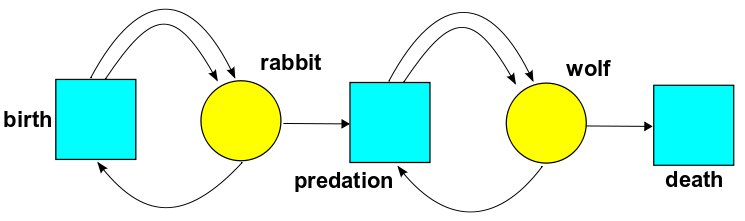}
\end{center}

\noindent Now, instead of different chemicals, the species really are different species of animals!  And instead of chemical reactions, the transitions are processes involving these species.  This Petri net has two species: {\bf rabbit} and {\bf wolf}.  It has three transitions:

\begin{itemize}
  \item In {\bf birth}, one rabbit comes in and two go out.  This is a caricature of reality: these bunnies reproduce asexually, splitting in two like amoebas.

\item In {\bf predation}, one wolf and one rabbit come in and two wolves go out.  This is a caricature of how predators need to eat prey to reproduce.    Biologists might use `biomass' to make this sort of idea more precise: a certain amount of mass will go from being rabbit to being wolf.

\item In {\bf death}, one wolf comes in and nothing goes out.   Note that we're pretending rabbits don't die unless they're eaten by wolves.  
\end{itemize}

If we labelled each transition with a number called a rate constant, we'd have a `stochastic' Petri net.  

To make this Petri net more realistic, we'd have to make it more complicated.  We're trying to explain general ideas here, not realistic models of specific situations.  Nonetheless, this Petri net already leads to an interesting model of population dynamics: a special case of the so-called `Lotka--Volterra predator-prey model\index{Lotka--Volterra!model}'.  We'll see the details soon.  

More to the point, this Petri net illustrates some possibilities that our previous example neglected.  Every transition has some `input' species and some `output' species.  But a species can show up more than once as the output (or input) of some transition.  And as we see in `death', we can have a transition with no outputs (or inputs) at all.

But let's stop beating around the bush, and give you the formal definitions.  They're simple enough:

\begin{definition}
A {\bf Petri net}\index{Petri net!definition of} consists of a set $S$ of {\bf species} and a set $T$ of {\bf transitions}, together with a function 
$$ i \colon S \times T \to \mathbb{N} $$
saying how many copies of each species shows up as {\bf input} for each transition, and a function
$$ o \colon S \times T \to \mathbb{N} $$
saying how many times it shows up as {\bf output}.
\end{definition} 

\begin{definition}
A {\bf stochastic Petri net}\index{Petri net!stochastic} is a Petri net together with a function 
$$ r \colon T \to (0,\infty) $$
giving a {\bf rate constant} for each transition.
\end{definition}

Starting from any stochastic Petri net, we can get two things.  First:

\begin{itemize}
\item The \href{http://en.wikipedia.org/wiki/Master_equation}{master equation}.\index{Petri net!master equation}   This says how the \emph{probability that we have a given number of things of each species} changes with time.     
\end{itemize}
\index{Petri net|)}

Since stochastic means `random', the master equation is what gives stochastic Petri nets their name.  The master equation is the main thing we'll be talking about in future chapters.  But not right away!  

Why not?

In chemistry, we typically have a huge number of things of each species.  For example, a gram of water contains about $ 3 \times 10^{22}$ water molecules, and a smaller but still enormous number of hydroxide ions (OH$^-$), hydronium ions (H$_3$O$^+$), and other scarier things.  These things blunder around randomly, bump into each other, and sometimes react and turn into other things.   There's a stochastic Petri net describing all this, as we'll eventually see.  But in this situation, we don't usually want to know the probability that there are, say, exactly $ 31,849,578,476,264$ hydronium ions.   That would be too much information!  We'd be quite happy knowing the \emph{expected value}\index{expected value} of the number of hydronium ions, so we'd be delighted to have a differential equation that says how this changes with time.   \index{chemistry!water}

And luckily, such an equation exists; and it's much simpler than the master equation.  So, in this section we'll talk about:

\begin{itemize}
\item The \href{http://en.wikipedia.org/wiki/Rate_equation}{rate equation}.\index{rate equation!definition of}  This says how the \emph{expected number of things of each species} changes with time. 
\end{itemize}

\noindent But first, we hope you get the overall idea.  The master equation is stochastic: at each time the number of things of each species is a random variable taking values in $ \mathbb{N}$, the set of natural numbers.  The rate equation is deterministic: at each time the expected number of things of each species is a non-random variable taking values in $ [0,\infty)$, the set of nonnegative real numbers.  If the master equation is the true story, the rate equation is only approximately true; but the approximation becomes good in some limit where the expected value\index{expected value!standard deviation} of the number of things of each species is large, and the standard deviation is comparatively small.  

If you've studied physics, this should remind you of other things.  The master equation should remind you of the quantum harmonic oscillator, where energy levels are discrete, and probabilities are involved.   The rate equation should remind you of the classical harmonic oscillator, where energy levels are continuous, and everything is deterministic.  

When we get to the `original research' part of our story, we'll see this analogy is fairly precise!  We'll take a bunch of ideas from quantum mechanics and quantum field theory, and tweak them a bit, and show how we can use them to describe the master equation for a stochastic Petri net.  

Indeed, the random processes that the master equation describes can be drawn as pictures:

\begin{center}
 \includegraphics[width=87mm]{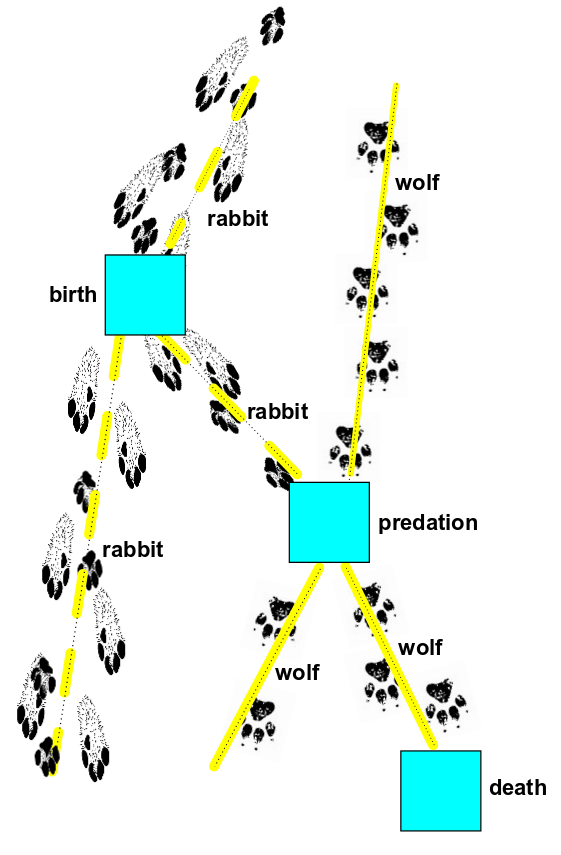}
\end{center}

\noindent
This looks like a Feynman diagram, with animals instead of particles!   It's pretty funny, but the resemblance is no joke: the math will back it up.\index{Feynman diagrams!with animals} 

We're dying to explain all the details.  But just as classical field theory is easier than quantum field theory, the rate equation is simpler than the master equation.\index{classical field theory}  So we should start there.

\subsection{The rate equation}\index{rate equation!definition of!stochastic Petri net}

If you handed over a stochastic Petri net, we can write down its rate equation.   Instead of telling you the general rule, which sounds rather complicated at first, let's do an example.  Take the Petri net we were just looking at:

\begin{center}
 \includegraphics[width=114mm]{wolf-rabbit.png}
\end{center}

\noindent
We can make it into a stochastic Petri net by choosing a number for each transition:

\begin{itemize} 
\item the birth rate constant $ \beta$
\item the predation rate constant $ \gamma$
\item the death rate constant $ \delta$
\end{itemize} 

Let $ x(t)$ be the number of rabbits and let $ y(t)$ be the number of wolves at time $ t$.  Then the rate equation looks like this:
$$ \frac{d x}{d t} = \beta x - \gamma x y $$
$$ \frac{d y}{d t} = \gamma x y - \delta y$$
It's really a system of equations, but we'll call the whole thing `the rate equation' because later we may get smart and write it as a single equation.

See how it works?  

\begin{itemize}
\item We get a term $ \beta x$ in the equation for rabbits, because rabbits are born at a rate equal to the number of rabbits times the birth rate constant $ \beta$.  
\item We get a term $ - \delta y$ in the equation for wolves, because wolves die at a rate equal to the number of wolves times the death rate constant $ \delta$.
\item We get a term $ - \gamma x y$ in the equation for rabbits, because rabbits die at a rate equal to the number of rabbits times the number of wolves times the predation rate constant $ \gamma$.  
\item We also get a term $ \gamma x y$ in the equation for wolves, because wolves are born at a rate equal to the number of rabbits times the number of wolves times $ \gamma$.  
\end{itemize}

Of course we're \emph{not} claiming that this rate equation makes any sense biologically!  For example, think about predation.  The $ \gamma x y$ terms in the above equation would make sense if rabbits and wolves roamed around randomly, and whenever a wolf and a rabbit came within a certain distance, the wolf had a certain probability of eating the rabbit and giving birth to another wolf.  At least it would make sense in the limit of large numbers of rabbits and wolves, where we can treat $ x$ and $ y$ as varying continuously rather than discretely.  That's a reasonable approximation to make sometimes.  Unfortunately, rabbits and wolves \emph{don't} roam around randomly, and a wolf \emph{doesn't} spit out a new wolf each time it eats a rabbit.

Despite that, the equations 
$$ \frac{d x}{d t} = \beta x - \gamma x y $$
$$ \frac{d y}{d t} = \gamma x y - \delta y$$
are actually studied in population biology.  As we said, they're a special case of the \href{http://en.wikipedia.org/wiki/Lotka\%E2\%80\%93Volterra\_equation}{\bf{Lotka--Volterra predator-prey model}}, which looks like this:
$$ \frac{d x}{d t} = \beta x - \gamma x y $$
$$ \frac{d y}{d t} = \epsilon x y - \delta y$$
The point is that while these models are hideously oversimplified and thus quantitatively inaccurate, they exhibit interesting \emph{qualititative} behavior that's fairly robust.  Depending on the rate constants, these equations can show either a stable equilibrium or stable periodic behavior.  And we go from one regime to another, we see a kind of catastrophe called a `Hopf bifurcation'.  You can read about this in 
\href{http://johncarlosbaez.wordpress.com/2010/12/24/this-weeks-finds-week-308}{week308} and 
\href{http://johncarlosbaez.wordpress.com/2011/02/17/this-weeks-finds-week-309/}{week309} of \textsl{This Week's Finds}.  Those consider some \emph{other} equations, not the Lotka--Volterra equations.  But their qualitative behavior is the same!

If you want stochastic Petri nets that give \emph{quantitatively} accurate models, it's better to retreat to chemistry.  Compared to animals, molecules come a lot closer to roaming around randomly and having a chance of reacting when they come within a certain distance.  So in chemistry, rate equations can be used to make accurate predictions.  

But we're digressing.  We should be explaining the \emph{general recipe for getting a rate equation from a stochastic Petri net!}  You might not be able to guess it from just one example.  In the next section, we'll do more examples, and maybe even write down a general formula.  But if you're feeling ambitious, you can try this now:

\begin{problem}\label{prob:1}
Can you write down a stochastic Petri net whose rate equation is the Lotka--Volterra predator-prey model:
$$ \frac{d x}{d t} = \beta x - \gamma x y $$
$$ \frac{d y}{d t} = \epsilon x y - \delta y$$
for arbitrary $ \beta, \gamma, \delta, \epsilon > 0$?  If not, for which values of these rate constants can you do it?
\end{problem} 

\subsection{References}

Here is free online introduction to stochastic Petri nets and their rate equations:

\begin{enumerate} 
\item[\cite{GP98}] Peter J. E. Goss and Jean Peccoud, \href{http://cmbi.bjmu.edu.cn/cmbidata/vcell/computer/petri-net.pdf}{Quantitative modeling of stochastic systems in molecular biology by using stochastic Petri nets}, {\sl Proc. Natl. Acad. Sci. USA} {\bf 95} (1988), 6750--6755. 
\end{enumerate} 

\noindent
We should admit that Petri net people say {\bf place} where we're saying {\bf species}!  The term `species' is used in the literature on chemical reaction networks, which we discuss starting in Section \ref{sec:17}.

Here are some other introductions to the subject:

\begin{enumerate} 
\item[\cite{Haa02}] Peter J. Haas, {\sl Stochastic Petri Nets: Modelling, Stability, Simulation}, Springer, Berlin, 2002. 

\item[\cite{K10}] Ina Koch, Petri nets---a mathematical formalism to analyze chemical reaction networks, {\sl Molecular Informatics} {\bf 29}, 838--843, 2010. 
\item[\cite{Wil06}] Darren James Wilkinson, \emph{Stochastic Modelling for Systems Biology}, Taylor and Francis, New York, 2006.

\end{enumerate} 

\subsection{Answers}

Here is the answer to the problem:

\vskip 1em \noindent {\bf Problem 1.}  Can you write down a stochastic Petri net whose rate equation is the Lotka--Volterra predator-prey model\index{Lotka--Volterra!model}:
$$ \frac{d x}{d t} = \beta x - \gamma x y$$
$$ \frac{d y}{d t} = \epsilon x y - \delta y$$
for arbitrary $\beta, \gamma, \delta, \epsilon > 0$? If not, for which values of these rate constants can you do it?

\vskip 1em

\begin{answer}
 We can find a stochastic Petri net that does the job for any $\beta, \gamma, \delta, \epsilon > 0$.  In fact we can find one that does the job for any possible value of $\beta, \gamma, \delta, \epsilon$.  But to keep things simple, let's just solve the original problem.

We'll consider a stochastic Petri net with two species, {\bf rabbit}  and {\bf wolf}, and four transitions:

\begin{itemize}
\item {\bf birth} (1 rabbit in, 2 rabbits out), with rate constant $\beta$
\item {\bf death} (1 wolf in, 0 wolves out), with rate constant $\delta$
\item {\bf jousting} (1 wolf and 1 rabbit in, $R$ rabbits and $W$ wolves out, where $R,W$ are arbitrary natural numbers), with rate constant $\kappa$
\item {\bf archery} (1 wolf and 1 rabbit in, $R'$ rabbits and $W'$ wolves out, where $R',W'$ are arbitrary natural numbers) with rate constant $\kappa'$.
\end{itemize} 
\noindent
All these rate constants are positive.

This gives the rate equation:
$$ \frac{dx}{dt} = \beta x + (R-1) \kappa x y + (R' - 1)\kappa' x y$$
$$ \frac{d y}{d t} = (W - 1) \kappa x y + (W' -1)\kappa' x y- \delta y $$
This is flexible enough to do the job.  

For example, let's assume that when they joust, the massive, powerful wolf always kills the rabbit, and then eats the rabbit and has one offspring ($R= 0$ and $W =2$).  And let's assume that in an archery battle, fighting with bow and arrow, the quick, crafty rabbit always kills the wolf, but does not reproduce afterward ($R' = 1$, $W' = 0$).  

\begin{center}
 \includegraphics[width=90mm]{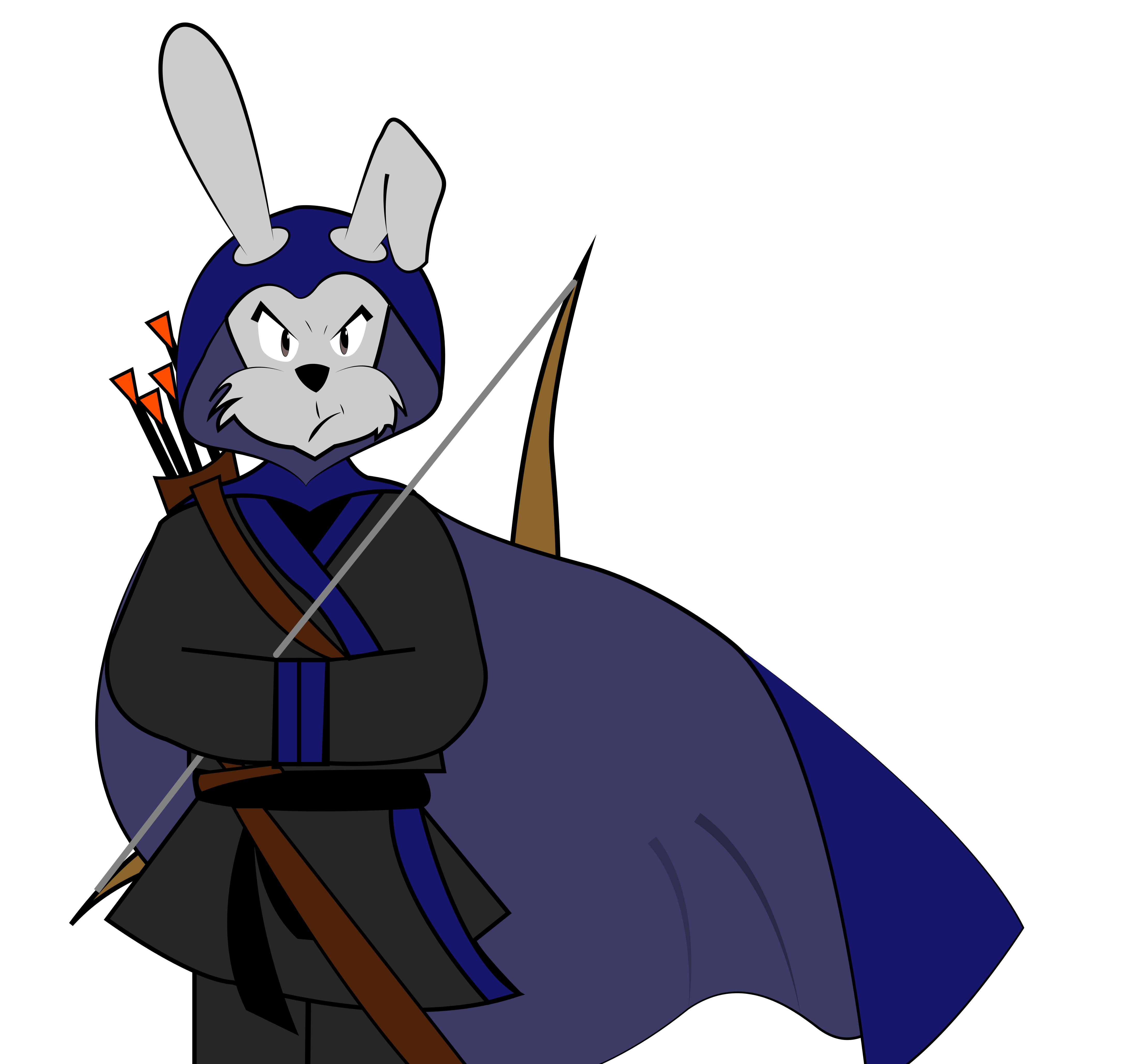}
\end{center}

Then we get
$$ \frac{dx}{dt} = \beta x - \kappa x y$$
$$ \frac{d y}{d t} = (\kappa - \kappa') x y- \delta y $$
This handles the equations 
$$ \frac{d x}{d t} = \beta x - \gamma x y$$
$$ \frac{d y}{d t} = \epsilon x y - \delta y$$
where $\beta,\gamma,\delta,\epsilon > 0$ and $\gamma > \epsilon$.  In other words, the cases where more rabbits die due to combat than wolves get born!  

We'll let you handle the cases where fewer rabbits die than wolves get born.

If we also include a death process for rabbits and birth process for wolves, we can get the fully general Lotka--Volterra equations\index{Lotka--Volterra!equations}:
$$ \frac{d x}{d t} = \beta x - \gamma x y$$
$$ \frac{d y}{d t} = \epsilon x y - \delta y$$
It's worth noting that biologists like to study these equations with different choices of sign for the constants involved: the 
\href{http://en.wikipedia.org/wiki/Lotka\%E2\%80\%93Volterra_equation}{predator-prey Lotka--Volterra equations} and the \href{http://en.wikipedia.org/wiki/Competitive_Lotka\%E2\%80\%93Volterra_equations}{competitive Lotka--Volterra equations}. 
\end{answer}

%%%%%%%%%%%%%%%%%%%%%%%%%%% SECTION 2 %%%%%%%%%%%%%%%%%%%%%%%%%%%%%%%%%%

\newpage
\section[The rate equation]{The rate equation} 
\label{sec:2}

As we saw previously in Section \ref{sec:1}, a Petri net is a picture that shows different kinds of things and processes that turn bunches of things into other bunches of things, like this:

\begin{center}
 \includegraphics[width=114mm]{wolf-rabbit.png}
\end{center}

\noindent
The kinds of things are called {\bf species} and the processes are called {\bf transitions}.   We see such transitions in chemistry:
\begin{center}
H + OH $\rightarrow$ H$_2$O
\end{center}
\noindent and population biology:
\begin{center}
amoeba $\rightarrow$ amoeba + amoeba
\end{center}
\noindent and the study of infectious diseases:
\begin{center}
infected + susceptible $\rightarrow$ infected + infected
\end{center}
\noindent and many other situations.  

A `stochastic' Petri net says the rate at which each transition
occurs.  We can think of these transitions as occurring randomly at a
certain rate---and then we get a stochastic process described by
something called the `master equation'.  But for starters, we've been
thinking about the limit where there are very many things of each
species.  Then the randomness washes out, and the expected number of
things of each species changes deterministically in a manner described
by the `rate equation'.
 
It's time to explain the general recipe for getting this rate
equation!  It looks complicated at first glance, so we'll briefly
state it, then illustrate it with tons of examples, and then state it
again.

One nice thing about stochastic Petri nets is that they let you dabble
in many sciences.  Last time we got a tiny taste of how they show up
in population biology.  This time we'll look at chemistry and models
of infectious diseases.  We won't dig very deep, but trust us: you can
do a lot with stochastic Petri nets in these subjects!  We'll give
some references in case you want to learn more.

\subsection{Rate equations: the general recipe}

Here's the recipe, really quickly:

A stochastic Petri net has a set of {\bf species} and a set of {\bf transitions}.   Let's concentrate our attention on a particular transition.  Then the $i$th species will appear $m_i$ times as the {\bf input} to that transition, and $n_i$ times as the {\bf output}.  Our transition also has a {\bf reaction rate} $0 < r < \infty$.  

The rate equation answers this question: 
$$ \frac{d x_i}{d t} = ???  $$
where $x_i(t)$ is the number of {\bf things} of the $i$th species at time $t$.   The answer is a sum of terms, one for each transition.  Each term works the same way.  For the transition we're looking at, it's
$$ r (n_i - m_i) x_1^{m_1} \cdots x_k^{m_k} $$
The factor of $(n_i - m_i)$ shows up because our transition destroys $m_i$ things of the $i$th species and creates $n_i$ of them.  The big product over all species, $x_1^{m_1} \cdots x_k^{m_k} $, shows up because our transition occurs at a rate proportional to the product of the numbers of things it takes as inputs.  The constant of proportionality is the reaction rate $r$.

\subsection{The formation of water (1)}

But let's do an example.  Here's a naive model for the formation of water from atomic hydrogen and oxygen:

\begin{center}
\includegraphics[width=61mm]{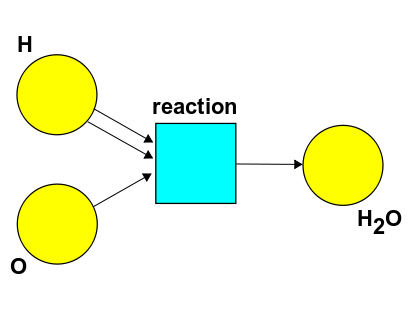}
\end{center}

\noindent
This Petri net has just one transition: two hydrogen atoms and an oxygen atom collide simultaneously and form a molecule of water.  That's not really how it goes... but what if it were?  Let's use $[\mathrm{H}]$ for the number of hydrogen atoms, and so on, and let the reaction rate be $\alpha$.  Then we get this rate equation:
$$ \begin{array}{ccl}
\displaystyle{\frac{d [\mathrm{H}]}{d t} }&=& - 2 \alpha [\mathrm{H}]^2 [\mathrm{O}] 
\\
\\
\displaystyle{\frac{d [\mathrm{O}]}{d t} }&=& - \alpha [\mathrm{H}]^2 [\mathrm{O}] 
\\
\\
\displaystyle{\frac{d [\mathrm{H}_2\mathrm{O}]}{d t} }&=& \alpha [\mathrm{H}]^2 [\mathrm{O}] \end{array}
$$
See how it works?  The reaction occurs at a rate proportional to the product of the numbers of things that appear as \emph{inputs}: two H's and one O.  The constant of proportionality is the rate constant $\alpha$.  So, the reaction occurs at a rate equal to $\alpha [\mathrm{H}]^2 [\mathrm{O}] $.   Then:

\begin{itemize}
 \item Since two hydrogen atoms get used up in this reaction, we get a factor of $-2$ in the first equation.  
 \item Since one oxygen atom gets used up, we get a factor of $ -1$ in the second equation.
 \item Since one water molecule is formed, we get a factor of $+1$ in the third equation. 
\end{itemize} 

\subsection{The formation of water (2)}

Let's do another example.  Chemical reactions rarely proceed by having \emph{three} things collide simultaneously---it's too unlikely.  So, for the formation of water from atomic hydrogen and oxygen, there will typically be an intermediate step.  Maybe something like this:

\begin{center}
 \includegraphics[width=114mm]{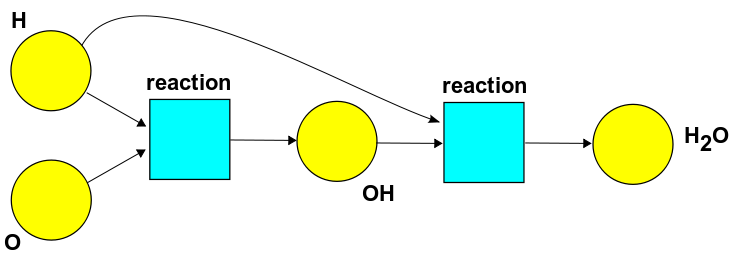}
\end{center}

\noindent
Here OH is called a `hydroxyl radical'.  We're not sure this is the \emph{most} likely pathway, but never mind---it's a good excuse to work out another rate equation.  If the first reaction has rate constant $\alpha$ and the second has rate constant $\beta$, here's what we get:
$$ \begin{array}{ccl}
\displaystyle{\frac{d [\mathrm{H}]}{d t}} &=& - \alpha [\mathrm{H}] [\mathrm{O}] - \beta [\mathrm{H}] [\mathrm{OH}] 
\\
\\
\displaystyle{\frac{d [\mathrm{OH}]}{d t}} &=&  \alpha [\mathrm{H}] [\mathrm{O}] -  \beta [\mathrm{H}] [\mathrm{OH}] 
\\
\\
\displaystyle{\frac{d [\mathrm{O}]}{d t}} &=& - \alpha [\mathrm{H}] [\mathrm{O}] 
\\
\\
\displaystyle{\frac{d [\mathrm{H}_2\mathrm{O}]}{d t}} &=& 
\beta [\mathrm{H}] [\mathrm{OH}]
\end{array} 
$$
See how it works?  Each reaction occurs at a rate proportional to the product of the numbers of things that appear as inputs.  We get minus signs when a reaction destroys one thing of a given kind, and plus signs when it creates one.  We don't get factors of 2 as we did last time, because now no reaction creates or destroys \emph{two} of anything.

\subsection{The dissociation of water (1)}

In chemistry every reaction comes with a reverse reaction.  So, if hydrogen and oxygen atoms can combine to form water, a water molecule can also `dissociate' into hydrogen and oxygen atoms.  The rate constants for the reverse reaction can be different than for the original reaction... and all these rate constants depend on the temperature.  At room temperature, the rate constant for hydrogen and oxygen to form water is a lot higher than the rate constant for the reverse reaction.  That's why we see a lot of water, and not many lone hydrogen or oxygen atoms.  But at sufficiently high temperatures, the rate constants change, and water molecules become more eager to dissociate.

Calculating these rate constants is a big subject.  We're just starting to read this book, which looked like the easiest one on the library shelf:

\begin{itemize}
\item[\cite{Log96}]  S. R. Logan, \emph{Chemical Reaction Kinetics}, Longman, Essex, (1996).
\end{itemize}
\noindent But let's not delve into these mysteries yet.  Let's just take our naive Petri net for the formation of water and turn around all the arrows, to get the reverse reaction:

\begin{center}
 \includegraphics[width=64mm]{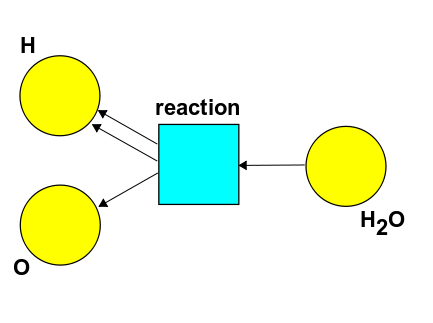}
\end{center}

If the reaction rate is $\alpha$, here's the rate equation:
$$ \begin{array}{ccl}
\displaystyle{\frac{d [\mathrm{H}]}{d t}} &=& 2 \alpha [\mathrm{H}^2\mathrm{O}] 
\\
\\
\displaystyle{\frac{d [\mathrm{O}]}{d t}} &=& \alpha [\mathrm{H}^2 \mathrm{O}] 
\\
\\
\displaystyle{\frac{d [\mathrm{H}_2\mathrm{O}]}{d t}} &=& - \alpha [\mathrm{H}^2 \mathrm{O}]   \end{array}
$$
See how it works?  The reaction occurs at a rate proportional to $[\mathrm{H}^2\mathrm{O}]$, since it has just a single water molecule as input.  That's where the $\alpha [\mathrm{H}^2\mathrm{O}]$ comes from.  Then:

\begin{itemize}
 \item Since two hydrogen atoms get formed in this reaction, we get a factor of +2 in the first equation.  
 \item Since one oxygen atom gets formed, we get a factor of +1 in the second equation.
 \item Since one water molecule gets used up, we get a factor of +1 in the third equation. 
\end{itemize} 

\subsection{The dissociation of water (2)}

Of course, we can also look at the reverse of the more realistic reaction involving a hydroxyl radical as an intermediate.  Again, we just turn around the arrows in the Petri net we had:

\begin{center}
 \includegraphics[width=114mm]{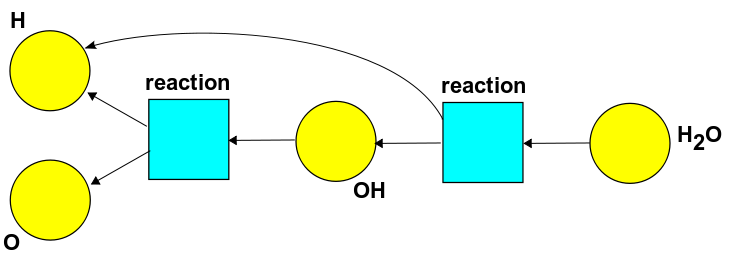}
\end{center}

Now the rate equation looks like this:
$$ \begin{array}{ccl}
\displaystyle{\frac{d [\mathrm{H}]}{d t}} &=& + \alpha [\mathrm{OH}] + \beta [\mathrm{H}_2\mathrm{O}] 
\\
\\
\displaystyle{\frac{d [\mathrm{OH}]}{d t}} &=& - \alpha [\mathrm{OH}] + \beta [\mathrm{H}_2 \mathrm{O}] 
\\
\\
\displaystyle{\frac{d [\mathrm{O}]}{d t}} &=& + \alpha [\mathrm{OH}] 
\\
\\
\displaystyle{\frac{d [\mathrm{H}_2\mathrm{O}]}{d t}} &=& - \beta [\mathrm{H}_2\mathrm{O}]
\end{array} 
$$
Do you see why?  Test your understanding of the general recipe.

By the way: if you're a category theorist, when we said ``turn around all the arrows" you probably thought \href{http://ncatlab.org/nlab/show/opposite\%20category}{``opposite category"}.\index{category theory!opposite category} \index{opposite category} And you'd be right!  A Petri net is just a way of presenting of a strict \href{http://ncatlab.org/nlab/show/symmetric+monoidal+category}{symmetric monoidal category}\index{category theory!symmetric monoidal category} \index{symmetric monoidal category} that's freely generated by some objects (the species) and some morphisms (the transitions).  When we turn around all the arrows in our Petri net, we're getting a presentation of the \emph{opposite} symmetric monoidal category.  

If you don't know what this means, don't worry.  We won't use category theory
in this book, even though it's lurking right beneath the surface.  Only at the end,
in Section \ref{sec:25_2}, will we say a bit more about Petri nets
and symmetric monoidal categories.

\subsection{The SI model}

The {\bf SI model} is an extremely simple model of an infectious disease.  We can describe it using this Petri net:  \index{SI model}

\begin{center}
 \includegraphics[width=61mm]{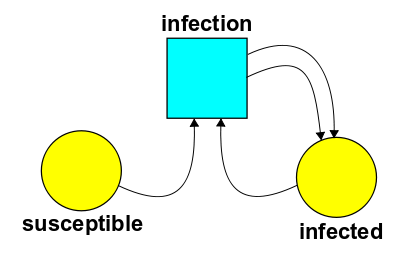}
\end{center}

There are two species: {\bf susceptible} and {\bf infected}.  And there's a transition called {\bf infection}, where an infected person meets a susceptible person and infects them.  

Suppose $S$ is the number of susceptible people and $I$ the number of infected ones.  If the rate constant for infection is $\beta$, the rate equation is
$$ \begin{array}{ccl}
\displaystyle{\frac{d S}{d t}} &=& - \beta S I \\ \\
\displaystyle{\frac{d I}{d t}} &=&  \beta S I 
\end{array}
$$
Do you see why?

By the way, it's easy to solve these equations exactly.  The total number of people doesn't change, so $S + I$ is a conserved quantity.  Use this to get rid of one of the variables.  You'll get a version of the famous \href{http://en.wikipedia.org/wiki/Logistic_function\#In_ecology:_modeling_population\_growth}{logistic equation}, so the fraction of people infected must grow sort of like this:

\begin{center}
 \includegraphics[width=80mm]{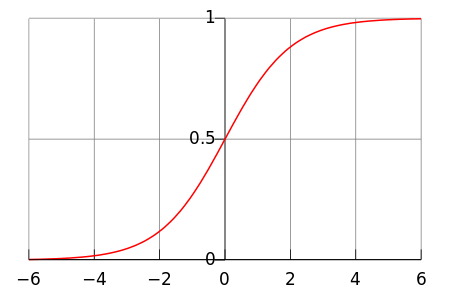}
\end{center}

\begin{problem}\label{prob:2} 
Is there a stochastic Petri net with just one species whose rate equation is the logistic equation:
$$ \frac{d P}{d t} = \alpha P - \beta P^2  ?$$
\end{problem} 

\subsection{The SIR model}
\label{sec:2_SIR}
\index{SIR model|(}

The SI model is just a warmup for the more interesting {\bf SIR model}, which was invented by Kermack and McKendrick in 1927:
\begin{enumerate} 
 \item[\cite{KM27}]  W. O. Kermack and A. G. McKendrick, A contribution to the mathematical theory of epidemics, {\sl Proc. Roy. Soc. Lond. A} {\bf 115}, 700-721, (1927).
\end{enumerate} 
\noindent
This is the only mathematical model we know to have been \emph{knighted}: Sir Model.

This model has an extra species, called {\bf resistant}, and an extra transition, called {\bf recovery}, where an infected person gets better and develops resistance to the disease:

\begin{center}
 \includegraphics[width=101mm]{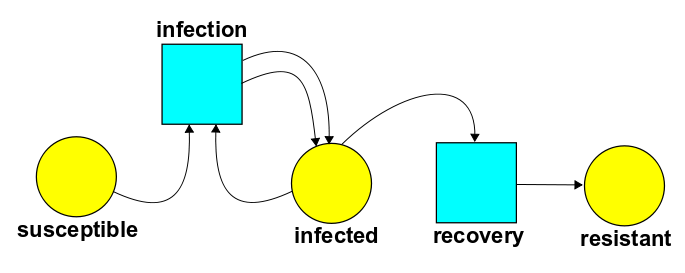}
\end{center}

If the rate constant for infection is $\beta$ and the rate constant for recovery is $\alpha$, the rate equation for this stochastic Petri net is:
$$ 
\begin{array}{ccl}
\displaystyle{\frac{d S}{d t}} &=& - \beta S I \\  \\
\displaystyle{\frac{d I}{d t}} &=&  \beta S I - \alpha I \\   \\
\displaystyle{\frac{d R}{d t}} &=&  \alpha I
\end{array}
$$
See why?

We don't know a `closed-form' solution to these equations.   But Kermack and McKendrick found an approximate solution in their original paper.  They used this to model the death rate from bubonic plague\index{bubonic plague} during an outbreak in Bombay, and got pretty good agreement.  Nowadays, of course, we can solve these equations numerically on the computer.

\index{SIR model|)}

\subsection{The SIRS model}
\index{SIRS model|(}

There's an even more interesting model of infectious disease called the {\bf SIRS model}.  This has one more transition, called {\bf losing resistance}, where a resistant person can go back to being susceptible.  Here's the Petri net:

\begin{center}
 \includegraphics[width=80mm]{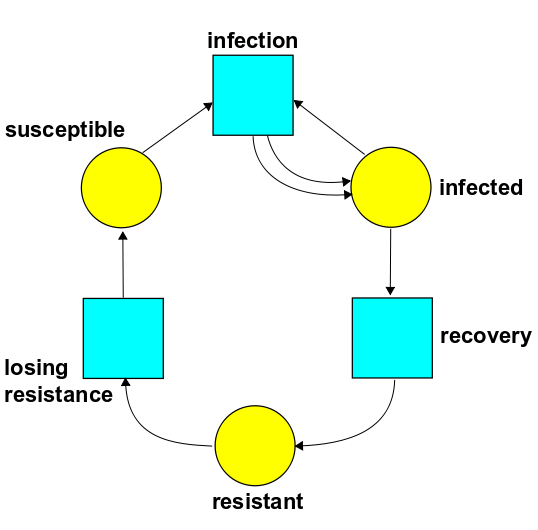}
\end{center}

\begin{problem}\label{prob:3}
 If the rate constants for recovery, infection and loss of resistance are $\alpha, \beta,$ and $\gamma$, write down the rate equations for this stochastic Petri net. 
\end{problem}

In the SIRS model we see something new: cyclic behavior!  Say you start with a few infected people and a lot of susceptible ones.  Then lots of people get infected... then lots get resistant... and then, much later, if you set the rate constants right, they lose their resistance and they're ready to get sick all over again!  You can sort of see it from the Petri net, which looks like a cycle.

You can learn about the SI, SIR and SIRS models here:

\begin{enumerate}
\item[\cite{Man06}]  Marc Mangel, \textsl{The Theoretical Biologist's Toolbox: Quantitative Methods for Ecology and Evolutionary Biology}, Cambridge U. Press, Cambridge, (2006). 
\end{enumerate} 

\noindent
For more models of this type, see:

\begin{itemize} 
 \item \href{http://en.wikipedia.org/wiki/Compartmental_models_in_epidemiology}{Compartmental models in epidemiology}, Wikipedia.
\end{itemize} 

\noindent
A `compartmental model' is closely related to a stochastic Petri net, but beware: the pictures in this article are not really Petri nets!\index{compartmental model}

\index{SIRS model|)}

\subsection{The general recipe revisited}

Now we'll remind you of the general recipe and polish it up a bit.  So, suppose we have a stochastic Petri net with $k$ species.  Let $x_i$ be the number of things of the $i$th species.   Then the rate equation looks like:
$$ \frac{d x_i}{d t} = ???  $$
It's really a bunch of equations, one for each $ 1 \le i \le k$.  But what is the right-hand side?

The right-hand side is a sum of terms, one for each transition in our Petri net.  So, let's assume our Petri net has just one transition!   (If there are more, consider one at a time, and add up the results.)  

Suppose the $i$th species appears as input to this transition $m_i$ times, and as output $n_i$ times.  Then the rate equation is
$$ \frac{d x_i}{d t} = r (n_i - m_i) x_1^{m_1} \cdots x_k^{m_k} $$
where $r$ is the rate constant for this transition.  

That's really all there is to it!  But subscripts make you eyes hurt more and as you get older---this is the real reason for using index-free notation, despite any sophisticated rationales you may have heard---so let's define a vector\index{index free vector notation}  
$$ x = (x_1, \dots , x_k) $$
that keeps track of how many things there are in each species.  Similarly let's make up an {\bf input vector}:
$$ m = (m_1, \dots, m_k) $$
and an {\bf output vector}:
$$ n = (n_1, \dots, n_k) $$
for our transition.  And a bit more unconventionally, let's define
$$ x^m = x_1^{m_1} \cdots x_k^{m_k} $$
Then we can write the rate equation for a single transition as
$$ \frac{d x}{d t} = r (n-m) x^m $$
This looks a lot nicer!  

Indeed, this emboldens us to consider a general stochastic Petri net with lots of transitions, each with their own rate constant.  Let's write $T$ for the set of transitions and $r(\tau)$ for the rate constant of the transition $\tau \in T$.  Let $n(\tau)$ and $m(\tau)$ be the input and output vectors of the transition $\tau$.  Then the rate equation for our stochastic Petri net is
$$ \frac{d x}{d t} = \sum_{\tau \in T} r(\tau) (n(\tau) - m(\tau)) x^{m(\tau)} $$
That's the fully general recipe in a nutshell.  We're not sure yet how helpful this notation will be, but it's here whenever we want it.

In Section \ref{sec:2} we'll get to the really interesting part, where ideas from quantum theory enter the game!  We'll see how things of different species randomly transform into each other via the transitions in our Petri net.  And someday we'll check that the \emph{expected} number of things in each state evolves according to the rate equation we just wrote down... at least in the limit where there are lots of things in each state.  

\subsection{Answers}

Here are the answers to the problems:

\vskip 1em \noindent {\bf Problem 2.} Is there a stochastic Petri net with just one state whose rate equation is the logistic equation:
$$ \frac{d P}{d t} = \alpha P - \beta P^2 ?$$
\vskip 1em

\begin{answer}
Yes.  Use the Petri net with one species, say {\bf amoeba}, and two transitions:
\begin{itemize} 
 \item {\bf fission}, with one amoeba as input and two as output, with rate constant $\alpha$.
 \item {\bf competition}, with two amoebas as input and one as output, with rate constant $\beta$.
\end{itemize} 
\end{answer} 
\noindent
The idea of `competition' is that when two amoebas are competing for limited resources, one may die.

\vskip 1em \noindent {\bf Problem 3.} If the rate constants for recovery, infection and loss of resistance are $\alpha, \beta$, and $\gamma$, write down the rate equations for this stochastic Petri net:

\begin{center}
 \includegraphics[width=80mm]{SIRS.png}
\end{center}

\vskip 1em

\begin{answer}
The rate equation is:
$$ 
\begin{array}{ccl} 
\displaystyle{\frac{d S}{d t}} &=& - \beta S I + \gamma R \\ \\
\displaystyle{ \frac{d I}{d t}} &=& \beta S I - \alpha I  \\ \\ 
\displaystyle{\frac{d R}{d t}} &=& \alpha I - \gamma R
\end{array} 
$$
\end{answer}

%%%%%%%%%%%%%%%%%%%  SECTION 3 %%%%%%%%%%%%%%%%%%%%%%

\newpage
\section[The master equation]{The master equation} 
\label{sec:3}

In Section \ref{sec:2} we explained the rate equation of a stochastic Petri net.  But now let's get serious: let's see what's \emph{stochastic}---that is, random--- about a stochastic Petri net.  For this we need to forget the rate equation (temporarily) and learn about the `master equation'.  This is where ideas from quantum field theory start showing up!

A Petri net has a bunch of {\bf species} and a bunch of {\bf transitions}.  Here's an example we've already seen, from chemistry:

\begin{center}
 \includegraphics[width=119.0625mm]{chemistryNetBasicA.png}
\end{center}

\noindent
The species are in yellow, the transitions in blue.  A {\bf labelling} \index{labelling}
\index{Petri net!labelling} of our Petri net is a way of having some number of things of each species.   We can draw these things as little black dots, often called {\bf tokens}: \index{token} \index{Petri net!token}

\begin{center}
 \includegraphics[width=119.0625mm]{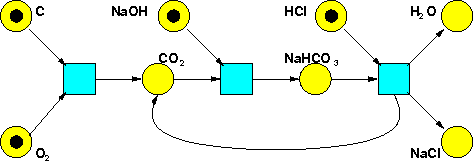}
\end{center}
\noindent
In this example there are only 0 or 1 things of each species: we've got one atom of carbon, one molecule of oxygen, one molecule of sodium hydroxide, one molecule of hydrochloric acid, and nothing else.  But in general, we can have any natural number of things of each species.

In a stochastic Petri net, the transitions occur randomly as time passes.  For example, as time passes we might see a sequence of transitions like this:

\begin{center}
 \includegraphics[width=80mm]{chemistryNetDot1A.png}
\end{center}

\begin{center}
 \includegraphics[width=80mm]{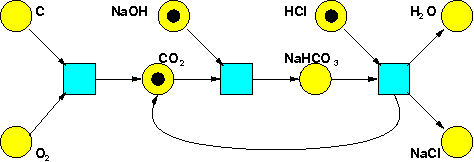}
\end{center}

\begin{center}
 \includegraphics[width=80mm]{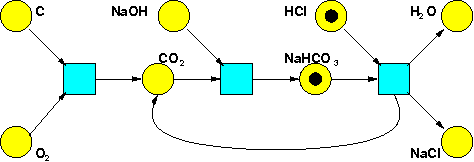}
\end{center}

\begin{center}
 \includegraphics[width=80mm]{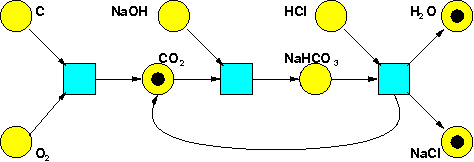}
\end{center}

Each time a transition occurs, the number of things of each species changes in an obvious way.

\subsection{The master equation} 

Now, we said the transitions occur `randomly', but that dosen't mean there's no rhyme or reason to them!  The miracle of probability theory is that it lets us state precise laws about random events.   The law governing the random behavior of a stochastic Petri net is called the `master equation'.  

In a stochastic Petri net, each transition has a {\bf rate constant}, a positive real number.  Roughly speaking, this determines the probability of that transition.

A bit more precisely: suppose we have a Petri net that is labelled in some way at some moment.  Then the probability that a given transition occurs in a short time $\Delta t$ is approximately:

\begin{itemize} 
 \item the rate constant for that transition, times
 \item the time $\Delta t$, times
 \item the number of ways the transition can occur.
\end{itemize} 

More precisely still: this formula is correct up to terms of order $(\Delta t)^2$.  So, taking the limit as $\Delta t \to 0$, we get a differential equation describing precisely how the probability of the Petri net having a given labelling changes with time!  And this is the {\bf master equation}.\index{master equation}

Now, you might be impatient to actually \emph{see} the master equation, but that would be rash.    The true master doesn't need to see the master equation.   It sounds like a Zen proverb, but it's true. The raw beginner in mathematics wants to see the solutions of an equation.  The more advanced student is content to prove that the solution exists.  But the master is content to prove that the \emph{equation} exists.  

A bit more seriously: what matters is understanding the rules that inevitably lead to some equation: actually \emph{writing it down} is then straightforward.  
And you see, there's something we haven't explained yet: ``the number of ways the transition can occur".  This involves a bit of counting.  Consider, for example, this Petri net:

\begin{center}
 \includegraphics[width=113mm]{wolf-rabbit.png}
\end{center}

\noindent
Suppose there are 10 rabbits and 5 wolves.  

\begin{itemize} 
\item How many ways can the {\bf birth} transition occur?  Since birth takes one rabbit as input, it can occur in 10 ways.
\item How many ways can {\bf predation} occur?  Since predation takes one rabbit and one wolf as inputs, it can occur in 10 $\times$ 5 = 50 ways.
\item How many ways can {\bf death} occur?  Since death takes one wolf as input, it can occur in 5 ways.  
\end{itemize} 

Or consider this one:

\begin{center}
 \includegraphics[width=63mm]{H2O-II.png}
\end{center} 

\noindent
Suppose there are 10 hydrogen atoms and 5 oxygen atoms.  How many ways can they form a water molecule?  There are 10 ways to pick the first hydrogen, 9 ways to pick the second hydrogen, and 5 ways to pick the oxygen.  So, there are
$$ 10 \times 9 \times 5 = 450 $$
ways.

Note that we're treating the hydrogen atoms as \emph{distinguishable}, so there are $10 \times 9$ ways to pick them, not $\frac{10 \times 9}{2} = {10 \choose 2}$.  In general, the number of ways to choose $M$ distinguishable things from a collection of $L$ is the \href{http://www.stanford.edu/~dgleich/publications/finite-calculus.pdf\#page=6}{{\bf falling power}} 
$$ L^{\underline{M}} = L \cdot (L - 1) 
\cdots (L - M + 1) $$\index{falling power} 
where there are $M$ factors in the product, but each is 1 less than the preceding one---hence the term `falling'.

Okay, now we've given you all the raw ingredients to work out the master equation for any stochastic Petri net.  The previous paragraph was a big fat hint.  One more nudge and you're on your own:

\begin{problem}\label{prob:4}
Suppose we have a stochastic Petri net with $k$ species and one transition with rate constant $r$.  Suppose the $i$th species appears $m_i$ times as the input of this transition and $n_i$ times as the output.   A labelling of this stochastic Petri net is a $k$-tuple of natural numbers $\ell = (\ell_1, \dots, \ell_k)$ saying how many things are in each species  Let $\psi_\ell(t)$ be the probability that the labelling is $\ell$ at time $t$.  Then the {\bf master equation} looks like this:\index{master equation}
$$ \frac{d}{d t} \psi_{\ell'}(t)  = \sum_{\ell} H_{\ell' \ell} \psi_{\ell}(t) $$
for some matrix of real numbers $H_{\ell' \ell}.$  What is this matrix?
\end{problem} 

You can write down a formula for this matrix using what we've told you.    And then, if you have a stochastic Petri net with more transitions, you can just compute the matrix for each transition using this formula, and add them all up.

There's a straightforward way to solve this problem, but we want to get the solution by a strange route: we want to \emph{guess} the master equation using ideas from \emph{quantum field theory!}

Why?  Well, if we think about a stochastic Petri net whose labelling undergoes random transitions as we've described, you'll see that any possible `history' for the labelling can be drawn in a way that looks like a \href{http://en.wikipedia.org/wiki/Feynman\_diagram}{Feynman diagram}.  In quantum field theory, Feynman diagrams show how things interact and turn into other things.  But that's what stochastic Petri nets do, too!\index{Feynman diagrams!and stochastic Petri nets}  

For example, if our Petri net looks like this:

\begin{center}
 \includegraphics[width=114mm]{wolf-rabbit.png}
\end{center}

\noindent
then a typical history can be drawn like this:

\begin{center}
 \includegraphics[width=80mm]{rabbit_wolf_feynman_diagram.png}
\end{center}

\noindent
Some rabbits and wolves come in on top.  They undergo some transitions as time passes, and go out on the bottom. The vertical coordinate is time, while the horizontal coordinate doesn't really mean anything: it just makes the diagram easier to draw.  

If we ignore all the artistry that makes it cute, this Feynman diagram is just a graph with species as edges and transitions as vertices.\index{Feynman diagrams!and probabilities}    Each transition occurs at a specific time.

We can use these Feynman diagrams to compute the probability that if we start it off with some labelling at time $t_1$, our stochastic Petri net will wind up with some other labelling at time $t_2$.   To do this, we just take a sum over Feynman diagrams that start and end with the given labellings.  For each Feynman diagram, we integrate over all possible times at which the transitions occur.  And what do we integrate?  Just the product of the rate constants for those transitions!

That was a bit of a mouthful, and it doesn't really matter if you followed it in detail.  What matters is that it \emph{sounds a lot like stuff you learn when you study quantum field theory!}   

That's one clue that something cool is going on here.  Another is the master equation itself:
$$ \frac{d}{d t} \psi_{\ell'}(t) = \sum_{\ell} H_{\ell' \ell} \psi_{\ell}(t) $$
This looks a lot like \href{http://en.wikipedia.org/wiki/Schr\%C3\%B6dinger_equation}{Schr\"{o}dinger's equation}, the basic equation describing how a quantum system changes with the passage of time.  

We can make it look even more like Schr\"{o}dinger's equation\index{quantum mechanics!Schr\"{o}dinger's equation} if we create a vector space with the labellings $\ell$ as a basis.  The numbers $\psi_\ell(t)$ will be the components of some vector $\psi(t)$ in this vector space.   The numbers $H_{\ell' \ell}$ will be the matrix entries of some operator $H$ on that vector space.  And the master equation becomes:
$$ \frac{d}{d t} \psi(t) = H \psi(t) $$
Compare this with Schr\"{o}dinger's equation\index{quantum mechanics!Schr\"{o}dinger's equation!definition of}:
$$ i \frac{d}{d t} \psi(t) = H \psi(t) $$
The only visible difference is that factor of $i$!  

But of course this is linked to another big difference: in the master equation $\psi$ describes probabilities, so it's a vector in a \emph{real} vector space.   In quantum theory $\psi$ describes \href{http://en.wikipedia.org/wiki/Probability_amplitude}{amplitudes}, so it's a vector in a \emph{complex} Hilbert space.\index{Hilbert space!complex}

Apart from this huge difference, everything is a lot like quantum field theory.  In particular, our vector space is a lot like the \href{http://en.wikipedia.org/wiki/Fock\_space}{Fock space} one sees in quantum field theory.   Suppose we have a quantum particle that can be in $k$ different states.  Then its \index{quantum field theory!Fock space} Fock space is the Hilbert space\index{Hilbert space!Fock space} we use to describe an arbitrary collection of such particles.  It has an orthonormal basis denoted
$$ | \ell_1 \cdots \ell_k \rangle $$
where $\ell_1, \dots, \ell_k$ are natural numbers saying how many particles there are in each state.  So, any vector in Fock space looks like this:
$$ \psi = \sum_{\ell_1, \dots, \ell_k}  \psi_{\ell_1 , \dots, \ell_k} \,  | \ell_1 \cdots \ell_k \rangle $$
But if write the whole list $\ell_1, \dots, \ell_k$ simply as $\ell$, this becomes
$$ \psi = \sum_{\ell} \psi_{\ell}   | \ell \rangle $$
This is almost like what we've been doing with Petri nets!---except we hadn't gotten around to giving names to the basis vectors.

In quantum field theory class, you would learn lots of interesting operators on Fock space: annihilation and creation operators, number operators, and so on.  So, when considering this master equation
$$ \frac{d}{d t} \psi(t) = H \psi(t) $$
it seemed natural to take the operator $H$ and write it in terms of these.  There was an obvious first guess, which didn't quite work... but thinking a bit harder eventually led to the right answer.   Later, it turned out people had already thought about similar things. So, we want to explain this.   

When we first started working on this stuff, we focused on the difference between collections of \emph{indistinguishable} things, like bosons or fermions, and collections of \emph{distinguishable} things, like rabbits or wolves.  But with the benefit of hindsight, it's even more important to think about the difference between quantum theory, which is all about \emph{probability amplitudes}, and the game we're playing now, which is all about \emph{probabilities}.  So, in the next Section, we'll explain how we need to modify quantum theory so that it's about probabilities.  This will make it easier to guess a nice formula for $H$.

\subsection{Answers}

Here is the answer to the problem:

\vskip 1em \noindent {\bf Problem 4.}  Suppose we have a stochastic Petri net with $k$ species and just one transition, whose rate constant is $r$. Suppose the $i$th species appears $m_i$ times as the input of this transition and $n_i$ times as the output. A labelling of this stochastic Petri net is a $k$-tuple of natural numbers $\ell = (\ell_1, \dots, \ell_k)$ saying how many things there are of each species. Let $\psi_\ell(t)$ be the probability that the labelling is $\ell$ at time $t$. Then the master equation looks like this:
$$\frac{d}{d t} \psi_{\ell'}(t) = \sum_{\ell} H_{\ell' \ell} \psi_{\ell}(t)$$
for some matrix of real numbers $ H_{\ell' \ell}$. What is this matrix?

\begin{answer}
To compute $H_{\ell' \ell}$ it's enough to start the Petri net in a definite labelling $\ell$ and see how fast the probability of being in some labelling $\ell'$ changes.  In other words, if at some time $ t$ we have
$$ \psi_{\ell}(t) = 1$$
then
$$ \frac{d}{dt} \psi_{\ell'}(t) = H_{\ell' \ell}$$
at this time.

Now, suppose we have a Petri net that is labelled in some way at some moment.  Then the probability that the transition occurs in a short time $ \Delta t$ is approximately:

\begin{itemize} 
 \item the rate constant $ r$, times
 \item the time $ \Delta t$, times
 \item the number of ways the transition can occur, which is the product of falling powers $ \ell_1^{\underline{m_1}} \cdots \ell_k^{\underline{m_k}}$.  Let's call this product $ \ell^{\underline{m}}$ for short.
\end{itemize} 

Multiplying these 3 things we get
$$ r \ell^{\underline{m}} \Delta t$$
So, the \emph{rate} at which the transition occurs is just:
$$ r \ell^{\underline{m}}$$
And when the transition occurs, it eats up $m_i$ things of the \emph{i}th species, and produces $n_i$ things of that species.  So, it carries our system from the original labelling $ \ell$ to the new labelling 
$$ \ell' = \ell  + n - m$$
So, \emph{in this case} we have
$$  \frac{d}{dt} \psi_{\ell'}(t)  = r \ell^{\underline{m}}$$
and thus
$$ H_{\ell' \ell} = r \ell^{\underline{m}}$$
However, that's not all: there's another case to consider!  Since the probability of the Petri net being in this new labelling $ \ell'$ is going up, the probability of it staying in the original labelling $ \ell$ must be going down by the same amount.  So we must also have
$$ H_{\ell \ell} = - r \ell^{\underline{m}}$$
We can combine both cases into one formula like this:
$$ H_{\ell' \ell} = r \ell^{\underline{m}} \left(\delta_{\ell', \ell + n - m} - \delta_{\ell', \ell}\right)$$
Here the first term tells us how fast the probability of being in the new labelling is going up.  The second term tells us how fast the probability of staying in the original labelling is going down.  

Note: each column in the matrix $ H_{\ell' \ell}$ sums to zero, and all the off-diagonal entries\index{matrix!off-diagonal entries} are nonnegative.  That's good: in the next section we'll show that this matrix must be `infinitesimal stochastic', meaning precisely that it has these properties!
\end{answer}

%%%%%%%%%%  SECTION 4  %%%%%%%%%%%%%%%

\newpage
\section[Probabilities vs amplitudes]{Probabilities vs amplitudes} 
\label{sec:4}

In Section \ref{sec:3} we saw clues that stochastic Petri nets are a lot like quantum field theory, but with probabilities replacing amplitudes.  There's a powerful analogy at work here, which can help us a lot.  It's time to make that analogy precise. But first, let us quickly sketch why it could be worthwhile.

\subsection{A Poisson process}

Consider this stochastic Petri net with rate constant $r$:

\begin{center}
 \includegraphics[width=62mm]{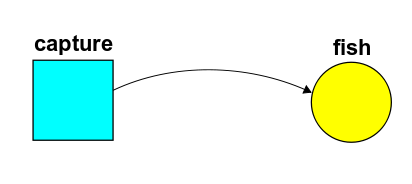}
\end{center}

\noindent
It describes an inexhaustible supply of fish swimming down a river, and getting caught when they run into a fisherman's net.  In any short time $\Delta t$ there's a chance of about $r \Delta t$ of a fish getting caught.  There's also a chance of two or more fish getting caught, but this becomes negligible by comparison as $\Delta t \to 0$.  Moreover, the chance of a fish getting caught during this interval of time is independent of what happens before or afterwards.  This sort of process is called a \href{http://en.wikipedia.org/wiki/Poisson\_process}{{\bf Poisson process}}.

\begin{problem}\label{prob:5}
Suppose we start out knowing for sure there are no fish in the fisherman's net.  What's the probability that he has caught $n$ fish at time $t$?
\end{problem}

\begin{answer}\index{power series!probability distribution} 
At any time there will be some probability of having caught $n$ fish; let's call this probability $\psi(n,t)$.  We can summarize all these probabilities in a single power series, called a \href{http://en.wikipedia.org/wiki/Probability-generating_function}{{\bf generating function}}:
$$\Psi(t) = \sum_{n=0}^\infty \psi(n,t) \, z^n $$\index{formal variable}
Here $z$ is a formal variable---don't ask what it means, for now it's just a trick.   In quantum theory we use this trick when talking about collections of photons rather than fish, but then the numbers $\psi(n,t)$ are complex `amplitudes'. Now they are real probabilities, but we can still copy what the physicists do, and use this trick to rewrite the master equation as follows:
$$\frac{d}{d t} \Psi(t) = H \Psi(t) $$
This describes how the probability of having caught any given number of fish changes with time.

What's the operator $H$?   Well, in quantum theory we describe the creation of photons using a certain operator on power series called the {\bf creation operator}:\index{quantum field theory!creation operator} \index{creation operator}
$$a^\dagger \Psi = z \Psi $$
We can try to apply this to our fish.  If at some time we're 100\% sure we have $n$ fish, we have 
$$\Psi = z^n$$
so applying the creation operator gives
$$a^\dagger \Psi = z^{n+1}$$
One more fish!  That's good.  So, an obvious wild guess is 
$$H = r a^\dagger$$
where $r$ is the rate at which we're catching fish.  Let's see how well this guess works. 

If you know how to exponentiate operators, you know to solve this equation:
$$\frac{d}{d t} \Psi(t) = H \Psi(t) $$
It's easy:
$$\Psi(t) = \mathrm{exp}(t H) \Psi(0)$$
Since we start out knowing there are no fish in the net, we have
$$\Psi(0) = 1$$
so with our guess for $H$ we get
$$\Psi(t) = \mathrm{exp}(r t a^\dagger) 1$$
But $a^\dagger$ is the operator of multiplication by $z$, so $\mathrm{exp}(r t a^\dagger)$ is multiplication by $e^{r t z}$, and
$$\Psi(t) = e^{r t z} = \sum_{n = 0}^\infty \frac{(r t)^n}{n!} \, z^n$$
So, if our guess is right, the probability of having caught $n$ fish at time $t$ is 
$$\frac{(r t)^n}{n!}$$
Unfortunately, this can't be right, because these probabilities don't sum to 1!  Instead their sum is
$$\sum_{n=0}^\infty \frac{(r t)^n}{n!} = e^{r t}$$
We can try to wriggle out of the mess we're in by dividing our answer by this fudge factor.  It sounds like a desperate measure, but we've got to try something!  

This amounts to guessing that the probability of having caught $n$ fish by time $t$ is
$$\frac{(r t)^n}{n!} \, e^{-r t}$$
And this guess is \emph{right!}  This is called the \href{http://en.wikipedia.org/wiki/Poisson\_distribution}{{\bf Poisson distribution}}: it's famous for being precisely the answer to the problem we're facing.\index{Poisson distribution}   

So on the one hand our wild guess about $H$ was wrong, but on the other hand it was not so far off.  We can fix it as follows:
$$H = r (a^\dagger - 1)$$
The extra $-1$ gives us the fudge factor we need.
\end{answer}

So, a wild guess corrected by an ad hoc procedure seems to have worked!  But what's really going on?  

What's really going on is that $a^\dagger$, or any multiple of this, is not a legitimate Hamiltonian for a master equation: if we define a time evolution operator $\exp(t H)$ using a Hamiltonian like this, probabilities won't sum to 1!  But $a^\dagger - 1$ is okay.  So, we need to think about which Hamiltonians are okay.

In quantum theory, self-adjoint Hamiltonians are okay.  But in probability theory, we need some other kind of Hamiltonian.  Let's figure it out.

\subsection{Probability theory vs quantum theory}

Suppose we have a system of any kind: physical, chemical, biological, economic, whatever.   The system can be in different states.  In the simplest sort of model, we say there's some set $X$ of states, and say that at any moment in time the system is definitely in one of these states.  But we want to compare two other options:

\begin{itemize} 

\item In a {\bf probabilistic} model, we may instead say that the system has a {\bf probability} $\psi(x)$ of being in any state $x \in X$.  These probabilities are nonnegative real numbers with 
$$\sum_{x \in X} \psi(x) = 1$$
\item In a {\bf quantum} model, we may instead say that the system has an {\bf amplitude} $\psi(x)$ of being in any state $x \in X$.  These amplitudes are complex numbers with 
$$\sum_{x \in X} | \psi(x) |^2 = 1$$
\end{itemize} 

Probabilities and amplitudes are similar yet strangely different.  Of course given an amplitude we can get a probability by taking its absolute value and squaring it.  This is a vital bridge from quantum theory to probability theory.  In the present section, however, we don't want to focus on the bridges, but rather the \emph{parallels} between these theories.

We often want to replace the sums above by integrals.  For that we need to replace our set $X$ by a \href{http://en.wikipedia.org/wiki/Measure\_\%28mathematics\%29}{measure space},\index{measure space} which is a set equipped with enough structure that you can integrate  real or complex functions defined on it.  Well, at least you can integrate so-called `integrable' functions---but we'll neglect all issues of analytical rigor here.  Then:\index{integration}

\begin{itemize} 
\item In a {\bf probabilistic} model, the system has a {\bf probability distribution} $\psi \colon X \to \mathbb{R}$, which obeys $\psi \ge 0$ and
$$\int_X \psi(x) \, d x = 1$$
\item In a {\bf quantum} model, the system has a {\bf wavefunction} $\psi \colon X \to \mathbb{C}$, which obeys 
$$\int_X | \psi(x) |^2 \, d x= 1$$
\end{itemize} 

In probability theory, we integrate $\psi$ over a set $S \subset X$ to find out the probability that our systems state is in this set.  In quantum theory we integrate $|\psi|^2$ over the set to answer the same question.

We don't need to think about sums over sets and integrals over measure spaces separately: there's a way to make any set $X$ into a measure space\index{measure space} such that by definition,
$$\int_X \psi(x) \, dx = \sum_{x \in X} \psi(x)$$
In short, integrals are more general than sums!  So, we'll mainly talk about integrals, until the very end.\index{integration!vs sums}

In probability theory, we want our probability distributions to be vectors in some vector space.  Ditto for wave functions in quantum theory!  So, we make up some vector spaces:

\begin{itemize} 
\item In probability theory, the probability distribution $\psi$ is a vector in the space
$$L^1(X) = \{ \psi \colon X \to \mathbb{C} : \int_X |\psi(x)| \, d x < \infty \} $$
\item In quantum theory, the wavefunction $\psi$ is a vector in the space
$$L^2(X) = \{ \psi \colon X \to \mathbb{C}  :  \int_X |\psi(x)|^2 \, d x < \infty \} $$
\index{$L^1$} \index{$L^2$}
\end{itemize}  

You may wonder why we defined $L^1(X)$ to consist of \emph{complex} functions when probability distributions are real.  We're just struggling to make the analogy seem as strong as possible.  In fact probability distributions are not just real but nonnegative.  We need to say this somewhere... but we can, if we like, start by saying they're complex-valued functions, but then whisper that they must in fact be nonnegative (and thus real).  It's not the most elegant solution, but that's what we'll do for now.  Now:

\begin{itemize}
\item The main thing we can do with elements of $L^1(X)$, besides what we can do with vectors in any vector space, is integrate one.  This gives a linear map:
$$\int \colon L^1(X) \to \mathbb{C} $$
\item The main thing we can with elements of $L^2(X)$, besides the things we can do with vectors in any vector space, is take the inner product of two:
$$\langle \psi, \phi \rangle = \int_X \overline{\psi}(x) \phi(x) \, d x $$
This gives a map that's linear in one slot and conjugate-linear in the other:
$$\langle - , - \rangle \colon L^2(X) \times L^2(X) \to \mathbb{C} $$
\index{$L^1$} \index{$L^2$}
\end{itemize}

First came probability theory with $L^1(X)$; then came quantum theory with $L^2(X)$.  Naive extrapolation would say it's about time for someone to invent an even more bizarre theory of reality based on $L^3(X).$ In this, you'd have to integrate the product of \emph{three} wavefunctions to get a number!  The math of \href{http://en.wikipedia.org/wiki/Lp_space}{\emph{$L^p$ spaces}} is already well-developed, so give it a try if you want.  We'll stick to $L^1$ and $L^2$.
\index{$L^1$} \index{$L^2$}

\subsection[Stochastic vs unitary operators]{Stochastic versus unitary operators}

Now let's think about time evolution:

\begin{itemize} 
\item In probability theory, the passage of time is described by a map sending probability distributions to probability distributions.  This is described using a \href{http://en.wikipedia.org/wiki/Stochastic\_matrix}{{\bf stochastic operator}}
\index{stochastic operator} \index{operator!stochastic}
\index{stochastic operator!definition of} \index{operator!stochastic!definition of}
$$U \colon L^1(X) \to L^1(X)$$
meaning a linear operator such that
$$\int U \psi = \int \psi $$
and
$$\psi \ge 0 \quad \Rightarrow \quad U \psi \ge 0$$
\item In quantum theory the passage of time is described by a map sending wavefunction to wavefunctions.  This is described using an \href{http://en.wikipedia.org/wiki/Unitary_operator}{{\bf isometry}}
$$U \colon L^2(X) \to L^2(X)$$\index{isometry}
meaning a linear operator such that
$$\langle U \psi , U \phi \rangle = \langle \psi , \phi \rangle $$
In quantum theory we usually want time evolution to be reversible, so we focus on isometries that have inverses: these are called \href{http://en.wikipedia.org/wiki/Unitary_operator}{ {\bf unitary}} operators.  In probability theory we often consider stochastic operators that are \emph{not} invertible.
\end{itemize} 

\subsection{Infinitesimal stochastic versus self-adjoint operators}
\label{sec:4_hamiltonians}

Sometimes it's nice to think of time coming in discrete steps.  But in theories where we treat time as continuous, to describe time evolution we usually need to solve a differential equation.  This is true in both probability theory and quantum theory.

In probability theory we often describe time evolution using a differential equation called the \href{http://en.wikipedia.org/wiki/Master_equation}{ {\bf master equation}}:\index{master equation}
$$\frac{d}{d t} \psi(t) = H \psi(t) $$
whose solution is 
$$\psi(t) = \exp(t H)\psi(0) $$
In quantum theory we often describe time evolution using a differential equation called \href{http://en.wikipedia.org/wiki/Schr\%C3\%B6dinger_equation}{{\bf Schr\"odinger's equation}}:
$$i \frac{d}{d t} \psi(t) = H \psi(t) $$
whose solution is 
$$\psi(t) = \exp(-i t H)\psi(0) $$
In both cases, we call the operator $H$ the {\bf Hamiltonian}.  In fact the appearance of $i$ in the quantum case is purely conventional; we could drop it to make the analogy better, but then we'd have to work with `skew-adjoint' operators instead of self-adjoint ones in what follows.  

Let's guess what properties an operator $H$ should have to make $\exp(-i t H)$ unitary for all $t$.  We start by assuming it's an isometry:\index{quantum mechanics!isometry}
$$\langle \exp(-i t H) \psi, \exp(-i t H) \phi \rangle = \langle \psi, \phi \rangle$$
Then we differentiate this with respect to $t$ and set $t = 0$, getting
$$\langle -i H \psi, \phi \rangle +  \langle \psi, -i H \phi \rangle = 0$$
or in other words:
$$\langle H \psi, \phi \rangle = \langle \psi, H \phi \rangle $$
Physicists call an operator obeying this condition \href{http://en.wikipedia.org/wiki/Self-adjoint\_operator}{self-adjoint}.   Mathematicians know there's \href{http://en.wikipedia.org/wiki/Self-adjoint\_operator\#Self-adjoint\_operators}{more to it}, but now is not the time to discuss such subtleties, intriguing though they be.  All that matters now is that there is, indeed, a correspondence between self-adjoint operators and well-behaved `1-parameter unitary groups' $\exp(-i t H)$.  This is called \href{http://en.wikipedia.org/wiki/Stone\%27s_theorem_on\_one-parameter_unitary_groups}{Stone's Theorem}.  \index{one-parameter unitary group}

But now let's copy this argument to guess what properties an operator $H$ must have to make $\exp(t H)$ stochastic.   We start by assuming $\exp(t H)$ is stochastic, so 
$$\int \exp(t H) \psi = \int \psi $$
and 
$$\psi \ge 0 \quad \Rightarrow \quad \exp(t H) \psi \ge 0$$
We can differentiate the first equation with respect to $t$ and set $t = 0$, getting
$$\int H \psi = 0$$
for all $\psi$.

But what about the second condition, 
$$\psi \ge 0 \quad \Rightarrow \quad \exp(t H) \psi \ge 0? $$
It seems easier to deal with this in the special case when integrals over $X$ reduce to sums.  So let's suppose that happens... and let's start by seeing what the \emph{first} condition says in this case.  

In this case, $L^1(X)$ has a basis of  `Kronecker delta functions':  The Kronecker delta function $\delta_i$ vanishes everywhere except at one point $i \in X$, where it equals 1.  Using this basis, we can write any operator on $L^1(X)$ as a matrix.

As a warmup, let's see what it means for an operator 
$$U \colon L^1(X) \to L^1(X)$$
to be stochastic in this case.  We'll take the conditions 
$$\int U \psi = \int \psi $$
and
$$\psi \ge 0 \quad \Rightarrow \quad U \psi \ge 0$$
and rewrite them using matrices.  For both, it's enough to consider the case where $\psi$ is a Kronecker delta, say $\delta_j$.  

In these terms, the first condition says 
$$\sum_{i \in X} U_{i j} = 1 $$
for each column $j$.  The second says
$$ U_{i j} \ge 0$$
for all $i, j$.  So in this case, a stochastic operator is just a square matrix where each column sums to 1 and all the entries are nonnegative.  (Such matrices are often called \href{http://en.wikipedia.org/wiki/Stochastic\_matrix}{left stochastic}.) \index{stochastic operator} \index{operator!stochastic}

Next, let's see what we need for an operator $H$ to have the property that $\exp(t H)$is stochastic for all $t \ge 0$.  It's enough to assume $t$ is very small, which lets us use the approximation
$$\exp(t H) = 1 + t H + \cdots $$
and work to first order in $t$.  Saying that each column of this matrix sums to 1 then amounts to
$$\sum_{i \in X} \delta_{i j} + t H_{i j} + \cdots  = 1$$
which requires
$$ \sum_{i \in X} H_{i j} = 0$$
Saying that each entry is nonnegative amounts to
$$\delta_{i j} + t H_{i j} + \cdots  \ge 0$$
When $i = j$ this will be automatic when $t$ is small enough, so the meat of this condition is
$$H_{i j} \ge 0~\mathrm{if}~i \ne j$$
So, let's say $H$ is {\bf infinitesimal stochastic} if its columns sum to zero and its off-diagonal entries are nonnegative. 
\index{infinitesimal stochastic operator!definition of} \index{operator!infinitesimal stochastic!definition of}
The idea is that any infinitesimal stochastic operator should be the \href{http://en.wikipedia.org/wiki/Infinitesimal_generator\_\%28stochastic_processes\%29}{infinitesimal generator} of a stochastic process.  The adjective `infinitesimal stochastic' doesn't roll off the tongue, but we don't know one we like better.  Some people call an infinitesimal stochastic operator an {\bf intensity matrix}\index{intensity matrix},
while others call it a {\bf stochastic Hamiltonian}, but we really want an adjective.

Anyway, when we get the details straightened out, any 1-parameter family of stochastic operators
$$ U(t) \colon L^1(X) \to L^1(X) \qquad   t \ge 0 $$
obeying 
$$ U(0) = I $$
$$ U(t) U(s) = U(t+s) $$
and continuity:
$$ t_i \to t \quad \Rightarrow \quad U(t_i) \psi \to U(t)\psi $$
should be of the form
$$ U(t) = \exp(t H)$$
for a unique `infinitesimal stochastic operator' $H$.

When $X$ is a finite set, this is true---and an infinitesimal stochastic operator is just a square matrix whose columns sum to zero and whose off-diagonal entries are nonnegative.\index{matrix!off-diagonal entries}   But we don't know a really 
good theorem characterizing infinite stochastic operators for general measure
spaces $X$.  We shall say what we know about this in Section \ref{sec:11_hille-yosida}.

Luckily, for our work on stochastic Petri nets we only need to understand the case where $X$ is a countable set and our integrals are really just sums.  This should be almost like the case where $X$ is a finite set---but we'll need to take care that all our sums converge.

\subsection{The moral}

Now we can see why a Hamiltonian like $a^\dagger $ is no good, while $ a^\dagger - 1$ is good.  (We'll ignore the rate constant $r$ since it's irrelevant here.)  The first one is not infinitesimal stochastic, while the second one is!

In this example, our set of states is the natural numbers:
$$X = \mathbb{N}$$
The probability distribution 
$$\psi \colon \mathbb{N} \to \mathbb{C}$$
tells us the probability of having caught any specific number of fish. 

The creation operator is not infinitesimal stochastic: in fact, it's stochastic!  Why?  Well, when we apply the creation operator, what was the probability of having $n$ fish now becomes the probability of having $n+1$ fish.  So, the probabilities remain nonnegative, and their sum over all $n$ is unchanged.  Those two conditions are all we need for a stochastic operator.

Using our fancy abstract notation, these conditions say:
$$\int a^\dagger \psi = \int \psi $$
and
$$ \psi \ge 0 \; \Rightarrow \; a^\dagger \psi \ge 0 $$
So, precisely by virtue of being stochastic, the creation operator fails to be infinitesimal stochastic:
$$ \int a^\dagger \psi \ne 0 $$
Thus it's a bad Hamiltonian for our stochastic Petri net.  

On the other hand, $a^\dagger - 1$ \emph{is} infinitesimal stochastic.  Its off-diagonal entries are the same as those of $a^\dagger$, so they're nonnegative.\index{matrix!off-diagonal entries}   Moreover:
$$ \int (a^\dagger - 1) \psi = 0 $$
precisely because
$$\int a^\dagger \psi = \int \psi $$
You may be thinking: all this fancy math just to understand a single stochastic Petri net, the simplest one of all!  

\begin{center}
 \includegraphics[width=62mm]{fisher.png}
\end{center}

But next we'll explain a general recipe which will let you write down the Hamiltonian for \emph{any} stochastic Petri net.  The lessons we've just learned will make this much easier. And pondering the analogy between probability theory and quantum theory will also be good for our bigger project of unifying the applications of network diagrams to dozens of different subjects.  

%%%%%%%%%% SECTION 5 %%%%%%%%%%% 

\newpage
\section{Annihilation and creation operators}
\label{sec:5}

Now comes the fun part.   Let's see how tricks from quantum theory can be used to describe random processes.  We'll try to make this section completely self-contained, except at the very end.  So, even if you skipped a bunch of the previous ones, this should make sense.

You'll need to know a bit of math: calculus, a tiny bit probability theory, and linear operators on vector spaces.  You don't need to know quantum theory, though you'll have more fun if you do.  What we're doing here is very similar, but also strangely different---for reasons explained in Section \ref{sec:4}.

\subsection[Rabbits and quantum mechanics]{Rabbits and quantum mechanics}

Suppose we have a population of rabbits in a cage and we'd like to describe its growth in a stochastic way, using probability theory.  Let $\psi_n$ be the probability of having $n$ rabbits.  We can borrow a trick from quantum theory, and summarize all these probabilities in a  \href{http://en.wikipedia.org/wiki/Formal_power_series}{formal power series} like this:\index{power series!formal}\index{power series!convergence of} 
$$ \Psi = \sum_{n = 0}^\infty \psi_n z^n $$
The variable $z$ doesn't mean anything in particular, and we don't care if the power series converges.  See, in math `formal' means ``it's only symbols on the page, just follow the rules".  It's like if someone says a party is `formal', so need to wear a white tie: you're not supposed to ask what the tie \emph{means}.

However, there's a good reason for this trick.  We can define two operators on formal power series, called the {\bf annihilation operator}:\index{annihilation operator}\index{quantum field theory!annihilation operator}
$$ a \Psi = \frac{d}{d z} \Psi $$
and the {\bf creation operator}:\index{creation operator}\index{quantum field 
theory!creation operator}
$$ a^\dagger \Psi = z \Psi $$
They're just differentiation and multiplication by $z$, respectively.   So, for example, suppose we start out being 100\% sure we have $n$ rabbits for some particular number $n$.  Then $\psi_n = 1$, while all the other probabilities are 0, so: 
$$ \Psi = z^n $$
If we then apply the creation operator, we obtain
$$ a^\dagger \Psi = z^{n+1} $$
Voil\`{a}!  One more rabbit!    

\begin{center}
\includegraphics[width=60mm]{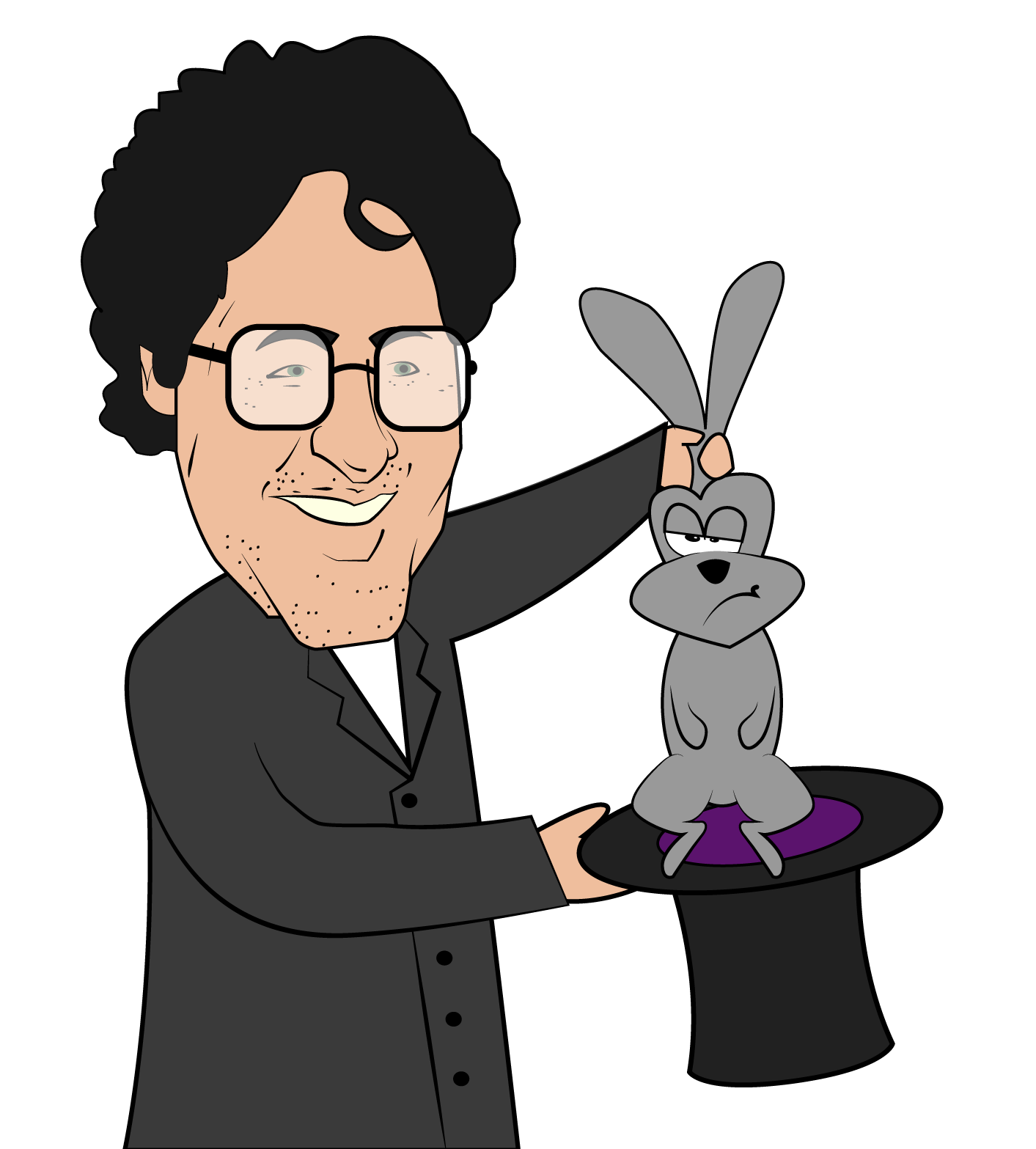}
\end{center}
%<img border = "2" src = "http://math.ucr.edu/home/baez/networks/rabbit-in-hat.jpg" alt = "}{
%

The annihilation operator is more subtle.  If we start out with $n$ rabbits:
$$ \Psi = z^n $$
and then apply the annihilation operator, we obtain
$$ a \Psi = n z^{n-1} $$
What does this mean?  The $z^{n-1}$ means we have one fewer rabbit than before.  But what about the factor of $n$?  It means there were $n$ different ways we could pick a rabbit and make it disappear!   This should seem a bit mysterious, for various reasons... but we'll see how it works soon enough.

The creation and annihilation operators don't commute:
$$(a a^\dagger - a^\dagger a) \Psi = \frac{d}{d z} (z \Psi) - z \frac{d}{d z} \Psi = \Psi 
$$
so for short we say:
$$ a a^\dagger - a^\dagger a = 1 $$
or even shorter:
$$ [a, a^\dagger] = 1 $$
where the \href{http://en.wikipedia.org/wiki/Commutator\#Ring_theory}{{\bf commutator}} of two operators is $ [S,T] = S T - T S $.  

The noncommutativity of operators is often claimed to be a special feature of \emph{quantum} physics, and the \href{http://en.wikipedia.org/wiki/Creation\_and_annihilation\_operators}{creation and annihilation operators} are fundamental to understanding the \href{http://en.wikipedia.org/wiki/Quantum_harmonic_oscillator}{quantum harmonic oscillator}.  There, instead of rabbits, we're studying \href{http://en.wikipedia.org/wiki/Quantum}{quanta of energy}, which are peculiarly abstract entities obeying rather counterintuitive laws.  So, it's cool that the same math applies to purely classical entities, like rabbits!

In particular, the equation $[a, a^\dagger] = 1$ just says that there's one more way to put a rabbit in a cage of rabbits, and then take one out, than to take one out and then put one in.

But how do we actually \emph{use} this setup?  We want to describe how the probabilities $\psi_n$ change with time, so we write
$$ \Psi(t) = \sum_{n = 0}^\infty \psi_n(t) z^n $$
Then, we write down an equation describing the rate of change of $\Psi$:
$$ \frac{d}{d t} \Psi(t) = H \Psi(t) $$
Here $H$ is an operator called the {\bf Hamiltonian}, and the equation is called the {\bf master equation}.\index{master equation!vs Hamiltonian}  The details of the Hamiltonian depend on our problem!  But we can often write it down using creation and annihilation operators.  Let's do some examples, and then we'll tell you the general rule.

\subsection{Catching rabbits}

\begin{center}
 \includegraphics[width=18.5mm]{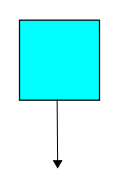}
\end{center}

In Section \ref{sec:4} we told you what happens when we stand in a river and catch fish as they randomly swim past.  Let us remind you of how that works.  But now let's use rabbits.

So, suppose an inexhaustible supply of rabbits are randomly roaming around a huge field, and each time a rabbit enters a certain area, we catch it and add it to our population of caged rabbits.   Suppose that on average we catch one rabbit per unit time.  Suppose the chance of catching a rabbit during any interval of time is independent of what happens before or afterwards.   What is the Hamiltonian describing the probability distribution of caged rabbits, as a function of time?

There's an obvious dumb guess: the creation operator!  However, we saw last time that this doesn't work, and we saw how to fix it.  The right answer is
$$H = a^\dagger - 1$$
To see why, suppose for example that at some time $t$ we have $n$ rabbits, so:
$$\Psi(t) = z^n$$
Then the master equation says that at this moment,
$$ \frac{d}{d t} \Psi(t) = (a^\dagger - 1) \Psi(t) =  z^{n+1} - z^n$$  
Since $\Psi = \sum_{n = 0}^\infty \psi_n(t) z^n$, this implies that the coefficients of our formal power series are changing like this:\index{power series!derivative of}
$$\frac{d}{d t} \psi_{n+1}(t) = 1 $$
$$\frac{d}{d t} \psi_{n}(t) = -1$$
while all the rest have zero derivative at this moment.  And that's exactly right!   See, $\psi_{n+1}(t)$ is the probability of having one more rabbit, and this is going up at rate 1.  Meanwhile, $\psi_n(t)$ is the probability of having $n$ rabbits,  and this is going down at the same rate.

\begin{problem}\label{prob:6} 
Show that with this Hamiltonian and any initial conditions, the master equation predicts that the expected number of rabbits grows linearly.
\end{problem} 

\subsection{Dying rabbits}

\begin{center}
 \includegraphics[width=21mm]{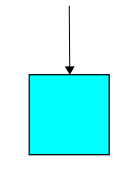}
\end{center}

Don't worry: no rabbits are actually injured in the research that we're doing here at the Centre for Quantum Technologies.   This is just a thought experiment.

Suppose a mean nasty guy had a population of rabbits in a cage and didn't feed them at all.  Suppose that each rabbit has a unit probability of dying per unit time.   And as always, suppose the probability of this happening in any interval of time is independent of what happens before or after that time.  

What is the Hamiltonian?   Again there's a dumb guess: the annihilation operator!  And again this guess is wrong, but it's not far off.   As before, the right answer includes a `correction term':
$$ H = a - N $$
This time the correction term is famous in its own right.  It's called the \href{http://en.wikipedia.org/wiki/Particle\_number_operator}{{\bf number operator}}:
$$N = a^\dagger a $$
\index{number operator} \index{quantum field theory!number operator}
The reason is that if we start with $n$ rabbits, and apply this operator, it amounts to multiplication by $n$:
$$ N z^n = z \frac{d}{d z} z^n = n z^n $$
Let's see why this guess is right.    Again, suppose that at some particular time $t$ we have $n$ rabbits, so
$$\Psi(t) = z^n$$
Then the master equation says that at this time
$$ \frac{d}{d t} \Psi(t) = (a - N) \Psi(t) = n z^{n-1} - n z^n$$ 
So, our probabilities are changing like this:
$$\frac{d}{d t} \psi_{n-1}(t) = n $$
$$\frac{d}{d t} \psi_n(t) = -n$$
while the rest have zero derivative.   And this is good!   We're starting with $n$ rabbits, and each has a unit probability per unit time of dying. So, the chance of having one less should be going up at rate $n$.   And the chance of having the same number we started with should be going \emph{down} at the same rate.

\begin{problem}\label{prob:7} 
Show that with this Hamiltonian and any initial conditions, the master equation predicts that the expected number of rabbits decays exponentially.
\end{problem} 

\subsection{Breeding rabbits}
\begin{center}
 \includegraphics[width=18.8mm]{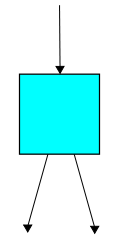}
\end{center}

Suppose we have a strange breed of rabbits that reproduce asexually.  Suppose that each rabbit has a unit probability per unit time of having a baby rabbit, thus effectively duplicating itself.

As you can see from the cryptic picture above, this `duplication' process takes one rabbit as input and has two rabbits as output.  So, if you've been paying attention, you should be ready with a dumb guess for the Hamiltonian: $a^\dagger a^\dagger a$.  This operator annihilates one rabbit and then creates two!

But you should also suspect that this dumb guess will need a `correction term'.   And you're right!  As always, the correction terms makes the probability of things staying the same \emph{go down} at exactly the rate that the probability of things changing \emph{goes up}.

You should guess the correction term... but we'll just tell you:
$$ H = a^\dagger a^\dagger a - N $$
We can check this in the usual way, by seeing what it does when we have $n$ rabbits:
$$ H z^n =  z^2 \frac{d}{d z} z^n - n z^n = n z^{n+1} - n z^n $$
That's good: since there are $n$ rabbits, the rate of rabbit duplication is $n$.  This is the rate at which the probability of having one more rabbit goes up... and also the rate at which the probability of having $n$ rabbits goes down.

\begin{problem}\label{prob:8} 
Show that with this Hamiltonian and any initial conditions, the master equation predicts that the expected number of rabbits grows exponentially.
\end{problem} 

\subsection{Dueling rabbits}

Let's do some stranger examples, just so you can see the general pattern.

\begin{center}
 \includegraphics[width=19.8mm]{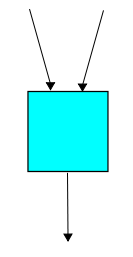}
\end{center}

\noindent
Here each pair of rabbits has a unit probability per unit time of fighting a duel with only one survivor.  You might guess the Hamiltonian $a^\dagger a a $, but in fact:
$$  H = a^\dagger a a - N(N-1) $$
Let's see why this is right!  Let's see what it does when we have $n$ rabbits:
$$ H z^n = z \frac{d^2}{d z^2} z^n - n(n-1)z^n = n(n-1) z^{n-1} - n(n-1)z^n $$
That's good: since there are $n(n-1)$ ordered pairs of rabbits, the rate at which duels take place is $n(n-1)$.   This is the rate at which the probability of having one less rabbit goes up... and also the rate at which the probability of having $n$ rabbits goes down.

(If you prefer \emph{unordered} pairs of rabbits, just divide the Hamiltonian by 2.  We should talk about this more, but not now.)

\subsection{Brawling rabbits}
\begin{center}
 \includegraphics[width=20.4mm]{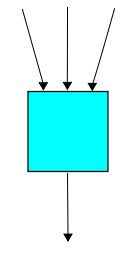}
\end{center}

\noindent 
Now each \emph{triple} of rabbits has a unit probability per unit time of getting into a fight with only one survivor!  We don't know the technical term for a three-way fight, but perhaps it counts as a small `brawl' or `melee'.    

\begin{center}
\includegraphics[width=110mm]{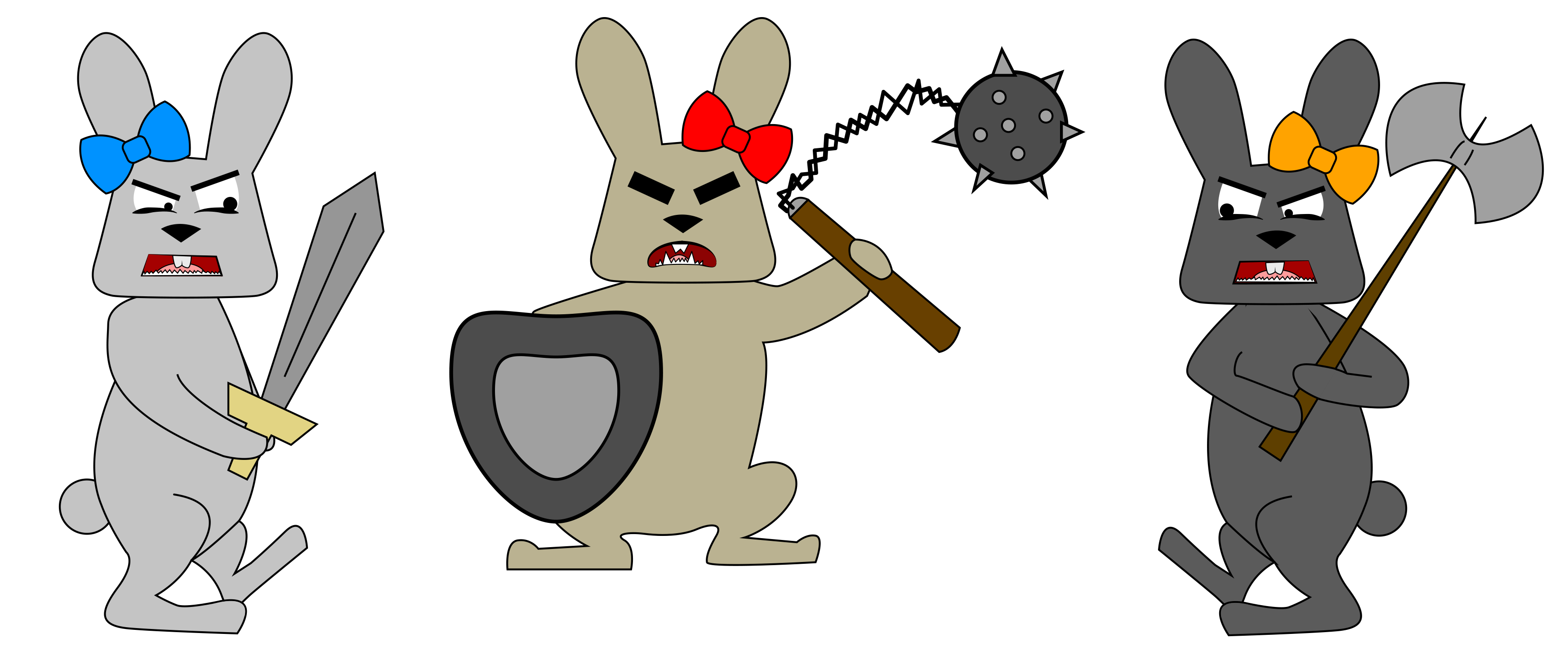}
\end{center}

\noindent
Now the Hamiltonian is:
$$ H = a^\dagger a^3 - N(N-1)(N-2) $$
You can check that:
$$ H z^n = n(n-1)(n-2) z^{n-2} - n(n-1)(n-2) z^n $$
and this is good, because $n(n-1)(n-2)$ is the number of ordered triples of rabbits.  You can see how this number shows up from the math, too:
$$ a^3 z^n = \frac{d^3}{d z^3} z^n = n(n-1)(n-2) z^{n-3} $$

\subsection{The general rule}
\label{sec:3_general_rule}

Suppose we have a process taking $k$ rabbits as input and having $j$ rabbits as output:

\begin{center}
 \includegraphics[width=30mm]{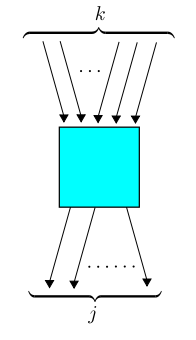}
\end{center}

\noindent
By now you can probably guess the Hamiltonian we'll use for this:
$$  H = {a^{\dagger}}^j a^k - N(N-1) \cdots (N-k+1) $$
This works because
$$ a^k z^n = \frac{d^k}{d z^k} z^n = n(n-1) \cdots (n-k+1) z^{n-k} $$
so that if we apply our Hamiltonian to $n$ rabbits, we get
$$ H z^n =  n(n-1) \cdots (n-k+1) (z^{n+j-k} - z^n) $$
See?  As the probability of having $n+j-k$ rabbits goes up, the probability of having $n$ rabbits goes down, at an equal rate.  This sort of balance is necessary for $H$ to be a sensible Hamiltonian in this sort of stochastic theory `infinitesimal stochastic operator', to be precise).  And the rate is exactly the number of ordered $k$-tuples taken from a collection of $n$ rabbits.  This is called the $k$th \href{http://www.stanford.edu/~dgleich/publications/finite-calculus.pdf\#page=6}{{\bf falling power}} of $n$, and written as follows:
$$ n^{\underline{k}} = n(n-1) \cdots (n-k+1) $$
Since we can apply functions to operators as well as numbers, we can write our Hamiltonian as: 
$$  H = {a^{\dagger}}^j a^k - N^{\underline{k}} $$

\subsection{Kissing rabbits}
\begin{center}
 \includegraphics[width=20mm]{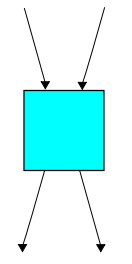}
\end{center}

Let's do one more example just to test our understanding.  This time each pair of rabbits has a unit probability per unit time of bumping into each other, exchanging a friendly kiss and walking off.  This shouldn't affect the rabbit population at all!  But let's follow the rules and see what they say.

According to our rules, the Hamiltonian should be:
$$  H = {a^{\dagger}}^2 a^2 - N(N-1) $$
However, 
$$ {a^{\dagger}}^2 a^2 z^n = z^2 \frac{d^2}{dz^2} z^n = n(n-1) z^n = N(N-1) z^n $$
and since $z^n$ form a `basis' for the formal power series, we see that:\index{power series!basis of}
$$ {a^{\dagger}}^2 a^2 = N(N-1)$$
so in fact:
$$ H = 0 $$
That's good: if the Hamiltonian is zero, the master equation will say 
$$ \frac{d}{d t} \Psi(t) = 0 $$
so the population, or more precisely the probability of having any given number of rabbits, will be constant.

There's another nice little lesson here.  Copying the calculation we just did, it's easy to see that:
$$  {a^{\dagger}}^k a^k = N^{\underline{k}}$$
This is a cute formula for falling powers of the number operator in terms of annihilation and creation operators.  It means that for the general transition we saw before:
\begin{center}
 \includegraphics[width=29mm]{k-in-j-out-labels.png}
\end{center}

\noindent
we can write the Hamiltonian in two equivalent ways:
$$  H = {a^{\dagger}}^j a^k - N^{\underline{k}} =  {a^{\dagger}}^j a^k - {a^{\dagger}}^k a^k  $$
Okay, that's it for now!  We can, and will, generalize all this stuff to stochastic Petri nets where there are things of many different kinds---not just rabbits.  And we'll see that the master equation we get matches the answer to the problem in Section \ref{sec:3}.  That's pretty easy.  

\subsection{References}

For a general introduction to stochastic processes, try this:\index{diffusion} 

\begin{enumerate}
\item[\cite{Van07}]
Nico Van Kampen, \textsl{Stochastic Processes in Physics and Chemistry},
North-Holland, New York, 2007.
\end{enumerate}

There has been a lot of work on using annihilation and creation operators for stochastic systems.  Here are some papers on the foundations:

\begin{enumerate}
\item[\cite{Doi76a}]
M.\ Doi,  Second-quantization representation for classical many-particle systems, \textsl{Jour.\  Phys.\ A} \textbf{9} 1465--1477, 1976.
\item[\cite{Doi76b}]
M.\ Doi, Stochastic theory of diffusion-controlled reactions. \textsl{Jour.\ Phys.\ A} \textbf{9} 1479--1495, 1976.
\item[\cite{Pel85}]
L.\ Peliti, Path integral approach to birth-death processes on a lattice, \textsl{Jour.\ Physique} \textbf{46} 1469-83, 1985. 
\item[\cite{Car96}]\index{renormalization group} 
J.\ Cardy, Renormalization group approach to reaction-diffusion problems.  Available as \href{http://arxiv.org/abs/cond-mat/9607163v2}{arXiv:cond-mat/9607163}. 
\item[\cite{MG98}]
D.\ C.\ Mattis and M.\ L.\ Glasser, The uses of quantum field theory in diffusion-limited reactions, \textsl{Rev.\ Mod.\ Phys.\ } \textbf{70}, 979--1001, 1998.
\end{enumerate}

Here are some papers on applications:\index{diffusion} 

\begin{enumerate}
\item[\cite{Ta05}]\index{renormalization group} 
Uwe T\"auber, Martin Howard and Benjamin P. Vollmayr-Lee,  Applications of field-theoretic renormalization group methods to reaction-diffusion problems, \textsl{Jour.\ Phys.\ A} \textbf{38} (2005), R79.  Also available as \href{http://arxiv.org/abs/cond-mat/0501678}{arXiv:cond-mat/0501678}. 
\item[\cite{BC07}]
M.\ A. Buice, and J.\ D. Cowan, \href{http://www.gatsby.ucl.ac.uk/~beierh/neuro_jc/BuiceCowan07_FieldTheoreticFluctuation.pdf}{Field-theoretic approach to fluctuation effects in neural networks}, \textsl{Phys.\ Rev.\ E} \textbf{75} (2007), 051919.  
\item[\cite{BC09}]
M.\ A.\ Buice and J.\ D.\ Cowan, Statistical mechanics of the neocortex, \textsl{Prog.\ Biophysics \& Mol.\ Bio.\ } \textbf{99}, 53--86, 2009.
\item[\cite{DF09}]
Peter J.\ Dodd and Neil M.\ Ferguson, A many-body field theory approach to stochastic models in population biology, \href{http://dx.plos.org/10.1371/journal.pone.0006855}{\sl PLoS ONE} \textbf{4}, e6855, 2009.
\end{enumerate}

\subsection{Answers}
\label{sec:5_answers}

Here are the answers to the problems:

\vskip 1em \noindent {\bf Problem 6.}  
Show that with the Hamiltonian 
$$H = a^\dagger - 1$$
and any initial conditions, the master equation predicts that the expected number of rabbits grows linearly.

\begin{answer}
Here is one answer, thanks to David Corfield\index{Corfield, David} on Azimuth.  If at some time $t$ we have $n$ rabbits, so that $\Psi(t) = z^n$, we have seen that the probability of having any number of rabbits changes as follows:
$$\frac{d}{d t} \psi_{n+1}(t) = 1, \qquad  \frac{d}{d t} \psi_{n}(t) = -1, \qquad \frac{d}{d t} \psi_m(t) = 0 \mathrm{  ~otherwise} $$
Thus the rate of increase of the expected number of rabbits is $(n +
1) - n = 1$.  But any probability distribution is a linear combination
of these basis vector $z^n$, so the rate of increase of the expected
number of rabbits is always 
$$\sum_n \psi_{n}(t) = 1$$
so the expected number grows linearly.
\end{answer}

Here is a second solution, using some more machinery.  This machinery is overkill here, but it will be useful for solving the next two problems and also many other problems.  

\begin{answer}
In the general formalism described in Section \ref{sec:4}, we used $\int \psi$ to mean the integral of a function over some measure space\index{measure space}, so that probability distributions are the functions obeying 
$$\int \psi = 1$$
and 
$$\psi \ge 0$$  

In the examples we've been looking at, this integral is really a sum over $n = 0,1,2, \dots$, and it might be confusing to use integral notation since we're using derivatives for a completely different purpose.  So let us define a sum notation as follows:
$$\sum \Psi = \sum_{n = 0}^\infty \psi_n $$
This may be annoying, since after all we really have
$$\Psi|_{z = 1} =  \sum_{n = 0}^\infty \psi_n $$
but please humor us.  

To work with this sum notation, two rules are very handy.  

\begin{enumerate}\index{power series!properties of|(} 
\item[{\bf Rule 1:}] For any formal power series $\Phi$,
$$  \sum a^\dagger \Phi = \sum \Phi $$
We mentioned this rule in Section~\ref{sec:4}: it's part of the creation operator being a stochastic operator.  It's easy to check: 
$$ \sum a^\dagger \Phi = z \Phi|_{z = 1} = \Phi|_{z=1} = \sum \Phi $$
%\vskip 1em
\item[{\bf Rule 2:}] For any formal power series $\Phi$,
$$       \sum a \Phi = \sum N \Phi   $$
Again this is easy to check:
$$\sum N \Phi = \sum a^\dagger a \Phi = \sum a \Phi $$
\end{enumerate} 
% \vskip 1em
These rules can be used together with the commutation relation
$$ [a , a^\dagger] = 1 $$
and its consequences 
$$ [a, N] = a, \qquad [a^\dagger, N] = - a^\dagger $$
to do many interesting things.

Let's see how!  Suppose we have some observable $O$ that we can write as an operator on formal power series: for example, the number operator, or any power of that.  The expected value\index{expected value!time derivative of} of this observable in the probability distribution $\Psi$ is
$$\sum O \Psi $$
So, if we're trying to work out the time derivative of the expected value of $O$, we can start by using the master equation:
$$ \frac{d}{dt} \sum O \Psi(t) = \sum O \frac{d}{dt} \Psi(t) = \sum O H \Psi(t) $$
Then we can write $O$ and $H$ using annihilation and creation operators and use our rules.

For example, in the problem at hand, we have 
$$H = a^\dagger - 1 $$
and the observable we're interested in is the number of rabbits
$$O = N$$
so we want to compute 
$$\sum O H \Psi(t) = \sum N(a^\dagger - 1) \Psi(t) $$
There are many ways to use our rules to evaluate this.  For example, Rule 1 implies that
$$\sum N(a^\dagger - 1) \Psi(t) = \sum a(a^\dagger - 1)\Psi(t) $$
but the commutation relations say $a a^\dagger = a^\dagger a + 1 = N+1$, so
$$  \sum a(a^\dagger - 1)\Psi(t) = \sum (N + 1 - a) \Psi(t) $$
and using Rule 1 again we see this equals 
$$ \sum \Psi(t) = 1 $$
Thus we have
$$ \frac{d}{dt} \sum N \Psi(t) = 1 $$
It follows that the expected number of rabbits grows linearly:
$$ \sum N \Psi(t) = t + c $$
\end{answer}

\index{power series!properties of|)} 

\vskip 1em \noindent {\bf Problem 7.} Show that with the Hamiltonian 
$$ H = a - N $$
and any initial conditions, the master equation predicts that the expected number of rabbits decays exponentially.

\begin{answer}
We use the machinery developed in our answer to Problem \ref{prob:6}.  We want to compute the time derivative of the expected number of rabbits:
$$ \frac{d}{dt} \sum N \Psi = \sum N H \Psi = \sum N (a - N) \Psi $$
The commutation relation $[a, N] = a$ implies that $Na = aN - a$.   So:
$$ \sum N(a-N) \Psi = \sum (a N - N - N^2) \Psi$$
but now Rule 2 says:
$$ \sum (a N - N - N^2)\Psi = \sum (N^2 - N - N^2) \Psi = - \sum N \Psi $$
so we see
$$\frac{d}{dt} \sum N \Psi = - \sum N \Psi$$
It follows that the expected number of rabbits decreases exponentially:
$$ \sum N \Psi(t) = c e^{-t} $$
\end{answer} 

\vskip 1em \noindent {\bf Problem 8.} Show that with the Hamiltonian 
$$H = {a^\dagger}^2 a - N $$
and any initial conditions, the master equation predicts that the expected number of rabbits grows exponentially.
\vskip 1em

\begin{answer}
We use the same technique to compute
$$ \frac{d}{dt} \sum N \Psi(t) = \sum N H \Psi(t) = \sum (N {a^\dagger}^2 a - N) \Psi(t) $$
First use the commutation relations to note:
$$ N({a^\dagger}^2 a - N) = N a^\dagger N - N^2 = a^\dagger (N+1) N - N^2$$
Then:
$$ \sum (a^\dagger (N+1) N - N^2) \Psi(t) = \sum ((N+1)N -N^2) \Psi(t) = \sum N \Psi(t) $$
So, we have
$$ \frac{d}{dt} \sum N \Psi(t) = \sum N \Psi(t) $$
It follows that the expected number of rabbits grows exponentially:
$$  \sum N \Psi(t) = c e^t $$\end{answer}

%%%%% SECTION 6 %%%%%%%

\newpage
\section[Population biology]{An example from population biology} 
\label{sec:6}

Stochastic Petri nets can be used to model everything from vending machines to chemical reactions.   Chemists have proven some powerful theorems about when these systems have equilibrium states.  We'd like to fit these results into our framework.   We'll do this soon.  But first, let's look at an example from population biology.

\subsection{Amoeba fission and competition}

Here's a stochastic Petri net:

\begin{center}
 \includegraphics[width=85mm]{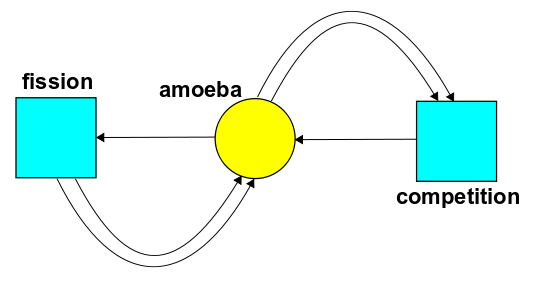}
\end{center}

\noindent
It shows a world with one species, {\bf amoeba}, and two transitions:
\begin{itemize} 
\item {\bf fission}, where one amoeba turns into two.  Let's call the rate constant for this transition $\alpha$.
\item {\bf competition}, where two amoebas battle for resources and only one survives. Let's call the rate constant for this transition $\beta$.
\end{itemize} 
We are going to analyse this example in several ways.  First we'll study the \emph{deterministic}  dynamics it describes: we'll look at its rate equation, which turns out to be the \href{http://en.wikipedia.org/wiki/Logistic\_function\#Logistic\_differential\_equation}{logistic equation}, familiar in population biology.    Then we'll study the \emph{stochastic}  dynamics, meaning its master equation. That's where the ideas from quantum field theory come in. 

\subsection{The rate equation}

If $P(t)$ is the population of amoebas at time $t$, we can follow the rules explained in Section~\ref{sec:2} and crank out this {\bf rate equation}:
$$ \frac{d P}{d t} = \alpha P - \beta P^2 $$
We can rewrite this as
$$ \frac{d P}{d t}= k P(1-\frac{P}{Q}) $$ 

where 
$$ Q = \frac{\alpha}{\beta}, \qquad k = \alpha $$
What's the meaning of $Q$ and $k$?

\begin{itemize} 
\item $Q$ is the {\bf carrying capacity}, that is, the maximum sustainable population the environment can support. 
\item $k$ is the {\bf growth rate} describing the approximately exponential growth of population when $P(t)$ is small.
\end{itemize} 

It's a rare treat to find such an important differential equation that can be solved by analytical methods.  Let's enjoy solving it.  We start by separating variables and integrating both sides:
$$\int \frac{d P}{P (1-P/Q)} = \int k d t$$
We need to use partial fractions on the left side above, resulting in\index{partial fractions}\index{integration!partial fractions}   
$$\int \frac{d P}{P} + \int \frac{d P}{Q-P} = \int k d t$$
and so we pick up a constant $C$, let $A=\pm e^{-C}$, and rearrange things as  
$$ \frac{Q-P}{P}=A e^{-k t} $$
so the population as a function of time becomes 
$$ P(t) = \frac{Q}{1+A e^{-k t}} $$
At $t=0$ we can determine $A$ uniquely.  We write $P_0 := P(0)$ and $A$ becomes 
$$ A = \frac{Q-P_0}{P_0} $$
The model now becomes very intuitive.  Let's set $Q = k=1$ and make a plot for various values of $A$:

\begin{center}
 \includegraphics[width=60mm]{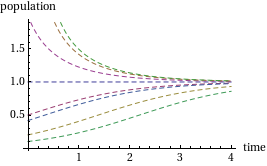}
\end{center}

We arrive at three distinct cases:

\begin{itemize} 
\item  {\bf equilibrium} ($A=0$). The horizontal blue line corresponds to the case where the initial population $P_0$ exactly equals the carrying capacity.  In this case the population is constant. 

\item  {\bf dieoff} ($A < 0$). The three decaying curves above the horizontal blue line correspond to cases where initial population is higher than the carrying capacity.  The population dies off over time and approaches the carrying capacity.

\item  {\bf growth} ($A > 0$). The four increasing curves below the horizontal blue line represent cases where the initial population is lower than the carrying capacity.  Now the population grows over time and approaches the carrying capacity.
\end{itemize} 

\subsection{The master equation}
\label{sec:6_master}

Next, let us follow the rules explained in Section \ref{sec:5} to write down the master equation for our example.  Remember, now we write:
$$ \Psi(t) = \sum_{n = 0}^\infty \psi_n(t) z^n $$
where $\psi_n(t)$ is the probability of having $n$ amoebas at time $t$, and $z$ is a formal variable.\index{formal variable}  The {\bf master equation} says:
$$ \frac{d}{d t} \Psi(t) = H \Psi(t)$$
where $H$ is an operator on formal power series called the {\bf Hamiltonian}.\index{power series!and the Hamiltonian}   To get the Hamiltonian we take each transition in our Petri net and build an operator from creation and annihilation operators, as follows.  Fission works like this:

\begin{center}
 \includegraphics[width=39mm]{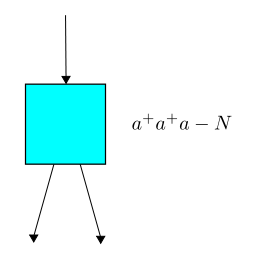}
\end{center}

\noindent
while competition works like this:

\begin{center}
 \includegraphics[width=49mm]{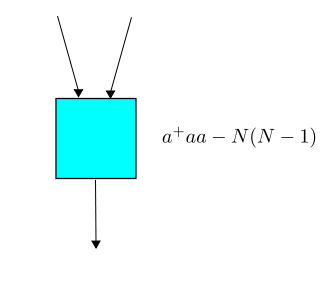}
\end{center}

\noindent
Here $a$ is the annihilation operator, $a^\dagger$ is the creation operator and $N = a^\dagger a$ is the number operator.  In Section \ref{sec:5} we explained precisely how the $N$'s arise.  So the theory is already in place, and we arrive at this Hamiltonian:
$$ H = \alpha (a^\dagger a^\dagger a - N)  +  \beta(a^\dagger a a - N(N-1))$$
Remember, $\alpha$ is the rate constant for fission, while $\beta$ is the rate constant for competition.

The master equation can be solved: it's equivalent to $\frac{d}{d t}(e^{-t H}\Psi(t))=0$ so that $e^{-t H}\Psi(t)$ is constant, and so
$$ \Psi(t) = e^{t H}\Psi(0) $$
and that's it!  We can calculate the time evolution starting from any initial probability distribution of populations.  Maybe everyone is already used to this, but we find it rather remarkable.  

Here's how it works.  We pick a population, say $n$ amoebas at $t=0$.  This would mean $\Psi(0) = z^n$.  We then evolve this state using $e^{t H}$.  We expand this operator as   
$$ \begin{array}{ccl} e^{t H} &=& \sum_{n=0}^\infty \frac{1}{n!} t^n H^n 
\\
\\
&=& 1 + t H  + \frac{1}{2}t^2 H^2 + \cdots \end{array} $$
This operator contains the full information for the evolution of the system.  It contains the \emph{histories}  of all possible amoeba populations---an amoeba mosaic if you will.  From this, we can construct amoeba Feynman diagrams.
\index{Feynman diagrams!and logistic growth|(}  

To do this, we work out each of the $H^n$ terms in the expansion above. The first-order terms correspond to the Hamiltonian acting once.  These are proportional to either $\alpha$ or $\beta$.  The second-order terms correspond to the Hamiltonian acting twice.  These are proportional to either $\alpha^2$, $\alpha\beta$ or $\beta^2$.   And so on.

This is where things start to get interesting!  To illustrate how it works, we will consider two possibilities for the second-order terms: 

\begin{enumerate}
\item We start with a lone amoeba, so $\Psi(0) = z$.  It reproduces and splits into two.  In the battle of the century, the resulting amoebas compete and one dies.  At the end we have:
$$ \frac{\alpha \beta}{2}  (a^\dagger a a)(a^\dagger a^\dagger a) z$$  
We can draw this as a Feynman diagram: 
 
\begin{center}
 \includegraphics[width=67.5mm]{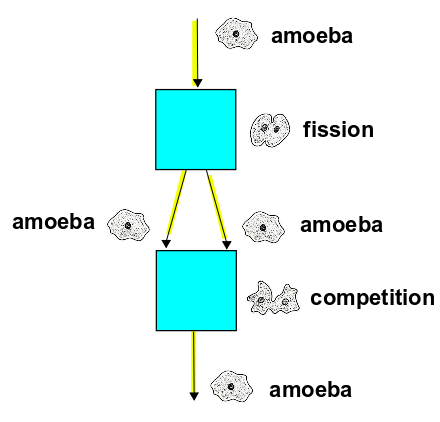}
\end{center}

\noindent
You might find this tale grim, and you may not like the odds either.  It's true, the odds could be better, but people are worse off than amoebas!  The great Japanese swordsman \href{http://en.wikipedia.org/wiki/Miyamoto\_Musashi}{Miyamoto Musashi} quoted the survival odds of fair sword duels as 1/3, seeing that 1/3 of the time both participants die.  A remedy is to cheat, but these amoeba are competing \emph{honestly}. 

\item  We start with two amoebas, so the initial state is $\Psi(0) = z^2$.  One of these amoebas splits into two.  One of these then gets into an argument about Petri nets with the original amoeba.  The amoeba who solved all the problems in this book survives.  At the end we have
$$ \frac{\alpha \beta}{2} (a^\dagger a a)(a^\dagger a^\dagger a) z^2 $$
with corresponding Feynman diagram:

\begin{center}
 \includegraphics[width=75mm]{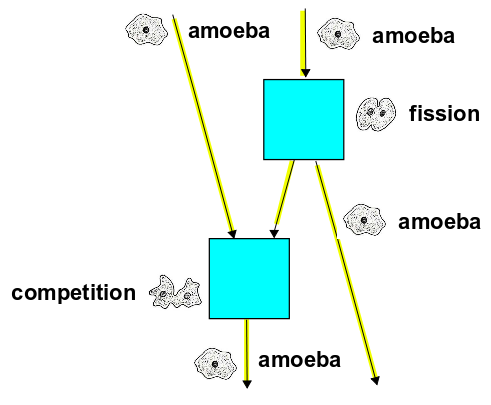}
\end{center}
\end{enumerate}

This should give an idea of how this all works.  The exponential of the Hamiltonian gives all possible histories, and each of these can be translated into a Feynman diagram.  

\index{Feynman diagrams!and logistic growth|)}

\subsection{An equilibrium state}

We've seen the equilibrium solution for the rate equation; now let's look for equilibrium solutions of the master equation.  This paper:

\begin{enumerate} 
\item[\cite{ACK08}] David F.\ Anderson, Georghe Craciun and Thomas G.\ Kurtz,
Product-form stationary distributions for deficiency zero chemical reaction networks. 
Bulletin of Mathematical Biology, Volume 72, Issue 8, pages 1947-1970 (2010). 
Also available as \href{http://www.arxiv.org/abs/0803.3042}{arXiv:0803.3042}.  
\index{Anderson--Craciun--Kurtz theorem}
\end{enumerate} 

\noindent
proves that for a large class of stochastic Petri nets, there exists an equilibrium solution of the master equation where the number of things of each species is distributed according to a \href{http://en.wikipedia.org/wiki/Poisson\_distribution}{Poisson distribution}.  Even more remarkably, these probability distributions are  \href{http://en.wikipedia.org/wiki/Independence\_\%28probability\_theory\%29}{\emph{independent}}, so knowing how many things are in one species tells you nothing about how many are in another!\index{Poisson distribution!and the deficiency zero theorem}  

Here's a nice quote from this paper:  \index{Deficiency Zero Theorem!history|(} 

\begin{quote}
``The surprising aspect of the deficiency zero theorem is that the assumptions of the theorem are completely related to the network of the system whereas the conclusions of the theorem are related to the dynamical properties of the system.''
\end{quote}

\noindent
The `deficiency zero theorem' is a result of Feinberg, Horn and Jackson, which says that for a large class of stochastic Petri nets, the rate equation has a unique equilibrium solution.  Anderson, Craciun and Kurtz showed how to use this fact to get equilibrium solutions of the master equation!  

Though they did not put it this way, their result relies crucially on the stochastic analogue of `\href{http://en.wikipedia.org/wiki/Coherent\_state}{coherent states}'.  
Coherent states\index{quantum field theory!coherent state} are important in quantum theory. Legend (or at least Wikipedia) has it that Erwin Schr\"{o}dinger\index{Schr\"odinger, Erwin} himself discovered them when he was looking for states of a quantum system that look `as classical as possible'.  Suppose you have a quantum harmonic oscillator.\index{quantum mechanics!harmonic oscillator}  Then the \href{http://en.wikipedia.org/wiki/Uncertainty\_principle}{uncertainty principle} says that
$$ \Delta p \Delta q \ge \hbar/2 $$
where $\Delta p$ is the uncertainty in the momentum and $\Delta q$ is the uncertainty in position.  Suppose we want to make $\Delta p \Delta q$ as small as possible, and suppose we also want $\Delta p = \Delta q$.  Then we need our particle to be in a `coherent state'.  That's the definition.  For the quantum harmonic oscillator, there's a way to write quantum states as formal power series
$$ \Psi = \sum_{n = 0}^\infty \psi_n z^n $$\index{power series!quantum harmonic oscillator} 
where $\psi_n$ is the amplitude for having $n$ quanta of energy.  A coherent state then looks like this:
$$ \Psi = e^{c z} = \sum_{n = 0}^\infty \frac{c^n}{n!} z^n $$
where $c$ can be any complex number.  Here we have omitted a constant factor necessary to normalize the state.  
\index{Deficiency Zero Theorem!history|)}

We can also use coherent states in classical stochastic systems like collections of amoebas!  Now the coefficient of $z^n$ tells us the probability of having $n$ amoebas, so $c$ had better be real.  And probabilities should sum to 1, so we really should normalize $\Psi$ as follows:
$$ \Psi = \frac{e^{c z}}{e^c} = e^{-c} \sum_{n = 0}^\infty \frac{c^n}{n!} z^n  $$
Now, the probability distribution
$$ \psi_n = e^{-c}  \frac{c^n}{n!} $$
is called a \href{http://en.wikipedia.org/wiki/Poisson\_distribution}{{\bf Poisson distribution}}.  So, for starters you can think of a `coherent state' as an over-educated way of talking about a Poisson distribution.\index{coherent state}\index{Poisson distribution!and coherent states}\index{quantum field theory!coherent state}   

Let's work out the expected number of amoebas in this Poisson distribution. In Section~\ref{sec:5_answers}, we started using this abbreviation:
$$ \sum \Psi = \sum_{n = 0}^\infty \psi_n $$
We also saw that the expected number of amoebas in the probability distribution $\Psi$ is
$$ \sum N \Psi $$
What does this equal?  Remember that $N = a^\dagger a$.  The annihilation operator $a$ is just $\frac{d}{d z}$, so
$$ a \Psi = c \Psi$$
and we get
$$ \sum N \Psi = \sum a^\dagger a \Psi = c \sum a^\dagger \Psi $$
But we saw in Section~\ref{sec:4} that $a^\dagger$ is stochastic, 
meaning 
$$\sum a^\dagger \Psi = \sum \Psi$$
for any $\Psi$.  Furthermore, our $\Psi$ here has
$$\sum \Psi = 1$$ 
since it's a probability distribution.  So:
$$ \sum N \Psi = c \sum a^\dagger \Psi = c \sum \Psi = c $$
The expected number of amoebas is just $c$.

\begin{problem} \label{prob:9} 
This calculation must be wrong if $c$ is negative: there can't be a negative number of amoebas.  What goes wrong then?
\end{problem} 

\begin{problem}\label{prob:10}\index{standard deviation!Poisson distribution}\index{Poisson distribution!standard deviation}  
Use the same tricks to calculate the standard deviation of the number of amoebas in the Poisson distribution $\Psi$.  
\end{problem}  

Now let's return to our problem and consider the initial amoeba state
$$ \Psi = e^{c z} $$ 
Here we aren't bothering to normalize it, because we're going to look for equilibrium solutions to the master equation, meaning solutions where $\Psi(t)$ doesn't change with time.  So, we want to solve
$$ H \Psi = 0 $$
Since this equation is linear, the normalization of $\Psi$ doesn't really matter.  

Remember, 
$$ H\Psi = \alpha (a^\dagger a^\dagger a - N)\Psi + \beta(a^\dagger a a - N(N-1)) \Psi $$
Let's work this out.  First consider the two $\alpha$ terms:
$$ a^\dagger a^\dagger a \Psi = c z^2 \Psi $$
and 
$$ -N \Psi = -a^\dagger a\Psi = -c z \Psi $$
Likewise for the $\beta$ terms we find 
$$a^\dagger a a\Psi=c^2 z \Psi$$
and
$$-N(N-1)\psi = -a^\dagger a^\dagger a a \Psi = -c^2 z^2\Psi $$
Here we're using a result from Section~\ref{sec:5}: the product ${a^\dagger}^2 a^2$ equals the `falling power' $N(N-1)$.  

The sum of all four terms must vanish.  This happens whenever 
$$\alpha(c z^2 - c z)+\beta(c^2 z-c^2 z^2) = 0$$
which is satisfied for 
$$c= \frac{\alpha}{\beta} $$
Yipee!  We've found an equilibrium solution, since we found a value for $c$ that makes $H \Psi = 0$.    Even better, we've seen that the expected number of amoebas in this equilibrium state is 
$$ \frac{\alpha}{\beta} $$ 
This is just the same as the equilibrium population we saw in the \emph{rate equation} ---that is, the logistic equation!  That's pretty cool, but it's no coincidence: as we shall see, Anderson, Craciun and Kurtz proved it works like this for lots of stochastic Petri nets.

\subsection{Answers}

Here are the answers to the problems, provided by David Corfield: \index{Corfield, David}

\vskip 1em \noindent {\bf Problem 9.} We calculated that the expected number of amoebas in the Poisson distribution with parameter $c$ is equal to $c$.  But this can't be true if $c$ is negative: there can't be a negative number of amoebas.  What goes wrong then?
\vskip 1em

\begin{answer}
If the probability of having $n$ amoebas is given by the Poisson distribution
$$ \psi_n = e^{-c}  \frac{c^n}{n!} $$
then $c$ had better be nonnegative for the probability to be nonnegative when $n = 1$.
\end{answer} 

\vskip 1em \noindent {\bf Problem 10.} Calculate the standard deviation of the number of amoebas in the Poisson distribution.\index{standard deviation!Poisson distribution}\index{Poisson distribution!standard deviation}     
\vskip 1em

\begin{answer}
The standard deviation is the square root of the variance, which is\index{standard deviation!and variance}  
$$ \sum N^2 \Psi - (\sum N \Psi)^2 $$
We have seen that for the Poisson distribution,\index{Poisson distribution!standard deviation}
$$ \sum N \Psi = c $$
and using the same tricks we see
$$\begin{array}{ccl} \sum N^2 \Psi &=& \sum a^{\dagger} a a^{\dagger} a \Psi
 \\  &=& c \sum a^{\dagger} a a^{\dagger} \Psi \\
  &=& c \sum a a^{\dagger} \Psi \\
  &=& c \sum (a^{\dagger} a + 1) \Psi \\
  &=& c(c + 1)\\
\end{array}$$
So, the variance is $c(c+1) - c^2 = c$ and the standard deviation is $\sqrt{c}$.
\end{answer}

%%%%% SECTION 7 %%%%%%%

\newpage
\section[Feynman diagrams]{Feynman diagrams}\index{Feynman diagrams|(}
\label{sec:7}

We've already begun to see how Feynman diagrams show up in the study of stochastic Petri nets: in Section \ref{sec:6_master} we saw an example from population biology, involving amoebas.  Now let's dig a bit deeper, and look at an example involving more than one species.  For this we need to generalize some of our notation a bit.

\subsection{Stochastic Petri nets revisited}

\index{Petri net!stochastic!definition of} \index{Petri net!definition of}
First, recall that a {\bf Petri net} consists of a set $S$ of {\bf species} and a set $T$ of {\bf transitions}, together with a function  \index{Petri net!species} \index{species}
$$ i \colon S \times T \to \mathbb{N} $$
saying how many things of each species appear in the {\bf input} for each transition, and a function \index{input} \index{transition} \index{Petri net!transition} 
$$ o \colon S \times T \to \mathbb{N}$$
saying how many things of each species appear in the {\bf output}.  \index{output}
A {\bf stochastic Petri net} is a Petri net together with a function 
$$ r \colon T \to (0,\infty) $$
giving a {\bf rate constant} for each transition.  \index{Petri net!rate constant} \index{rate constant}

In this section, we'll consider a Petri net with two species and three transitions:

\begin{center}
 \includegraphics[width=113.5mm]{wolf-rabbit.png}
\end{center}

\noindent
It should be clear that the transition `predation' has one wolf and one rabbit as input, and two wolves as output.  

As we've seen, starting from any stochastic Petri net we can get two things.  First:
\begin{itemize} 
\item The \href{http://en.wikipedia.org/wiki/Master\_equation}{{\bf master equation}}.\index{master equation!of stochastic Petri net}  This says how the \emph{probability that we have a given number of things of each species}  changes with time.  
\item The \href{http://en.wikipedia.org/wiki/Rate\_equation}{{\bf rate equation}}.  This says how the \emph{expected number of things of each species}  changes with time.  
\end{itemize} 
The master equation is stochastic: it describes how probabilities change with time.  The rate equation is deterministic.  

The master equation is more fundamental.  It's like the equations of quantum electrodynamics, which describe the amplitudes for creating and annihilating individual photons.  The rate equation is less fundamental.  It's like the classical Maxwell equations, which describe changes in the electromagnetic field in a deterministic way.  The classical Maxwell equations are an approximation to quantum electrodynamics.  This approximation gets good \emph{in the limit}  where there are lots of photons all piling on top of each other to form nice waves. 

Similarly, the rate equation can be derived from the master equation \emph{in the limit}  where the number of things of each species become large, and the fluctuations in these numbers become negligible.  

But we won't do this derivation!  Nor will we probe more deeply into the analogy with quantum field theory, even though that's the ultimate goal.  For now we'll be content to write down the master equation in a terse notation, and say a bit about Feynman diagrams.   But to set up our notation, let's start by recalling the rate equation.

\subsection{The rate equation}

Suppose we have a stochastic Petri net with $k$ different species.  Let $x_i$ be the number of things of the $i$th species.   Then the rate equation looks like this:
$$ \frac{d x_i}{d t} = ???  $$
It's really a bunch of equations, one for each $1 \le i \le k$.  But what is the right-hand side?

The right-hand side is a sum of terms, one for each transition in our Petri net.  So, let's start by assuming our Petri net has just one transition.   

Suppose the $i$th species appears as input to this transition $m_i$ times, and as output $n_i$ times.  Then the rate equation is
$$ \frac{d x_i}{d t} = r (n_i - m_i) x_1^{m_1} \cdots x_k^{m_k} $$
where $r$ is the rate constant for this transition.  

That's really all there is to it!  But we can make it look nicer.  Let's make up a vector 
$$ x = (x_1, \dots , x_k) \in [0,\infty)^k $$ 
that says how many things  there are of each species.  Similarly let's make up an {\bf input vector}
$$ m = (m_1, \dots, m_k) \in \mathbb{N}^k $$
and an {\bf output vector}
$$ n = (n_1, \dots, n_k) \in \mathbb{N}^k $$
for our transition.  To be cute, let's also define
$$ x^m = x_1^{m_1} \cdots x_k^{m_k} $$
Then we can write the rate equation for a single transition like this:
$$ \frac{d x}{d t} = r (n-m) x^m $$
Next let's do a general stochastic Petri net, with lots of transitions.  Let's write $T$ for the set of transitions and $r(\tau)$ for the rate constant of the transition $\tau \in T$.  Let $m(\tau)$ and $n(\tau)$ be the input and output vectors of the transition $\tau$.  Then the rate equation is:
$$ \frac{d x}{d t} = \sum_{\tau \in T} r(\tau) \, (n(\tau) - m(\tau)) \, x^{m(\tau)} $$
For example, consider our rabbits and wolves:

\begin{center}
 \includegraphics[width=113.5mm]{wolf-rabbit.png}
\end{center}

Suppose:

\begin{itemize} 
\item the rate constant for `birth' is $\beta$,
\item the rate constant for `predation' is $\gamma$,
\item the rate constant for `death' is $\delta$.
\end{itemize} 

Let $x_1(t)$ be the number of rabbits and $x_2(t)$ the number of wolves at time $t$.  Then the rate equation looks like this:
$$ \frac{d x_1}{d t} = \beta x_1 - \gamma x_1 x_2 $$
$$ \frac{d x_2}{d t} = \gamma x_1 x_2 - \delta x_2 $$
If you stare at this, and think about it, it should make perfect sense.  If it doesn't, go back and read Section \ref{sec:2}.

\subsection{The master equation}
\label{sec:7_master}

Now let's do something new.  In Section \ref{sec:5} we explained how to write down the master equation for a stochastic Petri net with just \emph{one}  species.  Now let's generalize that.  Luckily, the ideas are exactly the same.

So, suppose we have a stochastic Petri net with $k$ different species.
Let $\psi_{n_1, \dots, n_k}$ be the probability that we have $n_1$ things of the first species, $n_2$ of the second species, and so on.   The master equation will say how all these probabilities change with time.

To keep the notation clean, let's introduce a vector
$$ n = (n_1, \dots, n_k) \in \mathbb{N}^k $$
and let
$$ \psi_n = \psi_{n_1, \dots, n_k} $$  

Then, let's take all these probabilities and cook up a formal power series that has them as coefficients: as we've seen, this is a powerful trick.    To do this, we'll bring in some variables $z_1, \dots, z_k$ and write
$$ z^n = z_1^{n_1} \cdots z_k^{n_k} $$
as a convenient abbreviation.  Then any formal power series in these variables looks like this:
$$ \Psi = \sum_{n \in \mathbb{N}^k} \psi_n z^n $$\index{power series!probability distribution}  

We call $\Psi$ a {\bf state} if the probabilities sum to 1 as they should:
\index{state!stochastic}
$$ \sum_n \psi_n = 1 $$
The simplest example of a state is a monomial:
$$  z^n = z_1^{n_1} \cdots z_k^{n_k} $$
This is a state where we are 100\% sure that there are $n_1$ things of the first species, $n_2$ of the second species, and so on.  We call such a state a {\bf pure state}, since physicists use this term to describe a state where we know for sure exactly what's going on.  Sometimes a general state, one that might not be pure, is called {\bf mixed}.

The master equation says how a state evolves in time.  It looks like this:
$$ \frac{d}{d t} \Psi(t) = H \Psi(t) $$
So, we just need to tell you what $H$ is!  

It's called the {\bf Hamiltonian}.  \index{Hamiltonian}
It's a linear operator built from special operators that annihilate and create things of various species.  Namely, for each state $1 \le i \le k$ we have a {\bf annihilation operator}:
\index{quantum field theory!annihilation operator} \index{annihilation operator} 
$$ a_i \Psi = \frac{d}{d z_i} \Psi $$
and a {\bf creation operator}:\index{quantum field theory!creation operator}\index{creation operator}
$$ a_i^\dagger \Psi = z_i \Psi $$
How do we build $H$ from these?   Suppose we've got a stochastic Petri net whose set of transitions is $T$.  As before, write $r(\tau)$ for the rate constant of the transition $\tau \in T$, and let $n(\tau)$ and $m(\tau)$ be the input and output vectors of this transition.  Then:
$$ H = \sum_{\tau \in T} r(\tau) \, ({a^\dagger}^{n(\tau)} - {a^\dagger}^{m(\tau)}) \, a^{m(\tau)}  $$
where as usual we've introduce some shorthand notations to keep from going insane.  For example:
$$ a^{m(\tau)} = a_1^{m_1(\tau)} \cdots  a_k^{m_k(\tau)} $$
and
$$ {a^\dagger}^{m(\tau)} = {a_1^\dagger}^{m_1(\tau)} \cdots  {a_k^\dagger}^{m_k(\tau)} $$
Now, it's not surprising that each transition $\tau$ contributes a term to $H$.  It's also not surprising that this term is proportional to the rate constant $r(\tau)$.   The only tricky thing is the expression 
$$ ({a^\dagger}^{n(\tau)} - {a^\dagger}^{m(\tau)})a^{m(\tau)} $$
How can we understand it?  The basic idea is this.   We've got two terms here.  The first term:
$$  {a^\dagger}^{n(\tau)} a^{m(\tau)} $$
describes how $m_i(\tau)$ things of the $i$th species get annihilated, and $n_i(\tau)$ things of the $i$th species get created.   Of course this happens thanks to our transition $\tau$.  The second term:
$$  - {a^\dagger}^{m(\tau)} a^{m(\tau)} $$
is a bit harder to understand, but it says how the probability that \emph{nothing}  happens---that we remain in the same pure state---\emph{decreases}  as time passes.   Again this happens due to our transition $\tau$.

In fact, the second term must take precisely the form it does to ensure `conservation of total probability'.  In other words: if the probabilities $\psi_n$ sum to 1 at time zero, we want these probabilities to still sum to 1 at any later time.  And for this, we need that second term to be what it is!  In Section \ref{sec:5} we saw this in the special case where there's only one species.  The general case works the same way.

Let's look at an example.  Consider our rabbits and wolves yet again:

\begin{center}
 \includegraphics[width=113.5mm]{wolf-rabbit.png}
\end{center}

\noindent
and again suppose the rate constants for birth, predation and death are $\beta$, $\gamma$ and $\delta$, respectively.  We have
$$ \Psi = \sum_n \psi_n z^n $$ 

where 
$$ z^n = z_1^{n_1} z_2^{n_2} $$
and $\psi_n = \psi_{n_1, n_2}$ is the probability of having $n_1$ rabbits and $n_2$ wolves.  These probabilities evolve according to the equation 
$$ \frac{d}{d t} \Psi(t) = H \Psi(t) $$
where the Hamiltonian is 
$$  H = \beta B + \gamma C + \delta D $$
and $B$, $C$ and $D$ are operators describing birth, predation and death, respectively.  ($B$ stands for birth, $D$ stands for death... and you can call predation `consumption' if you want something that starts with $C$.  Besides, `consumer' is a nice euphemism for `predator'.)  What are these operators?  Just follow the rules we described:
$$ B = {a_1^\dagger}^2 a_1 - a_1^\dagger a_1$$
$$ C = {a_2^\dagger}^2 a_1 a_2 - a_1^\dagger a_2^\dagger a_1 a_2 $$
$$ D = a_2 -  a_2^\dagger a_2 $$
In each case, the first term is easy to understand:

\begin{itemize} 
\item Birth annihilates one rabbit and creates two rabbits.
\item Predation annihilates one rabbit and one wolf and creates two wolves.
\item Death annihilates one wolf.
\end{itemize} 

The second term is trickier, but we told you how it works.  

\subsection{Feynman diagrams}  

How do we solve the master equation?  If we don't worry about mathematical rigor too much, it's easy.  The solution of 
$$ \frac{d}{d t} \Psi(t) = H \Psi(t) $$
should be
$$ \Psi(t) = e^{t H} \Psi(0) $$
and we can hope that
$$ e^{t H} = 1 + t H + \frac{(t H)^2}{2!} + \cdots $$
so that
$$ \Psi(t) = \Psi(0) + t H \Psi(0) + \frac{t^2}{2!} H^2 \Psi(0) + \cdots $$
Of course there's always the question of whether this power series converges.\index{power series!convergence of}   In many contexts it doesn't, but that's not necessarily a disaster: the series can still be \href{http://en.wikipedia.org/wiki/Asymptotic\_expansion}{asymptotic} to the right answer, or even better, \href{http://en.wikipedia.org/wiki/Borel\_summation}{Borel summable} to the right answer.  

But let's not worry about these subtleties yet!  Let's just imagine our rabbits and wolves, with Hamiltonian 
$$  H = \beta B + \gamma C + \delta D $$
Now, imagine working out
$$ \Psi(t) = \Psi(0) + t H \Psi(0) + \frac{t^2}{2!} H^2 \Psi(0) + \frac{t^3}{3!} H^3 \Psi(0) + \cdots $$
We'll get lots of terms involving products of $B$, $C$ and $D$ hitting our original state $\Psi(0)$.   And we can draw these as diagrams!  For example, suppose we start with one rabbit and one wolf.  Then
$$ \Psi(0) = z_1 z_2 $$
And suppose we want to compute 
$$ H^3 \Psi(0) = (\beta B + \gamma C + \delta D)^3 \Psi(0) $$
as part of the task of computing $\Psi(t)$.  Then we'll get lots of terms: 27, in fact, though many will turn out to be zero.  Let's take one of these terms, for example the one proportional to:
$$ D C B \Psi(0) $$
We can draw this as a sum of Feynman diagrams, including this:

\begin{center}
 \includegraphics[width=82.5mm]{rabbit_wolf_feynman_diagram.png}
\end{center}

\noindent
In this diagram, we start with one rabbit and one wolf at top.  As we read the diagram from top to bottom, first a rabbit is born ($B$), then predation occur ($C$), and finally a wolf dies ($D$).  The end result is again a rabbit and a wolf.  

This is just one of four Feynman diagrams we should draw in our sum for $D C B \Psi(0)$, since either of the two rabbits could have been eaten, and either wolf could have died.  So, the end result of computing
$$ H^3 \Psi(0) $$
will involve a lot of Feynman diagrams... and of course computing 
$$ \Psi(t) = \Psi(0) + t H \Psi(0) + \frac{t^2}{2!} H^2 \Psi(0) + \frac{t^3}{3!} H^3 \Psi(0) + \cdots $$
will involve even more, even if we get tired and give up after the first few terms.  So, this Feynman diagram business may seem quite tedious... and it may not be obvious how it helps.  But it does, sometimes!  

Now is not the time to describe `practical' benefits of Feynman diagrams.  Instead, we'll just point out one conceptual benefit.  We started with what seemed like a purely computational chore, namely computing 
$$ \Psi(t) = \Psi(0) + t H \Psi(0) + \frac{t^2}{2!} H^2 \Psi(0) + \cdots $$
But then we saw how this series has a clear meaning!  It can be written as a sum over diagrams, each of which represents a \emph{possible history}  of rabbits and wolves.  So, it's what physicists call a \href{http://en.wikipedia.org/wiki/Path\_integral\_formulation}{`sum over histories'}.  

Feynman invented the idea of a sum over histories in the context of quantum field theory.  At the time this idea seemed quite mind-blowing, for various reasons.   First, it involved elementary particles instead of everyday things like rabbits and wolves.  Second, it involved complex `amplitudes' instead of real probabilities.  Third, it actually involved integrals instead of sums.  And fourth, a lot of these integrals diverged, giving infinite answers that needed to be `cured' somehow.

Now we're seeing a sum over histories in a more down-to-earth context without all these complications.  A lot of the underlying math is analogous... but now there's nothing mind-blowing about it: it's quite easy to understand.  So, we can use this analogy to demystify quantum field theory a bit.  On the other hand, thanks to this analogy, all sorts of clever ideas invented by quantum field theorists will turn out to have applications to biology and chemistry!  So it's a double win.

\index{Feynman diagrams|)}   

%%%%%% SECTION 8 %%%%%%%

\newpage
\section{The Anderson--Craciun--Kurtz theorem}
\label{sec:8}

\index{Anderson--Craciun--Kurtz theorem|(}
In Section \ref{sec:7} we reviewed the rate equation and the master equation.   Both of them describe processes where things of various kinds can react and turn into other things.  But:

\begin{itemize} 
\item In the rate equation, we assume the number of things varies continuously and is known precisely.  
\item In the master equation, we assume the number of things varies discretely and is known only probabilistically.
\end{itemize} 

This should remind you of the difference between classical mechanics and quantum mechanics. But the master equation is not \emph{quantum}, it's \emph{stochastic}: it involves probabilities, but there's no uncertainty principle going on.  
Still, a lot of the math is similar.  

Now, given an equilibrium solution to the rate equation---one that doesn't change with time---we'll try to find a solution to the master equation with the same property.  We won't \emph{always}  succeed---but  we often can!  The theorem saying how was proved here:

\begin{enumerate} 
\item[\cite{ACK08}] David F.\ Anderson, Georghe Craciun and Thomas G.\ Kurtz,
Product-form stationary distributions for deficiency zero chemical reaction networks.  Available as \href{http://arxiv.org/abs/0803.3042}{arXiv:0803.3042}.  
\end{enumerate} 

To emphasize the analogy to quantum mechanics, Brendan Fong\index{Fong, Brendan} has translated their proof into the language of annihilation and creation operators.  In particular, our equilibrium solution of the master equation is just like what people call a `\href{http://en.wikipedia.org/wiki/Coherent\_states}{coherent state}' in quantum mechanics.  

In what follows, we'll present Fong's `quantum proof' of the Anderson--Craciun--Kurtz theorem.  So if you know about quantum mechanics and coherent states, you should be happy.  But if you don't, fear not!---we're not assuming you do.

\subsection{The rate equation}

To construct our equilibrium solution of the master equation, we need a special type of solution to our rate equation.  We call this type a `complex balanced solution'.  This means that not only is the net rate of production of each species zero, but the net rate of production of each possible \emph{bunch}  of species is zero.  

Before we make this more precise, let's remind ourselves of the basic setup. 

We'll consider a stochastic Petri net with a finite set $ S$ of species and a finite set $ T$ of transitions.  For convenience let's take $ S = \{1,\dots, k\}$, so our species are numbered from 1 to $ k$.  Then each transition $ \tau$ has an input vector $ m(\tau) \in \mathbb{N}^k$ and output vector $ n(\tau) \in \mathbb{N}^k$.  These say how many things of each species go in, and how many go out.  Each transition also has rate constant $ r(\tau) \in (0,\infty)$, which says how rapidly it happens.

The rate equation concerns a vector $ x(t) \in [0,\infty)^k$ whose $ i$th component is the number of things of the $ i$th species at time $ t$.   Note: we're assuming this number of things varies continuously and is known precisely!  This should remind you of classical mechanics.  So, we'll call $ x(t)$, or indeed any vector in $ [0,\infty)^k$, a {\bf classical state}.  

The {\bf rate equation} says how the classical state $ x(t)$ changes with time:
$$ \displaystyle{  \frac{d x}{d t} = \sum_{\tau \in T} r(\tau)\, (n(\tau)-m(\tau)) \, x^{m(\tau)}} $$
You may wonder what $ x^{m(\tau)}$ means: after all, we're taking a vector to a vector power!  It's just an abbreviation, which we've seen plenty of times before.  If $ x \in \mathbb{R}^k$ is a list of numbers and $ m \in \mathbb{N}^k$ is a list of natural numbers, we define
$$  x^m = x_1^{m_1} \cdots x_k^{m_k} $$ 

We'll also use this notation when $ x$ is a list of \emph{operators}.

\subsection{Complex balance}\index{complex balanced equilibrium|(}
\index{equilibrium! complex balanced|(}

The vectors $ m(\tau)$ and $ n(\tau)$ are examples of what chemists call {\bf complexes}.  A complex is a bunch of things of each species.  For example, if the set $ S$ consists of three species, the complex $ (1,0,5)$ is a bunch consisting of one thing of the first species, none of the second species, and five of the third species.  
\index{complex}

For our Petri net, the set of complexes is the set $ \mathbb{N}^k$, and the complexes of particular interest are the {\bf input complex} $ m(\tau)$ and the {\bf output complex} $ n(\tau)$ of each transition $ \tau$. \index{input complex} \index{output complex}

We say a classical state $ c \in [0,\infty)^k$ is {\bf complex balanced} if for all complexes $ \kappa \in \mathbb{N}^k$ we have
$$ \displaystyle{ \sum_{\{\tau : m(\tau) = \kappa\}} r(\tau) c^{m(\tau)} =\sum_{\{\tau : n(\tau) = \kappa\}} r(\tau) c^{m(\tau)} } $$
The left hand side of this equation, which sums over the transitions with  \emph{input complex}  $ \kappa$, gives the rate of consumption of the complex $ \kappa .$  The right hand side, which sums over the transitions with \emph{output complex}  $ \kappa$, gives the rate of production of $ \kappa .$  So, this equation requires that the net rate of production of the complex $ \kappa$ is zero in the classical state $ c .$
\index{complex balanced equilibrium} \label{input complex} \label{output complex}

\begin{problem}\label{prob:11} 
Show that if a classical state $ c$ is complex balanced, and we set $ x(t) = c$ for all $ t,$ then $ x(t)$ is a solution of the rate equation.
\end{problem} 

Since $ x(t)$ doesn't change with time here, we call it an {\bf equilibrium solution} of the rate equation.  Since $ x(t) = c$ is complex balanced, we call it {\bf complex balanced} equilibrium solution.\index{complex balanced equilibrium|)}\index{equilibrium!complex balanced|)}

\subsection{The master equation}

We've seen that any complex balanced classical state gives an equilibrium solution of the \emph{rate}  equation.  The Anderson--Craciun--Kurtz theorem says that it also gives an equilibrium solution of the \emph{master}  equation.  
 
The master equation concerns a formal power series\index{power series!master equation}  
$$ \displaystyle{ \Psi(t) = \sum_{n \in \mathbb{N}^k} \psi_n(t) z^n} $$
where
$$  z^n = z_1^{n_1} \cdots z_k^{n_k} $$
and
$$  \psi_n(t) = \psi_{n_1, \dots,n_k}(t) $$
is the probability that at time $ t$ we have $ n_1$ things of the first species,  $ n_2$ of the second species, and so on.  

Note: now we're assuming this number of things varies discretely and is known only probabilistically!  So, we'll call $ \Psi(t)$, or indeed any formal power series where the coefficients are probabilities summing to 1, a {\bf stochastic state}.  Earlier we just called it a `state', but that would get confusing now: we've got classical states and stochastic states, and we're trying to relate them.

The {\bf master equation} says how the stochastic state $ \Psi(t)$ changes with time:
$$ \displaystyle{ \frac{d}{d t} \Psi(t) = H \Psi(t)} $$
where the {\bf Hamiltonian} $ H$ is:
$$ \displaystyle{ H = \sum_{\tau \in T} r(\tau) \left({a^\dagger}^{n(\tau)} - {a^\dagger}^{m(\tau)} \right) a^{m(\tau)} } $$
The notation here is designed to neatly summarize some big products of annihilation and creation operators.  For any vector $ n \in \mathbb{N}^k$, we have
$$ a^n = a_1^{n_1} \cdots  a_k^{n_k} $$
and
$$ \displaystyle{  {a^\dagger}^n = {a_1^\dagger}^{n_1} \cdots  {a_k^\dagger}^{n_k}} $$

\subsection{Coherent states}

Now suppose $ c \in [0,\infty)^k$ is a complex balanced equilibrium solution of the rate equation.  We want to get an equilibrium solution of the master equation.  How do we do it?

For any $ c \in [0,\infty)^k$ there is a stochastic state called a {\bf coherent state}, defined by 
$$   \displaystyle{ \Psi_c = \frac{e^{c z}}{e^c}} $$
Here we are using some very terse abbreviations.  Namely, we are defining
$$  e^{c} = e^{c_1} \cdots e^{c_k} $$
and 
$$  e^{c z} = e^{c_1 z_1} \cdots e^{c_k z_k} $$ 
Equivalently,
$$ \displaystyle{ e^{c z} = \sum_{n \in \mathbb{N}^k} \frac{c^n}{n!}z^n} $$
where $c^n$ and $z^n$ are defined as products in our usual way, and
$$  n! = n_1! \, \cdots \, n_k! $$
Either way, if you unravel the abbrevations, here's what you get:
$$ \displaystyle{  \Psi_c = e^{-(c_1 + \cdots + c_k)} \, \sum_{n \in \mathbb{N}^k} \frac{c_1^{n_1} \cdots c_k^{n_k}} {n_1! \, \cdots \, n_k!} \, z_1^{n_1} \cdots z_k^{n_k}} $$
Maybe now you see why we like the abbreviations.

The name `\href{http://en.wikipedia.org/wiki/Coherent\_states}{coherent state}' comes from quantum mechanics.  In quantum mechanics, we think of a coherent state $ \Psi_c$ as the `quantum state' that best approximates the classical state $ c$.  But we're not doing quantum mechanics now, we're doing probability theory.  $ \Psi_c$ isn't a `quantum state', it's a stochastic state.  

In probability theory, people like Poisson distributions.  In the state $ \Psi_c$, the probability of having $ n_i$ things of the $ i$th species is equal to 
$$ \displaystyle{  e^{-c_i} \, \frac{c_i^{n_i}}{n_i!}} $$
This is precisely the definition of a \href{http://en.wikipedia.org/wiki/Poisson\_distribution}{{\bf Poisson distribution}} with mean equal to $ c_i$.   We can multiply a bunch of factors like this, one for each species, to get
$$ \displaystyle{ e^{-c} \, \frac{c^n}{n!}} $$
This is the probability of having $ n_1$ things of the first species, $ n_2$ things of the second, and so on, in the state $ \Psi_c$.  So, the state $ \Psi_c$ is a product of independent Poisson distributions.  In particular, knowing how many things there are of one species  says \emph{nothing all about}  how many things there are of any other species!

It is remarkable that such a simple state can give an equilibrium solution of the master equation, even for very complicated stochastic Petri nets.  But it's true---at least if $ c$ is complex balanced.

\subsection[The proof]{Proof of the Anderson--Craciun--Kurtz theorem}

In a post on the \href{http://johncarlosbaez.wordpress.com/2011/09/13/network-theory-part-9/}{Azimuth blog},\index{Azimuth Project!blog} Brendan Fong\index{Fong, Brendan} translated Anderson, Craciun and Kurtz's proof of their theorem into the language of annihilation and creation operators.  Now we're ready to present Fong's proof:

\begin{theorem}[{\bf Anderson--Craciun--Kurtz}]\index{Anderson--Craciun--Kurtz theorem!statement of} Suppose $ c \in [0,\infty)^k$  is a complex balanced equilibrium solution of the rate equation.    Then $ H \Psi_c = 0$. 
\end{theorem}

Of course, it follows that $ \Psi_c$ is an equilibrium solution\index{master equation!equilibrium} of the master equation.  In other words, if we take $ \Psi(t) = \Psi_c$ for all times $ t$, the master equation holds:
$$ \displaystyle{ \frac{d}{d t} \Psi(t) = H \Psi(t)} $$
since both sides are zero. 

\begin{proof} We just need to show that $ H \Psi_c = 0$.  Since $ \Psi_c$ is a constant times $ e^{c z}$, it suffices to show $ H e^{c z} = 0$.  Remember that
$$ \displaystyle{ H e^{c z} =  \sum_{\tau \in T} r(\tau) \left( {a^\dagger}^{n(\tau)} -{a^\dagger}^{m(\tau)} \right) \, a^{m(\tau)} \, e^{c z}} $$
Since the annihilation operator $a_i$ is given by differentiation with respect to $ z_i$, while the creation operator $ a^\dagger_i$ is just multiplying by $z_i$, we have:
$$ \displaystyle{ H e^{c z} = \sum_{\tau \in T} r(\tau) \, c^{m(\tau)} \left( z^{n(\tau)} - z^{m(\tau)} \right) e^{c z}} $$
Expanding out $ e^{c z}$ we get:
$$ \displaystyle{ H e^{c z} = \sum_{i \in \mathbb{N}^k} \sum_{\tau \in T} r(\tau)c^{m(\tau)}\left(z^{n(\tau)}\frac{c^i}{i!}z^i - z^{m(\tau)}\frac{c^i}{i!}z^i\right)} $$
Shifting indices and defining negative powers to be zero:
$$ \displaystyle{ H e^{c z}  = \sum_{i \in \mathbb{N}^k} \sum_{\tau \in T} r(\tau)c^{m(\tau)}\left(\frac{c^{i-n(\tau)}}{(i-n(\tau))!}z^i - \frac{c^{i-m(\tau)}}{(i-m(\tau))!}z^i\right)} $$
So, to show $ H e^{c z} = 0$, we need to show this:
$$ \displaystyle{ \sum_{i \in \mathbb{N}^k} \sum_{\tau \in T} r(\tau)c^{m(\tau)}\frac{c^{i-n(\tau)}}{(i-n(\tau))!} \, z^i =\sum_{i \in \mathbb{N}^k} \sum_{\tau \in T} r(\tau)c^{m(\tau)}\frac{c^{i-m(\tau)}}{(i-m(\tau))!} \,z^i} $$
We do this by splitting the sum over $T$ according to output and then input complexes, making use of the complex balanced condition:
$$ \begin{array}{ccl} \displaystyle{ \sum_{i \in \mathbb{N}^k} \sum_{\kappa \in \mathbb{N}^k} \sum_{\tau : n(\tau)=\kappa}  r(\tau)c^{m(\tau)}\frac{c^{i-n(\tau)}}{(i-n(\tau))!} \, z^i} &=& \displaystyle{\sum_{i \in \mathbb{N}^k} \sum_{\kappa \in \mathbb{N}^k} \frac{c^{i-\kappa}}{(i-\kappa)!}\, z^i \sum_{\tau : n(\tau) = \kappa}  r(\tau)c^{m(\tau)} }  \\ \\ &=& \displaystyle{ \sum_{i \in \mathbb{N}^k} \sum_{\kappa \in \mathbb{N}^k} \frac{c^{i-\kappa}}{(i-\kappa)!}\, z^i \sum_{\tau : m(\tau) = \kappa}  r(\tau)c^{m(\tau)} } \\ \\ &=& \displaystyle{ \sum_{i \in \mathbb{N}^k} \sum_{\kappa \in \mathbb{N}^k} \sum_{\tau : m(\tau) = \kappa}  r(\tau)c^{m(\tau)}\frac{c^{i-m(\tau)}}{(i-m(\tau))!}\, z^i  } \end{array} $$
This completes the proof!  \end{proof}

We hope you agree how amazing this result is.  If you know quantum mechanics and coherent states you'll understand what we mean.  A coherent state is the ``best quantum approximation''; to a classical state, but we don't expect this quantum state to be \emph{exactly}  time-independent when the corresponding classical state is, \emph{except}  in very special cases, like when the Hamiltonian is quadratic in the creation and annihilation operators.  Here we are getting a result like that much more generally... but only given the ``complex balanced'' condition. 
\index{Anderson--Craciun--Kurtz theorem|)}

\subsection{An example}

We've already seen one example of the Anderson--Craciun--Kurtz theorem back in Section~\ref{sec:6}.  We had this stochastic Petri net:

\begin{center}
 \includegraphics[width=85mm]{amoeba.png}
\end{center}

We saw that the rate equation is just the logistic equation, familiar from population biology.  The equilibrium solution is complex balanced, because pairs of amoebas are getting created at the same rate as they're getting destroyed, and \emph{single}  amoebas are  getting created at the same rate as \emph{they're}  getting destroyed.  

So, the Anderson--Craciun--Kurtz theorem guarantees that there's an equilibrium solution of the master equation where the number of amoebas is distributed according to a Poisson distribution.  And, we actually checked that this was true!

In the next section, we'll look at another example.

\subsection[Answers]{Answers}

Here is the answer to the problem, provided by David Corfield:

\vskip 1em \noindent {\bf Problem 11.} Show that if a classical state $c$ is complex balanced, and we set $x(t) = c$ for all $t,$ then $x(t)$ is a solution of the rate equation.
\vskip 1em

\begin{answer}
Assuming $c$ is complex balanced, we have:
$$ \begin{array}{ccl} \displaystyle{ \sum_{\tau \in T} r(\tau) m(\tau) c^{m(\tau)}} &=& \displaystyle{ \sum_{\kappa} \sum_{\tau : m(\tau) = \kappa} r(\tau) m(\tau) c^{m(\tau)} } \\ \\ 
&=& \displaystyle{\sum_{\kappa} \sum_{\tau : m(\tau) = \kappa} r(\tau) \kappa c^{m(\tau)}} \\ \\
&=&
\displaystyle{ \sum_{\kappa} \sum_{\tau : n(\tau) = \kappa} r(\tau) \kappa c^{m(\tau)}} \\  \\ 
&=& \displaystyle{ \sum_{\kappa} \sum_{\tau : n(\tau) = \kappa} r(\tau) n(\tau) c^{m(\tau)}} \\ \\
&=& \displaystyle{\sum_{\tau \in T} r(\tau) n(\tau) c^{m(\tau)}} \end{array} $$
So, we have
$$ \displaystyle{\sum_{\tau \in T} r(\tau) \, (n(\tau) - m(\tau)) \, c^{m(\tau)} = 0} $$
and thus if $x(t) = c$ for all $t$ then $x(t)$ is a solution of the rate
equation:
$$ \displaystyle{  \frac{d x}{d t} = 0 = \sum_{\tau \in T} r(\tau)\, (n(\tau)-m(\tau)) \, x^{m(\tau)}} $$
\end{answer}

%%%%% SECTION 9 %%%%%%%

\newpage
\section[A reversible reaction]{An example of the Anderson--Craciun--Kurtz theorem}\index{Anderson--Craciun--Kurtz theorem!example of|(}
\label{sec:9}

Let's look at an example that illustrates the Anderson--Craciun--Kurtz theorem on equilibrium states\index{master equation!equilibrium}.  This example brings up an interesting `paradox'---or at least a problem.  Resolving this will get us ready to think about a version of \emph{Noether's theorem}  relating conserved quantities and symmetries.  In Section \ref{sec:9} we'll discuss Noether's theorem in more detail.

\subsection{A reversible reaction}

In chemistry a type of atom, molecule, or ion is called a {\bf chemical species}\index{reaction network!chemical species}, or {\bf species} for short.  Since we're applying our ideas to both chemistry and biology, it's nice that `species' is also used for a type of organism in biology.  This stochastic Petri net describes the simplest reversible reaction of all, involving two species:

\begin{center}
 \includegraphics[width=90mm]{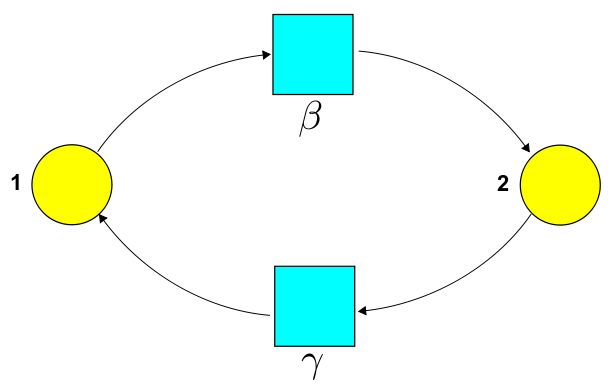}
\end{center}

We have species 1 turning into species 2 with rate constant $ \beta,$ and species 2 turning back into species 1 with rate constant $ \gamma.$  So, the rate equation is:
$$ \begin{array}{ccc} \displaystyle{ \frac{d x_1}{d t}} &=& -\beta x_1 + \gamma x_2  \\  \qquad \mathbf{} \\ \displaystyle{ \frac{d x_2}{d t}} &=&  \beta x_1 - \gamma x_2  \end{array} $$
where $ x_1$ and $ x_2$ are the amounts of species 1 and 2, respectively.

\subsection[Equilibria]{Equilibrium solutions of the rate equation}

Let's look for \emph{equilibrium}  solutions of the rate equation\index{rate equation!equilibrium}, meaning solutions where the amount of each species doesn't change with time.  Equilibrium occurs when each species is getting created at the same rate at which it's getting destroyed.

So, let's see when
$$ \displaystyle{ \frac{d x_1}{d t} = \frac{d x_2}{d t} = 0} $$
Clearly this happens precisely when
$$  \beta x_1 = \gamma x_2 $$
This says the rate at which 1's are turning into 2's equals the rate at which 2's are turning back into 1's.  That makes perfect sense.  

\subsection{Complex balanced equilibria}\index{complex balanced equilibrium|(}
\index{equilibrium!complex balanced|(}

In general, a chemical reaction involves a \emph{bunch}  of species turning into a \emph{bunch}  of species.  Since `bunch' is not a very dignified term, a bunch of species is usually called a {\bf complex}.  We saw in Section \ref{sec:8} that it's very interesting to study a strong version of equilibrium: {\bf complex balanced equilibrium}, in which each \emph{complex}  is being created at the same rate at which it's getting destroyed.
 
However, in the Petri net we're studying now, all the complexes being produced or destroyed consist of a single species.   In this situation, any equilibrium solution is automatically complex balanced.   This is great, because it means we can apply the Anderson--Craciun--Kurtz theorem from Section~\ref{sec:8}!  This says how to get from a complex balanced equilibrium solution of the \emph{rate equation}  to an equilibrium solution of the \emph{master equation}.  

First remember what the master equation says.  Let $ \psi_{n_1, n_2}(t)$ be the probability that we have $ n_1$ things of species 1 and $ n_2$ things of species 2 at time $ t$.  We summarize all this information in a formal power series:
$$ \Psi(t) = \sum_{n_1, n_2 = 0}^\infty \psi_{n_1, n_2}(t) z_1^{n_1} z_2^{n_2} $$
Then the master equation says
$$ \frac{d}{d t} \Psi(t) = H \Psi (t) $$
where following the general rules laid down in Section~\ref{sec:7}, 
$$ \begin{array}{ccl} H &=& \beta (a_2^\dagger - a_1^\dagger) a_1 + \gamma (a_1^\dagger - a_2^\dagger )a_2 \end{array} $$
This may look scary, but the {\bf annihilation operator} $ a_i$ and the {\bf creation operator} $ a_i^\dagger$ are just funny ways of writing the partial derivative $ \partial/\partial z_i$ and multiplication by $ z_i$, so 
$$ \begin{array}{ccl} H &=& \displaystyle{ \beta (z_2 - z_1) \frac{\partial}{\partial z_1} + \gamma (z_1 - z_2) \frac{\partial}{\partial z_2}} \end{array} $$
or if you prefer,
$$ \begin{array}{ccl} H &=& \displaystyle{ (z_2 - z_1) \, (\beta \frac{\partial}{\partial z_1} - \gamma \frac{\partial}{\partial z_2})} \end{array} $$
The first term describes species 1 turning into species 2.  The second describes species 2 turning back into species 1.

Now, the Anderson--Craciun--Kurtz theorem says that whenever $ (x_1,x_2)$ is a complex balanced solution of the rate equation, this recipe gives an equilibrium solution of the master equation:
$$ \displaystyle{ \Psi = \frac{e^{x_1 z_1 + x_2 z_2}}{e^{x_1 + x_2}}} $$
In other words: whenever $ \beta x_1 = \gamma x_2$, we have

we have 
$$ H\Psi = 0 $$ 
Let's check this!  For starters, the constant in the denominator of $ \Psi$ doesn't matter here, since $ H$ is linear.  It's just a normalizing constant, put in to make sure that our probabilities $ \psi_{n_1, n_2}$ sum to 1.  So, we just need to check that
$$ \displaystyle{ (z_2 - z_1) (\beta \frac{\partial}{\partial z_1} - \gamma \frac{\partial}{\partial z_2}) e^{x_1 z_1 + x_2 z_2} = 0} $$
If we do the derivatives on the left hand side, it's clear we want
$$ \displaystyle{ (z_2 - z_1) (\beta x_1 - \gamma x_2) e^{x_1 z_1 + x_2 z_2} = 0} $$
and this is indeed true when $ \beta x_1 = \gamma x_2 $.

So, the theorem works as advertised.  And now we can work out the probability $ \psi_{n_1,n_2}$ of having $ n_1$ things of species 1 and $ n_2$ of species 2 in our equilibrium state $ \Psi$.  To do this, we just expand the function $ \Psi$ as a power series and look at the coefficient of $ z_1^{n_1} z_2^{n_2}$.  We have
$$ \Psi = \displaystyle{ \frac{e^{x_1 z_1 + x_2 z_2}}{e^{x_1 + x_2}}} = \displaystyle{ \frac{1}{e^{x_1} e^{x_2}} \sum_{n_1,n_2}^\infty \frac{(x_1 z_1)^{n_1}}{n_1!} \frac{(x_2 z_2)^{n_2}}{n_2!}}  $$
so we get
$$ \psi_{n_1,n_2} = \displaystyle{ \frac{1}{e^{x_1}} \frac{x_1^{n_1}}{n_1!}  \cdot  \frac{1}{e^{x_2}} \frac{x_1^{n_2}}{n_2!}}$$
This is just a product of two independent Poisson distributions!\index{Poisson distribution!independent pair} 

\index{complex balanced equilibrium|)} \index{equilibrium!complex balanced|)}

In case you forget, a \href{http://en.wikipedia.org/wiki/Poisson\_distribution}{Poisson distribution} says the probability of $ k$ events occurring in some interval of time if they occur with a fixed average rate and independently of the time since the last event.  If the expected number of events is $ \lambda$, the Poisson distribution is 
$$ \displaystyle{ \frac{1}{e^\lambda} \frac{\lambda^k}{k!}} $$
and it looks like this for various values of $ \lambda$:

\begin{center}
 \includegraphics[width=90mm]{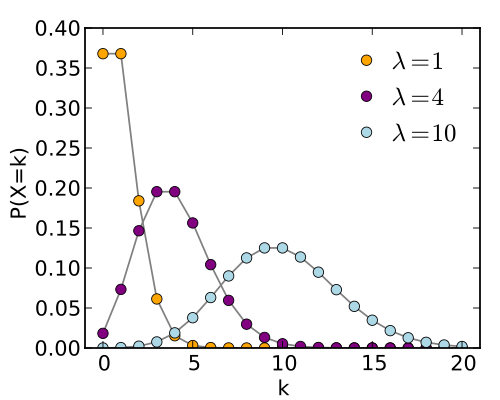}
\end{center}

\noindent
It looks almost like a Gaussian when $ \lambda$ is large, but when $ \lambda$ is small it becomes very lopsided.

Anyway: we've seen that in our equilibrium state, the number of things of species $ i = 1,2$ is given by a Poisson distribution with mean $ x_i$.   That's very nice and simple... but the amazing thing is that these distributions are \emph{independent}. 

Mathematically, this means we just \emph{multiply}  them to get the probability of finding $ n_1$ things of species 1 and $ n_2$ of species 2.  But it also means that knowing how many things there are of one species says nothing about the number of the other.

But something seems odd here.  One transition in our Petri net consumes a 1 and produces a 2, while the other consumes a 2 and produces a 1.  The total number of particles in the system never changes.  The more 1's there are, the fewer 2's there should be.   But we just said knowing how many 1's we have tells us nothing about how many 2's we have!  

At first this seems like a paradox.  Have we made a mistake?  Not exactly.   But we're neglecting something.

\index{Anderson--Craciun--Kurtz theorem!example of|)}

\subsection{Conserved quantities} \index{Noether's theorem|(}

Namely: the equilibrium solutions of the master equation we've found so far are not the only ones!  There are other solutions that fit our intuitions better.

Suppose we take any of our equilibrium solutions $ \Psi$ and change it like this: set the probability $ \psi_{n_1,n_2}$ equal to 0 unless
$$ n_1 + n_2 = n $$
but otherwise leave it unchanged.  Of course the probabilities no longer sum to 1, but we can rescale them so they do.  

The result is a new equilibrium solution, say $ \Psi_n.$ Why?  Because, as we've already seen, no transitions will carry us from one value of $ n_1 + n_2$ to another.  And in this new solution, the number of 1's is clearly \emph{not}  independent from the number of 2's.  The bigger one is, the smaller the other is.  

\begin{problem}\label{prob:12} 
Show that this new solution $ \Psi_n$ depends only on $ n$ and the ratio $ x_1/x_2 = \gamma/\beta$, not on anything more about the values of $ x_1$ and $ x_2$ in the original solution
$$ \Psi = \displaystyle{ \frac{e^{x_1 z_1 + x_2 z_2}}{e^{x_1 + x_2}}} $$\end{problem} 

\begin{problem}\label{prob:13} 
What is this new solution like when $ \beta = \gamma$?  (This particular choice makes the problem symmetrical when we interchange species 1 and 2.)
\end{problem} 

What's happening here is that this particular stochastic Petri net has a `conserved quantity': the total number of things never changes with time.  So, we can take any equilibrium solution of the master equation and---in the language of quantum mechanics---`project down to the subspace' where this conserved quantity takes a definite value, and get a new equilibrium solution.  In the language of probability theory, we say it a bit differently: we're `conditioning on' the conserved quantity taking a definite value.  But the idea is the same.

This important feature of conserved quantities suggests that we should try to invent a new version of \href{http://en.wikipedia.org/wiki/Noether\%27s_theorem}{Noether's theorem}.   This theorem links conserved quantities and \emph{symmetries}  of the Hamiltonian.

There are already a couple versions of Noether's theorem for classical mechanics, and for quantum mechanics... but now we want a version for \emph{stochastic}  mechanics.  And indeed one exists, and it's relevant to what we're doing here.  So let us turn to that.

%%%%% SECTION 10 %%%%%%%

\newpage
\section[Noether's theorem]{A stochastic version of Noether's theorem}
\label{sec:10}

Noether proved lots of theorems, but when people talk about \href{http://en.wikipedia.org/wiki/Noether\%27s\_theorem}{Noether's theorem}, they always seem to mean her result linking \emph{symmetries}  to \emph{conserved quantities}.  Her original result applied to classical mechanics, but now we'd like to present a version that applies to `stochastic mechanics'---or in other words, Markov processes.  This theorem was proved here:

\begin{enumerate}
\item[\cite{BF12}]  John Baez and Brendan Fong, A Noether theorem for Markov processes, {\sl J.\ Math.\ Phys.} {\bf 54} (2013), 013301.   Also available as \href{http://arxiv.org/abs/1203.2035}{arXiv:1203.2035}. 
\end{enumerate}

What's a \href{http://en.wikipedia.org/wiki/Markov\_process}{Markov process}?  We'll say more in a minute---but in plain English, it's a physical system where something hops around randomly from state to state, where its probability of hopping anywhere depends only on where it is \emph{now}, not its past history.   Markov processes include, as a special case, the stochastic Petri nets we've been talking about.  

Our stochastic version of Noether's theorem is copied after a well-known \emph{quantum}  version.  It's yet another example of how we can exploit the analogy between stochastic mechanics and quantum mechanics.  But for now we'll just present the stochastic version.  In the next section, we'll compare it to the quantum one.

\subsection{Markov processes}

We should and probably \emph{will}  be more general, but let's start by considering a finite set of states, say $X$.  To describe a Markov process we then need a matrix of real numbers $H = (H_{i j})_{i, j \in X}.$   The idea is this: suppose right now our system is in the state $i$.  Then the probability of being in some state $j$ changes as time goes by---and $H_{i j}$ is defined to be the \emph{time derivative}  of this probability right now.  

So, if $\psi_i(t)$ is the probability of being in the state $i$ at time $t$, we want the {\bf master equation} to hold:\index{master equation!and Markov processes}
$$ \frac{d}{d t} \psi_i(t) = \sum_{j \in X} H_{i j} \psi_j(t) $$
This motivates the definition of `infinitesimal stochastic', which we gave in a more general context back in Section~\ref{sec:4}:

\begin{definition}
 Given a finite set $X$, a matrix of real numbers $H = (H_{i j})_{i, j \in X}$ is {\bf infinitesimal stochastic} if 
$$ i \ne j  \Rightarrow  H_{i j} \ge 0 $$
and
$$ \sum_{i \in X} H_{i j} = 0  $$
for all $j \in X$. 
\end{definition}  

The inequality says that if we start in the state $i$, the probability of being in some other state, which starts at 0, can't go down, at least initially.  The equation says that the probability of being \emph{somewhere or other}  doesn't change.  Together, these facts imply that:
$$ H_{i i} \le 0 $$
That makes sense: the probability of being in the state $i$, which starts at 1, can't go up, at least initially.

Using the magic of matrix multiplication, we can rewrite the master equation as follows:
$$ \frac{d}{d t} \psi(t) = H \psi(t) $$\index{matrix!magic of matrix multiplication}
and we can solve it like this:
$$ \psi(t) = \exp(t H) \psi(0) $$ 

If $H$ is an infinitesimal stochastic operator, we will call $\exp(t H)$ a {\bf Markov process}, and $H$ its {\bf Hamiltonian}.

(Actually, most people call $\exp(t H)$ a {\bf Markov semigroup},\index{semigroup!Markov} \index{Markov semigroup} and reserve the term  \href{http://en.wikipedia.org/wiki/Markov\_process}{Markov process} for another way of looking at the same idea.  So, be careful.) \index{Markov process}

Noether's theorem is about `conserved quantities', that is, observables whose expected values don't change with time.\index{expected value!time independent}   To understand this theorem, you need to know a bit about observables.  In stochastic mechanics an {\bf observable} is simply a function assigning a number $O_i$ to each state $i \in X$. 

However, in quantum mechanics we often think of observables as matrices, so it's nice to do that here, too.  It's easy: we just create a matrix whose diagonal entries are the values of the function $O$.  And just to confuse you, we'll also call this matrix $O$.  So:
$$ O_{i j} = \left\{ \begin{array}{ccl}  O_i & \textrm{if} & i = j \\ 0 & \textrm{if} & i \neq j  \end{array} \right. $$
One advantage of this trick is that it lets us ask whether an observable commutes with the Hamiltonian.  Remember, the {\bf commutator} of matrices is defined by
$$ [O,H] = O H - H O$$
Noether's theorem will say that $[O,H] = 0$ if and only if $O$ is `conserved' in some sense.  What sense?  First, recall that a {\bf stochastic state} is just our fancy name for a probability distribution $\psi$ on the set $X$.  Second, the {\bf expected value}\index{expected value!definition of, stochastic} of an observable $O$ in the stochastic state $\psi$ is defined to be
$$ \sum_{i \in X} O_i \psi_i $$
In Section~\ref{sec:4} we introduced the notation
$$ \int \phi = \sum_{i \in X} \phi_i $$
for any function $\phi$ on $X$.   The reason is that later, when we generalize $X$ from a finite set to a measure space\index{measure space}, the sum at right will become an integral over $X$.  Indeed, a sum is just a special sort of integral!

Using this notation and the magic of matrix multiplication, we can write the expected value of $O$ in the stochastic state $\psi$ as\index{matrix!magic of matrix multiplication}
$$ \int O \psi $$ 
We can calculate how this changes in time if $\psi$ obeys the master equation... and we can write the answer using the commutator $[O,H]$:
\index{infinitesimal stochastic operator|(}
\index{operator!infinitesimal stochastic|(}
\begin{lemma}  Suppose $H$ is an infinitesimal stochastic operator and $O$ is an observable.  If $\psi(t)$ obeys the master equation, then
$$ \frac{d}{d t} \int O \psi(t) = \int [O,H] \psi(t) $$
\end{lemma}

\begin{proof}  Using the master equation we have
$$  \frac{d}{d t} \int O \psi(t) = \int O \frac{d}{d t} \psi(t) = \int O H \psi(t)  $$
But since $H$ is infinitesimal stochastic,
$$ \sum_{i \in X} H_{i j} = 0  $$
so for any function $\phi$ on $X$ we have
$$ \int H \phi = \sum_{i, j \in X} H_{i j} \phi_j = 0 $$
and in particular
\[   \int H O \psi(t) = 0    \]
Since $[O,H] = O H - H O $, we conclude from (1) and (2) that
$$ \frac{d}{d t} \int O \psi(t) = \int [O,H] \psi(t) $$
as desired.  
\end{proof}

The commutator doesn't look like it's doing much here, since we also have
$$  \frac{d}{d t} \int O \psi(t) = \int O H \psi(t)  $$
which is even simpler.  But the commutator will become useful when we get to Noether's theorem!

\subsection{Noether's theorem}

Here's a version of Noether's theorem for Markov processes.  It says an observable commutes with the Hamiltonian iff the expected values of that observable \emph{and its square}  don't change as time passes:\index{expected value!and Noether's theorem|(} 

\begin{theorem}\label{theorem:stochastic-noether}
\index{Noether's theorem!stochastic version} 
Suppose $H$ is an infinitesimal stochastic operator and $O$ is an observable.  Then 
$$ [O,H] =0 $$
if and only if 
$$ \frac{d}{d t} \int O\psi(t) = 0 $$
and 
$$ \frac{d}{d t} \int O^2\psi(t) = 0 $$
for all $\psi(t)$ obeying the master equation.
\end{theorem} 

If you know Noether's theorem from quantum mechanics, you might be surprised that in this version we need not only the observable \emph{but also its square}  to have an unchanging expected value!  We'll explain this, but first let's prove the theorem.

\begin{proof} 
The easy part is showing that if $[O,H]=0$ then $\frac{d}{d t} \int O\psi(t) = 0$ and $\frac{d}{d t} \int O^2\psi(t) = 0$. In fact there's nothing special about these two powers of $t$; we'll show that
$$ \frac{d}{d t} \int O^n \psi(t) = 0 $$
for all $n$.  The point is that since $H$ commutes with $O$, it commutes with all powers of $O$:
$$ [O^n, H] = 0$$
So, applying the Lemma to the observable $O^n$, we see
$$ \frac{d}{d t} \int O^n \psi(t) =  \int [O^n, H] \psi(t) = 0 $$
The backward direction is a bit trickier.   We now assume that 
$$ \frac{d}{d t} \int O\psi(t) = \frac{d}{d t} \int O^2\psi(t) = 0 $$
for all solutions $\psi(t)$ of the master equation.  This implies 
$$  \int O H\psi(t) = \int O^2H\psi(t) = 0 $$
or since this holds for \emph{all}  solutions,
\[ \sum_{i \in X} O_i H_{i j} = \sum_{i \in X} O_i^2H_{i j} = 0  \]
We wish to show that $[O,H]= 0$.  

First, recall that we can think of $O$ as a diagonal matrix with:
$$ O_{i j} = \left\{ \begin{array}{ccl}  O_i & \textrm{if} & i = j \\ 0 & \textrm{if} & i \ne j  \end{array} \right. $$
So, we have
$$ [O,H]_{i j} = \sum_{k \in X} (O_{i k}H_{k j} - H_{i k} O_{k j}) = O_i H_{i j} - H_{i j}O_j = (O_i-O_j)H_{i j} $$
To show this is zero for each pair of elements $i, j \in X$, it suffices to show that when $H_{i j} \ne 0$, then $O_j = O_i$.  That is, we need to show that if the system can move from state $j$ to state $i$, then the observable takes the same value on these two states.

In fact, it's enough to show that this sum is zero for any $j \in X$:
$$ \sum_{i \in X} (O_j-O_i)^2 H_{i j}  $$
Why?   When $i = j$, $O_j-O_i = 0$, so that term in the sum vanishes.  But when $i \ne j$, $(O_j-O_i)^2$ and $H_{i j}$ are both non-negative---the latter because $H$ is infinitesimal stochastic.  So if they sum to zero, they must each be individually zero. Thus for all $i \ne j$, we have $(O_j-O_i)^2H_{i j}=0$.  But this means that either $O_i = O_j$ or $H_{i j} = 0$, which is what we need to show.

So, let's take that sum and expand it: 
$$ \begin{array}{ccl} \displaystyle{ \sum_{i \in X} (O_j-O_i)^2 H_{i j}} &=& \displaystyle{ \sum_i (O_j^2 H_{i j}- 2O_j O_i H_{i j} +O_i^2 H_{i j})} \\  &=& \displaystyle{ O_j^2\sum_i H_{i j} - 2O_j \sum_i O_i H_{i j} + \sum_i O_i^2 H_{i j}} \end{array} $$
The three terms here are each zero: the first because $H$ is infinitesimal stochastic, and the latter two by equation (3). So, we're done! 
\end{proof}

\subsection{Markov chains}\index{Markov chain|(} 

So that's the proof... but why do we need both $O$ \emph{and its square}  to have an expected value\index{expected value!Markov chains} that doesn't change with time to conclude $[O,H] = 0$?  There's an easy counterexample if we leave out the condition involving $O^2$.  However, the underlying idea is clearer if we work with \href{http://en.wikipedia.org/wiki/Markov\_chain}{Markov chains} instead of Markov processes.  

In a Markov process, time passes by continuously.  In a Markov chain, time comes in discrete steps!  We get a Markov process by forming $\exp(t H)$ where $H$ is an infinitesimal stochastic operator.  We get a Markov chain by forming the operator $U, U^2, U^3, \dots$ where $U$ is a `stochastic operator'.  Remember:

\begin{definition}
Given a finite set $X$, a matrix of real numbers $U = (U_{i j})_{i, j \in X}$ is {\bf stochastic} if 
$$ U_{i j} \ge 0 $$
for all $i, j \in X$ and 
$$ \sum_{i \in X} U_{i j} = 1  $$
for all $j \in X$.  
\end{definition} 

The idea is that $U$ describes a random hop, with $U_{i j}$ being the probability of hopping to the state $i$ if you start at the state $j$.  These probabilities are nonnegative and sum to 1.

Any stochastic operator gives rise to a \href{http://en.wikipedia.org/wiki/Markov\_chain}{{\bf Markov chain}} $U, U^2, U^3, \dots .$  And in case it's not clear, that's how we're \emph{defining}  a Markov chain: the sequence of powers of a stochastic operator.  There are other definitions, but they're equivalent.

We can draw a Markov chain by drawing a bunch of states and arrows labelled by transition probabilities, which are the matrix elements $U_{i j}$:

%\todo{[PIC: markov\_chain.png]}

\begin{center}
 \includegraphics[width=50mm]{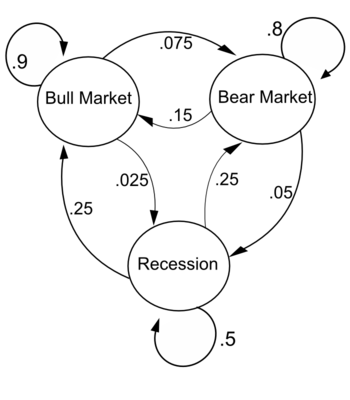}
\end{center}
%\href{http://en.wikipedia.org/wiki/Markov\_chain}{<img src = "http://math.ucr.edu/home/baez/networks/Markov\_chain.png" alt = ""/>} 
%
Here is Noether's theorem for Markov chains:

\begin{theorem}\index{Noether's theorem!for Markov chains} 

Suppose $U$ is a stochastic operator and $O$ is an observable.  Then  
$$ [O,U] =0 $$
if and only if 
$$  \int O U \psi = \int O \psi $$
and
$$ \int O^2 U \psi = \int O^2 \psi $$
for all stochastic states $\psi$.  
\end{theorem} 

In other words, an observable commutes with $U$ iff the expected values of that observable \emph{and its square}  don't change when we evolve our state one time step using $U$.  You can prove this theorem by copying the proof for Markov processes:

\begin{problem}\label{prob:14} 
Prove Noether's theorem for Markov chains.
\end{problem} 

But let's see why we need the condition on the square of observable!  That's the intriguing part. Here's a nice little Markov chain:

\begin{center}
 \includegraphics[width=60mm]{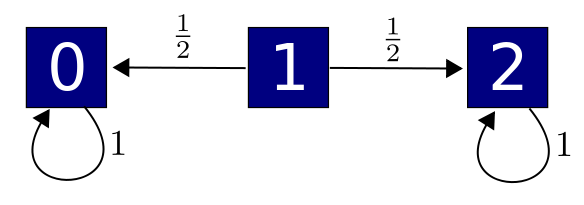}
\end{center}

\index{Markov chain|)} 

\noindent
where we haven't drawn arrows labelled by 0.  So, state 1 has a 50\% chance of hopping to state 0 and a 50\% chance of hopping to state 2; the other two states just sit there.  Now, consider the observable $O$ with 
$$O_i = i$$
It's easy to check that the expected value of this observable doesn't change with time:
$$  \int O U \psi = \int O \psi $$
for all $\psi$.  The reason, in plain English, is this.  Nothing at all happens if you start at states 0 or 2: you just sit there, so the expected value of $O$ doesn't change.  If you start at state 1, the observable equals 1.  You then have a 50\% chance of going to a state where the observable equals 0 and a 50\% chance of going to a state where it equals 2, so its \emph{expected}  value doesn't change: it still equals 1.  

On the other hand, we do \emph{not}  have $[O,U] = 0$ in this example, because we \emph{can}  hop between states where $O$ takes different values.  Furthermore, 
$$ \int O^2 U \psi \ne \int O^2 \psi $$
After all, if you start at state 1, $O^2$ equals 1 there.  You then have a 50\% chance of going to a state where $O^2$ equals 0 and a 50\% chance of going to a state where it equals 4, so its expected value changes!

So, that's why $\int O U \psi = \int O \psi $ for all $\psi$ is not enough to guarantee $[O,U] = 0$.  The same sort of counterexample works for Markov processes, too.

Finally, we should add that there's nothing terribly sacred about the \emph{square}  of the observable.  For example, we have:

\begin{theorem}\index{infinitesimal stochastic operator|)} 
\index{operator!infinitesimal stochastic|)} 
Suppose $H$ is an infinitesimal stochastic operator and $O$ is an observable.  Then  
$$ [O,H] =0 $$
if and only if 
$$ \frac{d}{d t} \int f(O) \psi(t) = 0 $$
for all smooth $f\colon \mathbb{R} \to \mathbb{R}$ and all $\psi(t)$ obeying the master equation.
\end{theorem} 

\begin{theorem} 
Suppose $U$ is a stochastic operator and $O$ is an observable.  Then  
$$ [O,U] =0 $$
if and only if 
$$  \int f(O) U \psi = \int f(O) \psi $$
for all smooth $f\colon \mathbb{R} \to \mathbb{R}$ and all stochastic states $\psi$.
\end{theorem} 

These make the `forward direction' of Noether's theorem stronger... and in fact, the forward direction, while easier, is probably more useful!  However, if you ever use Noether's theorem in the `reverse direction', it might be easier to check a condition involving only $O$ and its square.

\subsection{Answers}

Here's the answer to the problem:

\vskip 1em \noindent {\bf Problem 14.} 
Suppose $U$ is a stochastic operator and $O$ is an observable.   Show that $O$ commutes with $U$ iff the expected values of $O$ \emph{and its square}  don't change when we evolve our state one time step using $U$.  In other words, show that
$$ [O,U] =0 $$
if and only if
$$ \displaystyle{  \int O U \psi = \int O \psi} $$
and
$$ \displaystyle{ \int O^2 U \psi = \int O^2 \psi} $$
for all stochastic states $\psi.$

\index{expected value!and Noether's theorem|)} 

\begin{answer}
One direction is easy: if $[O,U] = 0$ then $[O^n,U] = 0$ for all $ n$, so
$$ \displaystyle{ \int O^n U \psi = \int U O^n \psi = \int O^n \psi} $$
where in the last step we use the fact that $U$ is stochastic.

For the converse we can use the same tricks that worked for Markov processes.  Assume that
$$ \displaystyle{  \int O U \psi = \int O \psi}  $$
and
$$ \displaystyle{ \int O^2 U \psi = \int O^2 \psi} $$
for all stochastic states $ \psi$.  These imply that
$$ \displaystyle{ \sum_{i \in X} O_i U_{i j} = O_j}  \quad \qquad (1) $$
and
$$ \displaystyle{ \sum_{i \in X} O^2_i U_{i j} = O^2_j} \qquad \qquad (2) $$
We wish to show that $[O,U]= 0$.   Note that
$$ \begin{array}{ccl} [O,U]_{i j} &=& \displaystyle{ \sum_{k \in X} (O_{i k}U_{k j} - U_{i k} O_{k j})} \\ \\  &=& (O_i-O_j)U_{i j} \end{array} $$
To show this is always zero, we'll show that when $ U_{i j} \ne 0$, then $ O_j = O_i$.  This says that when our system can hop from one state to another, the observable $ O$ must take the same value on these two states.

For this, in turn, it's enough to show that the following sum vanishes for any $ j \in X$:
$$ \displaystyle{ \sum_{i \in X} (O_j-O_i)^2 U_{i j}}  $$
Why?   The matrix elements $ U_{i j}$ are nonnegative since $ U$ is stochastic.  Thus the sum can only vanish if each term vanishes, meaning that $ O_j = O_i$ whenever $ U_{i j} \ne 0$.

To show the sum vanishes, let's expand it:
$$ \begin{array}{ccl} \displaystyle{ \sum_{i \in X} (O_j-O_i)^2 U_{i j}} &=& \displaystyle{ \sum_i (O_j^2 U_{i j}- 2O_j O_i U_{i j} +O_i^2 U_{i j})}  \\  \\ &=& \displaystyle{  O_j^2\sum_i U_{i j} - 2O_j \sum_i O_i U_{i j} + \sum_i O_i^2 U_{i j}} \end{array} $$
Now, since (1) and (2) hold for all stochastic states $ \psi$, this equals
$$ \displaystyle{  O_j^2\sum_i U_{i j} - 2O_j^2 + O_j^2 } $$
But this is zero because $ U$ is stochastic, which implies
$$ \sum_i U_{i j} = 1$$
So, we're done!

\end{answer} 

\index{Noether's theorem|)} 

%%%%%% SECTION 11 %%%%%%%

\newpage
\section{Quantum mechanics vs stochastic mechanics} 
\label{sec:11}

In Section \ref{sec:10} we proved a version of Noether's theorem for stochastic mechanics.  Now we want to compare that to the more familiar quantum version.  
But to do this, we need to say more about the analogy between stochastic mechanics and quantum mechanics.   So far, whenever we've tried, we feel pulled toward explaining some technical issues involving analysis: whether sums converge, whether derivatives exist, and so on.  We've been trying to avoid such stuff---not because we dislike it, but because we're afraid \emph{you} might.  But the more we put off discussing these issues, the more they fester.  So, this time we will gently explore some of these technical issues.  But don't be scared: we'll \emph{mainly} talk about some simple big ideas.  In the next section, we'll use these to compare the two versions of Noether's theorem.  

To beging with, we need to recall the analogy we began sketching in Section~\ref{sec:4}, and push it a bit further.   The idea is that stochastic mechanics differs from quantum mechanics in two big ways:

\begin{itemize}
\item  First, instead of complex amplitudes, stochastic mechanics uses nonnegative real probabilities.  The complex numbers form a \href{http://en.wikipedia.org/wiki/Ring\_\%28mathematics\%29}{ring}; the nonnegative real numbers form a mere \href{http://ncatlab.org/nlab/show/rig}{rig}, which is a `ri{\bf n}g without {\bf n}egatives'.  Rigs are much neglected in the typical math curriculum, but unjustly so: they're almost as good as rings in many ways, and there are lots of important examples, like the natural numbers $\mathbb{N}$ and the nonnegative real numbers, $[0,\infty)$.  For probability theory, we should learn to love rigs.  

But there are, alas, situations where we need to subtract probabilities, even when the answer comes out negative: namely when we're taking the \emph{time derivative}  of a probability.  So sometimes we need $\mathbb{R}$ instead of just $[0,\infty)$.

\item  Second, while in quantum mechanics a state is described using a `wavefunction', meaning a complex-valued function obeying
$$ \int |\psi|^2 = 1 $$
in stochastic mechanics it's described using a `probability distribution', meaning a nonnegative real function obeying
$$ \int \psi = 1 $$\end{itemize} 

So, let's try our best to present the theories in close analogy, while respecting these two differences.

\subsection{States}
\label{sec:11_states}

We'll start with a set $X$ whose points are {\bf states} that a system can be in.  So far, we assumed $X$ was a finite set, but this section is so mathematical we might as well  assume it's a \href{http://en.wikipedia.org/wiki/Measure\_\%28mathematics\%29}{measure space}.\index{measure space}  A measure space lets you do integrals, but a finite set is a special case, and then these integrals are just sums.\index{integration!finite sums}   So, we'll write things like\index{integration!and measure space}  
$$ \int f $$
and mean the integral of the function $f$ over the measure space $X$, but if $X$ is a finite set this just means 
$$ \sum_{x \in X} f(x) $$
Now, we've already defined the word `state', but both quantum and stochastic mechanics need a more general concept of state.  Let's call these `quantum states' and `stochastic states':

\begin{itemize} 
\item  In {\bf quantum mechanics}, the system has an {\bf amplitude} $\psi(x)$ of being in any state $x \in X$.  These amplitudes are complex numbers with 
$$\int | \psi |^2 = 1$$
We call $\psi \colon X \to \mathbb{C}$ obeying this equation a {\bf quantum state}.

\item  In {\bf stochastic mechanics}, the system has a {\bf probability} $\psi(x)$ of being in any state $x \in X$.  These probabilities are nonnegative real numbers with 
$$\int \psi = 1$$
We call $\psi \colon X \to [0,\infty)$ obeying this equation a {\bf stochastic state}.
\end{itemize} 

In quantum mechanics we often use this abbreviation:
$$ \langle \phi, \psi \rangle = \int \overline{\phi} \psi $$
so that a quantum state has
$$ \langle \psi, \psi \rangle = 1 $$
Similarly, we could introduce this notation in stochastic mechanics:
$$ \langle \psi \rangle = \int \psi $$
so that a stochastic state has
$$ \langle \psi \rangle = 1 $$
But this notation is a bit risky, since angle brackets of this sort often stand for expectation values of observables.  So, we've been writing $\int \psi$, and we'll
keep on doing this most of the time.

In quantum mechanics, $\langle \phi, \psi \rangle$ is well-defined whenever both $\phi$ and $\psi$ live in the vector space
$$L^2(X) = \{ \psi \colon X \to \mathbb{C} \; : \;  \int |\psi|^2 < \infty \} $$
In stochastic mechanics, $\langle \psi \rangle$ is well-defined whenever $\psi$ lives in the vector space
$$L^1(X) =  \{ \psi \colon X \to \mathbb{R}  : \; \int |\psi| < \infty \} $$
You'll notice we wrote $\mathbb{R}$ rather than $[0,\infty)$ here. That's because in some calculations we'll need functions that take negative values, even though our stochastic states are nonnegative.

\subsection{Observables}
\label{sec:11_observables}

A state is a way our system can be.  An observable is something we can measure about our system.   They fit together: we can measure an observable when our system is in some state.   If we repeat this we may get different answers, but there's a nice formula for average or `expected' answer.

In quantum mechanics, an {\bf observable} is a \href{http://en.wikipedia.org/wiki/Self-adjoint\_operator}{self-adjoint operator} $A$ on $L^2(X)$. \index{self-adjoint operator}\index{quantum mechanics!self-adjoint operator}\index{self-adjoint operator}\index{observable!quantum}\index{operator!self-adjoint} \index{self-adjoint operator} In the case where $X$ is a finite set, a self-adjoint operator on $L^2(X)$
is just one with
$$  \langle \psi, A \phi \rangle = \langle  A \psi, \phi \rangle $$
for all $\psi, \phi \in L^2(X)$.  
In general there are some extra subtleties, which we prefer to sidestep here. The {\bf expected value}\index{expected value!definition of, quantum} of an observable $A$ in the state $\psi$ is
$$ \langle \psi, A \psi \rangle = \int \overline{\psi} A \psi $$
Here we're assuming that we can apply $A$ to $\psi$ and get a new vector $A \psi \in L^2(X)$.   This is automatically true when $X$ is a finite set, but in general we need to be more careful.  

In stochastic mechanics, an {\bf observable} is a real-valued function $A$ on $X$.    The {\bf expected value} of $A$ in the state $\psi$ is\index{expected value} \index{observable!stochastic} 
$$ \langle A \psi \rangle = \int A \psi $$
Here we're using the fact that we can multiply $A$ and $\psi$ and get a new vector $A \psi \in L^1(X)$, at least if $A$ is bounded.  Again, this is automatic if $X$ is a finite set, but not otherwise.  

The left side of the equation above is just an abbreviation for the right side, where
we write the integral sign using angle brackets.  We mentioned in the previous section 
that this use of angle brackets is a bit risky, so we won't use 
it often.  It does make the analogy between quantum mechanics and stochastic mechanics rather pretty, though.  

There are a couple of other equivalent ways to write the expected value
of an observable in stochastic mechanics:
\begin{itemize}
\item Sometimes the stochastic state $\psi \in L^1(X)$ and 
the constant function $1$ both lie $L^2(X)$.  This happens 
whenever $X$ is a {\bf finite measure space}, \index{measure space!finite}
meaning one with $\int 1 < \infty$.   For example, it happens when $X$ is 
just a finite set with integration being summation.   

In this situation, we have
$$   \int A \psi = \langle 1, A \psi \rangle $$
This gives the expected value of stochastic observables a bit more of a quantum 
flavor!   We consider this formula to be more of a `trick'  than a deep insight.  Sometimes it may be useful, though.
\item Sometimes the stochastic state $\psi \in L^1(X)$ and the 
observable $A$, which is a bounded function on $X$, both lie in $L^2(X)$.
This is always the case when $X$ is a finite set with integration being summation.

In this situation we can think of $A$ as vector $\hat{A} \in L^2(X)$, and we
have
$$  \int A \psi = \langle \hat{A}, \psi \rangle $$
Again, this formula has a kind of quantum flavor.  We shall use this formula
in Section \ref{sec:24_3}.
\end{itemize}

\subsection{Symmetries}
\label{sec:11_symmetries}

Besides states and observables, we need `symmetries', which are transformations that map states to states.  For example, we use these to describe how our system changes when we wait a while.

In quantum mechanics, an {\bf isometry}\index{quantum mechanics!isometry}\index{operator!isometry}\index{isometry} is a linear map $U \colon L^2(X) \to L^2(X)$ such that
$$ \langle U \phi, U \psi \rangle = \langle \phi, \psi \rangle$$
for all $\psi, \phi \in L^2(X)$.  If $U$ is an isometry and $\psi$ is a quantum state, then $U \psi$ is again a quantum state.

In stochastic mechanics, a {\bf stochastic operator} is a linear map $U\colon L^1(X) \to L^1(X)$ such that \index{operator!stochastic}\index{stochastic operator}
$$ \int U \psi = \int \psi $$
and
$$ \psi \ge 0   \Rightarrow   U \psi \ge 0 $$
for all $\psi \in L^1(X)$.  If $U$ is stochastic and $\psi$ is a stochastic state, then $U \psi$ is again a stochastic state.

In quantum mechanics we are mainly interested in invertible isometries, which are called {\bf unitary} operators.\index{quantum mechanics!unitary operator}\index{unitary operator}\index{operator!unitary} There are lots of these, and their inverses are always isometries.   There are, however, very few stochastic operators whose inverses are stochastic:

\begin{problem}\label{prob:15} 
Suppose $X$ is a finite set.  Show that any isometry $U\colon L^2(X) \to L^2(X)$ is invertible, and its inverse is again an isometry.\index{quantum mechanics!isometry}
\end{problem} 

\begin{problem}\label{prob:16} 
Suppose $X$ is a finite set.  Which stochastic operators $U\colon L^1(X) \to L^1(X)$ have stochastic inverses?
\end{problem} 

This is why we usually think of time evolution as being reversible quantum mechanics, but not in stochastic mechanics!  In quantum mechanics we often describe time evolution using a `1-parameter group', while in stochastic mechanics we describe it using a 1-parameter \emph{semi}group... meaning that we can run time forwards, but not backwards.\index{semigroup!1-parameter} 

But let's see how this works in detail! 

\subsection[Quantum evolution]{Time evolution in quantum mechanics}

In quantum mechanics there's a beautiful relation between observables and symmetries, which goes like this.   Suppose that for each time $t$ we want a unitary operator $U(t)\colon  L^2(X) \to L^2(X)$ that describes time evolution.  Then it makes a lot of sense to demand that these operators form a 1-parameter group:

\begin{definition} 
A collection of linear operators U(t) ($t \in \mathbb{R}$) on some vector space forms a {\bf 1-parameter group} if\index{group theory!unitary!1-parameter}
$$ U(0) = 1 $$
and
$$ U(s+t) = U(s) U(t) $$
for all $s,t \in \mathbb{R}$.
\end{definition} 
\noindent
Note that these conditions force all the operators $U(t)$ to be invertible.  

Now suppose our vector space\index{vector space}\index{Hilbert space!from vector space} is a Hilbert space, like $L^2(X)$.  Then we call a 1-parameter group a {\bf 1-parameter unitary group}\index{group theory!unitary!1-parameter} if the operators involved are all unitary.  

It turns out that 1-parameter unitary groups\index{group theory!unitary!1-parameter!continuous} are either continuous in a certain way, or so pathological that you can't even prove they exist without the axiom of choice!  So, we always focus on the continuous case:

\begin{definition} 
A 1-parameter unitary group\index{one-parameter unitary group!continuous} is {\bf \href{http://en.wikipedia.org/wiki/C0-semigroup}{strongly continuous}} if $U(t) \psi$ depends continuously on $t$ for all $\psi$, in this sense:
$$ t_i \to t  \Rightarrow  \|U(t_i) \psi - U(t) \psi \| \to 0 $$
\end{definition} 

Then we get a classic result proved by Marshall Stone back in the early 1930s.  You may not know him, but he was so influential at the University of Chicago during this period that it's often called the `Stone Age'.  And here's one reason why:

\begin{theorem}[{\bf \href{http://en.wikipedia.org/wiki/Stone\%27s\_theorem\_on\_one-parameter\_unitary\_groups}{Stone's Theorem}}]\index{Stone's theorem!for unitary groups}  There is a one-to-one correspondence between strongly continuous 1-parameter unitary groups on a Hilbert space\index{Hilbert space!one-parameter unitary groups on} and self-adjoint operators on that Hilbert space, given as follows.  Given a strongly continuous 1-parameter unitary group\index{one-parameter unitary group} $U(t)$ we can always write 
$$ U(t) = \exp(-i t H)$$ 
for a unique self-adjoint operator $H$.  Conversely, any self-adjoint operator determines a strongly continuous 1-parameter group this way.  For all vectors $\psi$ for which $H \psi$ is well-defined, we have
$$ \left.\frac{d}{d t} U(t) \psi \right|_{t = 0} = -i H \psi  $$
Moreover, for any of these vectors, if we set
$$ \psi(t) = \exp(-i t H) \psi $$
we have
$$ \frac{d}{d t} \psi(t) = - i H \psi(t) $$
\end{theorem}

When $U(t) = \exp(-i t H)$ describes the evolution of a system in time, $H$ is is called the {\bf Hamiltonian}\index{quantum mechanics!Hamiltonian}, and it has the physical meaning of `energy'.  The equation we just wrote down is then called {\bf Schr\"{o}dinger's equation}\index{quantum mechanics!Schr\"{o}dinger's equation}.

So, simply put, in quantum mechanics we have a correspondence between observables and nice one-parameter groups of symmetries.  Not surprisingly, our favorite observable, energy, corresponds to our favorite symmetry: time evolution!

However, if you were paying attention, you noticed that we carefully avoided explaining how we define $\exp(- i t H)$.  We didn't even say what a self-adjoint operator is.    This is where the technicalities come in: they arise when $H$ is \href{http://en.wikipedia.org/wiki/Unbounded\_operator}{unbounded}\index{unbounded operator}, and not defined on all vectors in our Hilbert space.

Luckily, these technicalities evaporate for \emph{finite-dimensional}  Hilbert spaces\index{Hilbert space!finite-dimensional}, such as $L^2(X)$ for a finite set $X$.  Then we get:

\begin{theorem}[{\bf Stone's Theorem---Baby Version}]\index{Stone's theorem!baby version}
  Suppose we are given a \emph{finite-dimensional}  Hilbert space.  In this case, a linear operator $H$ on this space is self-adjoint iff it's defined on the whole space and
$$ \langle \phi , H \psi \rangle = \langle H \phi, \psi \rangle $$
for all vectors $\phi, \psi$.  Given a strongly continuous\index{group theory!unitary!1-parameter} 1-parameter unitary group $U(t)$ we can always write
$$ U(t) = \exp(- i t H) $$
for a unique self-adjoint operator $H$, where
$$ \exp(-i t H) \psi = \sum_{n = 0}^\infty \frac{(-i t H)^n}{n!} \psi $$
with the sum converging for all $\psi$.  Conversely, any self-adjoint operator on our space determines a strongly continuous 1-parameter group this way.  For all vectors $\psi$ in our space we then have
$$ \left.\frac{d}{d t} U(t) \psi \right|_{t = 0} = -i H \psi  $$
and if we set
$$ \psi(t) = \exp(-i t H) \psi $$
we have
$$ \frac{d}{d t} \psi(t) = - i H \psi(t) $$
\end{theorem}

\subsection[Stochastic evolution]{Time evolution in stochastic mechanics}

We've seen that in quantum mechanics, time evolution is usually described by a 1-parameter group of operators that comes from an observable: the Hamiltonian.  Stochastic mechanics is different!   

First, since stochastic operators aren't usually invertible, we typically describe time evolution by a mere `semigroup':

\begin{definition}
A collection of linear operators $U(t)$ on some vector space,
where $t \in [0,\infty)$, forms a {\bf 1-parameter semigroup}\index{semigroup!1-parameter} if
$$ U(0) = 1 $$
and
$$ U(s+t) = U(s) U(t) $$
for all $s, t \ge 0$.
\end{definition} 

Now suppose this vector space is $L^1(X)$ for some measure space $X$.\index{measure space}  We want to focus on the case where the operators $U(t)$ are stochastic and depend continuously on $t$ in the same sense we discussed earlier.  

\begin{definition}
 A 1-parameter strongly continuous semigroup of stochastic operators $U(t)\colon L^1(X) \to L^1(X)$ is called a {\bf Markov semigroup}.\index{semigroup!Markov}
\index{Markov semigroup}
 \end{definition} 

What's the analogue of Stone's theorem for Markov semigroups?  We don't know a fully satisfactory answer!  If you know, please tell us.  In Section \label{sec:11_hille-yosida} we'll say what we \emph{do} know---we're not \emph{completely} clueless---but for now let's look at the `baby' case where $X$ is a finite set.  Then the story is neat and complete:

\begin{theorem}\index{infinitesimal stochastic operator|(} 
\index{operator!infinitesimal stochastic|(}
Suppose we are given a \emph{finite} set $X$.  In this case, a linear operator $H$ on $L^1(X)$ is {\bf infinitesimal stochastic} iff it's defined on the whole space,
$$ \int H \psi = 0 $$
for all $\psi \in L^1(X)$, and the matrix of $H$ in terms of the obvious basis obeys
$$ H_{i j} \ge 0 $$
for all $j \ne i$.  Given a Markov semigroup $U(t)$ on $L^1(X)$, we can always write
$$ U(t) = \exp(t H) $$
\index{semigroup!Markov} \index{Markov semigroup}
for a unique infinitesimal stochastic operator $H$, where
$$ \exp(t H) \psi = \sum_{n = 0}^\infty \frac{(t H)^n}{n!} \psi $$
with the sum converging for all $\psi$.  Conversely, any infinitesimal stochastic operator on our space determines a Markov semigroup this way.  For all $\psi \in L^1(X)$ we then have
$$ \left.\frac{d}{d t} U(t) \psi \right|_{t = 0} = H \psi  $$
and if we set
$$ \psi(t) = \exp(t H) \psi $$
we have the {\bf master equation}:
$$ \frac{d}{d t} \psi(t) = H \psi(t) $$
\end{theorem} 

In short, time evolution in stochastic mechanics is a lot like time evolution in quantum mechanics, except it's typically not invertible, and the Hamiltonian is typically \emph{not an observable}.  Why not?  Because we defined an observable to be a function $A \colon X \to \mathbb{R}$.  We can think of this as giving an operator on $L^1(X)$, namely the operator of multiplication by $A$.  That's a nice trick, which we used to good effect in Section~\ref{sec:10}.  However, at least when $X$ is a finite set, this operator will be diagonal in the obvious basis consisting of functions that equal 1 one point of $X$ and zero elsewhere.  So, it can only be infinitesimal stochastic if it's zero!

\begin{problem}\label{prob:17} 
If $X$ is a finite set, show that any operator on $L^1(X)$ that's both diagonal and infinitesimal stochastic must be zero.
\end{problem} 

\subsection[The Hille--Yosida theorem]{The Hille--Yosida theorem}\index{Hille--Yosida theorem|(}
\label{sec:11_hille-yosida}

We've now told you everything you really need to know about the broad analogy between quantum mechanics and stochastic mechanics... but not everything we want to say.   What happens when $X$ is not a finite set?   What are Markov semigroups like then?  We can't abide letting this question go unresolved!  Unfortunately we only know a partial answer.   We can get a certain distance using the `Hille--Yosida theorem', but this theorem does not completely characterize Markov semigroups.  Still, it's worth knowing.  

We need a few concepts to get started:

\begin{definition} 
A \href{http://en.wikipedia.org/wiki/Banach_space}{\bf Banach space} is a vector space with a \href{http://en.wikipedia.org/wiki/Norm\_\%28mathematics\%29\#Definition}{norm} such that any Cauchy sequence\index{Cauchy sequence} converges.
\end{definition}

\noindent
Examples include Hilbert spaces like $L^2(X)$ for any measure space\index{measure space}, but also other spaces like $L^1(X)$ for any measure space!

\begin{definition} 
If $V$ is a Banach space, a 1-parameter semigroup of operators $U(t)\colon V \to V$ is called a {\bf contraction semigroup} \index{semigroup!contraction}\index{contraction semigroup} if it is strongly continuous and 
$$ \| U(t) \psi \| \le \| \psi \| $$
for all $t \ge 0 $ and all $\psi \in V$. 
\end{definition} 

\noindent
Examples include strongly continuous 1-parameter unitary groups, but also Markov semigroups!\index{group theory!unitary!1-parameter}

\begin{problem}\label{prob:18} 
Show any Markov semigroup is a contraction semigroup. \index{semigroup!Markov}
\index{semigroup!contraction} \index{Markov semigroup} \index{contraction semigroup}
\end{problem} 

The Hille--Yosida theorem generalizes Stone's theorem to contraction semigroups:\index{Stone's theorem!generalization to contraction semigroups}\index{semigroup!contraction}  \index{contraction semigroup}

\begin{theorem}[{\bf  \href{http://en.wikipedia.org/wiki/Hille\%E2\%80\%93Yosida\_theorem}{Hille--Yosida Theorem}}] \index{Hille--Yosida theorem}
Given a contraction semigroup $U(t)$ we can always write 
$$ U(t) = \exp(t H)$$ 
for some \href{http://en.wikipedia.org/wiki/Densely\_defined\_operator}{densely defined} operator $H$ such that $H - \lambda I $ has an inverse and
$$ \displaystyle{ \| (H - \lambda I)^{-1} \psi \| \le \frac{1}{\lambda} \| \psi \|} $$
for all $\lambda > 0 $ and $\psi \in V$.  Conversely, any such operator determines a strongly continuous 1-parameter group.  For all vectors $\psi$ for which $H \psi$ is well-defined, we have
$$ \left.\frac{d}{d t} U(t) \psi \right|_{t = 0} = H \psi  $$
Moreover, for any of these vectors, if we set
$$ \psi(t) = U(t) \psi $$
we have
$$ \frac{d}{d t} \psi(t) = H \psi(t) $$
\end{theorem}

\noindent
If you like, you can take the stuff at the end of this theorem to be what we mean by saying $U(t) = \exp(t H)$.   When $ U(t) = \exp(t H)$, we say that $ H$ {\bf generates} the semigroup $ U(t)$.

But now suppose $V = L^1(X)$.  Besides the conditions in the Hille--Yosida theorem, what \emph{extra}  conditions on $H$ are necessary and sufficient for $H$ to generate a Markov semigroup?  In other words, what's a definition of `infinitesimal stochastic operator' that's suitable not only when $X$ is a finite set, but an arbitrary measure space?\index{measure space}

The best answer we know, not completely satisfactory, seems to be here:

\begin{enumerate}
\item[\cite{MR92}]  Zhi-Ming Ma and Michael R\"{o}ckner, \emph{Introduction to the Theory of (Non-Symmetric) Dirichlet Forms}, Springer, Berlin, 1992.
\end{enumerate}

\noindent 
This book provides a very nice self-contained proof of the
Hille--Yosida theorem.  On the other hand, it does \emph{not} answer
the question in general, but only when the skew-symmetric part of $H$
is dominated (in a certain sense) by the symmetric part.  So, we're
stuck on this front, but that needn't bring the whole project to a
halt.  We'll just sidestep this question.

For a good well-rounded introduction to Markov semigroups and what
they're good for, try:

\begin{enumerate}
\item[\cite{RPT02}]  Ryszard Rudnicki, Katarzyna Pich\`{o}r and Marta Tyran-Kam\`{i}nska, \href{http://www.impan.pl/~rams/r48-ladek.pdf}{Markov semigroups and their applications}.
\end{enumerate} 

\index{Hille--Yosida theorem|)} 

\subsection[Answers]{Answers}
\label{sec:11_answers}

\vskip 1em \noindent
{\bf Problem 15.} Suppose $X$ is a finite set.  Show that any isometry $U \colon L^2(X) \to L^2(X)$ is invertible, and its inverse is again an isometry.\index{quantum mechanics!isometry}

\begin{answer}
Remember that $U$ being an isometry means that it preserves the inner product:\index{quantum mechanics!isometry}
$$\langle U \psi, U \phi \rangle = \langle \psi, \phi \rangle $$
and thus it preserves the $L^2$ norm
$$\|U \psi \| = \| \psi \|  $$
given by $\| \psi \| = \langle \psi, \psi \rangle^{1/2}$.  It follows that if $U\psi = 0$, then $\psi = 0$, so $U$ is one-to-one.  Since $U$ is a linear operator from a \emph{finite-dimensional}  vector space to itself, $U$ must therefore also be onto.  Thus $U$ is invertible, and because $U$ preserves the inner product, so does its inverse: given $\psi, \phi \in L^2(X)$ we have
$$\langle U^{-1} \phi, U^{-1} \psi \rangle = \langle \phi, \psi \rangle $$
since we can write $\phi' = U^{-1} \phi,$ $\psi' = U^{-1} \psi$  and then the above equation says
$$ \langle \phi' , \psi' \rangle = \langle U \phi' , U \psi' \rangle $$
\end{answer} 

\vskip 1em \noindent {\bf Problem 16.} 
Suppose $X$ is a finite set.  Which stochastic operators $U \colon L^1(X) \to L^1(X)$ have stochastic inverses?

\begin{answer}
Suppose the set $ X$ has $ n$ points.  Then the set of stochastic states 
$$ S = \{ \psi \colon X \to \mathbb{R} \; : \; \psi \ge 0, \quad \int \psi = 1 \} $$
is a \href{http://en.wikipedia.org/wiki/Simplex}{simplex}.  That is, it's an equilateral triangle when $ n = 3$, a regular tetrahedron when $ n = 4$, and so on.
\end{answer} 

\begin{center}
 \includegraphics[width=30mm]{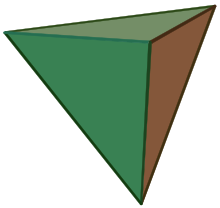}
\end{center}

\noindent
In general, $S$ has $ n$ corners, which are the functions $ \psi$ that equal 1 at one point of $ S$ and zero elsewhere.  Mathematically speaking, $S$ is a {\bf \href{http://en.wikipedia.org/wiki/Convex\_set}{convex set}}, and its corners are its {\bf \href{http://en.wikipedia.org/wiki/Extreme_point}{extreme points}}: the points that can't be written as {\bf \href{http://en.wikipedia.org/wiki/Convex\_combination}{convex combinations}} of other points of $S$ in a nontrivial way.

Any stochastic operator $ U$ must map $ S$ into itself, so if $U$ has an inverse that's also a stochastic operator, it must give a bijection $ U \colon S \to S$.   Any linear transformation acting as a bijection between convex sets must map extreme points to extreme points (this is easy to check), so $ U$ must map corners to corners in a bijective way.  This implies that it comes from a permutation of the points in $ X$.

In other words, any stochastic matrix with an inverse that's also stochastic is a {\bf \href{http://en.wikipedia.org/wiki/Permutation\_matrix}{permutation matrix}}: a square matrix with every entry 0 except for a single 1 in each row and each column.   \index{matrix!permutation} \index{permutation matrix}  For another proof, see the answer to Problem \ref{prob:23}, which appears in Section \ref{sec:16_answers}.

It is worth adding that there are lots of stochastic operators whose inverses are not, in general, stochastic.  We can see this in at least two ways.

First, for any measure space\index{measure space} $X$, every stochastic operator $U \colon L^1(X) \to L^1(X)$ that's `close to the identity' in this sense:
$$ \| U - I \| < 1 $$
(where the norm is the \href{http://en.wikipedia.org/wiki/Operator\_norm}{operator norm}) will be invertible, simply because \emph{every}  operator obeying this inequality is invertible!  After all, if this inequality holds, we have a convergent geometric series:
$$ \displaystyle{ U^{-1} = \frac{1}{I - (I - U)} = \sum_{n = 0}^\infty (I - U)^n} $$
Second, suppose $X$ is a finite set and $H$ is infinitesimal stochastic operator on $ L^1(X)$.  Then $H$ is bounded, so the stochastic operator $\exp(t H)$ where $t \ge 0$ will always have an inverse, namely $\exp(-t H)$.  But for $t$ sufficiently small, this inverse $\exp(-tH)$ will only be stochastic if $-H$ is infinitesimal stochastic, and that's only true if $H = 0$.

In something more like plain English: when you've got a finite set of states, you can formally run any Markov process backwards in time, but a lot of those `backwards-in-time' operators will involve negative probabilities for the system to hop from one state to another!

\vskip 1em \noindent {\bf Problem 17.}  
If $X$ is a finite set, show that any operator on $L^1(X)$ that's both diagonal and infinitesimal stochastic must be zero.

\begin{answer}
We are thinking of operators on $L^1(X)$ as matrices with respect to the obvious basis of functions that equal 1 at one point and 0 elsewhere.  If $H_{i j}$ is an infinitesimal stochastic matrix, the sum of the entries in each column is zero.  If it's diagonal, there's at most one nonzero entry in each column.  So, we must have $H = 0$.
\end{answer} 

\vskip 1em \noindent {\bf Problem 18.} 
Show any Markov semigroup $ U(t)\colon L^1(X) \to L^1(X)$ is a contraction semigroup.

\begin{answer}
We need to show
$$ \|U(t) \psi\| \le \| \psi \| $$
for all $t \ge 0$ and $\psi \in L^1(X)$.  Here the norm is the $L^1$ norm, so more explicitly we need to show
$$ \int |U(t) \psi | \le \int |\psi| $$
We can split $ \psi$ into its positive and negative parts:
$$\psi = \psi_+ - \psi_-$$
where
$$ \psi_{\pm} \ge 0$$
Since $ U(t)$ is stochastic we have
$$ U(t) \psi_{\pm} \ge 0$$
and
$$ \int U(t) \psi_\pm = \int \psi_\pm $$
so
$$ \begin{array}{ccl}  \int |U(t) \psi | &=& \int |U(t) \psi_+ - U(t) \psi_-| \\  &\leq & \int |U(t) \psi_+| + |U(t) \psi_-| \\  &=& \int  U(t) \psi_+ + U(t) \psi_-  \\ &=& \int \psi_+ + \psi_- \\ &=&  \int |\psi|  \end{array} $$
\end{answer} 

%%%%% SECTION 12 %%%%%%%

\newpage
\section[Noether's theorem: quantum vs stochastic]{Noether's theorem: quantum vs stochastic}
\label{sec:12}\index{Noether's theorem!quantum vs stochastic|(}  

In this chapter we merely want to show you the quantum and stochastic
versions of Noether's theorem, side by side.  Having made our
sacrificial offering to the math gods last time by explaining how
everything generalizes when we replace our finite set $X$ of states by
an infinite set or an even more general
\href{http://en.wikipedia.org/wiki/Measure\_space}{measure space},\index{measure space}
we'll now relax and state Noether's theorem only for a finite set.  If
you're the sort of person who finds that unsatisfactory, you can do
the generalization yourself.

\subsection{Two versions of Noether's theorem}

Let us write the quantum and stochastic Noether's theorem so they look almost the same:

\begin{theorem}\index{Noether's theorem!quantum version} 
\label{theorem:quantum-noether}  
Let $X$ be a finite set.  Suppose $H$ is a self-adjoint operator on $L^2(X)$, and let $O$ be an observable.  Then 
$$[O,H] = 0 $$
if and only if for all states $\psi(t)$ obeying Schr\"{o}dinger's equation
$$ \displaystyle{ \frac{d}{d t} \psi(t) = -i H \psi(t)} $$
the expected value of $O$ in the state $\psi(t)$ does not change with $t.$
\end{theorem}

\begin{theorem}  
Let $X$ be a finite set.  Suppose $H$ is an infinitesimal stochastic operator on $L^1(X)$, and let $O$ be an observable.  Then 
 $$[O,H] =0 $$
if and only if for all states $\psi(t)$ obeying the master equation
$$ \displaystyle{ \frac{d}{d t} \psi(t) = H \psi(t)} $$
the expected values of $O$ and $O^2$ in the state $\psi(t)$ do not change with $t.$
\end{theorem} 

This makes the big difference stick out like a sore thumb: in the quantum version we only need the expected value of $O$, while in the stochastic version we need the expected values of $O$ and $O^2$!

We've already proved the stochastic version of Noether's theorem.  Now let's do the quantum version.

\subsection[Proof]{Proof of the quantum version}

Our statement of the quantum version was silly in a couple of ways.  First, we spoke of the Hilbert space $L^2(X)$ for a finite set $X$, but any finite-dimensional Hilbert space will do equally well.  Second, we spoke of the ``self-adjoint operator" $H$ and the ``observable" $O$, but in quantum mechanics an observable is the \emph{same thing}  as a self-adjoint operator! 

Why did we talk in such a silly way?  Because we were attempting to emphasize the similarity between quantum mechanics and stochastic mechanics.   But they're somewhat different.  For example, in stochastic mechanics we have two very different concepts: infinitesimal stochastic operators, which \emph{generate symmetries}, and functions on our set $X,$ which \emph{are observables}.  But in quantum mechanics something wonderful happens: self-adjoint operators both \emph{generate symmetries}  and \emph{are observables!}   So, our attempt was a bit strained.  
\index{infinitesimal stochastic operator|)} 
\index{operator!infinitesimal stochastic|)}

Let us state and prove a less silly quantum version of Noether's theorem, which implies the one above:

\begin{theorem}\index{Noether's theorem!quantum version} 
Suppose $H$ and $O$ are self-adjoint operators on a finite-dimensional Hilbert space.  Then 
$$[O,H] = 0 $$
if and only if for all states $\psi(t)$ obeying Schr\"{o}dinger's equation\index{quantum mechanics!Schr\"{o}dinger's equation}
$$ \displaystyle{ \frac{d}{d t} \psi(t) = -i H \psi(t)} $$
the expected value of $O$ in the state $\psi(t)$ does not change with $t$:
$$ \displaystyle{ \frac{d}{d t} \langle \psi(t), O \psi(t) \rangle = 0} $$\end{theorem} 

\begin{proof} 
The trick is to compute the time derivative we just wrote down.  Using Schr\"{o}dinger's equation\index{quantum mechanics!Schr\"{o}dinger's equation}, the product rule, and the fact that $H$ is self-adjoint we get: 
$$\begin{array}{ccl}  \displaystyle{ \frac{d}{d t} \langle \psi(t), O \psi(t) \rangle} &=& 
\langle -i H \psi(t) , O \psi(t) \rangle + \langle \psi(t) , O (- i H \psi(t)) \rangle \\  \\
% do we have an extra )? 
&=& i \langle \psi(t) , H O \psi(t) \rangle -i \langle \psi(t) , O H \psi(t) \rangle \\  \\
%&=& i \langle \psi(t) , H O \psi(t) \rangle -i \langle \psi(t) , O H \psi(t)) \rangle \\  \\
&=& - i \langle \psi(t), [O,H] \psi(t) \rangle  \end{array} $$
So, if $[O,H] = 0$, clearly the above time derivative vanishes.  Conversely, if this time derivative vanishes for all states $\psi(t)$ obeying Schr\"{o}dinger's equation\index{quantum mechanics!Schr\"{o}dinger's equation}, we know 
$$\langle \psi, [O,H] \psi \rangle = 0 $$
for all states $\psi$ and thus all vectors in our Hilbert space.  Does this imply $[O,H] = 0$?  Yes, because $i$ times a commutator of a self-adjoint operator is self-adjoint, and for any self-adjoint operator $A$ we have
$$\forall \psi  \;  \langle \psi, A \psi \rangle = 0 \qquad \Rightarrow \qquad A = 0 $$
This is a well-known fact whose proof goes like this.  Assume $ \langle \psi, A \psi \rangle = 0$ for all $\psi.$ Then to show $A = 0,$ it is enough to show $\langle \phi, A \psi \rangle = 0$ for all $\phi$ and $\psi$.  But we have a marvelous identity:
$$ \begin{array}{ccl} \langle \phi, A \psi \rangle &=& \frac{1}{4} \left( \langle \psi + \phi, \, A (\psi + \phi) \rangle  -  \langle \psi - \phi, \, A (\psi - \phi) \rangle \right. \\ && \left. +i \langle \psi + i \phi, \, A (\psi + i \phi) \rangle  -  i\langle \psi - i \phi, \, A (\psi - i \phi) \rangle \right) \end{array} $$
and all four terms on the right vanish by our assumption.   
\end{proof} 

The marvelous identity\index{polarization identity} up there is called the {\bf \href{http://en.wikipedia.org/wiki/Polarization\_identity}{polarization identity}}.  In plain English, it says: if you know the diagonal entries of a self-adjoint matrix in every basis, you can figure out \emph{all}  the entries of that matrix in every basis.  Why is it called the `polarization identity'?  Presumably because it shows up in optics, in the study of polarized light.

\subsection[Comparison]{Comparison}

In both the quantum and stochastic cases, the time derivative of the expected value of an observable $O$ is expressed in terms of its commutator with the Hamiltonian.  In the quantum case we have
$$ \displaystyle{ \frac{d}{d t} \langle \psi(t), O \psi(t) \rangle = - i \langle \psi(t), [O,H] \psi(t) \rangle} $$
and for the right side to \emph{always}  vanish, we need $[O,H] = 0 $, thanks to the polarization identity.\index{polarization identity}  In the stochastic case, a perfectly analogous equation holds:
$$\displaystyle{ \frac{d}{d t} \int O \psi(t) = \int [O,H] \psi(t)} $$
but now the right side can always vanish even without $[O,H] = 0.$  We saw a counterexample in Section~\ref{sec:10}.  There is nothing like the polarization identity to save us!  To get $[O,H] = 0$ we need a supplementary hypothesis, for example the vanishing of 
$$\displaystyle{ \frac{d}{d t} \int O^2 \psi(t)} $$
Okay!  Starting in the next Section we'll change gears and look at some more examples of stochastic Petri nets and Markov processes, including some from chemistry.  After some more of that, we'll move on to networks of other sorts.  There's a really big picture here, and we're afraid we've been getting caught up in the details of a tiny corner.

\index{Noether's theorem!quantum vs stochastic|)}  

%%%%%% SECTION 13 %%%%%%%

\newpage
\section[Chemistry and the Desargues graph]{Chemistry and the Desargues graph}\index{chemistry!Desargues graph}
\label{sec:13}

We've been doing a lot of hard work lately.  Let's take a break and
think about a fun example from chemistry!

\subsection{The ethyl cation}\index{chemistry!ethyl cation} 

Suppose you start with a molecule of \href{http://en.wikipedia.org/wiki/Ethane}{ethane}, which has 2 carbons and 6 hydrogens arranged like this:\index{chemistry!ethane} 

\begin{center}
\includegraphics[width=50mm]{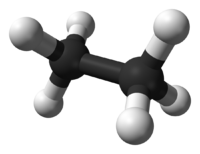}
\end{center}

\noindent 
Then suppose you remove one hydrogen.  The result is a positively charged ion, or `\href{http://en.wikipedia.org/wiki/Ion}{cation}'.  It would be cute if the opposite of a cation was called a `dogion'.  Alas, it's \href{http://en.wikipedia.org/wiki/Ion\#Anions_and_cations}{not}.

This particular cation, formed from removing one hydrogen from an ethane molecule,\index{chemistry!ethyl cation} 
is called an `ethyl cation'.  People used to think it looked like this:

\begin{center}
 \includegraphics[width=50mm]{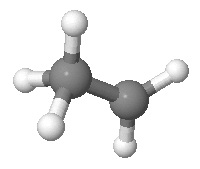}
\end{center}

They also thought a hydrogen could hop from the carbon with 3 hydrogens attached to it to the carbon with 2.  So, they drew a graph with a vertex for each way the hydrogens could be arranged, and an edge for each hop.  It looks really cool:

\begin{center}
 \includegraphics[width=80mm]{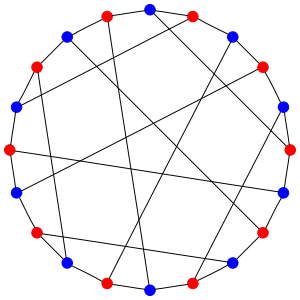}
\end{center}

\noindent
The red vertices come from arrangements where the \emph{first}  carbon has 2 hydrogens attached to it, and the blue vertices come from those where the \emph{second}  carbon has 2 hydrogens attached to it.  So, each edge goes between a red vertex and a blue vertex.  

This graph has 20 vertices, which are arrangements or `states' of the ethyl cation.  It has 30 edges, which are hops or `transitions'.  Let's see why those numbers are right.  

First we need to explain the rules of the game.  The rules say that the 2 carbon atoms are distinguishable: there's a `first' one and a `second' one.  The 5 hydrogen atoms are also distinguishable.  But, all we care about is which carbon atom each hydrogen is bonded to: we don't care about further details of its location.  And we require that 2 of the hydrogens are bonded to one carbon, and 3 to the other.

If you're a physicist, you may wonder why the rules work this way: after all, at a fundamental level, identical particles \emph{aren't}  really distinguishable.  We're afraid we can't give a fully convincing explanation right now: we're just reporting the rules!

Given these rules, there are 2 choices of which carbon has two hydrogens attached to it. Then there are 
$$ \displaystyle{ {5 \choose 2} = \frac{5 \times 4}{2 \times 1} = 10} $$ 
choices of which two hydrogens are attached to it.  This gives a total of $2 \times 10 = 20$ states.  These are the vertices of our graph: 10 red and 10 blue.

The edges of the graph are transitions between states.  Any hydrogen in the group of 3 can hop over to the group of 2.  There are 3 choices for which hydrogen atom makes the jump.  So, starting from any vertex in the graph there are 3 edges.  This means there are $3 \times 20 / 2 = 30$ edges.  

Why divide by 2?  Because each edge touches 2 vertices.  We have to avoid double-counting them.

\subsection{The Desargues graph}

The idea of using this graph in chemistry goes back to this paper:

\begin{itemize} 
 \item[\cite{BFB66}] A. T. Balaban, D. F\v{a}rca\c{s}iu, and R. B\v{a}nic\v{a}, Graphs of multiple 1,2-shifts in carbonium ions and related systems, \emph{Rev. Roum. Chim.\ }{\bf 11} (1966), 1205.
 \end{itemize} 
\noindent
This paper is famous because it was the first to use graphs in chemistry to describe molecular transitions, as opposed to using them as pictures of molecules!  But this particular graph was already famous for other reasons.  It's called the {\bf Desargues--Levi graph}, or {\bf Desargues graph} for short: \index{graph theory!Desargues graph|(}\index{Desargues graph|(}

\begin{itemize} 
 \item \href{http://en.wikipedia.org/wiki/Desargues\_graph}{Desargues graph}, Wikipedia.
\end{itemize} 

\noindent
Later we'll say why it's called this.  

There are lots of nice ways to draw the Desargues graph.  For example:

\begin{center}
 \includegraphics[width=80mm]{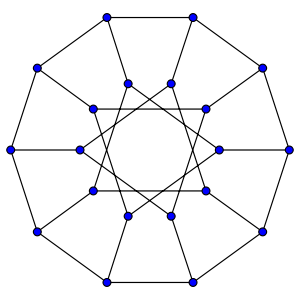}
\end{center}

\noindent
The reason why we can draw such pretty pictures is that the Desargues graph is very symmetrical.  Clearly any permutation of the 5 hydrogens acts as a symmetry of the graph, and so does any permutation of the 2 carbons.  This gives a symmetry group $S_5 \times S_2$, which has $5! \times 2! = 240$ elements.  And in fact this turns out to be the full symmetry group of the Desargues graph.

The Desargues graph, its symmetry group, and its applications to chemistry are discussed here:
\begin{itemize} 
 \item[\cite{Ran97}] Milan Randic, Symmetry properties of graphs of interest in chemistry: II: Desargues--Levi graph, \href{http://onlinelibrary.wiley.com/doi/10.1002/qua.560150610/abstract;jsessionid=0FB39FF9D7D985C1B6454CF4E30D9E92.d03t02}{\emph{Int. Jour. Quantum Chem.}}  {\bf 15} (1997), 663--682.
\end{itemize} 

\subsection{The ethyl cation, revisited}\index{chemistry!ethyl cation|(} 

We can try to describe the ethyl cation using probability theory.  If at any moment its state corresponds to some vertex of the Desargues graph, and it hops randomly along edges as time goes by, it will trace out a random walk\index{random walk!Desargues graph} on the Desargues graph.  This is a nice example of a \href{http://en.wikipedia.org/wiki/Continuous-time\_Markov\_process}{Markov process}!  

We could also try to describe the ethyl cation using quantum mechanics.   Then, instead of having a \emph{probability}  of hopping along an edge, it has an \emph{amplitude}  of doing so.   But as we've seen, a lot of similar math will still apply.  

It should be fun to compare the two approaches.  But we bet you're wondering which approach is correct.  This is a somewhat tricky question.  The answer would seem to depend on how much the ethyl cation is interacting with its environment---for example, bouncing off other molecules.  When a system is interacting a lot with its environment, a probabilistic approach seems to be more appropriate.  The relevant buzzword is \href{http://en.wikipedia.org/wiki/Quantum\_decoherence}{`environmentally induced decoherence'}.\index{quantum mechanics!decoherence}    
 
However, there's something much more basic we have tell you about.

After the paper by Balaban, Farcasiu, and Banic came out, people gradually realized that the ethyl cation \emph{doesn't really look like the drawing we showed you!}   It's what chemists call `nonclassical' ion.\index{chemistry!nonclassical ion}  What they mean is this: its actual structure is not what you get by taking the traditional ball-and-stick model of an ethane molecule and ripping off a hydrogen.  The ethyl cation really looks like this:

%\todo{[PIC ethyl\_cation\_nonclassical.jpg]}

\begin{center}
 \includegraphics[width=50mm]{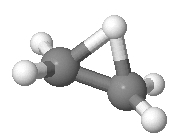}
\end{center}

\noindent

This picture comes from here:
\begin{itemize} 
 \item Stephen Bacharach, Ethyl cation, \href{http://comporgchem.com/blog/?p=60}{Computational Organic Chemistry}.
\end{itemize} 
If you go to this website, you'll find more details, and pictures that you can actually rotate.

\subsection{Trigonal bipyramidal molecules}

In short: if we stubbornly insist on applying the Desargues graph to realistic chemistry, the ethyl cation won't do: we need to find some other molecule to apply it to.  Luckily, there are lots of options!  They're called \href{http://en.wikipedia.org/wiki/Trigonal\_bipyramidal\_molecular\_geometry}{trigonal bipyramidal molecules}.  They look like this:

\begin{center}
 \includegraphics[width=40mm]{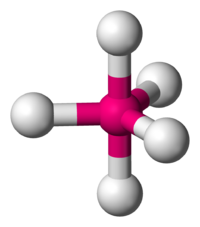}
\end{center}

\noindent
The 5 balls\index{chemistry!ligand} on the outside are called `ligands': they could be atoms or bunches of atoms.   In chemistry, `\href{http://en.wikipedia.org/wiki/Ligand}{ligand}' just means something that's stuck onto a central thing.   For example, in \href{http://en.wikipedia.org/wiki/Phosphorus\_pentachloride}{phosphorus pentachloride} the ligands are chlorine atoms, all attached to a central phosphorus atom:

\begin{center}
 \includegraphics[width=40mm]{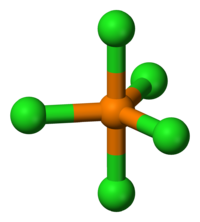}
\end{center}

\noindent
It's a colorless solid\index{chemistry!solid}, but as you might expect, it's pretty \href{http://www.cbwinfo.com/Chemical/Precursors/p36.html}{nasty stuff}: it's not flammable, but it reacts with water or heat to produce toxic chemicals like hydrogen chloride\index{chemistry!hydrogen chloride}.

Another example is \href{http://en.wikipedia.org/wiki/Iron\_pentacarbonyl}{iron pentacarbonyl}, where 5 carbon-oxygen ligands are attached to a central iron atom:

\begin{center}
 \includegraphics[width=40mm]{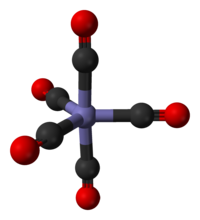}
\end{center}

\noindent
You can make this stuff by letting powdered iron react with carbon monoxide.  It's a straw-colored liquid with a pungent smell!

Whenever you've got a molecule of this shape, the ligands come in two kinds.  There are the 2 `axial' ones, and the 3 `equatorial' ones:

\begin{center}
 \includegraphics[width=40mm]{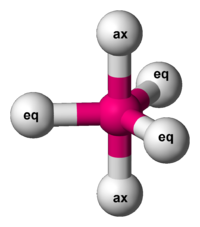}
\end{center}

The molecule has 20 states... but only if count the states a certain way.   We have to treat all 5 ligands as distinguishable, but think of two arrangements of them as the same if we can rotate one to get the other.   The trigonal bipyramid has a rotational symmetry group with 6 elements.  So, there are 5! / 6 = 20 states.  

The transitions between states are devilishly tricky.  They're called {\bf pseudorotations}, and they look like this:

\begin{center}
 \includegraphics[width=80mm]{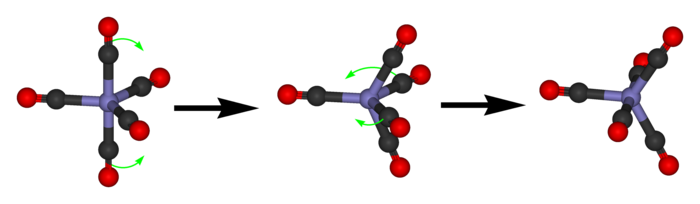}
\end{center}

\noindent
If you look very carefully, you'll see what's going on.  First the 2 axial ligands move towards each other to become equatorial.  Now the equatorial ones are no longer in the horizontal plane: they're in the plane facing us!   Then 2 of the 3 equatorial ones swing out to become axial.    This fancy dance is called the \href{http://en.wikipedia.org/wiki/Berry\_mechanism}{Berry pseudorotation mechanism}.

To get from one state to another this way, we have to pick 2 of the 3 equatorial ligands to swing out and become axial.  There are 3 choices here.   So, if we draw a graph with states as vertices and transitions as edges, it will have 20 vertices and $20 \times 3 / 2 = 30$ edges.  That sounds suspiciously like the  Desargues graph\index{graph theory!Desargues graph}!

\begin{problem}\label{prob:19} 
Show that the graph with states of a trigonal bipyramidal molecule as vertices and pseudorotations as edges is indeed the Desargues graph.
\end{problem} 

It seems this fact was first noticed here:

\begin{enumerate} 
 \item[\cite{LR68}] Paul C. Lauterbur and Fausto Ramirez, Pseudorotation in trigonal-bipyramidal molecules, \href{http://pubs.acs.org/doi/abs/10.1021/ja01026a029}{\emph{J. Am. Chem. Soc.}} {\bf 90} (1968), 6722.6726.
\end{enumerate} 

In the next section we'll say more about the Markov process or quantum process corresponding to a random walk on the Desargues graph.  But since the Berry pseudorotation mechanism is so hard to visualize, we'll \emph{pretend} that the ethyl cation looks like this:

\begin{center}
 \includegraphics[width=50mm]{ethyl_cation_classical.jpg}
\end{center}

\noindent
and we'll use this picture to help us think about the Desargues graph.

That's okay: everything we'll figure out can easily be translated to apply to the real-world situation of a trigonal bipyramidal molecule.  The virtue of math is that when two situations are `mathematically the same', or `isomorphic', we can talk about either one, and the results automatically apply to the other.  This is true even if the one we talk about doesn't actually exist in the real world!

\subsection{Drawing the  Desargues graph}

Before we quit, let's think a bit more about how we draw the Desargues graph.  For this it's probably easiest to go back to our naive model of an ethyl cation:

\begin{center}
 \includegraphics[width=50mm]{ethyl_cation_classical.jpg}
\end{center}

Even though ethyl cations don't really look like this, and we should be talking about some trigonal bipyramidal molecule instead, it won't affect the math to come.  Mathematically, the two problems are isomorphic!  So let's stick with this nice simple picture.

We can be a bit more abstract, though.  A state of the ethyl cation is like having 5 balls, with 3 in one pile and 2 in the other.   And we can focus on the first pile and forget the second, because whatever isn't in the first pile must be in the second.

Of course a mathematician calls a pile of things a `set', and calls those things `elements'.  So let's say we've got a set with 5 elements.  Draw a red dot for each 2-element subset, and a blue dot for each 3-element subset.   Draw an edge between a red dot and a blue dot whenever the 2-element subset is \emph{contained}  in the 3-element subset.  We get the Desargues graph.

That's true by definition.  But we never proved that this graph looks like this:

\begin{center}
 \includegraphics[width=80mm]{300px-Desargues_graph_2COL.png}
\end{center}

We won't now, either.  We'll just leave it as a problem:

\begin{problem}\label{prob:20} 
If we define the Desargues graph to have vertices corresponding to 2- and 3-element subsets of a 5-element set, with an edge between vertices when one subset is contained in another, why does it look like the picture above?
\end{problem} 

To draw a picture we know is correct, it's actually easier to start with a big graph that has vertices for \emph{all}  the subsets of our 5-element set.   If we draw an edge whenever an $ n$-element subset is contained in an $ (n+1)$-element subset, the Desargues graph will be sitting inside this big graph.   

Here's what the big graph looks like:

\begin{center}
 \includegraphics[width=60mm]{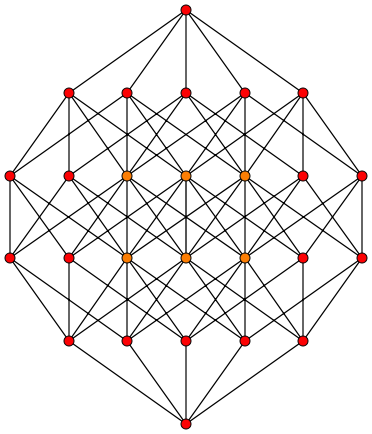}
\end{center}
%  This graph is missing an edge at left?!!

\noindent
This graph has $ 2^5$ vertices.  It's actually a picture of a \href{http://en.wikipedia.org/wiki/Penteract}{5-dimensional hypercube}!
Unfortunately, to make the picture pretty, some vertices are hidden behind others: each orange vertex has another behind it.   The vertices are arranged in rows.  There's

\begin{itemize} 
 \item one 0-element subset,
 \item five 1-element subsets,
 \item ten 2-element subsets,
 \item ten 3-element subsets,
 \item five 4-element subsets,
 \item one 5-element subset.
\end{itemize} 

\noindent
So, the numbers of vertices in each row should go like this:
$$ 1 \quad 5 \quad 10 \quad 10 \quad 5 \quad 1 $$
which is a row in \href{http://en.wikipedia.org/wiki/Pascal\%27s\_triangle}{Pascal's triangle}.    We get the Desargues graph if we keep only the vertices corresponding to 2- and 3-element subsets, like this:

\begin{center}
 \includegraphics[width=0.8\textwidth]{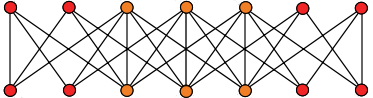}
\end{center}

\noindent
However, each orange vertex has another hiding behind it, so we only see 7 vertices in each row here, not 10.   Laurens Gunnarsen fixed this, and got the following picture of the Desargues graph:

\begin{center}
 \includegraphics[width=0.8\textwidth]{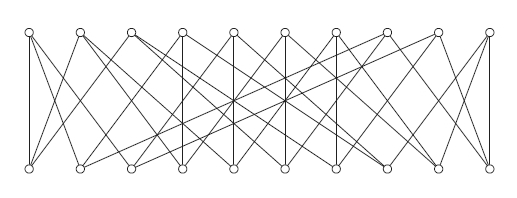}
\end{center}

\noindent
It's less pretty than our earlier picture of the Desargues graph, but at least there's no mystery to it.   Also, it shows that the Desargues graph can be generalized in various ways.  For example, there's a theory of \href{http://mathworld.wolfram.com/BipartiteKneserGraph.html}{bipartite Kneser graphs} $ H(n,k).$  The Desargues graph is $ H(5,2)$.

\subsection{Desargues' theorem}\index{graph theory!Desargues' theorem}

Finally, we can't resist answering this question: why is this graph called the `Desargues graph'?  This name comes from \href{http://en.wikipedia.org/wiki/Desargues\%27\_theorem}{Desargues' theorem}, a famous result in projective geometry\index{projective geometry}.  Suppose you have two triangles ABC and abc, like this:

\begin{center}
 \includegraphics[width=80mm]{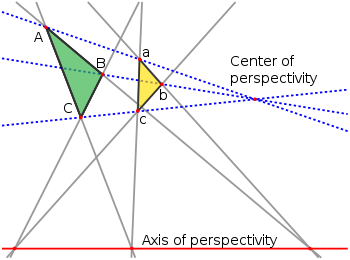}
\end{center}

\noindent
Suppose the lines Aa, Bb, and Cc all meet at a single point, the `center of perspectivity'.  Then the point of intersection of ab and AB, the point of intersection of ac and AC, and the point of intersection of bc and BC all lie on a single line, the `axis of perspectivity'.  The converse is true too.  Quite amazing!  

\begin{definition}The {\bf Desargues configuration}\index{Desargues configuration} consists of all the actors in this drama: 
\begin{itemize} 
\item 10 points: A, B, C, a, b, c, the center of perspectivity, and the three points on the axis of perspectivity and 
\item 10 lines: Aa, Bb, Cc, AB, AC, BC, ab, ac, bc and the axis of perspectivity
\end{itemize} 
\end{definition} 

Given any configuration of points and lines, we can form a graph called its \href{http://en.wikipedia.org/wiki/Levi\_graph}{{\bf Levi graph}} by drawing a vertex for each point or line, and drawing edges to indicate which points lie on which lines.  And now for the punchline: the Levi graph of the Desargues configuration is the `Desargues--Levi graph'!---or Desargues graph, for short.   \index{graph theory!Levi graph} \index{Levi graph}

Alas, we don't know how this is relevant to anything we've discussed.  For now it's just a tantalizing curiosity.

\subsection{Answers}

We didn't get any solutions of Problem \ref{prob:19} the first time we posted this on the Azimuth blog, so we reposed it on \href{http://johncarlosbaez.wordpress.com/2011/10/23/a-math-problem-coming-from-chemistry/}{Azimuth} and the \href{http://golem.ph.utexas.edu/category/2011/10/a\_math\_problem\_coming\_from\_chem.html}{\emph{n}-Category Caf\'{e}}.  That brought out lots of great solutions, which we've used to compose the answer here.  Special thanks go especially to J.\ M.\ Allen, Twan van Laarhoven, Peter Morgan and Johan Swanljung.  

\vskip 1em \noindent {\bf Problem 19.} 
Show that the graph with states of a trigonal bipyramidal molecule as vertices and pseudorotations as edges is indeed the Desargues graph.

\begin{answer}
To be specific, let's use iron pentacarbonyl as our example of a trigonal bipyramidal molecule:
\end{answer} 

\begin{center}
 \includegraphics[width=40mm]{200px-Iron-pentacarbonyl-from-xtal-3D-balls.png}
\end{center}

\noindent
It suffices to construct a 1-1 correspondence between the states of this molecule and those of the ethyl cation, such that two states of this molecule are connected by a transition if and only if the same holds for the corresponding states of the ethyl cation.

Here's the key idea: the ethyl cation has 5 hydrogens, with 2 attached to one carbon and 3 attached to the other.  Similarly, the trigonal bipyramidal molecule has 5 carbonyl grops, with 2 axial and 3 equatorial.  We'll use this resemblance to set up our correspondence.

There are various ways to describe states of the ethyl cation, but this is the best for us.  Number the hydrogens 1,2,3,4,5.  Then a state of the ethyl cation consists of a partition of the set $\{1,2,3,4,5\}$ into a 2-element set and a 3-element set, together with one extra bit of information, saying which carbon has 2 hydrogens attached to it.  This extra bit is the color here:

\begin{center}
 \includegraphics[width=80mm]{300px-Desargues_graph_2COL.png}
\end{center}

What do transitions look like in this description?  When a transition occurs, two hydrogens that belonged to the 3-element set now become part of the 2-element set.  Meanwhile, both hydrogens that belonged to the 2-element set now become part of the 3-element set.  (Ironically, the one hydrogen that hops is the one that stays in the 3-element set.)  Moreover, the extra bit of information changes.  That's why every edge goes from a red dot to a blue one, or vice versa.

So, to solve the problem, we need to show that the same description also works for the states and transitions of iron pentacarbonyl!

In other words, we need to describe its states as ways of partitioning the set {1,2,3,4,5} into a 2-element set and a 3-element set, together with one extra bit of information.  And we need its transitions to switch two elements of the 2-element set with two of the 3-element set, while changing that extra bit.

To do this, number the carbonyl groups 1,2,3,4,5.  The 2-element set consists of the axial ones, while the 3-element set consists of the equatorial ones.  When a transition occurs, two of axial ones trade places with two of the equatorial ones, like this:

\begin{center}
 \includegraphics[width=119.0625mm]{700px-Iron-pentacarbonyl-Berry-mechanism.png}
\end{center}

So, now we just need to figure out what that extra bit of information is, and why it always \emph{changes}  when a transition occurs.

Here's how we calculate that extra bit.  Hold the iron pentacarbonyl molecule vertically with one of the equatorial carbonyl groups pointing to your left.  Remember, the carbonyl groups are numbered.  So, write a list of these numbers, say (a,b,c,d,e), where a is the top axial one, b,c,d are the equatorial ones listed in counterclockwise order starting from the one pointing left, and e is the bottom axial one.  This list is some permutation of the list (1,2,3,4,5).  Take the sign of this permutation to be our bit!

Let's do an example:

\begin{center}
 \includegraphics[width=35mm]{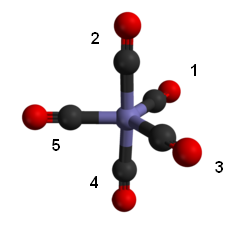}
\end{center}

\noindent
Here we get the list (2,5,3,1,4) since 2 is on top, 4 is on bottom, and 5,3,1 are the equatorial guys listed counterclockwise starting from the one at left.  The list (2,5,3,1,4) is an odd permutation of (1,2,3,4,5), so our bit of information is \emph{odd} .

Of course we must check that this bit is well-defined: namely, that it doesn't change if we \emph{rotate}  the molecule.  Rotating it a third of a turn gives an even permutation of the equatorial guys and leaves the axial ones alone, so this is an even permutation.  Flipping it over gives an odd permutation of the equatorial guys, but it also gives an odd permutation of the axial ones, so this too is an even permutation.  So, rotating the molecule doesn't change the sign of the permutation we compute from it.  The sign is thus a well-defined function of the state of the molecule.

Next we must to check that this sign changes whenever our molecule undergoes a transition.  For this we need to check that any transition changes our list of numbers by an odd permutation.  Since all transitions are conjugate in the permutation group, it suffices to consider one example:

\begin{center}
 \includegraphics[width=119.0625mm]{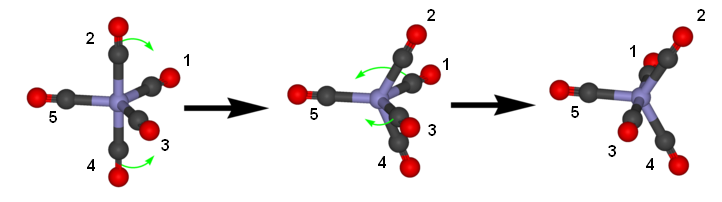}
\end{center}

Here we started with a state giving the list (2,5,3,1,4).  The transition ttakes us to a state that gives the list (3,5,4,2,1) if we hold the molecule so that 3 is pointing up and 5 to the left.  The reader can check that going from one list to another requires an \emph{odd}  permutation.  So we're done.   

Tracy Hall's discussion of the above problem was very interesting and erudite.  He also addressed an interesting subsidiary problem, namely: what graph do we get if we discard that extra bit of information that says which carbon in the ethyl cation has 3 hydrogens attached to it and which has 2?  The answer is the \href{http://en.wikipedia.org/wiki/Petersen\_graph}{Petersen graph}:

%\todo{[PIC petersen\_graph.png]}

\begin{center}
 \includegraphics[width=30mm]{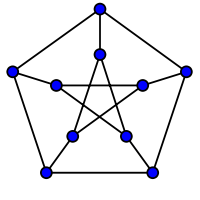}
\end{center}%\href{http://en.wikipedia.org/wiki/Petersen_graph" rel="nofollow}{<img src="http://upload.wikimedia.org/wikipedia/commons/thumb/9/91/Petersen1_tiny.svg/200px-Petersen1_tiny.svg.png" />
%}

Hall wrote:

\begin{quote}
As some comments have pointed out over on Azimuth,\index{Azimuth Project!blog} in both cases there are ten underlying states which simply pick out two of the five pendant atoms as special, together with an extra parity bit (which can take either value for any of the ten states), giving twenty states in total. The correspondence of the ten states is clear: an edge exists between state A and state B, in either case, if and only if the two special atoms of state A are disjoint from the two special atoms of state B. This is precisely one definition of the \href{http://en.wikipedia.org/wiki/Petersen\_graph}{Petersen graph} (a famous 3-valent graph on 10 vertices that shows up as a small counterexample to lots of naive conjectures). Thus the graph in either case is a double cover of the Petersen graph---but that does not uniquely specify it, since, for example, both the Desargues graph and the dodecahedron graph are double covers of the Petersen graph.

For a labeled graph, each double cover corresponds uniquely to an element of the $\mathbb{Z}/2\mathbb{Z}$ cohomology of the graph (for an unlabeled graph, some of the double covers defined in this way may turn out to be isomorphic). Cohomology over $\mathbb{Z}/2\mathbb{Z}$ takes any cycle as input and returns either 0 or 1, in a consistent way (the output of a $\mathbb{Z}/2\mathbb{Z}$ sum of cycles is the sum of the outputs on each cycle). The double cover has two copies of everything in the base (Petersen) graph, and as you follow all the way around a cycle in the base, the element of cohomology tells you whether you come back to the same copy (for 0) or the other copy (for 1) in the double cover, compared to where you started.

One well-defined double cover for any graph is the one which simply switches copies for every single edge (this corresponding to the element of cohomology which is 1 on all odd cycles and 0 on all even cycles). This always gives a double cover which is a bipartite graph, and which is connected if and only if the base graph is connected and not bipartite. So if we can show that in both cases (the fictitious ethyl cation and phosphorus pentachloride) the extra parity bit can be defined in such a way that it switches on every transition, that will show that we get the Desargues graph in both cases.

The fictitious ethyl cation is easy: the parity bit records which carbon is which, so we can define it as saying which carbon has three neighbors. This switches on every transition, so we are done. Phosphorus pentachloride is a bit trickier; the parity bit distinguishes a labeled molecule from its mirror image, or enantiomer. As has already been pointed out on both sites, we can use the parity of a permutation to distinguish this, since it happens that the orientation-preserving rotations of the molecule, generated by a three-fold rotation acting as a three-cycle and by a two-fold rotation acting as a pair of two-cycles, are all even permutations, while the mirror image that switches only the two special atoms is an odd permutation. The pseudorotation can be followed by a quarter turn to return the five chlorine atoms to the five places previously occupied by chlorine atoms, which makes it act as a four-cycle, an odd permutation. Since the parity bit in this case also can be defined in such a way that it switches on every transition, the particular double cover in each case is the Desargues graph---a graph one might be surprised to come across here. %, since just this past week I have been working out some combinatorial matrix theory for the same graph!

The five chlorine atoms in phosphorus pentachloride lie in six triangles which give a triangulation of the 2-sphere, and another way of thinking of the pseudorotation is that it corresponds to a \href{http://en.wikipedia.org/wiki/Pachner\_moves}{Pachner move} or bistellar flip on this triangulation---in particular, any bistellar flip on this triangulation that preserves the number of triangles and the property that all vertices in the triangulation have degree at least three corresponds to a pseudorotation as described.
\end{quote}

The next problem was both posed and answered by Greg Egan \href{http://johncarlosbaez.wordpress.com/2011/10/26/network-theory-part-15/\#comment-8869}{on Azimuth}. \index{Egan, Greg}

\vskip 1em \noindent {\bf Problem 20.} 
 If we define the Desargues graph to have vertices corresponding to 2- and 3-element subsets of a 5-element set, with an edge between vertices when one subset is contained in another, why does it look like this picture?

\begin{center}
 \includegraphics[width=80mm]{300px-Desargues_graph_2COL.png}
\end{center}

\begin{answer}
 For $i=0,...,4$, and with all the additions below modulo 5, define five pairs of red dots as:
$$ \{i, i+1\}, \{i+1,i+4\}$$
and five pairs of blue dots as:
$$ \{i, i+1, i+4\}, \{i+1, i+2, i+4\}.$$
The union of the $ i$th pair of red dots is the first of the $ i$th pair of blue dots, and the union of the second of the $ i$th pair of red dots and the first of the $ (i+1)$th pair of red dots is the second of the $ i$th pair of blue dots.  So if we form a 20-sided polygon whose vertices are alternating red and blue dots generated in this order, all the edges of the polygon will join red dots to blue dots of which they are subsets:

\begin{center}
 \includegraphics[width=105mm]{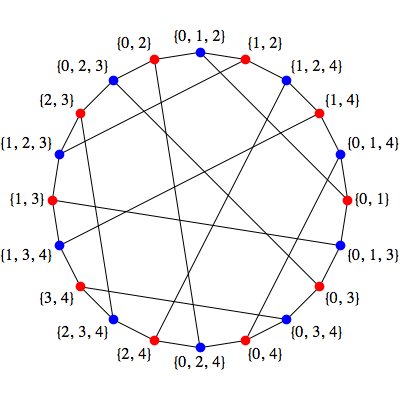}
\end{center}

\noindent
The first red dot of the $i$th pair is also a subset of the first blue dot of the $(i+1)$th pair:
$$\{i+1, i+2, i\}$$
which gives the five short chords in the picture, while the second red dot of the $ i$th pair is a subset of the second blue dot of the $(i+2)$ pair:
$$\{i+3, i+4, i+1\}$$
which gives the five long chords in the picture. 
\end{answer} 

\index{chemistry!ethyl cation|)} 
\index{graph theory!Desargues graph|)}\index{Desargues graph|)}

%%%%%% SECTION 14 %%%%%%%

\newpage
\section[Graph Laplacians]{Graph Laplacians}\index{Laplacian!graph|(} 
\label{sec:14}

In Section \ref{sec:13} we saw how to get a graph whose vertices are states of a molecule and whose edges are transitions between states.  We focused on two beautiful but not completely realistic examples that both give rise to the same highly symmetrical graph: the `Desargues graph'.  

In this section we'll consider how a molecule can carry out a random walk\index{random walk!Desargues graph} on this graph.  Then we'll get to work showing how \emph{any} graph gives:

\begin{itemize} 
\item A Markov process, namely a random walk\index{random walk!as Markov processes} on the graph.
\item A quantum process, where instead of having a \emph{probability} to hop from vertex to vertex as time passes, we have an \emph{amplitude}.\index{random walk!quantum walk}
\end{itemize} 

The trick is to use an operator called the `graph Laplacian', a discretized version of the Laplacian which happens to be both infinitesimal stochastic and self-adjoint.\index{infinitesimal stochastic!Laplacian}    By our results in in Section~\ref{sec:11}, such an operator will give rise \emph{both}  to a Markov process \emph{and} a quantum process (that is, a one-parameter unitary group).\index{one-parameter unitary group}

The most famous operator that's both infinitesimal stochastic\index{infinitesimal stochastic!Laplacian} and self-adjoint is the Laplacian, $ \nabla^2.$  Because it's both, the Laplacian shows up in \emph{two}  important equations: one in stochastic mechanics, the other in quantum mechanics.

\begin{itemize} 
\item The \href{http://en.wikipedia.org/wiki/Heat\_equation}{{\bf heat equation}}:\index{heat equation}
$$ \displaystyle{ \frac{d}{d t} \psi = \nabla^2 \psi} $$
describes how the probability $ \psi(x)$ of a particle being at the point $ x$ smears out as the particle randomly walks around.  The corresponding Markov process is called \href{http://en.wikipedia.org/wiki/Brownian\_motion}{`Brownian motion'}.\index{Brownian motion} 

\begin{center}
 \includegraphics[width=90mm]{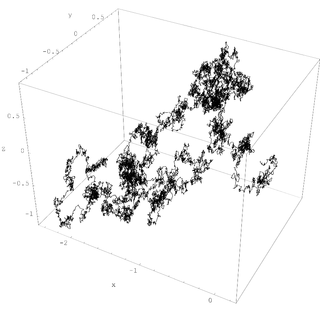}
\end{center}

\item
The \href{http://en.wikipedia.org/wiki/Schr\%C3\%B6dinger\_equation}{{\bf Schr\"{o}dinger equation}}:
$$ \displaystyle{ \frac{d}{d t} \psi = -i \nabla^2 \psi} $$
describes how the amplitude $ \psi(x)$ of a particle being at the point $ x$ wiggles about as the particle `quantumly' walks around.
\end{itemize} 

Both these equations have analogues where we replace space by a graph, and now we'll describe them.

\subsection{A random walk on the Desargues graph}\index{random walk!Desargues graph|(}

Back to business!  We've been telling you about the analogy between quantum mechanics and stochastic mechanics.  This analogy becomes especially interesting in chemistry, which lies on the uneasy borderline between quantum and stochastic mechanics.  

Fundamentally, of course, atoms and molecules are described by quantum mechanics.  But sometimes chemists describe chemical reactions using stochastic mechanics instead.  When can they get away with this?  Apparently whenever the molecules involved are big enough and interacting with their environment enough for `decoherence' to kick in.  We won't attempt to explain this now.

Let's imagine we have a molecule of iron pentacarbonyl with---here's the unrealistic part, but it's not really too bad---distinguishable carbonyl groups:

\begin{center}
 \includegraphics[width=40mm]{200px-Iron-pentacarbonyl-from-xtal-3D-balls.png}
\end{center}\

\noindent
Iron pentacarbonyl is liquid at room temperatures, so as time passes, each molecule will bounce around and occasionally do a maneuver called a `pseudorotation':

\begin{center}
 \includegraphics[width=119.0625mm]{700px-Iron-pentacarbonyl-Berry-mechanism.png}
\end{center}

\noindent
We can approximately describe this process as a random walk on a graph whose \emph{vertices}  are \emph{states}  of our molecule, and whose \emph{edges}  are \emph{transitions}  between states---namely, pseudorotations.  And as we saw in Section~\ref{sec:13}, this graph is the Desargues graph:\index{graph theory!Desargues graph!random walk on} \index{vertex} \index{edge}
\index{graph theory!vertex} \index{graph theory!edge}

\begin{center}
 \includegraphics[width=80mm]{300px-Desargues_graph_2COL.png}
\end{center}

\noindent
Note: all the transitions are reversible here.   And thanks to the enormous amount of symmetry, the rates of all these transitions must be equal.  

Let's write $ V$ for the set of vertices of the Desargues graph.  A probability distribution of states of our molecule is a function
$$ \psi \colon V \to [0,\infty) $$
with 
$$ \sum_{x \in V} \psi(x) = 1 $$
We can think of these probability distributions as living in this vector space:
$$ L^1(V) = \{ \psi \colon V \to \mathbb{R} \}$$
We're calling this space $L^1$ because of the general abstract nonsense explained in Section~\ref{sec:11}: probability distributions on any measure space\index{measure space!probability distributions on}\index{vector space!of probability distributions}\index{$L^1$} live in a vector space called $ L^1.$  Right now that notation is overkill, since \emph{every}  function on $ V$ lies in $ L^1.$  But please humor us.

The point is that we've got a general setup that applies here.  There's a Hamiltonian:
$$ H\colon L^1(V) \to L^1(V) $$
describing the rate at which the molecule randomly hops from one state to another... and the probability distribution $ \psi \in L^1(X)$ evolves in time according to the equation:
$$ \displaystyle{ \frac{d}{d t} \psi(t) = H \psi(t)} $$
But what's the Hamiltonian $ H$?  It's very simple, because it's equally likely for the state to hop from any vertex  to any other vertex that's connected to that one by an edge.  Why?  Because the problem has so much symmetry that nothing else makes sense.  

So, let's write $ E$ for the set of edges of the Desargues graph.  We can think of this as a subset of $ V \times V$ by saying $ (x,y) \in E$ when $ x$ is connected to $ y$ by an edge.  Then
$$ \displaystyle{ (H \psi)(x) =  \sum_{y : (x,y) \in E}~ \psi(y) \quad - \quad 3 \psi(x)} $$
We're subtracting $ 3 \psi(x)$ because there are 3 edges coming out of each vertex $ x,$ so this is the rate at which the probability of staying at $ x$ decreases.   We could multiply this Hamiltonian by a constant if we wanted the random walk to happen faster or slower... but let's not.

The next step is to solve this discretized version of the heat equation:\index{heat equation} 
$$ \displaystyle{ \frac{d}{d t} \psi(t) = H \psi(t)} $$
Abstractly, the solution is easy:
$$ \psi(t) = \exp(t H) \psi(0) $$
But to actually compute $ \exp(t H),$ we might want to diagonalize the operator $ H.$  In this particular example, we can do this by taking advantage of the enormous symmetry of the Desargues graph.  But let's not do this just yet.  First let's draw some general lessons from this example.

\subsection{Graph Laplacians}
\label{sec:14_laplacians}

The Hamiltonian we just saw is an example of a `graph Laplacian'.  We can write down such a Hamiltonian for any graph, but it gets a tiny bit more complicated when different vertices have different numbers of edges coming out of them.  

The word `graph' means lots of things, but right now we're talking about \href{http://en.wikipedia.org/wiki/Graph\_\%28mathematics\%29\#Simple\_graph}{{\bf simple graphs}}\index{graph theory!simple graph}\index{simple graph}.   Such a graph has a set of {\bf vertices} $V$ \index{vertex} \index{graph theory!vertex} and a set of {\bf edges} $E \subseteq V \times V,$ \index{edge} \index{graph theory!edge}
such that
$$ (x,y) \in E \Rightarrow (y,x) \in E $$
which says the edges are undirected, and
$$ (x,x) \notin E $$
which says there are no loops.  Let $ d(x)$ be the {\bf degree} of the vertex $ x,$ meaning the number of edges coming out of it.  \index{degree} \index{graph theory!vertex!degree of} \index{vertex!degree of} 

Then the \href{http://en.wikipedia.org/wiki/Laplacian\_matrix}{{\bf graph Laplacian}} is this operator on $ L^1(V)$:\index{Laplacian!graph!definition of}
$$ \displaystyle{ (H \psi)(x) =  \sum_{y  :  (x,y) \in E}  \psi(y) \quad - \quad d(x) \psi(x)} $$
There is a huge amount to say about graph Laplacians!  If you want, you can get started here:

\begin{enumerate}
 \item[\cite{New00}]
Michael William Newman, \href{http://www.seas.upenn.edu/\~jadbabai/ESE680/Laplacian\_Thesis.pdf}{\emph{The Laplacian Spectrum of Graphs}}, Masters Thesis, Department of Mathematics, University of Manitoba, 2000.
\end{enumerate}

But for now, let's just say that $ \exp(t H)$ is a Markov process describing a random walk\index{random walk!as Markov processes} on the graph, where hopping from one vertex to any neighboring vertex has unit probability per unit time.  We can make the hopping faster or slower by multiplying $ H$ by a constant.  And here is a good time to admit that most people use a graph Laplacian that's the negative of ours, and write time evolution as $ \exp(-t H)$.  The advantage is that then the eigenvalues of the Laplacian are $ \ge 0.$

But what matters most is this.  We can write the operator $ H$ as a matrix whose entry $ H_{x y}$ is 1 when there's an edge from $ x$ to $ y$ and 0 otherwise, except when $ x = y$, in which case the entry is $ -d(x).$  And then:

\begin{problem}\label{prob:21} 
Show that for any finite graph, the graph Laplacian $ H$ is infinitesimal stochastic, meaning that: 
$$ \displaystyle{ \sum_{x \in V} H_{x y} = 0} $$\index{infinitesimal stochastic!Laplacian}
and
$$ x \ne y \Rightarrow  H_{x y} \ge 0 $$
\end{problem} 

This fact implies that for any $ t \ge 0,$ the operator $ \exp(t H)$ is stochastic---just what we need for a Markov process.

But we could also use $ H$ as a Hamiltonian for a \emph{quantum}  system, if we wanted.  Now we think of $ \psi(x)$ as the \emph{amplitude}  for being in the state $ x \in V$.  But now $ \psi$  is a function
$$ \psi \colon V \to \mathbb{C} $$
with 
$$ \displaystyle{ \sum_{x \in V} |\psi(x)|^2 = 1} $$
We can think of this function as living in the Hilbert space
$$ L^2(V) = \{ \psi \colon V \to \mathbb{C} \} $$
where the inner product is
$$ \langle \phi, \psi \rangle = \sum_{x \in V} \overline{\phi(x)} \psi(x) $$

\begin{problem}\label{prob:22} 
Show that for any finite graph, the graph Laplacian $H \colon L^2(V) \to L^2(V)$ is self-adjoint, meaning that: 
$$ H_{x y} = \overline{H}_{y x} $$\end{problem} 

This implies that for any $ t \in \mathbb{R}$, the operator $ \exp(-i t H)$ is unitary---just what we need for one-parameter unitary group\index{one-parameter unitary group}.  So, we can take this version of Schr\"{o}dinger's equation\index{quantum mechanics!Schr\"{o}dinger's equation}:
$$ \displaystyle{ \frac{d}{d t} \psi = -i H \psi} $$
and solve it:
$$ \displaystyle{ \psi(t) = \exp(-i t H) \psi(0)} $$
and we'll know that time evolution is unitary!

So, we're in a dream world where we can do stochastic mechanics \emph{and}  quantum mechanics with the same Hamiltonian.  We'd like to exploit this somehow, but we're not quite sure how.  Of course physicists like to use a trick called \href{http://en.wikipedia.org/wiki/Wick\_rotation}{Wick rotation} where they turn quantum mechanics into stochastic mechanics by replacing time by imaginary time.  We can do that here.  But we'd like to do something new, special to this context.

You might want to learn more about the relation between chemistry and graph theory.  Of course, graphs show up in at least two ways: first for drawing molecules, and second for drawing states and transitions, as we've been doing.  These books seem to be good:

\begin{enumerate} 
\item[\cite{BR91}] Danail Bonchev and D.\ H.\ Rouvray, eds., \emph{Chemical Graph Theory: Introduction and Fundamentals}, Taylor and Francis, 1991. 
\item[\cite{Tr92}] Nenad Trinajstic, \emph{Chemical Graph Theory}, CRC Press, 1992. 
\item[\cite{Kin93}] R. Bruce King, \emph{Applications of Graph Theory and Topology in Inorganic Cluster Coordination Chemistry}, 
CRC Press, 1993. 
\end{enumerate} 
\index{random walk!Desargues graph|)}

\subsection{The Laplacian of the Desargues graph}\index{Laplacian!graph!Desargues graph}\index{Desargues graph!Laplacian of} 

\index{Egan, Greg}
Greg Egan figured out a nice way to explicitly describe the eigenvectors of the Laplacian $H$ for the Desargues graph.  This lets us explicitly solve the heat equation\index{heat equation} 
$$ \displaystyle{ \frac{d}{d t} \psi = H \psi} $$
and also the Schr\"{o}dinger equation on this graph.

First there is an obvious eigenvector: any constant function.  Indeed, for any finite graph, any constant function $\psi$ is an eigenvector of the graph Laplacian with eigenvalue zero:
$$H \psi = 0$$
so that
$$ \exp(t H) \psi = \psi $$
This reflects the fact that the evenly smeared probability distribution is an equilibrium state for the random walk\index{random walk!heat equation} described by heat equation.\index{heat equation}  For a connected graph this will be the only equilibrium state.  For a graph with several connected components, there will be an equilibrium state for each connected component, which equals some positive constant on that component and zero elsewhere.\index{connected component!equilibrium state} 

For the Desargues graph, or indeed any connected graph, all the other eigenvectors will be functions that undergo exponential decay at various rates:
$$H \psi = \lambda \psi $$
with $\lambda < 0$, so that
$$\exp(t H) \psi = \exp(\lambda t) \psi $$
decays as $t$ increases.  But since probability is conserved, any vector that undergoes exponential decay must have terms that sum to zero.  So a vector like this won't be a stochastic state, but rather a deviation from equilibrium.  We can write any stochastic state as the equilibrium state plus a sum of terms that decay with different rate constants.

If you put a value of $1$ on every red dot and $-1$ on every blue dot,  you get an eigenvector of the Hamiltonian with eigenvalue $-6$: $-3$ for the diagonal entry for each vertex, and $-3$ for the sum over the neighbors.

By trial and error, it's not too hard to find examples of variations on this where some vertices have a value of zero.  Every vertex with zero either has its neighbors all zero, or two neighbors of opposite signs:

\begin{itemize} 
\item Eigenvalue $-5$:  every non-zero vertex has two neighbors with the opposite value to it, and one neighbor of zero.
\item Eigenvalue $-4$:  every non-zero vertex has one neighbor with the opposite value to it, and two neighbors of zero.
\item Eigenvalue $-2$:  every non-zero vertex has one neighbor with the same value, and two neighbors of zero.
\item Eigenvalue $-1$: every non-zero vertex has two neighbors with the same value, and one neighbor of zero.
\end{itemize} 

So the general pattern is what you'd expect:  the more neighbors of equal value each non-zero vertex has, the more slowly that term will decay.

In more detail, the eigenvectors are given by the the following functions, together with functions obtained by rotating these pictures.  Unlabelled dots are labelled by zero:

\begin{center}
 \includegraphics[width=130mm]{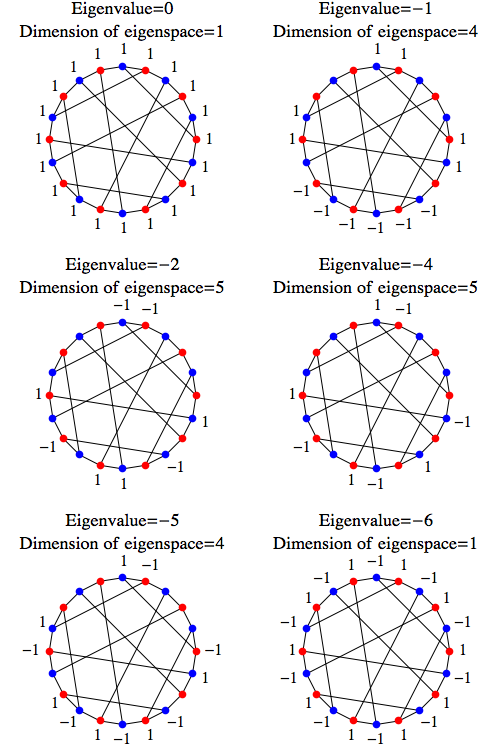}
\end{center}

\newpage 

In fact you can span the eigenspaces by rotating the particular eigenvectors we drew by successive multiples of $1/5$ of a rotation, i.e.\ $4$ dots.  If you do this for the examples with eigenvalues $-2$ and $-4$, it's not too hard to see that all five rotated diagrams are linearly independent.  If you do it for the examples with eigenvalues $-1$ and $-5$, four rotated diagrams are linearly independent.  So along with the two 1-dimensional spaces, that's enough to prove that you've exhausted the whole 20-dimensional space.

\index{Laplacian!graph|)}

%%%%% SECTION 15 %%%%%%%

\newpage
\section[Dirichlet operators]{Dirichlet operators and electrical circuits}\index{electrical circuits!and Dirichlet operators|(}\index{Dirichlet operator|(}
\label{sec:15}

We've been comparing two theories: stochastic mechanics and quantum mechanics.   Last time we saw that any graph gives us an example of \emph{both}  theories!  It's a bit peculiar, but now we'll explore the intersection of these theories a little further, and see that it has another interpretation.  It's also the theory of \emph{electrical circuits made of resistors!}   

That's nice, because we're supposed to be talking about `network theory', and electrical circuits are perhaps the most practical networks of all:

\begin{center}
 \includegraphics[width=80mm]{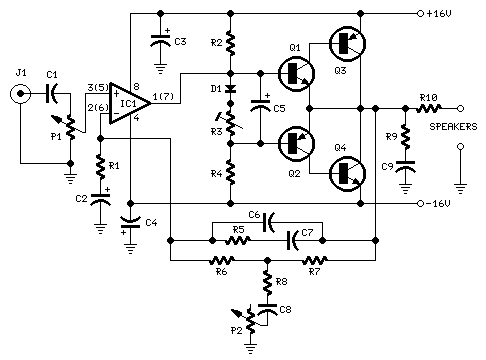}
\end{center}

\subsection{Dirichlet operators and electrical circuits}\index{Dirichlet operator!electrical circuits}

Last time we saw that any graph gives us an operator called the `graph Laplacian' that's both infinitesimal stochastic and self-adjoint.\index{infinitesimal stochastic!Laplacian}  That means we get both:

\begin{itemize} 
\item a \href{http://en.wikipedia.org/wiki/Markov\_process}{Markov process} describing the random walk of a classical particle on the graph.\index{random walk!classical particle} and
\item a \href{http://en.wikipedia.org/wiki/Stone\%27s\_theorem\_on\_one-parameter\_unitary\_groups}{1-parameter unitary group}\index{one-parameter unitary group} describing the motion of a quantum particle on the graph.\index{random walk!quantum walk}
\end{itemize} 

That's sort of neat, so it's natural to wonder what are \emph{all}  the operators that are both infinitesimal stochastic and self-adjoint.  They're called `Dirichlet operators', and at least in the finite-dimensional case we're considering, they're easy to completely understand.  Even better, it turns out they describe electrical circuits made of resistors! 
\index{Dirichlet operator}

Let's take a lowbrow attitude and think of a linear operator $H\colon \mathbb{C}^n \to \mathbb{C}^n$ as an $n \times n$ matrix with entries $H_{i j}$.  Then:

\begin{itemize}
\item  $H$ is {\bf self-adjoint} if it equals the conjugate of its transpose:
$$H_{i j} = \overline{H}_{j i}$$\index{self-adjoint operator!definition of}
\end{itemize} 

\begin{itemize} 
\item $H$ is {\bf infinitesimal stochastic} if its columns sum to zero and its off-diagonal entries are nonnegative:\index{infinitesimal stochastic operator} \index{operator!infinitesimal stochastic}
$$\sum_i H_{i j} = 0$$ 
 $$i \ne j \Rightarrow H_{i j} \ge 0 $$
\end{itemize} 

\begin{itemize} 
\item $H$ is a {\bf Dirichlet operator}\index{Dirichlet operator!definition of} if it's both self-adjoint and infinitesimal stochastic.
\end{itemize} 

What are Dirichlet operators like?  Suppose $H$ is a Dirichlet operator.  Then its off-diagonal entries are $\geq 0$, and since
$$ \sum_i H_{i j} = 0$$ 
its diagonal entries obey
$$  H_{i i} = - \sum_{ i \ne j} H_{i j} \le 0$$
So all the entries of the matrix $H$ are real, which in turn implies it's symmetric:
$$ H_{i j} = \overline{H}_{j i} = H_{j i} $$
So, we can build any Dirichlet operator $H$ as follows:

\begin{itemize} 
\item Choose the entries above the diagonal, $H_{i j}$ with $i < j$, to be arbitrary nonnegative real numbers.
\item The entries below the diagonal, $H_{i j}$ with $i > j$, are then forced on us by the requirement that $H$ be symmetric: $H_{i j} = H_{j i}$.
\item The diagonal entries are then forced on us by the requirement that the columns sum to zero: $H_{i i} = - \sum_{ i \ne j} H_{i j}$.
\end{itemize} 

Note that because the entries are real, we can think of a Dirichlet operator as a linear operator $H \colon \mathbb{R}^n \to \mathbb{R}^n$.  We'll do that for the rest of this section.

\subsection{Circuits made of resistors}
\label{sec:15_circuits}

Now for the fun part.  We can easily \emph{draw}  any Dirichlet operator!   To this we draw $n$ dots, connect each pair of distinct dots with an edge, and label the edge connecting the $i$th dot to the $j$th with any number $H_{i j} \ge 0$.   

\begin{center}
 \includegraphics[width=55mm]{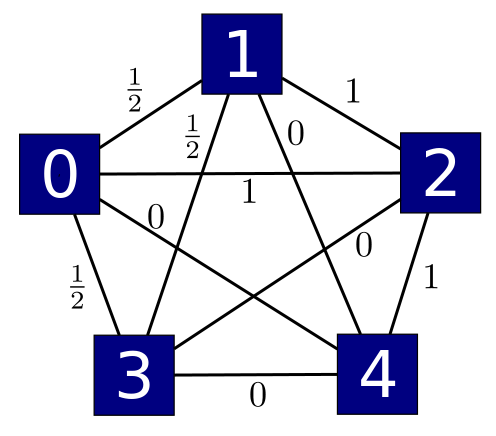}
\end{center}

\noindent
This contains all the information we need to build our Dirichlet operator.  To make the picture prettier, we can leave out the edges labelled by 0:

\begin{center}
 \includegraphics[width=55mm]{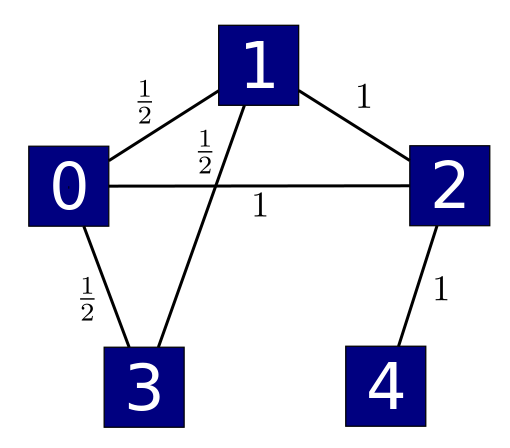}
\end{center}

\noindent
Like last time, the graphs we're talking about are \href{http://en.wikipedia.org/wiki/Simple\_graph\#Simple\_graph}{{\bf simple}}: undirected, with no edges from a vertex to itself, and at most one edge from one vertex to another.   So: \index{graph theory!simple graph} \index{simple graph}

\begin{theorem}
Any finite simple graph with edges labelled by positive numbers gives a Dirichlet operator, and conversely.
\end{theorem} 

We already talked about a special case last time: if we label all the edges by the number 1, our operator $H$ is called the {\bf graph Laplacian}.  So, now we're generalizing that idea by letting the edges have more interesting labels.

What's the meaning of this trick?  Well, we can think of our graph as an \emph{electrical circuit}  where the edges are \emph{wires}.  What do the numbers labelling these wires mean?  One obvious possibility is to put a \href{http://en.wikipedia.org/wiki/Resistor}{resistor} on each wire, and let that number be its \href{http://en.wikipedia.org/wiki/Electrical\_resistance\_and\_conductance}{resistance}.   But that doesn't make sense, since we're leaving out wires labelled by 0.  If we leave out a wire, that's not like having a wire of zero resistance: it's like having a wire of \emph{infinite}  resistance!  No current can go through when there's no wire.  So the number labelling an edge should be the \emph{conductance}  of the resistor on that wire.  Conductance is the reciprocal of resistance.

So, our Dirichlet operator above gives a circuit like this:

\begin{center}
 \includegraphics[width=60mm]{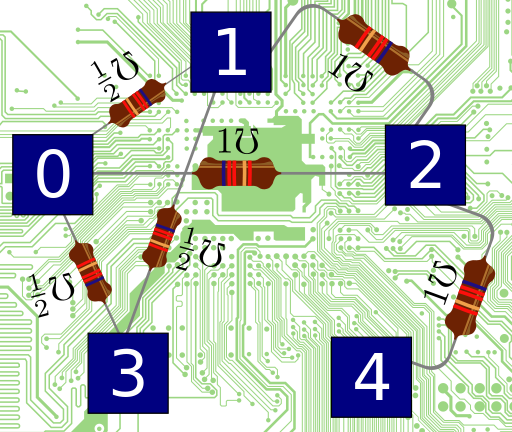}
\end{center}

\noindent
Here $\Omega$ is the symbol for an \href{http://en.wikipedia.org/wiki/Ohm}{`ohm'}, a unit of resistance... but the upside-down version, namely $\mho$, is the symbol for a \href{http://en.wikipedia.org/wiki/Siemens\_\%28unit\%29\#Mho}{`mho'}, a unit of conductance that's the reciprocal of an ohm.  (Actually the term \href{http://en.wikipedia.org/wiki/Siemens_\%28unit\%29}{`siemens'} is officially preferred over `mho', but we like the word `mho'.)

Let's see if this cute idea leads anywhere.  Think of a Dirichlet operator $H \colon \mathbb{R}^n \to \mathbb{R}^n$ as a circuit made of resistors.   What could a vector $\psi \in \mathbb{R}^n$ mean?  It assigns a real number to each vertex of our graph.  The only sensible option is for this number to be the \href{http://en.wikipedia.org/wiki/Electric_potential}{electric potential} at that point in our circuit.  So let's try that.

Now, what's
$$ \langle \psi, H \psi \rangle  ?$$
In quantum mechanics this would be a very sensible thing to look at: it would be gives us the expected value\index{expected value!of Hamiltonian, quantum} of the Hamiltonian $H$ in a state $\psi$.  But what does it mean in the land of electrical circuits?

Up to a constant factor, it turns out to be the \emph{power}  consumed by the electrical circuit!\index{electrical circuits!power}   

Let's see why.  First, remember that when a current flows along a wire, power gets consumed.  In other words, electrostatic potential energy gets turned into heat.  The power consumed is 
$$ P = V I $$
where $V$ is the voltage across the wire and $I$ is the current flowing along the wire.  If we assume our wire has resistance $R$ we also have Ohm's law:
$$ I = V / R $$
so 
$$ \displaystyle{ P = \frac{V^2}{R}} $$
If we write this using the conductance $C = 1/R$ instead of the resistance $R$, we get
$$ P = C V^2 $$
But our electrical circuit has \emph{lots}  of wires, so the power it consumes will be a sum of terms like this.   We're assuming  $H_{i j}$ is the conductance of the wire from the $i$th vertex to the $j$th, or zero if there's no wire connecting them.  And by definition, the voltage across this wire is the difference in electrostatic potentials at the two ends: $\psi_i - \psi_j$.  So, the total power consumed is
$$ \displaystyle{ P = \sum_{i \ne j}  H_{i j} (\psi_i - \psi_j)^2}$$
This is nice, but what does it have to do with $ \langle \psi , H \psi \rangle$? 

The answer is here:

\begin{theorem}  If $H \colon \mathbb{R}^n \to \mathbb{R}^n$ is any Dirichlet operator, and $\psi \in \mathbb{R}^n$ is any vector, then
$$ \displaystyle{ \langle \psi , H \psi \rangle = -\frac{1}{2} \sum_{i \ne j}  H_{i j} (\psi_i - \psi_j)^2}  $$
\end{theorem} 

\begin{proof} Let's start with the formula for power:
$$ \displaystyle{ P = \sum_{i \ne j}  H_{i j} (\psi_i - \psi_j)^2} $$
Note that this sum includes the condition $i \ne j$, since we only have wires going between distinct vertices.  But the summand is zero if $i = j$, so we also have
$$  \displaystyle{ P = \sum_{i, j}  H_{i j} (\psi_i - \psi_j)^2}$$
Expanding the square, we get
$$ \displaystyle{ P = \sum_{i, j}  H_{i j} \psi_i^2 - 2 H_{i j} \psi_i \psi_j + H_{i j} \psi_j^2} $$
The middle term looks promisingly similar to $\langle \psi, H \psi \rangle$, but what about the other two terms?  Because $H_{i j} = H_{j i}$, they're equal:
$$  \displaystyle{ P = \sum_{i, j} - 2 H_{i j} \psi_i \psi_j + 2 H_{i j} \psi_j^2 } $$
And in fact they're zero!  Since $H$ is infinitesimal stochastic, we have
$$ \displaystyle{ \sum_i H_{i j} = 0} $$
so 
$$ \displaystyle{ \sum_i H_{i j} \psi_j^2 = 0} $$
and it's still zero when we sum over $j$.  We thus have
$$  \displaystyle{ P = - 2 \sum_{i, j} H_{i j} \psi_i \psi_j} $$
But since $\psi_i$ is real, this is -2 times
$$ \displaystyle{ \langle \psi, H \psi \rangle  = \sum_{i, j}  H_{i j} \overline{\psi}_i \psi_j} $$
So, we're done.  

An instant consequence of this theorem is that a Dirichlet operator has
$$ \langle \psi , H \psi \rangle \le 0 $$
for all $\psi$.  Actually most people use the opposite sign convention\index{infinitesimal stochastic!Laplacian!sign convention} in defining infinitesimal stochastic operators.  This makes $H_{i j} \le 0$, which is mildly annoying, but it gives
$$ \langle \psi , H \psi \rangle \ge 0 $$
which is nice.  When $H$ is a Dirichlet operator, defined with this opposite sign convention, $\langle \psi , H \psi \rangle$ is called a {\bf Dirichlet form}. 
\end{proof}

\subsection{The big picture}

Maybe it's a good time to step back and see where we are.

So far we've been exploring the analogy between stochastic mechanics and quantum mechanics.  Where do networks come in?  Well, they've actually come in twice so far:

\begin{enumerate} 
\item  First we saw that Petri nets can be used to describe stochastic or quantum processes where things of different kinds randomly react and turn into other things.  A Petri net is a kind of network like this:

\begin{center}
 \includegraphics[width=50mm]{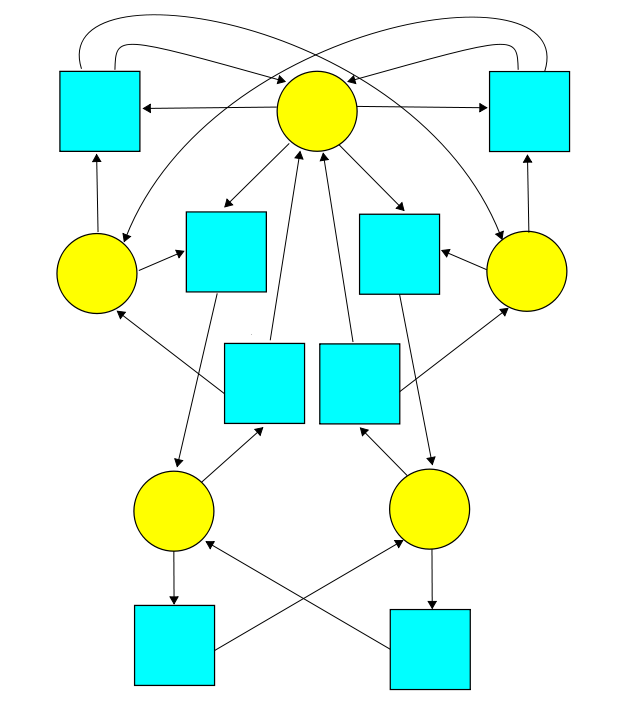}
\end{center}

The different kinds of things are the yellow circles; we called them {\bf species}.  The reactions where things turn into other things are the blue squares: we called them {\bf transitions}.  We label the transitions by numbers to say the rates at which they occur.

\item Then we looked at stochastic or quantum processes where in each transition a single thing turns into a single thing.  We can draw these as Petri nets where each transition has just one species as input and one species as output.  But we can also draw them as directed graphs with edges labelled by numbers:

\begin{center}
 \includegraphics[width=60mm]{markov_chain_noether_counterexample.png}
\end{center}

\noindent
Now the dark blue boxes are species and the \emph{edges} are transitions!  Furthermore, in this situation we often call the species {\bf states}.  
\end{enumerate}

In this section we looked at a special case of the second kind of network: the Dirichlet operators.  For these the `forward' transition rate $H_{i j}$ equals the `reverse' rate $H_{j i}$, so our graph can be undirected: no arrows on the edges.  And for these the rates $H_{i i}$ are determined by the rest, so we can omit the edges from vertices to themselves.  We also omit edges labelled by zero:

\begin{center}
 \includegraphics[width=55mm]{complete-graph-5-zero.png}
\end{center}

\noindent
The result can be seen as an electrical circuit made of resistors!  So we're building up a little dictionary:

\begin{itemize} 
\item Stochastic mechanics: $\psi_i$ is a probability and $H_{i j}$ is a transition rate (probability per time).
\item Quantum mechanics: $\psi_i$ is an amplitude and $H_{i j}$ is a transition rate (amplitude per time).
\item Circuits made of resistors: $\psi_i$ is a voltage and $H_{i j}$ is a conductance.
\end{itemize} 

This dictionary may seem rather odd---especially the third item, which looks completely different than the first two!  But that's good: when things aren't odd, we don't get many new ideas.  The whole point of this `network theory' business is to think about networks from many different viewpoints and let the sparks fly! 

Actually, this particular oddity is well-known in certain circles.  We've been looking at the \emph{discrete}  version, where we have a finite set of states.  But in the \emph{continuum}, the classic example of a Dirichlet operator is the Laplacian $H = \nabla^2$.  And then:

\begin{itemize} 
\item The \href{http://en.wikipedia.org/wiki/Heat_equation}{{\bf heat equation}}:
$$ \frac{d}{d t} \psi = \nabla^2 \psi $$
is fundamental to stochastic mechanics.

\item The \href{http://en.wikipedia.org/wiki/Schr\%C3\%B6dinger\_equation}{{\bf Schr\"{o}dinger equation}}:
$$ \frac{d}{d t} \psi = -i \nabla^2 \psi $$
is fundamental to quantum mechanics.

\item The \href{http://en.wikipedia.org/wiki/Poisson\%27s\_equation}{{\bf Poisson equation}}:
$$ \nabla^2 \psi = -\rho $$
is fundamental to electrostatics.
\end{itemize} 

Briefly speaking, \href{http://en.wikipedia.org/wiki/Electrostatics}{electrostatics} is the study of how the electric potential $\psi$ depends on the charge density $\rho$.  The theory of electrical circuits made of resistors can be seen as a special case, at least when the current isn't changing with time.

If you want to learn more, this is a great place to start:

\begin{enumerate} 
\item[\cite{DS84}] P. G. Doyle and J. L. Snell, \href{http://www.math.dartmouth.edu/~doyle/}{\sl Random Walks and Electrical Circuits}, Mathematical Association of America, Washington DC, 1984. 
\end{enumerate} 

\noindent
This free online book explains, in a really fun informal way, how random walks on graphs, are related to electrical circuits made of resistors.\index{electrical circuits!resistor networks}\index{random walk!and electrical networks}  To dig deeper into the continuum case, try:

\begin{enumerate} 
\item[\cite{Fuk80}]  M. Fukushima, {\sl Dirichlet Forms and Markov Processes}, North-Holland, Amsterdam, 1980.  
\end{enumerate} 

\index{electrical circuits!and Dirichlet operators|)}\index{Dirichlet operator|)}

%%%%%% SECTION 16 %%%%%%%% 

\newpage
\section[Perron--Frobenius theory]{Perron--Frobenius theory}\label{sec:16} 

We're in the middle of a battle: in addition to our typical man vs.\ equation scenario, it's a battle between two theories.  You should know the two opposing forces well by now.  It's our old friends, at it again:

\begin{center} 
\emph{Stochastic Mechanics vs Quantum Mechanics!}
\end{center} 

Today we're reporting live from a crossroads, and we're facing a skirmish that gives rise to what some might consider a paradox.  Let's sketch the main thesis before we get our hands dirty with the gory details.   

First we need to tell you that the battle takes place at the intersection of stochastic and quantum mechanics.  We recall from Section~\ref{sec:15} that there is a class of operators called `Dirichlet operators'\index{Dirichlet operator} that are valid Hamiltonians for both stochastic and quantum mechanics.  In other words, you can use them to generate time evolution both for old-fashioned random processes and for quantum processes.
 
Staying inside this class allows the theories to fight it out on the same turf.  We will be considering a special subclass of Dirichlet operators, which we call `irreducible' Dirichlet operators.  These are the ones that generate Markov processes where starting at any vertex, we have a nonzero chance of winding up at any other vertex. When considering this subclass, we see something interesting: 

\begin{thesis} Let $H$ be an irreducible Dirichlet operator with $n$ eigenstates.  In stochastic mechanics, there is only one valid state that is an eigenvector of $H$: the unique so-called `Perron--Frobenius state'.  The other $n-1$ eigenvectors are forbidden states of a stochastic system: the stochastic system is either in the Perron--Frobenius state, or in a superposition of at least two eigensvectors.  In quantum mechanics, all $n$ eigenstates of $H$ are valid states.  
\end{thesis}\index{Perron–Frobenius theorem!and quantum states} 

\noindent This might sound like a riddle, but now we'll prove, riddle or not, that it's a fact.  

\subsection{At the intersection of two theories} 

The \emph{typical} view of how quantum mechanics and probability theory come into contact looks like this:

\begin{center}
 \includegraphics[width=50mm]{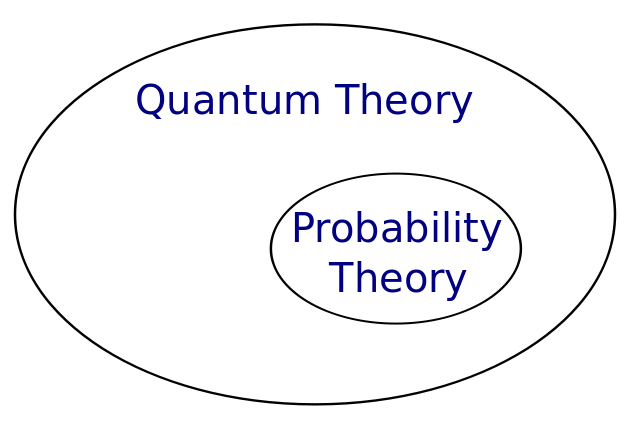}
\end{center}

\noindent
The idea is that quantum theory generalizes classical probability theory by considering observables that don't commute.

That's perfectly valid, but we've been exploring an alternative view in this series.  Here quantum theory doesn't subsume probability theory, but they intersect:

\begin{center}
 \includegraphics[width=70mm]{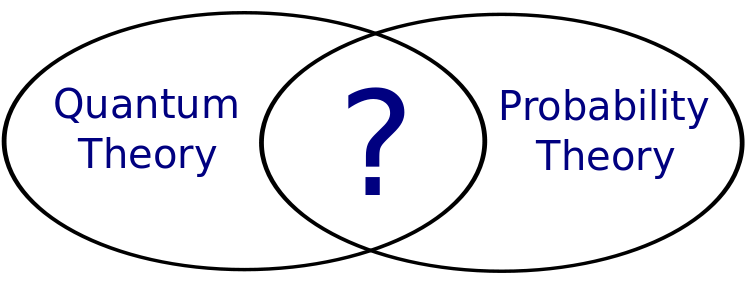}
\end{center}

\noindent
What goes in the middle, you might ask?  We saw in Section~\ref{sec:15_circuits} that electrical circuits made of resistors constitute the intersection!

\begin{center}
 \includegraphics[width=70mm]{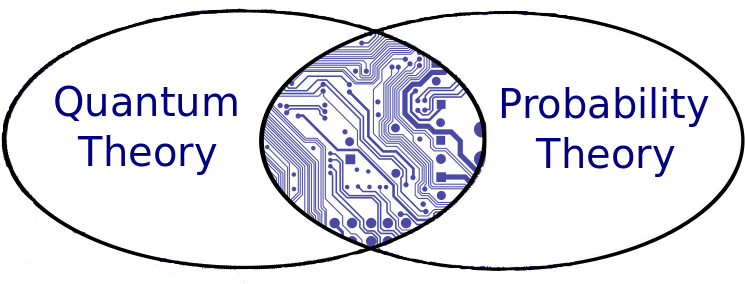}
\end{center}

\noindent
For example, a circuit like this:

\begin{center}
 \includegraphics[width=80mm]{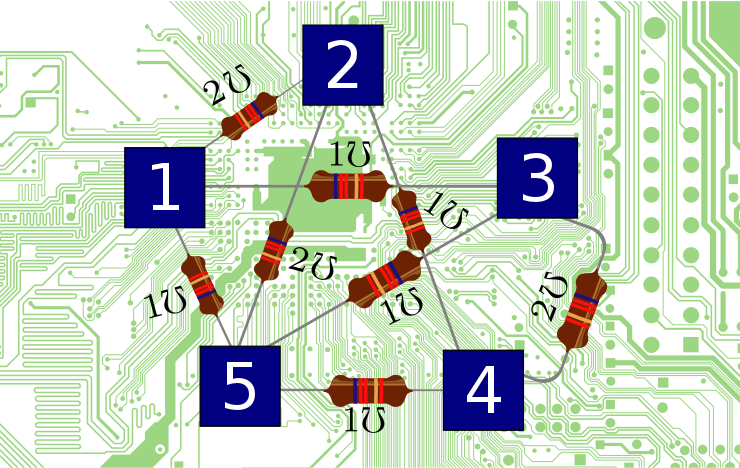}
\end{center}

\noindent
gives rise to a Hamiltonian $H$ that's good both for quantum mechanics and stochastic mechanics.   Indeed, we saw that the power dissipated by a circuit made of resistors is related to the familiar quantum theory concept known as the expectation value of the Hamiltonian!
$$ \textrm{power} = -2 \langle \psi, H \psi \rangle $$
Oh---and you might think we made a mistake and wrote our $\Omega$ (ohm) symbols upside down.  We didn't.  It happens that $\mho$ is the symbol for a `mho'---a unit of conductance that's the reciprocal of an ohm. Reread Section~\ref{sec:15} for the details.

\subsection{Stochastic mechanics versus quantum mechanics}

Before we plunge in, let's review how states, time evolution, symmetries and observables work in the two theories.   In this section we'll take a lowbrow approach, fixing a basis\index{vector space!basis} for our vector space of states, and assuming it's finite-dimensional so that all vectors have $n$ components which are either complex or real numbers.   In other words, we'll treat our space as either $ \mathbb{C}^n$ or $ \mathbb{R}^n$.  In this fashion, linear operators that map such spaces to themselves will be represented as square matrices.  

Vectors will be written as $\psi_i$ where the index $i$ runs from 1 to $n$, and we think of each choice of the index as a {\bf state} of our system---but since we'll be using that word in other ways too, let's call it a {\bf configuration}.  It's just a basic way our system can be.

\subsubsection*{States} 

Besides the configurations $i = 1,\dots, n$, we have more general states that tell us the probability or amplitude of finding our system in one of these configurations:
\begin{itemize}
\item {\bf Stochastic states} are $n$-tuples of nonnegative real numbers:  
\end{itemize} 
$$ \psi_i \in \mathbb{R}^+ $$
The probability of finding the system in the $i$th configuration is defined to be $\psi_i$.  For these probabilities to sum to one, $\psi_i$ needs to be normalized like this: 
$$ \sum_i \psi_i = 1 $$
or in the notation we're using: 
$$ \langle \psi \rangle = 1 $$
where we define
$$ \langle \psi \rangle = \sum_i \psi_i $$
\begin{itemize} 
\item {\bf Quantum states} are $n$-tuples of complex numbers:
\end{itemize} 
$$ \psi_i \in \mathbb{C} $$
The probability of finding a state in the $i$th configuration is defined to be $|\psi(x)|^2$.  For these probabilities to sum to one, $\psi$ needs to be normalized like this:
$$ \sum_i |\psi_i|^2 = 1 $$
or in other words 
$$ \langle \psi,  \psi \rangle = 1 $$
where the {\bf inner product} of two vectors $\psi$ and $\phi$ is defined by
$$ \langle \psi, \phi \rangle = \sum_i \overline{\psi}_i \phi_i $$
Now, the usual way to turn a quantum state $\psi$ into a stochastic state is to take the absolute value of each number $\psi_i$ and then square it.  However, if the numbers $\psi_i$ happen to be nonnegative, we can also turn $\psi$ into a stochastic state simply by multiplying it by a number to ensure $\langle \psi \rangle = 1$.

This is very unorthodox, but it lets us evolve the same vector $\psi$ either stochastically or quantum-mechanically, using the recipes we'll describe next.  In physics jargon these correspond to evolution in `real time' and `imaginary time'.  But don't ask which is which: from a quantum viewpoint stochastic mechanics uses imaginary time, but from a stochastic viewpoint it's the other way around! 

\subsubsection*{Time evolution}

Time evolution works similarly in stochastic and quantum mechanics, but with a few big differences:

\begin{itemize} 
\item In stochastic mechanics the state changes in time according to the {\bf master equation}:\index{stochastic mechanics!master equation} 
$$  \frac{d}{d t} \psi(t) = H \psi(t) $$
which has the solution
$$ \psi(t) = \exp(t H) \psi(0) $$
\item In quantum mechanics the state changes in time according to {\bf Schr\"{o}dinger's equation}\index{quantum mechanics!Schr\"{o}dinger's equation}:
$$  \frac{d}{d t} \psi(t) = -i H \psi(t) $$
which has the solution
$$ \psi(t) = \exp(-i t H) \psi(0) $$
\end{itemize} 

The operator $H$ is called the {\bf Hamiltonian}.  The properties it must have depend on whether we're doing stochastic mechanics or quantum mechanics:

\begin{itemize} 
\item  We need $H$ to be {\bf infinitesimal stochastic} for time evolution given by $\exp(t H)$ to send stochastic states to stochastic states.  In other words, we need that (i) its columns sum to zero and (ii) its off-diagonal entries are real and nonnegative:\index{operator!infinitesimal stochastic} \index{infinitesimal stochastic operator}
$$ \sum_i H_{i j}=0 $$
$$ i\neq j\Rightarrow H_{i j}> 0 $$
\item We need $ H$ to be {\bf self-adjoint} for time evolution given by $\exp(-i t H)$ to send quantum states to quantum states.  So, we need
\end{itemize} 
$$  H = H^\dagger $$
where we recall that the adjoint of a matrix is the conjugate of its transpose:
$$ (H^\dagger)_{i j} := \overline{H}_{j i} $$
\index{operator!self-adjoint} \index{self-adjoint operator}
We are concerned with the case where the operator $H$ generates both a valid quantum evolution and also a valid stochastic one:
\begin{itemize} 
\item $H$ is a {\bf Dirichlet operator} if it's both self-adjoint and infinitesimal stochastic.
\end{itemize} 
We will soon go further and zoom in on this intersection!  But first let's finish our review.  

\subsubsection*{Symmetries}

As explained in Section~\ref{sec:11_symmetries}, besides states and observables we need symmetries, which are transformations that map states to states.  These include the evolution operators which we only briefly discussed in the preceding subsection.

\begin{itemize} 
\item  A linear map $U$ that sends quantum states to quantum states is called an  {\bf isometry},\index{quantum mechanics!isometry!definition of} and isometries are characterized by this property:
$$ U^\dagger U = 1$$
\item A linear map $U$ that sends stochastic states to stochastic states is called a {\bf stochastic operator}, and stochastic operators are characterized by these properties:\index{stochastic operator}\index{operator!stochastic}
$$ \sum_i U_{i j} = 1 $$
and 
$$ U_{i j}> 0 $$
\end{itemize} 

A notable difference here is that in our finite-dimensional situation, isometries are always invertible, but stochastic operators may not be!  If $U$ is an $n \times n$ matrix that's an isometry,\index{matrix!isometry} $U^\dagger$ is its inverse.  So, we also have
$$ U U^\dagger = 1$$\index{quantum mechanics!isometry}
and we say $U$ is {\bf unitary}.  But if $U$ is stochastic, it may not have an inverse---and even if it does, its inverse is rarely stochastic.   This explains why in stochastic mechanics time evolution is often not reversible, while in quantum mechanics it always is.

\begin{problem}\label{prob:23} 
  Suppose $U$ is a stochastic $n \times n$ matrix whose inverse is stochastic.  What are the possibilities for $U$?
\end{problem}

It is quite hard for an operator to be a symmetry in both stochastic and quantum mechanics, especially in our finite-dimensional situation:

\begin{problem}\label{prob:24} 
 Suppose $U$ is an $n \times n$ matrix that is both stochastic and unitary.  What are the possibilities for $U$?
\end{problem}

\subsubsection*{Observables}

`Observables' are real-valued quantities that can be measured, or predicted, given a specific theory.\index{expected value!quantum vs stochastic}   
\begin{itemize} 
\item In quantum mechanics, an observable is given by a self-adjoint matrix $O$, and the expected value of the observable $O$ in the quantum state $\psi$ is
$$ \langle \psi , O \psi \rangle = \sum_{i,j} \overline{\psi}_i O_{i j} \psi_j $$
\item In stochastic mechanics, an observable $O$ has a value $O_i$ in each configuration $i$, and the expected value of the observable $O$ in the stochastic state $\psi$ is
$$ \langle O \psi \rangle = \sum_i O_i \psi_i $$
\end{itemize} 

We can turn an observable in stochastic mechanics into an observable in quantum mechanics by making a diagonal matrix whose diagonal entries are the numbers $O_i$.

\subsection{From graphs to matrices}
\label{sec:16_graphs_to_matrices}

In Section~\ref{sec:15} we explained how a graph with positive numbers on its edges gives rise to a Hamiltonian in both quantum and stochastic mechanics---in other words, a Dirichlet operator.

Recall how it works.  We'll consider \href{http://en.wikipedia.org/wiki/Simple_graph#Simple_graph}{{\bf simple graphs}}: graphs without arrows on their edges, with at most one edge from one vertex to another, and with no edge from a vertex to itself.  And we'll only look at graphs with finitely many vertices and edges.  We'll assume each edge is labelled by a positive number, like this:

\begin{center}
 \includegraphics[width=55mm]{complete-graph-5-zero.png}
\end{center}

\noindent
If our graph has $n$ vertices, we can create an $n \times n$ matrix $A$ where $A_{i j}$ is the number labelling the edge from $i$ to $j$, if there is such an edge, and 0 if there's not.  This matrix is symmetric, with real entries, so it's self-adjoint.   So $A$ is a valid Hamiltonian in quantum mechanics.  

How about stochastic mechanics?   Remember that a Hamiltonian in stochastic mechanics needs to be `infinitesimal stochastic'.  So, its off-diagonal entries must be nonnegative, which is indeed true for our $A$, but also the sums of its columns must be zero, which is not true when our $A$ is nonzero.

But now comes the best news you've heard all day: we can improve $A$ to a stochastic operator in a way that is completely determined by $A$ itself!  This is done by subtracting a diagonal matrix $L$ whose entries are the sums of the columns of $A$:
$$L_{i i} = \sum_i A_{i j} $$
$$  i \ne j \Rightarrow L_{i j} = 0 $$
It's easy to check that
$$ H = A - L $$
is still self-adjoint, but now also infinitesimal stochastic.  So, it's a \emph{Dirichlet operator}: a good Hamiltonian for \emph{both} stochastic and quantum mechanics!  
\index{Dirichlet operator}

In Section~\ref{sec:15}, we saw a bit more: \emph{every} Dirichlet operator arises this way.  It's easy to see.  You just take your Dirichlet operator and make a graph with one edge for each nonzero off-diagonal entry.  Then you label the edge with this entry.  So, Dirichlet operators are essentially the same as finite simple graphs with edges labelled by positive numbers. 

Now, a simple graph can consist of many separate `pieces', called {\bf components}.  Then there's no way for a particle hopping along the edges to get from one component to another, either in stochastic or quantum mechanics.  So we might as well focus our attention on graphs with just one component.  These graphs are called `connected'.  In other words:

\begin{definition} 
A simple graph is {\bf connected} if it is nonempty and there is a path of edges connecting any vertex to any other.  
\end{definition} 

Our goal now is to understand more about Dirichlet operators coming from connected graphs.  For this we need to learn the Perron--Frobenius theorem.  But let's start with something easier.

\subsection{Perron's theorem}\index{Perron's theorem|(}

In quantum mechanics it's good to think about observables that have positive expected values:\index{expected value!positive} 
$$   \langle \psi, O \psi \rangle \ge 0 $$
for every quantum state $\psi \in \mathbb{C}^n$.  These are called {\bf positive definite}.\index{matrix!positive definite}  But in stochastic mechanics it's good to think about matrices that are positive in a more naive sense:

\begin{definition}
An $n \times n$ real matrix $T$ is {\bf positive}\index{matrix!positive} if all its entries are positive: 
$$    T_{i j} \ge 0 $$
for all $1 \le i, j \le n$.
\end{definition} 

Similarly: 

\begin{definition}
A vector $\psi \in \mathbb{R}^n$ is {\bf positive}\index{vector!positive!definition of} if all its components are positive:
$$   \psi_i \ge 0 $$
for all $1 \le i \le n$.
\end{definition}

We'll also define {\bf nonnegative}\index{matrix!nonnegative}\index{vector!nonnegative} matrices and vectors in the same way, replacing $> 0$ by $\ge 0$.  A good example of a nonnegative vector is a stochastic state.

In 1907, Perron proved the following fundamental result about positive matrices:

\begin{theorem}[{\bf Perron's Theorem}]\index{Perron's theorem!statement of} 
  Given a positive square matrix $T$, there is a positive real number $r$, called the {\bf Perron--Frobenius eigenvalue} of $T$, such that $r$ is an eigenvalue of $T$ and any other eigenvalue $\lambda$ of $T$ has $ |\lambda| < r$.  Moreover, there is a positive vector $\psi \in \mathbb{R}^n$ with $T \psi = r \psi$.  Any other vector with this property is a scalar multiple of $\psi$.  Furthermore, any nonnegative vector that is an eigenvector of $T$ must be a scalar multiple of $\psi$.
\end{theorem} 

In other words, if $T$ is positive, it has a unique eigenvalue with the largest absolute value.  This eigenvalue is positive.  Up to a constant factor, it has an unique eigenvector.  We can choose this eigenvector to be positive.  And then, up to a constant factor, it's the \emph{only} nonnegative eigenvector of $T$.

\subsection{From matrices to graphs}
\label{sec:16_matrices_to_graphs}

The conclusions of Perron's theorem don't hold for matrices that are merely nonnegative.  For example, these matrices
$$ \left( \begin{array}{cc} 1 & 0 \\ 0 & 1 \end{array} \right) , \qquad \left( \begin{array}{cc} 0 & 1 \\ 0 & 0 \end{array} \right) $$
are nonnegative, but they violate lots of the conclusions of Perron's theorem.  

Nonetheless, in 1912 Frobenius published an impressive generalization of Perron's result.  In its strongest form, it doesn't apply to \emph{all} nonnegative matrices; only to those that are `irreducible'.  So, let us define those.

\index{Perron's theorem|)}

We've seen how to build a matrix from a graph.  Now we need to build a graph from a matrix!  Suppose we have an $n \times n$ matrix $T$.  Then we can build a graph $G_T$ with $n$ vertices where there is an edge from the $i$th vertex to the $j$th vertex if and only if $T_{i j} \ne 0$.  

But watch out: this is a different kind of graph!  It's a \href{http://en.wikipedia.org/wiki/Directed_graph}{{\bf directed graph}}, meaning the edges have directions, there's at most one edge going from any vertex to any vertex, and we do allow an edge going from a vertex to itself.  There's a stronger concept of `connectivity' for these graphs:

\begin{definition}
 A directed graph is {\bf strongly connected} if there is a directed path of edges going from any vertex to any other vertex.\index{graph theory!strongly connected graph}\index{strongly connected graph}    
\end{definition}

So, you have to be able to walk along edges from any vertex to any other vertex, but always following the direction of the edges!  Using this idea we define irreducible matrices:

\begin{definition}\index{strongly connected graph!and irreducibility} 
 A nonnegative square matrix $T$ is {\bf irreducible} if its graph $G_T$ is strongly connected.  
\end{definition}

\subsection{The Perron--Frobenius theorem} 
\label{sec:16_perron_frobenius}

Now we are ready to state:

\begin{theorem}[{\bf The Perron--Frobenius Theorem}]\index{Perron–Frobenius theorem!statement of} 
 Given an irreducible nonnegative square matrix $T$, there is a positive real number $r$, called the {\bf Perron--Frobenius eigenvalue} of $T$, such that $r$ is an eigenvalue of $T$ and any other eigenvalue $\lambda$ of $T$ has $|\lambda| \le r$.  Moreover, there is a positive vector $\psi \in \mathbb{R}^n$ with $T\psi = r \psi$.  Any other vector with this property is a scalar multiple of $\psi$.  Furthermore, any nonnegative vector that is an eigenvector of $T$ must be a scalar multiple of $\psi$.
\end{theorem}

The only conclusion of this theorem that's weaker than those of Perron's theorem is that there may be other eigenvalues with $|\lambda| = r$.  For example, this matrix is irreducible\index{matrix!irreducible} and nonnegative\index{matrix!nonnegative}:
$$ \left( \begin{array}{cc} 0 & 1 \\ 1 & 0 \end{array} \right) $$
Its Perron--Frobenius eigenvalue is 1, but it also has -1 as an eigenvalue.    In general, Perron--Frobenius theory says quite a lot about the other eigenvalues on the circle $|\lambda| = r,$ but we won't need that fancy stuff here.

Perron--Frobenius theory is useful in many ways, from highbrow math to ranking football teams.   We'll need it not just now but also in Section \ref{sec:22}.  There are many books and other sources of information for those that want to take a closer look at this subject.  If you're interested, you can \href{https://www.google.com/search?q=perron-frobenius}{search online} or take a look at these:

\begin{itemize} 
\item  Dimitrious Noutsos, \href{http://www.math.uoi.gr/~dnoutsos/Papers_pdf_files/slide_perron.pdf}{Perron Frobenius theory and some extensions}, 2008.  (Includes proofs of the basic theorems.)

\item  V. S. Sunder, \href{http://www.imsc.res.in/~sunder/pf.pdf}{Perron Frobenius theory}, 18 December 2009.  (Includes applications to graph theory, Markov chains and von Neumann algebras.)

\item  Stephen Boyd, \href{http://www.stanford.edu/class/ee363/lectures/pf.pdf}{Lecture 17: Perron Frobenius theory}, Winter 2008-2009.  (Includes a max-min characterization of the Perron--Frobenius eigenvalue and applications to Markov chains, economics, population growth and power control.)

\end{itemize} 

If you are interested and can read German, the original work appears here:

\begin{enumerate} 
\item[\cite{Per07}]  Oskar Perron, Zur Theorie der Matrizen, \emph{Math. Ann.} {\bf 64} (1907), 248--63.

\item[\cite{Fro12}]  Georg Frobenius, \"{U}ber Matrizen aus nicht negativen Elemente, \emph{S.-B. Preuss Acad. Wiss. Berlin} (1912), 456--477.   
\end{enumerate}

\noindent And, of course, there's this:

\begin{itemize} 
\item Wikipedia, \href{http://en.wikipedia.org/wiki/Perron%E2%80%93Frobenius_theorem}{Perron--Frobenius theorem}.
\end{itemize} 

\noindent It's quite good.

\subsection{Irreducible Dirichlet operators}\index{Dirichlet operator!irreducible}

Now comes the payoff.\index{Dirichlet operator!irreducible}  We saw how to get a Dirichlet operator $H$ from any finite simple graph with edges labelled by positive numbers.  Now let's apply Perron--Frobenius theory to prove our thesis.

Unfortunately, the matrix $H$ is rarely nonnegative.  If you remember how we built it, you'll see its off-diagonal entries will always be nonnegative... but its diagonal entries can be negative.  

Luckily, we can fix this just by adding a big enough multiple of the identity matrix to $H$!   The result is a \index{matrix!nonnegative}nonnegative matrix
$$  T = H + c I $$
where $c > 0$ is some large number.  This matrix $T$ has the same eigenvectors as $H$.   The off-diagonal matrix entries of $T$ are the same as those of $H$, so $T_{i j}$ is nonzero for $i \ne j$ exactly when the graph we started with has an edge from $i$ to $j$.  So, for $i \ne j$, the graph $G_T$ will have an directed edge going from $i$ to $j$ precisely when our original graph had an edge from $i$ to $j$.  And that means that if our original graph was connected, $G_T$ will be strongly connected.  Thus, by definition, the matrix $T$ is irreducible!\index{strongly connected graph!and irreducibility} 

Since $T$ is nonnegative and irreducible, the Perron--Frobenius theorem swings into action and we conclude:

\begin{lemma} 
Suppose $H$ is the Dirichlet operator coming from a connected finite simple graph with edges labelled by positive numbers.  Then the eigenvalues of $H$ are real.  Let $\lambda$ be the largest eigenvalue.  Then there is a positive vector $\psi \in \mathbb{R}^n$ with $H\psi = \lambda \psi$.  Any other vector with this property is a scalar multiple of $\psi$.  Furthermore, any nonnegative vector that is an eigenvector of $H$ must be a scalar multiple of $\psi$.
\end{lemma} 

\begin{theorem}
 The eigenvalues of $H$ are real since $H$ is self-adjoint.  Notice that if $r$ is the Perron--Frobenius eigenvalue of $T = H + c I$ and
$$ T \psi = r \psi$$
then 
$$ H \psi = (r - c)\psi $$
By the Perron--Frobenius theorem the number $r$ is positive, and it has the largest absolute value of any eigenvalue of $T$.  Thanks to the subtraction, the eigenvalue $r - c$ may not have the largest absolute value of any eigenvalue of $H$.  It is, however, the largest eigenvalue of $H$, so we take this as our $\lambda$.  The rest follows from the Perron--Frobenius theorem.  
\end{theorem} 

But in fact we can improve this result, since the largest eigenvalue $\lambda$ is just zero.  Let's also make up a definition, to make our result sound more slick:

\begin{definition} 
A Dirichlet operator is {\bf irreducible} if it comes from a connected finite simple graph with edges labelled by positive numbers.
\end{definition} 

\noindent
This meshes nicely with our earlier definition of irreducibility for nonnegative matrices.  Now:

\begin{theorem} 
\label{theorem_equilibrium}
Suppose $H$ is an irreducible Dirichlet operator.  Then $H$ has zero as its largest real eigenvalue.  There is a positive vector $\psi \in \mathbb{R}^n$ with $H\psi = 0$.  Any other vector with this property is a scalar multiple of $\psi$.  Furthermore, any nonnegative vector that is an eigenvector of $H$ must be a scalar multiple of $\psi$.
\end{theorem} 

\begin{proof} 
Choose $\lambda$ as in the Lemma, so that $H\psi = \lambda \psi$.  Since $\psi$ is positive we can normalize it to be a stochastic state:
$$ \sum_i \psi_i = 1 $$
Since $H$ is a Dirichlet operator, $\exp(t H)$ sends stochastic states to stochastic states, so
$$  \sum_i (\exp(t H) \psi)_i = 1 $$
for all $t \ge 0$.  On the other hand,
$$ \sum_i (\exp(t H)\psi)_i = \sum_i e^{t \lambda} \psi_i = e^{t \lambda} $$
so we must have $\lambda = 0$.  
\end{proof} 

What's the point of all this?  One point is that there's a unique stochastic state $\psi$ that's an {\bf equilibrium}, meaning that $H \psi = 0$, so doesn't change with time.  It's also {\bf globally stable}: since all the other eigenvalues of $H$ are negative, all other stochastic states converge to this one as time goes forward.

\subsection{An example}

There are many examples of irreducible Dirichlet operators.  For instance, in Section~\ref{sec:14} we talked about graph Laplacians.  The Laplacian of a connected simple graph is always irreducible.  But let us try a different sort of example, coming from the picture of the resistors we saw earlier:

\begin{center}
 \includegraphics[width=60mm]{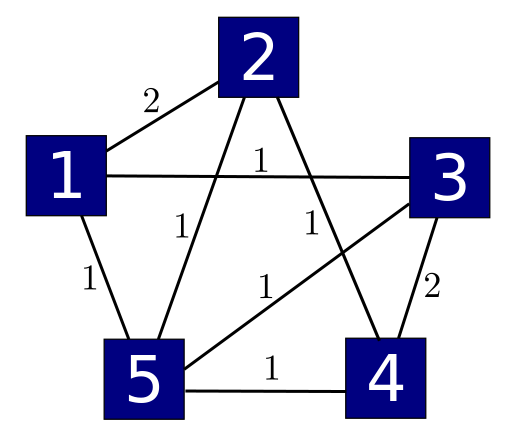}
\end{center}

\noindent
Let's create a matrix $A$ whose entry $A_{i j}$ is the number labelling the edge from $i$ to $j$ if there is such an edge, and zero otherwise:
$$A = \left(
\begin{array}{ccccc}
 0 & 2 & 1 & 0 & 1 \\
 2 & 0 & 0 & 1 & 1 \\
 1 & 0 & 0 & 2 & 1 \\
 0 & 1 & 2 & 0 & 1 \\
 1 & 1 & 1 & 1 & 0
\end{array}
\right) $$
Remember how the game works.  The matrix $A$ is already a valid Hamiltonian for quantum mechanics, since it's self adjoint.  However, to get a valid Hamiltonian for both stochastic and quantum mechanics---in other words, a Dirichlet operator---we subtract the diagonal matrix $L$ whose entries are the sums of the columns of $A.$  In this example it just so happens that the column sums are all 4, so $L = 4 I,$ and our Dirichlet operator is
$$ H = A - 4 I = \left(
\begin{array}{ccccc}
 -4 & 2 & 1 & 0 & 1 \\
 2 & -4 & 0 & 1 & 1 \\
 1 & 0 & -4 & 2 & 1 \\
 0 & 1 & 2 & -4 & 1 \\
 1 & 1 & 1 & 1 & -4
\end{array}
\right) $$
We've set up this example so it's easy to see that the vector $\psi = (1,1,1,1,1)$ has
$$  H \psi = 0 $$
So, this is the unique eigenvector for the eigenvalue 0.  We can use Mathematica to calculate the remaining eigenvalues of $H$.  The set of eigenvalues is
$$\{0, -7, -8, -8, -3 \} $$
As we expect from our theorem, the largest real eigenvalue is 0.   By design, the eigenstate associated to this eigenvalue is
$$ | v_0 \rangle = (1, 1, 1, 1, 1) $$
(This funny notation for vectors is common in quantum mechanics, so don't worry about it.)   All the other eigenvectors fail to be nonnegative, as predicted by the theorem.  They are:  
$$ \begin{array}{ccl} | v_1 \rangle &=& (1, -1, -1, 1, 0),\\
 | v_2 \rangle &=& (-1, 0, -1, 0, 2), \\ 
 | v_3 \rangle &=& (-1, 1, -1, 1, 0), \\
 | v_4 \rangle &=& (-1, -1, 1, 1, 0). \end{array} $$
To compare the quantum and stochastic states, consider first $ |v_0\rangle$.  This is the only eigenvector that can be normalized to a stochastic state.  Remember, a stochastic state must have nonnegative components.  This rules out $ |v_1\rangle$ through to $ |v_4\rangle$ as valid stochastic states, no matter how we normalize them!  However, these are allowed as states in quantum mechanics, once we normalize them correctly.  For a stochastic system to be in a state other than the Perron--Frobenius state, it must be a linear combination of least two eigenstates.  For instance,  
$$  \psi_a = (1-a)|v_0\rangle + a |v_1\rangle $$
can be normalized to give stochastic state only if $ 0 \leq a \leq \frac{1}{2}$.  

And, it's easy to see that it works this way for any irreducible Dirichlet operator, thanks to our theorem.  So, our thesis has been proved true!

\subsection{Problems} 

Let us conclude with a few more problems.  There are lots of ways to characterize irreducible nonnegative matrices; we don't need to mention graphs.  Here's one:

\begin{problem}\label{prob:25} 
Let $T$ be a nonnegative $n \times n$ matrix.  Show that $T$ is irreducible if and only if for all $i,j \ge 0$, $(T^m)_{i j} > 0$ for some natural number $m$.
\end{problem} 
\index{matrix!irreducible} \index{irreducible matrix}

You may be confused because we explained the usual concept of irreducibility for nonnegative matrices, but we also defined a concept of irreducibility for Dirichlet operators.  Luckily there's no conflict: Dirichlet operators aren't nonnegative matrices, but if we add a big multiple of the identity to a Dirichlet operator it becomes a nonnegative matrix, and then:

\begin{problem}\label{prob:26} 
Show that a Dirichlet operator $H$ is irreducible if and only if the nonnegative operator $H + c I$ (where $c$ is any sufficiently large constant) is irreducible.
 \end{problem} 

Irreducibility is also related to the nonexistence of interesting conserved quantities. In Section~\ref{sec:10} we saw a version of Noether's Theorem for stochastic mechanics.  Remember that an observable $O$ in stochastic mechanics assigns a number $O_i$ to each configuration $i = 1, \dots, n$.  We can make a diagonal matrix with $O_i$ as its diagonal entries, and by abuse of language we call this $O$ as well.  Then we say $O$ is a {\bf conserved quantity} for the Hamiltonian $H$ if the commutator $[O,H] = O H - H O$ vanishes.

\begin{problem}\label{prob:27} 
Let $H$ be a Dirichlet operator.  Show that $H$ is irreducible if and only if every conserved quantity $O$ for $H$ is a constant, meaning that for some  $c \in \mathbb{R}$ we have $O_i = c$ for all $i$.  (Hint: examine the proof of Noether's theorem.)
\end{problem} 

In fact this works more generally:

\begin{problem}\label{prob:28} 
Let $H$ be an infinitesimal stochastic matrix. 
\index{operator!infinitesimal stochastic} \index{infinitesimal stochastic operator} 
Show that $H + c I$ is an irreducible nonnegative matrix for all sufficiently large $c$ if and only if every conserved quantity $O$ for $H$ is a constant.
\end{problem} 

\subsection{Answers}
\label{sec:16_answers}

Here are the answers to the problems in this section, thanks to 
Arjun Jain.  \index{Jain, Arjun}

\vskip 1em \noindent {\bf Problem 23.} 
Suppose $U$ is a stochastic $n \times n$ matrix whose inverse is stochastic.  What are the possibilities for $U$?

\begin{answer}
\index{matrix!permutation} \index{permutation matrix}
The matrix $U$ must be a {\bf permutation matrix}: that is, all its entries must be $0$ 
except for a single $1$ in each row and each column.  This follows from the answer to 
Problem \ref{prob:16}, which appears in Section \ref{sec:11_answers}.  

In fact, Problem 23 is just Problem \ref{prob:16} stated in a less fancy way---but this 
prompted Arjun Jain \index{Jain, Arjun} to find the following less fancy solution. 
To see this, suppose that $U$ is a stochastic $n \times n$ matrix whose inverse, $V$, is also stochastic.   In particular, we have   
$$ U_{i j},  V_{i j} \geq 0  , \qquad \displaystyle{ \sum_i U_{i j} = 1 } $$
and
$$ \displaystyle{ \sum_k U_{i k} V_{k j} = \delta_{i j} }$$

Fix $i$ and $j$.  First, suppose $i \neq j$. We have
$$ \displaystyle{ \sum_k U_{i k} V_{k j} = 0 } $$
Since $U_{i k}, V_{k j} \geq 0$, for every choice of $k$ either $U_{i k}$ or $V_{k j}$ must be zero.  It follows that if some entry $U_{i k}$ is nonzero, then $V_{k j}$ must be zero for all $j$ except for $j = i$.  Further if $U_{i' k}$ is also nonzero where $i' \neq i$, then $V_{k j}$ must be zero except for $j = i'$. Since $i' \neq i$, $V_{k i}$ is also zero, and so is $V_{k i'}$. Thus, if any two entries of a column in $U$ were nonzero, some row of $V$ would be zero, implying that $V$ is noninvertible. 

As this is not the case, at most one entry in each column of $U$ can be nonzero.  Furthermore, since the sum of the entries in each column is $1$:
$$\displaystyle{ \sum_i U_{i j} = 1 }$$ 
the one nonzero entry in each column of $U$ must equal $1$.

Next suppose $i=j$.  Now we have
$$\displaystyle{ \sum_k U_{i k} V_{k i} = 1 }$$
Since each column of $U$ has just one nonzero entry,  for each $k$ we have one choice of index $I$ such that $U_{I j} \neq 0$. Suppose that $i$ is not equal to this choice of $I$ for any $k$. Then 
$$\displaystyle{ \sum_k U_{i k} V_{k i} = 0 } $$ 
which is a contradiction.  Thus, for every $i$ we have a unique $k$ such that $U_{i k}$ is nonzero.  In other words, there is exactly one nonzero entry in each row of $U$.

In short, a stochastic matrix with a stochastic inverse must be a permutation matrix.  Conversely, it is easy to check that any permutation matrix is stochastic, with its
inverse being another permutation matrix, and thus also stochastic.
\end{answer}

\vskip 1em \noindent {\bf Problem 24.} 
Suppose $U$ is an $n \times n$ matrix that is both stochastic and unitary.  What are the possibilities for $U$?
\index{matrix!stochastic} \index{stochastic matrix}
\index{matrix!unitary} \index{unitary matrix}

\begin{answer}
In the solution to Problem 23 we only used the fact that $U$ is a stochastic matrix
whose inverse has nonnegative entries.   If $U$ is stochastic and unitary this is
true, because its inverse is its transpose.  So, the argument in Problem 23 shows
that a stochastic unitary matrix must be a permutation matrix.  Conversely,
it is easy to check that any permutation matrix is stochastic and unitary.
\index{matrix!permutation} \index{permutation matrix}
\end{answer}

\vskip 1em \noindent {\bf Problem 25.}
Let $T$ be a nonnegative $n \times n$ matrix.  Show that $T$ is irreducible if and only if for all $i,j \ge 0$, $(T^m)_{i j} > 0$ for some natural number $m$.
\index{matrix!irreducible} \index{irreducible matrix}

\begin{answer}
As in Section \ref{sec:16_matrices_to_graphs}, we build a directed graph $G_T$ with $n$ vertices, with an edge from the $i$th vertex to the $j$th if and only if $T_{ij} 
\ne 0$.  We need to show this graph is strongly connected if and only if for all $i,j \ge 0$, 
$(T^m)_{i j} > 0$ for some natural number $m$.   Recall that a directed graph is {\bf strongly connected} when there is a directed edge path from any vertex to any other.

Notice that 
$$T^2_{i j} = \displaystyle{ \sum_k T_{i k} T_{k j} }$$ 
If $T_{i k} T_{k j}$ is nonzero, it means that neither $T_{i k}$ nor $T_{k j}$ is zero: that is, the graph $G_T$ has an edge from $i$ to $k$ and an edge from $k$ to $j$.  Thus $T_{i k} T_{k j}$ is nonzero if and only if there is a directed 2-edge path from $i$ to $j$.  Since all entries of $T$ are nonnegative, $T^2_{i j}$ is zero only if there is no directed 2-edge path from $i$ to $j$.  If it is nonzero, there is at least one directed 2-edge path from $i$ to $j$.

Next consider
$$ T^3_{i j} = \displaystyle{ \sum_k T^2_{i k} T_{k j} }$$ 
This is nonzero if and only if there is at least one directed 2-edge path from $i$ to $k$ and a directed edge from $k$ to $j$---and thus a directed 3-edge path from $i$ to $j$.

Using induction, we can show that $T^m_{i j}$ is nonzero if and only if the graph $G_T$ has a directed $m$-edge path from $i$ to $j$.  Thus, this graph is strongly connected if and only if $i,j \ge 0$, $(T^m)_{i j} > 0$ for some natural number $m$.
\index{strongly connected graph} \index{graph!strongly connected}
\end{answer}

\vskip 1em \noindent {\bf Problem 26.} 
Show that a Dirichlet operator $H$ is irreducible if and only if the nonnegative operator $H + c I$ (where $c$ is any sufficiently large constant) is irreducible.
\index{Dirichlet operator!irreducible} \index{irreducible Dirichlet operator}
\index{matrix!irreducible} \index{irreducible matrix}

\begin{answer}
Recall that by definition, an $n \times n$ Dirichlet operator $H$ is irreducible if and only 
if the simple graph $\Gamma_H$ with an edge from $i$ to $j$ when $H_{i j} > 0$ is connected.  Similarly, the nonnegative operator $T = H + c I$ is irreducible if and only if the directed graph $G_T$ with an edge from $i$ to $j$ when $T_{i j} > 0$ is strongly connected.   

Since the matrix $H$ is symmetric, so is $T$, so there is no difference between connectedness and strong connectedness for $G_T$.   Only the
off-diagonal entries of $H$ matter for determining the connectedness of $\Gamma_H$, and similarly, only the off-diagonal entries of $T$ matter for determining the connectedness of $G_T$.  Since $H$ and $T$ have the same off-diagonal entries,
$\Gamma_H$ is connected if and only if $G_T$ is.
\end{answer}

\vskip 1em \noindent {\bf Problem 27.}
Let $H$ be a Dirichlet operator.  Show that $H$ is irreducible if and only if every conserved quantity $O$ for $H$ is a constant, meaning that for some  $c \in \mathbb{R}$ we have $O_i = c$ for all $i$.  
\index{operator!Dirichlet} \index{Dirichlet operator}

\begin{answer}
This is a special case of the next problem.
\end{answer}

\vskip 1em \noindent {\bf Problem 28.}
Let $H$ be an infinitesimal stochastic matrix. 
\index{operator!infinitesimal stochastic} \index{infinitesimal stochastic operator} 
Show that $H + c I$ is an irreducible nonnegative matrix for all sufficiently large $c$ if and only if every conserved quantity $O$ for $H$ is a constant.

\begin{answer}
Let $T = H + c I$ where $c$ is chosen large enough to make $T$ be nonnegative.
If $O$ is a conserved quantity for $H$, then $O H = H O$, implying that $O H^m = H^m O$ for all $m \in \mathbb{N}$.  Since a scalar multiply of the identity matrix commutes with everything, it follows that $O T^m = T^m O$ for all $m$.  As $O$ is a diagonal matrix we have
$$   (O_{ii} - O_{jj}) T^m_{i j} = 0 $$
for all $i,j$.

If $T$ is irreducible, then from Problem \ref{prob:25} there exists a natural number m such that $T^m_{i j} > 0$ for every $i,j$ with $i \neq j$.  By the above equation this implies $(O_{ii} - O_{jj}) = 0$. Thus, $O$ is a scalar multiple of the
identity, or in other words, a constant.

Conversely, assume that $T$ is not irreducible.  Let $G_T$ be the directed graph associated to $T$.  Then there is a diagonal matrix $O$ that is not a scalar multiple
of the identity such that $O_{ii}$ is constant on each strongly connected component
of $G_T$.  This implies that 
$$   (O_{ii} - O_{jj}) T^m_{i j} = 0 $$
for all $i,j$ and all $m \in \mathbb{N}$.  This in turn implies that $O T^m = T^m O$
for all $m$, and thus $O H = H O$.
\end{answer}

%%%%% SECTION 17 %%%%%%%

\newpage
\section{The deficiency zero theorem}\index{Deficiency Zero Theorem!introduction to|(} 
\label{sec:17}

We've seen how Petri nets can be used to describe chemical reactions.  Indeed our very first example way back in Section \ref{sec:2} came from chemistry:

\begin{center}
 \includegraphics[width=119.0625mm]{chemistryNetBasicA.png}
\end{center}

\noindent
However, chemists rarely use the formalism of Petri nets.  They use a different but entirely equivalent formalism, called `reaction networks'.  So now we'd like to tell you about those. 

You may wonder: why bother with another formalism, if it's equivalent to the one we've already seen?   Well, one goal of this network theory program is to get people from different subjects to talk to each other---or at least be \emph{able}  to.  This requires setting up some dictionaries to translate between formalisms.  Furthermore, lots of deep results on stochastic Petri nets are being proved by chemists---but phrased in terms of reaction networks.  So you need to learn this other formalism to read their papers.  Finally, this other formalism is actually better in some ways!

\subsection{Reaction networks}\index{reaction network|(} 
\label{sec:17_reaction_networks}

Here's a reaction network:

\begin{center}
 \includegraphics[width=35mm]{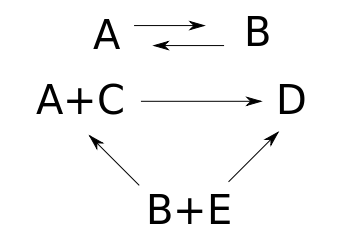}
\end{center}

This network involves five {\bf species}: that is, different kinds of things.  They could be atoms, molecules, ions or whatever: chemists call all of these \href{http://en.wikipedia.org/wiki/Chemical_species}{species}, but there's no need to limit the applications to chemistry: in population biology, they could even be biological species!  We're calling them A, B, C, D, and E, but in applications, we'd call them by specific names like CO$_2$ and HCO$_3$, or `rabbit' and `wolf', or whatever.  \index{species} \index{reaction network!species}

This network also involves five {\bf reactions}, which are shown as arrows.  Each reaction turns one bunch of species into another.  So, written out more longwindedly, we've got these reactions:
$$A \to B$$ 
$$B \to A$$
$$A + C \to D$$
$$B + E \to A + C$$
$$B + E \to D$$
If you remember how Petri nets work, you can see how to translate any reaction network into a Petri net, or vice versa.  For example, the reaction network we've just seen gives this Petri net:

\begin{center}
 \includegraphics[width=90mm]{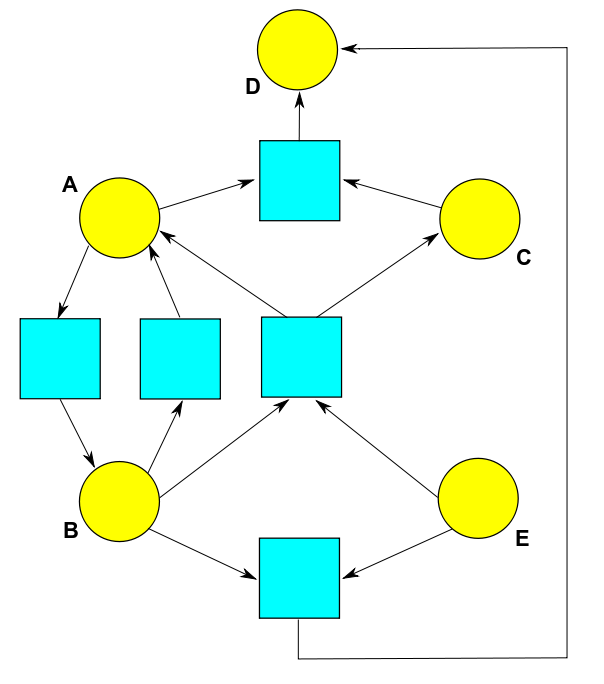}
\end{center}

Each species corresponds to a yellow circle in this Petri net.  And each reaction corresponds to {\bf transition} of this Petri net, drawn as a blue square.  The arrows say how many things of each species appear as input or output to each transition.  There's less explicit emphasis on complexes in the Petri net notation, but you can read them off if you want them.

In chemistry, a bunch of species is called a `complex'.  But what do we mean by `bunch', exactly?  Well, we mean that in a given complex, each species can show up 0,1,2,3...\ or any natural number of times. So, we can formalize things like this:

\begin{definition}  Given a set $S$ of {\bf species}, a {\bf complex} of those species is a function $C\colon S \to \mathbb{N}$.\index{complex} 
\end{definition}

Roughly speaking, a reaction network is a graph whose vertices are labelled by complexes.  Unfortunately, the word \href{http://en.wikipedia.org/wiki/Graph\_\%28mathematics\%29}{`graph'} means different things in mathematics---appallingly many things!  Everyone agrees that a graph has vertices and edges, but there are lots of choices about the details.  Most notably:

\begin{itemize} 
\item  We can either put arrows on the edges, or not.
\item  We can either allow more than one edge between vertices, or not.
\item  We can either allow edges from a vertex to itself, or not.
\end{itemize} 

If we say `no' in every case we get the concept of `simple graph', which we discussed in Sections~\ref{sec:15} and \ref{sec:16}.  At the other extreme, if we say `yes' in every case we get the concept of `directed multigraph', which is what we want now.  A bit more formally:

\begin{definition}\index{directed multigraph}  
\index{directed multigraph} \index{graph theory!directed multigraph}
\index{graph theory!vertex} \index{graph theory!edge}
A {\bf directed multigraph} consists of a set $V$ of {\bf vertices}, a set $E$ of {\bf edges},\index{vertex}\index{edge} and functions $s,t\colon E \to V$ saying the {\bf source} and {\bf target} of each edge.  \index{source} \index{target} 
\end{definition} 

Given this, we can say:

\begin{definition}\index{reaction network!definition of}  
A {\bf reaction network} is a set of species together with a directed multigraph whose vertices are labelled by complexes of those species.
\end{definition} 

You can now prove that reaction networks are equivalent to Petri nets:

\begin{problem}\label{prob:29}  
Show that any reaction network gives rise to a Petri net, and vice versa.  
\end{problem} 

In a stochastic Petri net each transition is labelled by a {\bf rate constant}: that is, a number in $(0,\infty)$.  This lets us write down some differential equations saying how species turn into each other.  So, let's make this definition (which is not standard, but will clarify things for us):

\begin{definition}\index{reaction network!stochastic!definition of}  
A {\bf stochastic reaction network} is a reaction network where each reaction is labelled by a rate constant. 
\end{definition} 

Now you can do this:

\begin{problem}\label{prob:30} 
Show that any stochastic reaction network gives rise to a stochastic Petri net, and vice versa.  
\end{problem} 

For extra credit, show that in each of these problems we actually get an equivalence of categories! For this you need to define morphisms between Petri nets, morphisms between reaction networks, and similarly for stochastic Petric nets and stochastic reaction networks.  If you get stuck, ask Eugene Lerman for advice.  There are different ways to define morphisms, but he knows a \href{http://arxiv.org/abs/1008.5359}{cool one}.

We've been downplaying category theory so far, but it's been lurking beneath everything we do, and someday it may rise to the surface.

\subsection{The deficiency zero theorem}
\label{sec:17_theorem}

You may already have noticed one advantage of reaction networks over Petri nets: they're quicker to draw.  This is true even for tiny examples.  For instance, this reaction network:
$$2 X_1 + X_2 \leftrightarrow 2 X_3 $$
corresponds to this Petri net:

\begin{center}
 \includegraphics[width=88.5mm]{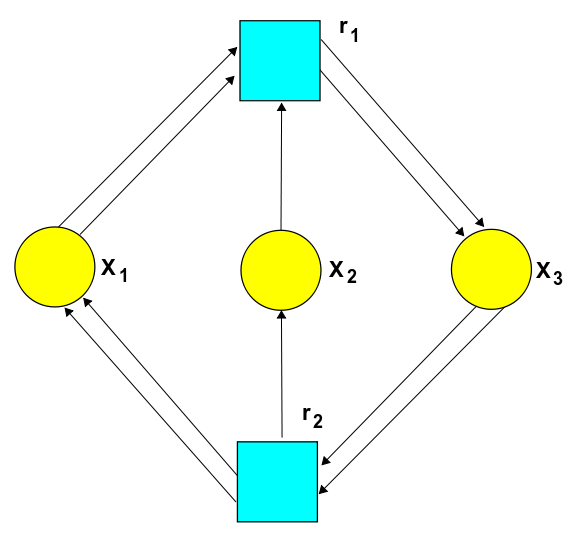}
\end{center}

\noindent
But there's also a deeper advantage.  As we know, any stochastic Petri net gives two equations:

\begin{itemize} 
\item  The \href{http://en.wikipedia.org/wiki/Master\_equation}{{\bf master equation}}, which says how the \emph{probability that we have a given number of things of each species}  changes with time.

\item  The \href{http://en.wikipedia.org/wiki/Rate\_equation}{{\bf rate equation}}, which says how the \emph{expected number of things of each species} changes with time.   
\end{itemize} 

The simplest solutions of these equations are the equilibrium solutions, where nothing depends on time. Back in Section~\ref{sec:8}, we explained when an equilibrium solution of the rate equation gives an equilibrium solution of the master equation.    But when is there an equilibrium solution of the rate equation in the first place?  

The `deficiency zero theorem' gives a handy sufficient condition.  And this condition is best stated using reaction networks!  But to understand it, we need to understand the `deficiency' of a reaction network.  So let's define that, and then say what all the words in the definition mean:

\begin{definition} \index{reaction network deficiency!definition of}\index{reaction network!deficiency of} 
The {\bf deficiency} of a reaction network is:
\begin{itemize} 
\item  the number of vertices minus 
\item  the number of connected components minus\index{connected component!deficiency}  
\item  the dimension of the stoichiometric subspace.  
\end{itemize} 
\end{definition} 

The first two concepts here are easy.  A reaction network is a graph (okay, a directed multigraph).  So, it has some number of vertices, and also some number of connected components.  Two vertices lie in the same {\bf connected component} \index{graph theory!connected component} \index{connected component} iff you can get from one to the other by a path \emph{where you don't care which way the arrows point}.  For example, this reaction network:

\begin{center}
 \includegraphics[width=35mm]{chemical_reaction_network_part_17_I.png}
\end{center}

\noindent
has 5 vertices and 2 connected components.  

So, what's the `stoichiometric subspace'?  `\href{http://en.wikipedia.org/wiki/Stoichiometry}{Stoichiometry}' is a scary-sounding word. It comes from the Greek words \emph{stoicheion}, meaning element, and \emph{metron}, meaning measure. In chemistry, it's the study of the relative quantities of reactants and products in chemical reactions.   But for us, stoichiometry is just the art of counting species.  To do this, we can form a\index{vector space} vector space $\mathbb{R}^S$ where $S$ is the set of species.  A vector in $\mathbb{R}^S$ is a function from species to real numbers, saying how much of each species is present.  Any complex gives a vector in $\mathbb{R}^S$, because it's actually a function from species to \emph{natural}  numbers.  

\begin{definition} 
The {\bf stoichiometric subspace} of a reaction network is the subspace $\mathrm{Stoch} \subseteq \mathbb{R}^S$ spanned by vectors of the form $x - y$ where $x$ and $y$ are complexes connected by a reaction.
\end{definition} 

`Complexes connected by a reaction' makes sense because vertices in the reaction network are complexes, and edges are reactions.  Let's see how it works in our example:

\begin{center}
 \includegraphics[width=35mm]{chemical_reaction_network_part_17_I.png}
\end{center}

\noindent
Each complex here can be seen as a vector in $\mathbb{R}^S$, which is a vector space whose basis we can call $A, B, C, D, E$.  Each reaction gives a difference of two vectors coming from complexes:

\begin{itemize} 
\item  The reaction $A \to B$ gives the vector $B - A.$

\item  The reaction $B \to A$ gives the vector $A - B.$

\item  The reaction $A + C \to D$ gives the vector $D - A - C.$

\item  The reaction $B + E \to A + C$ gives the vector $A + C - B - E.$

\item  The reaction $B + E \to D$ gives the vector $D - B - E.$
\end{itemize} 
\noindent
The pattern is obvious, we hope.

These 5 vectors span the stoichiometric subspace.   But this subspace isn't 5-dimensional, because these vectors are linearly dependent!  The first vector is the negative of the second one.  The last is the sum of the previous two.   And those are all the linear dependencies, so the stoichiometric subspace is 3 dimensional.  For example, it's spanned by these 3 linearly independent vectors: $A - B, D - A - C,$ and $D - B - E$.  

We hope you see the moral of this example: the stoichiometric subspace is the space of \emph{ways to move in}  $\mathbb{R}^S$ \emph{that are allowed by the reactions in our reaction network!}    And this is important because the rate equation describes how the amount of each species changes as time passes.  So, it describes a point moving around in $\mathbb{R}^S$.

Thus, if $\mathrm{Stoch} \subseteq \mathbb{R}^S$ is the stoichiometric subspace, and $x(t) \in \mathbb{R}^S$ is a solution of the rate equation, then $x(t)$ always stays within the set
$$x(0) + \mathrm{Stoch} = \{ x(0) + y \colon  y \in \mathrm{Stoch} \} $$
Mathematicians would call this set the \href{http://en.wikipedia.org/wiki/Coset}{coset} of $x(0)$, but chemists call it the {\bf stoichiometric compatibility class} of $x(0).$ \index{stoichiometric!compatibility class}

Anyway: what's the deficiency of the reaction network in our example?  It's 
$$5 - 2 - 3 = 0$$
since there are 5 complexes, 2 connected components and the dimension of the stoichiometric subspace is 3.   

But what's the deficiency zero theorem?  You're almost ready for it.  You just need to know \emph{one}  more piece of jargon!  A reaction network is {\bf weakly reversible} if whenever there's a reaction going from a complex $x$ to a complex $y$, there's a path of reactions going back from $y$ to $x$.  Here the paths need to follow the arrows.\index{reaction network!weakly reversible}\index{weakly reversible!reaction network} 

So, this reaction network is \emph{not}  weakly reversible:

\begin{center}
 \includegraphics[width=40mm]{chemical_reaction_network_part_17_I.png}
\end{center}

\noindent
since we can get from $A + C$ to $D$ but not back from $D$ to $A + C,$ and from $B+E$ to $D$ but not back, and so on.  However, the network becomes weakly reversible if we add a reaction going back from $D$ to $B + E$:

\begin{center}
 \includegraphics[width=40mm]{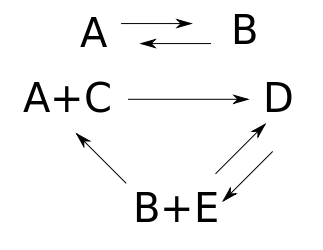}
\end{center}

If a reaction network isn't weakly reversible, one complex can  turn into another, but not vice versa.  In this situation, what typically happens is that as time goes on we have less and less of one species.  We could have an equilibrium where there's \emph{none}  of this species.  But we have little right to expect an equilibrium solution of the rate equation that's {\bf positive}, meaning that it sits at a point $x \in (0,\infty)^S$, where there's a nonzero amount of every species.  

The argument here is not watertight: you'll note that we fudged the difference between species and complexes.   But it can be made so when our reaction network has deficiency zero:

\begin{theorem}[{\bf Deficiency Zero Theorem}]\index{Deficiency Zero Theorem!statement of} 
  Suppose we are given a reaction network with a finite set of species $S$, and suppose its deficiency is zero.  Then: 
\begin{enumerate} 
\item[(i)] If the network is not weakly reversible and the rate constants are positive, the rate equation does not have a positive equilibrium solution.
\item[(ii)] If the network is not weakly reversible and the rate constants are positive, the rate equation does not have a positive periodic solution, that is, a periodic solution lying in $(0,\infty)^S$. 
\item[(iii)] If the network is weakly reversible and the rate constants are positive, the rate equation has exactly one equilibrium solution in each positive stoichiometric compatibility class.   This equilibrium solution is complex balanced.  Any sufficiently nearby solution that starts in the same stoichiometric compatibility class will approach this equilibrium as $t \to +\infty$.   Furthermore, there are no other positive periodic solutions.
\index{stoichiometric!compatibility class}
\end{enumerate}
\end{theorem} 

This is quite an impressive result.  Even better, the `complex balanced' condition means we can instantly turn the equilibrium solutions of the rate equation we get from this theorem into equilibrium solutions of the master equation, thanks to the Anderson--Craciun--Kurtz theorem!  If you don't remember what we're talking about here, go back to Section~\ref{sec:8}.

We'll look at an easy example of this theorem in Section \ref{sec:18}.

\subsection{References and remarks}

The deficiency zero theorem appears here:

\begin{enumerate} 
\item[\cite{Fei87}]  Martin Feinberg, \href{http://www.seas.upenn.edu/~jadbabai/ESE680/Fei87a.pdf}{Chemical reaction network structure and the stability of complex isothermal reactors: I. The deficiency zero and deficiency one theorems}, \emph{Chemical Engineering Science}  {\bf 42} (1987), 2229--2268. 

\item[\cite{HJ72}] F.\ Horn and Roy Jackson, General mass action kinetics, {\sl Archive for Rational Mechanics and Analysis}, {\bf 47},  81--116, 1972. 
\end{enumerate} 

You can learn more about chemical reaction networks and the deficiency zero theorem here:

\begin{enumerate} 
\item[\cite{Fei79}]  Martin Feinberg, \href{http://www.che.eng.ohio-state.edu/~FEINBERG/LecturesOnReactionNetworks/}{\sl Lectures on Reaction Networks}, 1979. 

\item[\cite{Gun03}]  Jeremy Gunawardena, \href{http://vcp.med.harvard.edu/papers/crnt.pdf}{Chemical reaction network theory for \emph{in-silico} biologists}, 2003. 
\end{enumerate} 

At first glance the deficiency zero theorem might seem to settle all the basic questions about the dynamics of reaction networks, or stochastic Petri nets... but actually, it just means that deficiency zero reaction networks don't display very interesting dynamics in the limit as $t \to +\infty.$  So, to get more interesting behavior, we need to look at reaction networks that don't have deficiency zero.  

For example, in biology it's interesting to think about `bistable' chemical reactions: reactions that have two stable equilibria.  An electrical switch of the usual sort is a bistable system: it has stable `on' and `off' positions.  A bistable chemical reaction can serve as a kind of biological switch:\index{bistability} 

\begin{enumerate} 
\item[\cite{CTF06}]  Gheorghe Craciun, Yangzhong Tang and Martin Feinberg, \href{http://www.pnas.org/content/103/23/8697.abstract}{Understanding bistability in complex enzyme-driven reaction networks}, \href{http://www.pnas.org/content/108/11/4281.abstract}{ \emph{PNAS}}  {\bf 103} (2006), 8697--8702.  
\end{enumerate} 

It's also interesting to think about chemical reactions with stable periodic solutions.  Such a reaction can serve as a biological clock:

\begin{enumerate} 
\item[\cite{For11}]  Daniel B. Forger, \href{http://www.pnas.org/content/108/11/4281.abstract}{Signal processing in cellular clocks}, \href{http://www.pnas.org/content/108/11/4281.abstract}{\emph{PNAS}}  {\bf 108} (2011), 4281--4285. 
\end{enumerate} 
\index{Deficiency Zero Theorem!introduction to|)} 

\index{reaction network|)} 

%%%%%% SECTION 18 %%%%%%%% 

\newpage
\section[Example of the deficiency zero theorem]{Example of the deficiency zero theorem}\label{sec:18}\index{Deficiency Zero Theorem!example of|(}  

In the last section we explained how `reaction networks', as used in chemistry, are just another way of talking about Petri nets.  We stated an amazing result on reaction networks: the deficiency zero theorem.  This settles quite a number of questions about chemical reactions.  Now let's illustrate it with an example.

Our example won't show how \emph{powerful} this theorem is: it's too simple.  But it'll help explain the ideas involved.

\subsection{Diatomic molecules}\index{chemistry!diatomic molecules|(}

A diatomic molecule consists of two atoms of the same kind, stuck together:

\begin{center}
 \includegraphics[width=40mm]{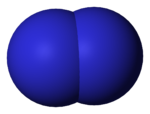}
\end{center}

\noindent At room temperature there are 5 elements that are diatomic gases: hydrogen, nitrogen, oxygen, fluorine, chlorine.  Bromine is a diatomic liquid, but easily evaporates into a diatomic gas:\index{chemistry!bromine}

%<div align = "center}<img border = "2" width = "230" src = "http://math.ucr.edu/home/baez/networks/bromine.jpg" alt = ""/></div>

\begin{center}
 \includegraphics[width=40mm]{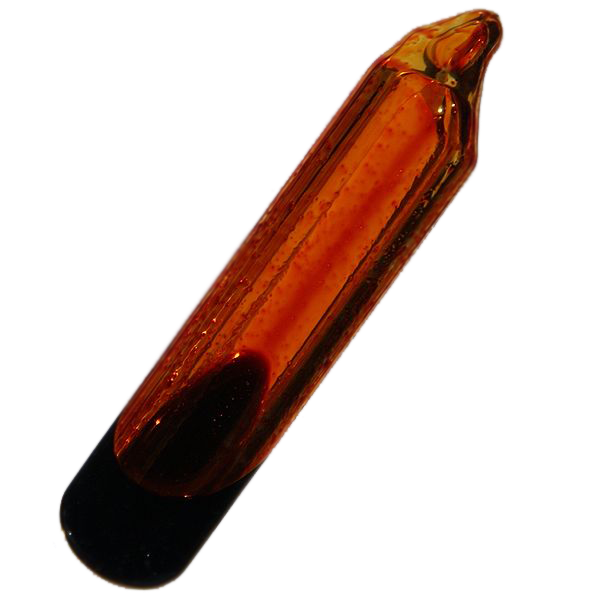}
\end{center}

\noindent
Iodine is a crystal at room temperatures:\index{chemistry!iodine}

\begin{center}
 \includegraphics[width=40mm]{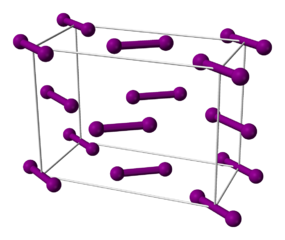}
\end{center}

\noindent
but if you heat it a bit, it becomes a diatomic liquid and then a gas:

\begin{center}
 \includegraphics[width=40mm]{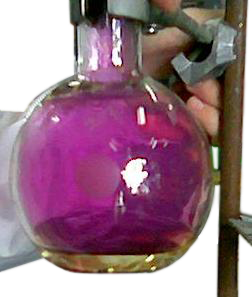}
\end{center}

\noindent
so people often list it as a seventh member of the diatomic club.  

When you heat any diatomic gas enough, it starts becoming a `monatomic' gas as molecules break down into individual atoms.  However, just as a diatomic molecule can break apart into two  atoms:
$$   A_2 \to A + A $$
two atoms can recombine to form a diatomic molecule:
$$  A + A \to A_2 $$
So in equilibrium, the gas will be a mixture of diatomic and monatomic forms.   The exact amount of each will depend on the temperature and pressure, since these affect the likelihood that two colliding atoms stick together, or a diatomic molecule splits apart.   The detailed nature of our gas also matters, of course.

But we don't need to get into these details here!  Instead, we can just write down the `rate equation' for the reactions we're talking about. All the details we're ignoring will be hiding in some constants called `rate constants'.  We won't try to compute these; we'll leave that to our chemist friends.

\subsection{A reaction network} 

To write down our rate equation, we start by drawing a `reaction network'.    For this, we can be a bit abstract and call the diatomic molecule $B$ instead of $A_2$.  Then it looks like this:

%<img width = "200" src = "http://math.ucr.edu/home/baez/networks/chemical_reaction_network_part_18.png" alt = ""/>

\begin{center}
 \includegraphics[width=40mm]{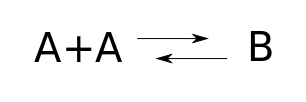}
\end{center}

We could write down the same information using a Petri net:

%<img width = "400" src = "http://math.ucr.edu/home/baez/networks/petri_net_part_18.png" alt = ""/>

\begin{center}
 \includegraphics[width=92mm]{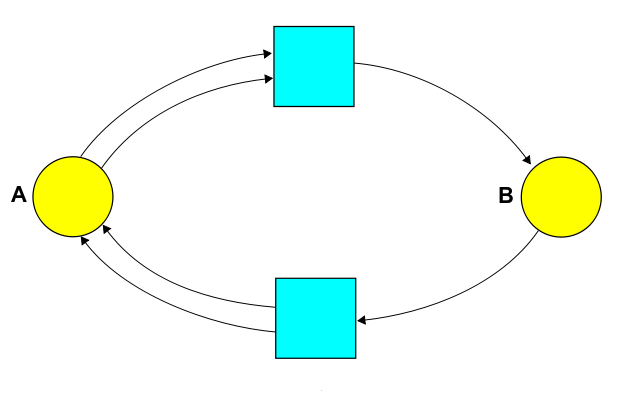}
\end{center}

But now let's focus on the reaction network!  Staring at this picture, we can read off various things: 

\begin{itemize} 
\item {\bf Species.}  The species are the different kinds of atoms, molecules, etc.  In our example the set of species is $S = \{A, B\}$.  

\item {\bf Complexes.} A complex is a finite sum of species, like $A$, or $A + A$, or for a fancier example using more efficient notation, $2 A + 3 B$.    So, we can think of a complex as a vector $v \in \mathbb{R}^S$.  The complexes that actually show up in our reaction network form a set $C \subseteq \mathbb{R}^S$.  In our example, $C = \{A+A, B\}$. 

\item {\bf Reactions.}  A reaction is an arrow going from one complex to another.  In our example we have two reactions: $A + A \to B$ and $B \to A + A$.    
\end{itemize}

Chemists define a \emph{reaction network} to be a triple $(S, C, T)$ where $S$ is a set of species, $C$ is the set of complexes that appear in the reactions, and $T$ is the set of reactions $v \to w$ where $v, w \in C$.  (Stochastic Petri net people call reactions \emph{transitions}, hence the letter $T$.)

So, in our example we have:

\begin{itemize}
\item Species: $S = \{A,B\}$.  
\item Complexes: $C= \{A+A, B\}$.  
\item Reactions: $T =  \{A+A\to B, B\to A+A\}$.  
\end{itemize}  
  
To get the rate equation, we also need one more piece of information: a \emph{rate constant} $r(\tau)$ for each reaction $\tau \in T$.  This is a nonnegative real number that affects how fast the reaction goes.  All the details of how our particular diatomic gas behaves at a given temperature and pressure are packed into these constants!  

\subsection{The rate equation}

The rate equation says how the expected numbers of the various species, atoms, molecules and the like changes with time.  This equation is deterministic.   It's a good approximation when the numbers are large and any fluctuations in these numbers are negligible by comparison. 

Here's the general form of the rate equation:
$$ \frac{dx_i}{d t} =  \sum_{\tau\in T} r(\tau) \, (n_i(\tau)-m_i(\tau)) \, x^{m(\tau)} $$ 

Let's take a closer look.  The quantity $x_i$ is the expected population of the $i$th species.  So, this equation tells us how that changes.   But what about the right hand side?  As you might expect, it's a sum over reactions.  And:

\begin{itemize} 
\item The term for the reaction $\tau$ is proportional to the rate constant $r(\tau)$. 

\item Each reaction $\tau$ goes between two complexes, so we can write it as $m(\tau) \to n(\tau)$.  Among chemists the input $m(\tau)$ is called the \emph{reactant complex}, and the output is called the \emph{product complex}.  The difference $n_i(\tau)-m_i(\tau)$ tells us how many items of species $i$ get created, minus how many get destroyed.  So, it's the net amount of this species that gets produced by the reaction $\tau$.  The term for the reaction $\tau$ is proportional to this, too.

\item Finally, the \href{http://en.wikipedia.org/wiki/Law_of_mass_action}{law of mass action} says that the rate of a reaction is proportional to the product of the concentrations of the species that enter as inputs.  More precisely, if we have a reaction $\tau$ where the input is the complex $m(\tau)$, we define $ x^{m(\tau)} = x_1^{m_1(\tau)} \cdots x_k^{m_k(\tau)}$.  The law of mass action says the term for the reaction $\tau$ is proportional to this, too!
\end{itemize} 

Let's see what this says for the reaction network we're studying:
  
%<img width = "200" src = "http://math.ucr.edu/home/baez/networks/chemical_reaction_network_part_18.png" alt = ""/>
\begin{center}
 \includegraphics[width=40mm]{chemical_reaction_network_part_18.png}
\end{center}

Let's write $x_1(t)$ for the number of $A$ atoms and $x_2(t)$ for the number of $B$ molecules.  Let the rate constant for the reaction $B \to A + A$ be $\alpha$, and let the rate constant for $A + A \to B$ be $\beta$.   Then the rate equation is this:
$$ \frac{d}{d t} x_1 =  2 \alpha x_2 - 2 \beta x_1^2 $$ 
$$ \frac{d}{d t} x_2 = -\alpha x_2 + \beta x_1^2 $$
This is a bit intimidating.  However, we can solve it in closed form thanks to something very precious: a \emph{conserved quantity}.    \index{conserved quantity}

We've got two species, $A$ and $B$.  But remember, $B$ is just an abbreviation for a molecule made of two $A$ atoms.   So, the total number of $A$ atoms is conserved by the reactions in our network.  This is the number of $A$'s plus twice the number of $B$'s: $x_1 + 2x_2$.  So, this should be a \emph{conserved quantity}: it should not change with time.  Indeed, by adding the first equation above to twice the second, we see:
$$ \frac{d}{d t} (x_1 + 2x_2) = 0 $$ 
As a consequence, any solution will stay on a line 
$$ x_1 + 2 x_2 = c $$
for some constant $c$.  We can use this fact to rewrite the rate equation just in terms of $x_1$:
$$ \frac{d}{d t} x_1 = \alpha (2c - x_1) - 2 \beta x_1^2 $$ 

This is a separable differential equation, so we can solve it if we can figure out how to do this integral 
$$ t = \int \frac{d x_1}{\alpha (2c - x_1) - 2 \beta x_1^2 } $$
and then solve for $x_1$.

This sort of trick won't work for more complicated examples. 
But the idea remains important: the numbers of atoms of various kinds---hydrogen, helium, lithium, and so on---are conserved by chemical reactions, so a solution of the rate equation can't roam freely in $\mathbb{R}^S$.  It will be trapped on some hypersurface, which is called a `stoichiometric compatibility class'.  And this is very important.

We don't feel like doing the integral required to solve our rate equation in closed form, because this idea doesn't generalize too much.  On the other hand, we can always solve the rate equation numerically.  So let's try that!

For example, suppose we set $\alpha = \beta = 1$.  We can plot the solutions for three different choices of initial conditions, say $(x_1,x_2) = (0,3), (4,0),$ and $(3,3)$.  We get these graphs:

\begin{center}
 \includegraphics[width=\textwidth]{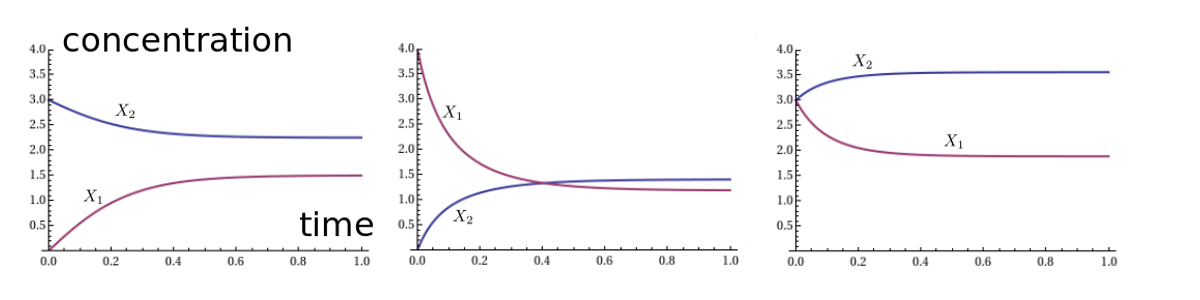}
\end{center}

\noindent
It looks like the solution always approaches an equilibrium.   We seem to be getting different equilibria for different initial conditions, and the pattern is a bit mysterious.   However, something nice happens when we plot the ratio $x_1^2 / x_2$:

%<img src = "http://math.ucr.edu/home/baez/networks/part-18-ratio-plots.png" alt = ""/>

\begin{center}
 \includegraphics[width=\textwidth]{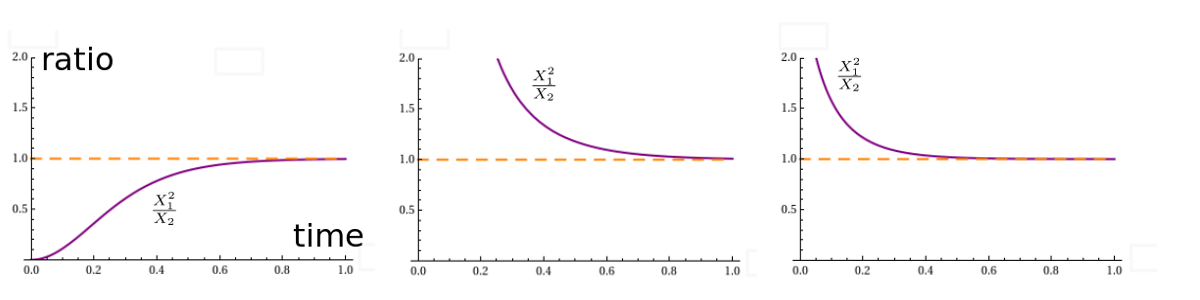}
\end{center}

Apparently it always converges to 1.   Why should that be?  It's not terribly surprising.  With both rate constants equal to 1, the reaction $A + A \to B$ proceeds at a rate equal to the square of the number of $A$'s, namely $x_1^2$.  The reverse reaction proceeds at a rate equal to the number of $B$'s, namely $x_2$.  So in equilibrium, we should have $x_1^2 = x_2$.

But why is the equilibrium \emph{stable}?  In this example we could see that using the closed-form solution, or maybe just common sense.   But it also follows from a powerful theorem that handles a \emph{lot} of reaction networks.

\subsection[Deficiency zero theorem]{The deficiency zero theorem} 
\index{Deficiency Zero Theorem|(}

It's called the deficiency zero theorem, and we saw it in Section \ref{sec:17_theorem}.  Very roughly, it says that if our reaction network is `weakly reversible' and has `deficiency zero', the rate equation will have equilibrium solutions that behave about as nicely as you could want.  

Let's see how this works.   We need to remember some jargon:

\begin{itemize} 

\item {\bf Weakly reversible.} A reaction network is {\bf weakly reversible} if for every reaction $v \to w$ in the network, there exists a path of reactions in the network starting at $w$ and leading back to $v$.

\item {\bf Reversible.}   A reaction network is {\bf reversible} if for every reaction $v \to w$ in the network, $w \to v$ is also a reaction in the network.    Any reversible reaction network is weakly reversible.   Our example is reversible, since it consists of reactions $A + A \to B$, $B \to A + A$.\index{reaction network!reversible}\index{reversible reaction network}  

\end{itemize} 

But what about `deficiency zero'?  We defined that concept in Section~\ref{sec:18}, but let's review:

\begin{itemize} 

\item {\bf Connected component.}\index{connected component!definition of}
  A reaction network gives a kind of graph with complexes as vertices and reactions as edges.  Two complexes lie in the same  {\bf connected component} if we can get from one to the other by a path of reactions, where at each step we're allowed to go either forward \emph{or backward} along a reaction.  Chemists call a connected component a {\bf linkage class}.  In our example there's just one:

%<div align = "center}<img width = "200" src = "http://math.ucr.edu/home/baez/networks/chemical_reaction_network_part_18.png" alt = ""/></div> 
\begin{center}
 \includegraphics[width=40mm]{chemical_reaction_network_part_18.png}
\end{center}

\item {\bf Stoichiometric subspace.}\index{stoichiometric!subspace}\index{reaction network!stoichiometric subspace of} 
  The {\bf stoichiometric subspace} is the subspace $\mathrm{Stoch} \subseteq \mathbb{R}^S$ spanned by the vectors of the form $w - v$ for all reactions $v \to w$ in our reaction network.    This subspace describes the directions in which a solution of the rate equation can move.  In our example, it's spanned by $B - 2 A$ and $2 A - B$, or if you prefer, $(-2,1)$ and $(2,-1)$.  These vectors are linearly dependent, so the stoichiometric subspace has dimension 1.  

\item {\bf Deficiency.}\index{deficiency}\index{reaction network!deficiency of}
  The {\bf deficiency} of a reaction network is the number of complexes, minus the number of connected components, minus the dimension of the stoichiometric subspace.  In our example there are 2 complexes, 1 connected component, and the dimension of the stoichiometric subspace is 1.  So, our reaction network has deficiency 2 - 1 - 1 = 0.  

\end{itemize} 

So, the deficiency zero theorem applies!  What does it say?  To understand it, we need a bit more jargon.  First of all, a vector $x \in \mathbb{R}^S$ tells us how much we've got of each species: the amount of species $i \in S$ is the number $x_i$.  And then:

\begin{itemize} 

\item {\bf Stoichiometric compatibility class.}\index{stoichiometric!compatibility class}\index{reaction network!stoichiometric compatibility class}  Given a vector $v\in \mathbb{R}^S$, its {\bf stoichiometric compatibility class} is the subset of all vectors that we could reach using the reactions in our reaction network:
$$ \{ v + w \; : \; w \in \mathrm{Stoch} \} $$
\end{itemize} 

In our example, where the stoichiometric subspace is spanned by $(2,-1)$, the stoichiometric compatibility class of the vector $(a,b)$ is the line consisting of points
$$ (x_1, x_2) = (a,b) + s(2,-1) $$
where the parameter $s$ ranges over all real numbers.  Notice that this line can also be written as
$$ x_1 + 2x_2 = c $$
We've already seen that if we start with initial conditions on such a line, the solution will stay on this line.  And that's how it always works: as time passes, any solution of the rate equation stays in the same stoichiometric compatibility class!  

In other words: the stoichiometric subspace is defined by a bunch of linear equations, one for each linear conservation law that all the reactions in our network obey.    Here a {\bf linear conservation law} is a law saying that some linear combination of the numbers of species does not change.

Next: 

\begin{itemize} 

\item {\bf Positivity.}\index{stoichiometric!compatibility class!positivity}  A vector in $\mathbb{R}^S$ is {\bf positive} if all its components are positive; this describes a container of chemicals where all the species are actually present.  The {\bf positive stoichiometric compatibility class} of $x\in \mathbb{R}^S$ consists of all positive vectors in its stoichiometric compatibility class.

\end{itemize} 

We finally have enough jargon in our arsenal to state the deficiency zero theorem.  We'll only state the part we need now:

\begin{theorem}[{\bf Deficiency Zero Theorem}]\index{Deficiency Zero Theorem!statement of}  If a reaction network is weakly reversible and the rate constants are positive, the rate equation has exactly one equilibrium solution in each positive stoichiometric compatibility class.  Any sufficiently nearby solution that starts in the same class will approach this equilibrium as $t \to +\infty$. 
\end{theorem} 

In our example, this theorem says there's just one positive
equilibrium $(x_1,x_2)$ in each line
$$ x_1 + 2x_2 = c $$
We can find it by setting the time derivatives to zero:
$$ \frac{d}{d t} x_1 =  2 \alpha x_2 - 2 \beta x_1^2 = 0 $$ 
$$ \frac{d}{d t} x_2 = -\alpha x_2 + \beta x_1^2 = 0 $$
Solving these, we get
$$ \frac{x_1^2}{x_2} = \frac{\alpha}{\beta} $$ 

So, these are our equilibrium solutions.  It's easy to verify that indeed, there's one of these in each stoichiometric compatibility class $x_1 + 2x_2 = c$.  And the deficiency zero theorem also tells us that any sufficiently nearby solution that starts in the same class will approach this equilibrium as $t \to \infty$.\index{Deficiency Zero Theorem} 

This partially explains what we saw before in our graphs.  It shows that in the case $\alpha = \beta = 1$, any solution that starts by \emph{nearly} having
$$ \frac{x_1^2}{x_2} = 1 $$
will actually have
$$ \lim_{t \to +\infty} \frac{x_1^2}{x_2} = 1 $$
But in fact, in this example we don't even need to start \emph{near} the equilibrium for our solution to approach the equilibrium!  

What about in general?  For many years this was an open question:

\begin{conjecture}[{\bf Global Attractor Conjecture}]\index{Global Attractor Conjecture|(}  If a reaction network is weakly reversible and the rate constants are positive, the rate equation has exactly one equilibrium solution in each positive stoichiometric compatibility class, and \emph{any} positive solution that starts in the same class will approach this equilibrium as $t \to +\infty$. 
\end{conjecture}

In their groundbreaking 1972 paper \cite{HJ72}, Horn and Jackson thought they had proved this, but in 1974 Horn realized they had not:

\begin{enumerate}
\item[\cite{Hor74}] Fritz Horn, The dynamics of open reaction systems, in \textsl{Mathematical Aspects of Chemical and Biochemical Problems and Quantum Chemistry}, ed.\ Donald S.\ Cohen, \textsl{SIAM--AMS Proceedings} \textbf{8}, American Mathematical Society, Providence, R.I., 1974, pp.\ 125--137.
\end{enumerate}

\noindent
It was dubbed the `global attractor conjecture' in the following paper, which proved it in a special case:

\begin{enumerate} 
\item[\cite{CDSS07}] Gheorghe Craciun, Alicia Dickenstein, Anne Shiu and Bernd Sturmfels, Toric dynamical systems.  Available as \href{http://arxiv.org/abs/0708.3431}{arXiv:0708.3431}.
\end{enumerate} 

\noindent
In 2011 Anderson proved the conjecture for reaction networks with a single connected component, or `linkage class':

\begin{enumerate} 
\item[\cite{And11}] David F.\ Anderson, A proof of the Global Attractor Conjecture in the single linkage class case.  Available as \href{http://arxiv.org/abs/1101.0761}{arXiv:1101.0761}.
\end{enumerate} 

\noindent
In 2015, Craciun proved the conjecture in general:

\begin{enumerate}
\item[\cite{Cra15}] Gheorghe Craciun, Toric differential inclusions and a proof of the Global Attractor Conjecture.  Available as \href{https://arxiv.org/abs/1501.02860}{arXiv:1501.02860}.
\end{enumerate}

\index{Deficiency Zero Theorem|)}
\index{Global Attractor Conjecture|)} 
\index{Deficiency Zero Theorem!introduction to|)} 

%%%%%%%% SECTION 19 %%%%%%%%%%%%% 

\newpage
\section[Example of the Anderson--Craciun--Kurtz theorem]{Example of the Anderson--Craciun--Kurtz \\ theorem}\label{sec:19}\index{Anderson--Craciun--Kurtz theorem!example of|(} 

In Section \ref{sec:18} we started looking at a simple example: a diatomic gas.  
%<div align="center}<img src="http://math.ucr.edu/home/baez/networks/nitrogen.png" alt="" /></div>

\begin{center}
 \includegraphics[width=40mm]{nitrogen.png}
\end{center}

\noindent
A diatomic molecule of this gas can break apart into two atoms:
$$   A_2 \to A + A $$
and conversely, two atoms can combine to form a diatomic molecule:
$$  A + A \to A_2 $$
We can draw both these reactions using a chemical reaction network:

%<img width="170" src="http://math.ucr.edu/home/baez/networks/chemical_reaction_network_part_18.png" alt="" />
\begin{center}
 \includegraphics[width=40mm]{chemical_reaction_network_part_18.png}
\end{center}

\noindent
where we're writing $B$ instead of $A_2$ to abstract away some detail that's just distracting here. 

In Section \ref{sec:18} we looked at the rate equation for this chemical reaction network, and found equilibrium solutions of that equation.  Now let's look at the master equation, and find equilibrium solutions of that.  This will illustrate the Anderson--Craciun--Kurtz theorem.  We'll also see how the conservation law we noticed in the last section is related to Noether's theorem for Markov processes.

\subsection[The master equation]{The master equation}

We'll start from scratch.  The master equation is all about how atoms or molecules or rabbits or wolves or other things interact randomly and turn into other things.  So, let's write $\psi_{m,n}$ for the probability that we have $m$ atoms of $A$ and $n$ molecules of $B$ in our container.  These probabilities are functions of time, and the master equation will say how they change.

First we need to pick a {\bf rate constant} for each reaction.  Let's say the rate constant for the reaction that produces $A$s is some number $\alpha > 0$:
$$   B \to A + A $$
while the rate constant for the reaction that produces $B$s is some number $\beta > 0$:
$$  A + A \to B $$
Before we make it pretty using the ideas we've been explaining all along, the master equation says:
$$  \displaystyle{ \frac{d}{d t} \psi_{m,n} (t)} \; = \;  \alpha (n+1) \, \psi_{m-2,n+1} \; - \; \alpha n \, \psi_{m,n}  \; + \; \beta (m+2)(m+1) \, \psi_{m+2,n-1} \; - \;\beta m(m-1) \, \psi_{m,n} \; $$

Yuck!  Normally we don't show you such nasty equations.  Indeed the whole point of our work has been to demonstrate that by packaging the equations in a better way, we can understand them using high-level concepts instead of mucking around with millions of scribbled symbols.  But we thought we'd show you what's secretly lying behind our beautiful abstract formalism, just once. 

Each term has a meaning.  For example, the third one:
$$ \beta (m+2)(m+1)\psi_{m+2,n-1}(t) $$
means that the reaction $A + A \to B$ will tend to increase the probability of there being $m$ atoms of $A$ and $n$ molecules of $B$ if we start with $m+2$ atoms of $A$ and $n-1$ molecules of $B.$  This reaction can happen in $(m+2)(m+1)$ ways.   And it happens at a probabilistic rate proportional to the rate constant for this reaction, $\beta$.  

We won't go through the rest of the terms.  It's a good exercise to do so, but there could easily be a typo in the formula, since it's so long and messy.  So let us know if you find one!

To simplify this mess, the key trick is to introduce a {\bf generating function} that summarizes all the probabilities in a single power series:
$$ \Psi = \sum_{m,n \ge 0} \psi_{m,n} y^m \, z^n $$
It's a power series in two variables, $y$ and $z,$ since we have two chemical species: $A$s and $B$s.\index{power series}   

Using this trick, the master equation looks like
$$ \displaystyle{ \frac{d}{d t} \Psi(t) = H \Psi(t) } $$
where the {\bf Hamiltonian} $H$ is a sum of terms, one for each reaction.  This Hamiltonian is built from operators that annihilate and create $A$s and $B$s.  The annihilation and creation operators for $A$ atoms are:
$$ \displaystyle{ a = \frac{\partial}{\partial y} , \qquad a^\dagger = y } $$
The annihilation operator differentiates our power series with respect to the variable $y.$  The creation operator multiplies it by that variable.  Similarly, the annihilation and creation operators for $B$ molecules are:
$$ \displaystyle{ b = \frac{\partial}{\partial z} , \qquad b^\dagger = z } $$
In Section~\ref{sec:7_master} we explained a recipe that lets us stare at our chemical reaction network and write down this Hamiltonian:
$$ H = \alpha ({a^\dagger}^2 b - b^\dagger b) + \beta (b^\dagger a^2 - {a^\dagger}^2 a^2) $$
As promised, there's one term for each reaction.  But each term is itself a sum of two: one that increases the probability that our container of chemicals will be in a new state, and another that decreases the probability that it's in its original state.  We get a total of four terms, which correspond to the four terms in our previous way of writing the master equation.

\begin{problem}\label{prob:31} 
Show that this way of writing the master equation is equivalent to the previous one. 
\end{problem}

\subsection[Equilibrium solutions]{Equilibrium solutions}

Now we will look for all {\bf equilibrium} solutions of the master equation: in other words, solutions that don't change with time.  So, we're trying to solve
$$ H \Psi = 0 $$
Given the rather complicated form of the Hamiltonian, this seems tough.  The challenge looks more concrete but even more scary if we go back to our original formulation.  We're looking for probabilities $\psi_{m,n},$ nonnegative numbers that sum to one, such that
$$  \alpha (n+1) \, \psi_{m-2,n+1} \; - \; \alpha n \, \psi_{m,n}  \; + \; \beta (m+2)(m+1) \, \psi_{m+2,n-1} \; - \;\beta m(m-1) \, \psi_{m,n} = 0 $$
This equation is horrid!  But the good news is that it's \emph{linear}, so a linear combination of solutions is again a solution.  This lets us simplify the problem using a conserved quantity.

Clearly, there's a quantity that the reactions here don't change:

%<img width="170" src="http://math.ucr.edu/home/baez/networks/chemical_reaction_network_part_18.png" alt="" />
\begin{center}
 \includegraphics[width=40mm]{chemical_reaction_network_part_18.png}
\end{center}

\noindent
What's that?  It's the number of $A$s plus \emph{twice} the number of $B$s.  After all, a $B$ can turn into two $A$s, or vice versa.  So a $B$ counts twice.

Of course the secret reason is that $B$ is a diatomic molecule made of two $A$s.  But you'd be able to follow the logic here even if you didn't know that, just by looking at the chemical reaction network...\ and sometimes this more abstract approach is handy!  Indeed, the way chemists first discovered that certain molecules are made of certain atoms is by seeing which reactions were possible and which weren't.

\index{chemistry!diatomic molecules|)}

Suppose we start in a situation where we know \emph{for sure} that the number of $A$s plus twice the number of $B$s equals some number $k$:  
$$   \psi_{m,n} = 0  \;  \textrm{unless} \; m+2n = k $$
Then we know $\Psi$ is initially of the form
$$  \Psi = \sum_{m+2n = k} \psi_{m,n} \, y^m z^n   $$
But since the number of $A$s plus twice the number of $B$s is conserved, if $\Psi$ obeys the master equation it will \emph{continue} to be of this form!

Put a fancier way, we know that if a solution of the master equation starts in this subspace:
$$  L_k =  \{ \Psi  \;: \; \Psi = \sum_{m+2n = k} \psi_{m,n} y^m z^n \; \textrm{for some} \; \psi_{m,n} \}  $$
it will \emph{stay} in this subspace.  So, because the master equation is linear, we can take any solution $\Psi$ and write it as a linear combination of solutions $\Psi_k,$ one in each subspace $L_k.$

In particular, we can do this for an equilibrium solution $\Psi.$  And then all the solutions $\Psi_k$ are also equilibrium solutions: they're linearly independent, so if one of them changed with time, $\Psi$ would too.

This means we can just look for equilibrium solutions in the subspaces $L_k.$ If we find these, we can get \emph{all} equilibrium solutions by taking linear combinations.

Once we've noticed that, our horrid equation makes a bit more sense:
$$  \alpha (n+1) \, \psi_{m-2,n+1} \; - \; \alpha n \, \psi_{m,n}  \; + \; \beta (m+2)(m+1) \, \psi_{m+2,n-1} \; - \;\beta m(m-1) \, \psi_{m,n} = 0 $$
Note that if the pair of subscripts $m, n$ obey $m + 2n = k,$ the same is true for the other pairs of subscripts here!  So our equation relates the values of $\psi_{m,n}$ for all the points $(m,n)$ with integer coordinates lying on this line segment:
$$ m+2n = k , \qquad m ,n \ge 0 $$
 
You should be visualizing something like this:

%<div align="center}<img width="250" src="http://math.ucr.edu/home/baez/networks/diatomic_gas_subspace.png" /></div>
\begin{center}
 \includegraphics[width=40mm]{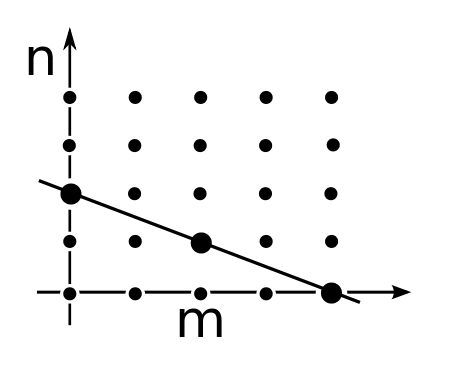}
\end{center}

\noindent
If you think about it a minute, you'll see that if we know $\psi_{m,n}$ at two points on such a line, we can keep using our equation to recursively work out all the rest.  So, there are \emph{at most two} linearly independent equilibrium solutions of the master equation in each subspace $L_k.$

Why \emph{at most} two?  Why not two?   Well, we have to be a bit careful about what happens at the ends of the line segment: remember that $\psi_{m,n}$ is defined to be zero when $m$ or $n$ becomes negative.   If we think very hard about this, we'll see there's just \emph{one} linearly independent equilibrium solution of the master equation in each subspace $L_k.$ But this is the sort of nitty-gritty calculation that's not fun to watch someone else do, so we won't bore you with that.

Soon we'll move on to a more high-level approach to this problem.  But first, one remark.  Our horrid equation is like a fancy version of the usual discretized form of the equation
$$ \displaystyle {\frac{d^2 \psi}{d x^2} = 0 } $$
namely:
$$  \psi_{n-1} - 2 \psi_{n} + \psi_{n+1} = 0 $$
And this makes sense, since we get
$$ \displaystyle {\frac{d^2 \psi}{d x^2} = 0 } $$
by taking the \href{http://en.wikipedia.org/wiki/Heat_equation}{\bf{heat equation}}:
$$ \displaystyle \frac{\partial \psi}{\partial t} = {\frac{\partial^2 \psi}{\partial x^2} } $$
and assuming $\psi$ doesn't depend on time.  So what we're doing is a lot like looking for equilibrium solutions of the heat equation.  The heat equation describes how heat smears out as little particles of heat randomly move around.  True, there don't really exist `little particles of heat', but this equation also describes the diffusion\index{diffusion}  of any other kind of particles as they randomly move around undergoing Brownian motion.\index{Brownian motion}   Similarly, our master equation describes a random walk\index{random walk!on line segment} on this line segment:
$$ m+2n = k , \qquad m , n \ge 0 $$
or more precisely, the points on this segment with integer coordinates.  The equilibrium solutions arise when the probabilities $\psi_{m,n}$ have diffused as much as possible.  

If you think about it this way, it should be physically obvious that there's just \emph{one} linearly independent equilibrium solution of the master equation for each value of $k.$ 

There's a general moral here, too, which we're seeing in a special case: the master equation for a chemical reaction network really describes a bunch of random walks, one for each allowed value of the conserved quantities that can be built as linear combinations of number operators.  In our case we have one such conserved quantity, but in general there may be more (or none).  

Furthermore, these `random walks' are what we've been calling Markov processes in Section~\ref{sec:10}.\index{random walk!as Markov processes}

\subsection[Noether's theorem]{Noether's theorem}

We simplified our task of finding equilibrium solutions of the master equation by finding a conserved quantity.  The idea of simplifying problems using conserved quantities is fundamental to physics: this is why physicists are so enamored with quantities like energy, momentum, angular momentum and so on.  

Nowadays physicists often use `Noether's theorem' to get conserved quantities from symmetries.  There's a very simple version of Noether's theorem for quantum mechanics, but in Section~\ref{sec:10} we saw a version for stochastic mechanics, and it's that version that is relevant now.  We don't really \emph{need} Noether's theorem now, since we found the conserved quantity and exploited it without even noticing the symmetry.  Nonetheless it's interesting to see how it relates to what we're doing.  

For the reaction we're looking at now, the idea is that the subspaces $L_k$ are eigenspaces of an operator that commutes with the Hamiltonian $H.$  It follows from standard math that a solution of the master equation that starts in one of these subspaces, stays in that subspace.

What is this operator?  It's built from `number operators'.  The {\bf number operator} for $A$s is
$$    N_A = a^\dagger a $$
and the number operator for $B$s is 
$$    N_B = b^\dagger b $$
A little calculation shows
$$  N_A \,y^m z^n = m \, y^m z^n, \quad \qquad  N_B\, y^m z^n = n \,y^m z^n $$
so the eigenvalue of $N_A$ is the number of $A$s, while the eigenvalue of $N_B$ is the number of $B$s.  This is why they're called number operators. 
\index{number operator}\index{quantum field theory!number operator}

As a consequence, the eigenvalue of the operator $N_A + 2N_B$ is the number of $A$s plus twice the number of $B$s:
$$  (N_A + 2N_B) \, y^m z^n = (m + 2n) \, y^m z^n $$
Let's call this operator $O,$ since it's so important:
$$ O = N_A + 2N_B $$
If you think about it, the spaces $L_k$ we saw a minute ago are precisely the eigenspaces of this operator:
$$  L_k = \{ \Psi \; : \; O \Psi = k \Psi \} $$
As we've seen, solutions of the master equation that start in one of these eigenspaces will stay there.  This lets us take some techniques that are very familiar in quantum mechanics, and apply them to this stochastic situation.  

First of all, time evolution as described by the master equation is given by the operators $\exp(t H).$  In other words, 
$$ \displaystyle{ \frac{d}{d t} \Psi(t) } = H \Psi(t) \quad \textrm{and} \quad  \Psi(0) = \Phi \quad  \Rightarrow \quad \Psi(t) = \exp(t H) \Phi $$
But if you start in some eigenspace of $O,$ you stay there.  Thus if $\Phi$ is an eigenvector of $O,$ so is $\exp(t H) \Phi,$ with the same eigenvalue.  In other words,
$$  O \Phi = k \Phi  $$
implies
$$  O \exp(t H) \Phi = k \exp(t H) \Phi = \exp(t H) O \Phi $$
But since we can choose a basis consisting of eigenvectors of $O,$ we must have
$$ O \exp(t H) = \exp(t H) O $$
or, throwing caution to the winds and differentiating:
$$  O H = H O $$
So, as we'd expect from Noether's theorem, our conserved quantity commutes with the Hamiltonian!   This in turn implies that $H$ commutes with any polynomial in $O,$ which in turn suggests that
$$  \exp(s O) H = H \exp(s O) $$
and also
$$ \exp(s O) \exp(t H) = \exp(t H) \exp(s O) $$
The last equation says that $O$ generates a 1-parameter family of `symmetries': operators $\exp(s O)$ that commute with time evolution.  But what do these symmetries actually do?  Since
$$  O y^m z^n = (m + 2n) y^m z^n $$
we have
$$ \exp(s O) y^m z^n = e^{s(m + 2n)}\, y^m z^n $$
So, this symmetry takes any probability distribution $\psi_{m,n}$ and multiplies it by $e^{s(m + 2n)}.$

In other words, our symmetry multiplies the relative probability of finding our container of gas in a given state by a factor of $e^s$ for each $A$ atom, and by a factor of $e^{2s}$ for each $B$ molecule.   It might not seem obvious that this operation commutes with time evolution!  However, experts on chemical reaction theory are familiar with this fact.

Finally, a couple of technical points.  Starting where we said ``throwing caution to the winds", our treatment has not been rigorous, since $O$ and $H$ are unbounded operators, and these must be handled with caution.  Nonetheless, all the commutation relations we wrote down are true.

The operators $\exp(s O)$ are unbounded for positive $s.$  They're bounded for negative $s,$ so they give a 1-parameter \emph{semi}group of bounded operators.  But they're not stochastic operators: even for $s$ negative, they don't map probability distributions to probability distributions.  However, they do map any nonzero vector $\Psi$ with $\psi_{m,n} \ge 0$ to a vector $\exp(s O) \Psi$ with the same properties.  So, we can just normalize this vector and get a probability distribution.  The need for this normalization is why we spoke of \emph{relative} probabilities.   

\subsection{The Anderson--Craciun--Kurtz theorem}

Now we'll actually find all equilibrium solutions of the master equation in closed form.  To understand this final section, you really do need to remember some things we've discussed earlier.  In Section~\ref{sec:17} we considered the same chemical reaction network we're studying now, but we looked at its rate equation, which looks like this:
$$ \displaystyle{ \frac{d}{d t} x_1 =  2 \alpha x_2 - 2 \beta x_1^2} $$
$$ \displaystyle{ \frac{d}{d t} x_2 = - \alpha x_2 + \beta x_1^2 } $$
This describes how the number of $A$s and $B$s changes in the limit where there are lots of them and we can treat them as varying continuously, in a deterministic way.  The number of $A$s is $x_1,$ and the number of $B$s is $x_2.$ 

We saw that the quantity
$$  x_1 + 2 x_2 $$
is conserved, just as now we've seen that $N_A + 2 N_B$ is conserved.  We saw that the rate equation has one equilibrium solution for each choice of $x_1 + 2 x_2.$ And we saw that these equilibrium solutions obey
$$ \displaystyle{ \frac{x_1^2}{x_2} = \frac{\alpha}{\beta} } $$
The Anderson--Craciun--Kurtz theorem, discussed in Section~\ref{sec:8}, is a powerful result that gets equilibrium solution of the master equation from equilibrium solutions of the rate equation.  It only applies to equilibrium solutions that are `complex balanced', but that's okay:

\begin{problem}\label{prob:32} 
Show that the equilibrium solutions of the rate equation for the chemical reaction network 

%<img width="170" src="http://math.ucr.edu/home/baez/networks/chemical_reaction_network_part_18.png" alt="" />
\begin{center}
 \includegraphics[width=40mm]{chemical_reaction_network_part_18.png}
\end{center}

\noindent
are complex balanced.
\end{problem} 

So, given any equilibrium solution $(x_1,x_2)$ of our rate equation, we can hit it with the Anderson--Craciun--Kurtz theorem and get an equilibrium solution of the master equation!  And it looks like this:
$$ \displaystyle{  \Psi = e^{-(x_1 + x_2)} \, \sum_{m,n \ge 0} \frac{x_1^m x_2^n} {m! n! } \, y^m z^n } $$
In this solution, the probability distribution
$$ \displaystyle{ \psi_{m,n} = e^{-(x_1 + x_2)} \,  \frac{x_1^m x_2^n} {m! n! } } $$
is a product of Poisson distributions.  The factor in front is there to make the numbers $\psi_{m,n}$ add up to one.   And remember, $x_1, x_2$ are any nonnegative numbers with 
$$ \displaystyle{ \frac{x_1^2}{x_2} = \frac{\alpha}{\beta} } $$
So from all we've said, the above formula gives an explicit closed-form solution of the horrid equation 
$$  \alpha (m+2)(m+1) \, \psi_{m+2,n-1} \; - \;\alpha m(m-1) \, \psi_{m,n} 
\;  + \; \beta (n+1) \, \psi_{m-2,n+1} \; - \; \beta n \, \psi_{m,n} = 0 $$
That's pretty nice.  We found some solutions without ever doing any nasty calculations.

But we've really done better than getting \emph{some} equilibrium solutions of the master equation.  By restricting attention to $n,m$ with $m+2n = k,$ our formula for $\psi_{m,n}$ gives an equilibrium solution that lives in the eigenspace $L_k$:
$$  \displaystyle{  \Psi_k = e^{-(x_1 + x_2)} \, \sum_{m+2n =k} \frac{x_1^m x_2^n} {m! n! } \, y^m z^n } $$
And by what we've said, linear combinations of these give \emph{all} equilibrium solutions of the master equation.   

And we got them with very little work!  Despite all our fancy talk, we essentially just took the equilibrium solutions of the rate equation and plugged them into a straightforward formula to get equilibrium solutions of the master equation.  This is why the Anderson--Craciun--Kurtz theorem is so nice.  And of course we're looking at a very simple reaction network: for more complicated ones it becomes even better to use this theorem to avoid painful calculations.

We could go further.  For example, we could study nonequilibrium solutions using Feynman diagrams like this:

%<img width="350" src="http://math.ucr.edu/home/baez/networks/feynman_diagram_hydrogen.png" alt="" />
\begin{center}
 \includegraphics[width=80mm]{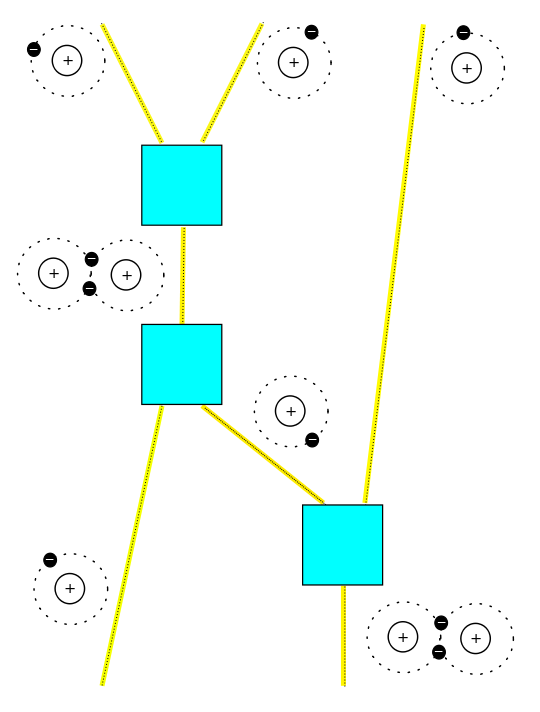}
\end{center}

But instead, we will leave off with another puzzle.  We introduced some symmetries, but we haven't really explored them yet:

\begin{problem}\label{prob:33} 
What do the symmetries associated to the conserved quantity $O$ do to the equilibrium solutions of the master equation given by
$$ \displaystyle{  \Psi = e^{-(x_1 + x_2)} \, \sum_{m,n \ge 0} \frac{x_1^m x_2^n} {m! n! } \, y^m z^n } $$
where $(x_1,x_2)$ is an equilibrium solution of the rate equation?  In other words, what is the significance of the 1-parameter family of solutions $ \exp(s O) \Psi$?
\end{problem}

Also, we used a conceptual argument to check that $H$ commutes with $O$, but it's good to know that we can check this sort of thing directly:

\begin{problem}\label{prob:34} 
Compute the commutator 
$$ [H, O] = H O - O H $$
and show it vanishes.
\end{problem}

\subsection{Answers}

\index{Egan, Greg}  \index{Jain, Arjun}
The answers to Problems 33 and 34 were written with help from Greg Egan
and Arjun Jain, respectively.

\vskip 1em \noindent {\bf Problem 33.} 
What do the symmetries associated to the conserved quantity $O$ do to the equilibrium solutions of the master equation given by
$$ \displaystyle{  \Psi = e^{-(x_1 + x_2)} \, \sum_{m,n \ge 0} \frac{x_1^m x_2^n} {m! n! } \, y^m z^n } $$
where $(x_1,x_2)$ is an equilibrium solution of the rate equation?  In other words, what is the significance of the 1-parameter family of solutions $ \exp(s O) \Psi$?

\begin{answer}
The symmetry $\exp(s O)$ maps the equilibrium solution of the master equation associated with the solution $(x_1, x_2)$ of the rate equation to that associated with $(e^s x_1, e^{2s} x_2)$.  Clearly the equation 
$$ \displaystyle{ \frac{x_1^2}{x_2} = \frac{\alpha}{\beta}} $$
is still satisfied by the new concentrations $x_1'=e^s x_1$ and $x_2'=e^{2s} x_2$.  

Indeed, the symmetries $\exp(s O)$ are related to a 1-parameter group of symmetries of the rate equation
$$  (x_1, x_2) \mapsto (x_1', x_2') = (e^s x_1, e^{2s} x_2) $$
These symmetries map the parabola of equilibrium solutions
$$  \displaystyle{ \frac{x_1^2}{x_2} = \frac{\alpha}{\beta}, \qquad x_1, x_2 \ge 0 } $$
to itself.  For example, if we have an equilibrium solution of the rate equation, we can multiply the number of lone atoms by 1.5 and multiply the number of molecules by 2.25, and get a new solution.

What's surprising is that this symmetry exists even when we consider small numbers of atoms and molecules, where we treat these numbers as integers instead of real numbers.  If we have 3 atoms, we can't multiply the number of atoms by 1.5.  So this is a bit shocking at first!

The trick is to treat the gas stochastically using the master equation rather than deterministically using the rate equation.  What our symmetry does is multiply the relative probability of finding our container of gas in a given state by a factor of $e^s$ for each lone atom, and by a factor of $e^{2s}$ for each molecule.

This symmetry commutes with time evolution as given by the master equation.  And for probability distributions that are products of Poisson distributions, this symmetry has the effect of multiplying the \emph{mean} number of lone atoms by $e^s$, and the \emph{mean} number of molecules by $e^{2s}$.

On the other hand, the symmetry $\exp(s O)$ maps each subspace $L_k$ to itself.  So this symmetry has the property that if we start in a state with a \emph{definite} total number of atoms (that is, lone atoms plus twice the number of molecules), it will map us to another state with the same total number of molecules!  

And if we start in a state with a definite number of lone atoms \emph{and} a definite number of molecules, the symmetry will leave this state completely unchanged!

These facts sound paradoxical at first, but of course they're not.  They're just a bit weird.   They're closely related to another weird fact, which however is well-known.  If we take a quantum system and start it off in an eigenstate of energy, it will never change, except for an unobservable phase.  Every state is a superposition of energy eigenstates.  So you might think that nothing can ever change in quantum mechanics.  But that's wrong: the phases that are unobservable in a single energy eigenstate become observable \emph{relative} phases in a superposition.  

Indeed the math is exactly the same, except now we're multiplying relative probabilities by positive real numbers, instead of multiplying amplitudes by complex numbers!
\end{answer}

\vskip 1em \noindent {\bf Problem 34.} 
Compute the commutator 
$$ [H, O] = H O - O H $$
and show it vanishes.

\begin{answer} 
One approach is to directly compute the commutator.  Recall that 
$$ H = \alpha ({a^\dagger}^2 b - b^\dagger b) + \beta (b^\dagger a^2 - {a^\dagger}^2 a^2) $$
and 
$$ O = a^\dagger a + 2 b^\dagger b. $$
Let us show that the term in $H$ proportional to $\alpha$ commutes with $O$.  For this
we need to show that
$$  ({a^\dagger}^2 b - b^\dagger b)(a^\dagger a + 2 b^\dagger b) - (a^\dagger a + 2 b^\dagger b)({a^\dagger}^2 b - b^\dagger b) $$
vanishes.
As $a$ and $a^\dagger$ commute with $b$ and $b^\dagger$, we can put the $a$ and $a^\dagger$ terms in front, obtaining
$$ ({a^\dagger}^3 a b - a^\dagger a b^\dagger b + 2 {a^\dagger}^2 b b^\dagger b - 2 b^\dagger b b^\dagger b) - (a^\dagger {a ^\dagger}^2 b - a^\dagger a b^\dagger b + 2 {a^\dagger}^2 b^\dagger b^2 - 2 b^\dagger b b^\dagger b) .$$
The second term cancels the sixth, and the fourth cancels the eighth, leaving
$$ a^\dagger ( {a^\dagger}^2 a + 2 a^\dagger b b^\dagger - a {a^\dagger}^2 - 2 a^\dagger b^\dagger b ) b$$
The second and fourth terms in the parentheses sum to $2 a^\dagger$, while the first and third terms give $ -2 a^\dagger$, so we get $0$.  A similar argument shows that the term in $H$ proportional to $ \beta$ commutes with $O$.

A more generally applicable strategy begins by showing that every formal power series is a (typically infinite) linear combination of vectors in the eigenspaces
$$  L_k = \{ \Psi \; : \; O \Psi = k \Psi \} $$ 
and noting that an operator commutes with $O$ if it maps each eigenspace to itself.  We can check that
\[      [N_A , a] = -a, \qquad [N_A, a^\dagger] = a, \qquad [N_A, b] = 0, \qquad
[N_A, b^\dagger] = 0 , \]
\[      [N_B , b] = -b, \qquad [N_B, b^\dagger] = b, \qquad [N_B, a] = 0, \qquad
[N_B, a^\dagger] = 0 , \]
and use these to show
$$ a \colon L_k \to L_{k-1} , \qquad a^\dagger \colon L_k \to L_{k+1} ,$$
$$ b \colon L_k \to L_{k-2} , \qquad b^\dagger \colon L_k \to L_{k+2} .$$
For example, if $\Psi \in L_k$ then $O\Psi = k \Psi$ so
$$ O b \Psi = (N_A + 2N_B) b \Psi = (b N_A + 2 (b N_B - 1))\Psi = (k-2)b\Psi $$
so $b \Psi \in L_{k-2}$.

Using this, we can see that each term in 
$$ H = \alpha ({a^\dagger}^2 b - b^\dagger b) + \beta (b^\dagger a^2 - {a^\dagger}^2 a^2) $$
commutes with $O$, so $[H,O] = 0$.  For example, consider the first term, involving
${a^\dagger}^2 b$.   If $\Psi \in L_k$ then $b\Psi \in L_{k-2}$, so $a^\dagger b \Psi \in L_{k-1}$ and ${a^\dagger}^2 b \in L_k$ again.  Thus $\alpha {a^\dagger}^2 b$ maps each eigenspace of $O$ to itself.  Each of the other terms does this as well.
\end{answer}

\index{Anderson--Craciun--Kurtz theorem!example of|)} 

%%%%%%% SECTION 20 %%%%%%%%% 

\newpage
\section{The deficiency of a reaction network}\label{sec:20} 

In the last few sections we've been talking about `reaction networks', like this:

%<div align = "center}
%<img width = "200" src = "http://math.ucr.edu/home/baez/networks/chemical_reaction_network_part_20_V.png" alt = ""/>
%</div>
\begin{center}
 \includegraphics[width=40mm]{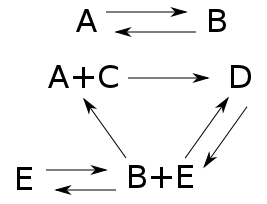}
\end{center}

\noindent
Here $A,B,C,D,$ and $E$ are names of chemicals, and the arrows are chemical reactions.   If we know how fast each reaction goes, we can write down a `rate equation' describing how the amount of each chemical changes with time.   

\index{Deficiency Zero Theorem|(}
In Section~\ref{sec:17} we met the deficiency zero theorem, a powerful tool for finding equilibrium solutions of the rate equation: in other words, solutions where the amounts of the various chemicals don't change at all.   To apply this theorem, two conditions must hold.   Both are quite easy to check:
\begin{itemize}
\item  Your reaction network needs to be `weakly reversible': if you have a reaction that takes one bunch of chemicals to another, there's a series of reactions that takes that other bunch back to the one you started with.  
\item  A number called the `deficiency' that you can compute from your reaction network needs to be zero.
\end{itemize}
The first condition makes a lot of sense, intuitively: you won't get an equilibrium with all the different chemicals present if some chemicals can turn into others but not the other way around.   But the second condition, and indeed the concept of `deficiency', seems mysterious.

Luckily, when you work through the proof of the deficiency zero theorem, the mystery evaporates.   It turns out that there are two equivalent ways to define the deficiency of a reaction network.  One makes it easy to compute, and that's the definition people usually give.  But the other explains why it's important.

In fact the whole proof of the deficiency zero theorem is full of great insights, so we want to show it to you.  This will be the climax of the book: we'll see that Markov processes, Noether's theorem and Perron--Frobenius theory play crucial roles, even though the deficiency zero theorem concerns not the master equation but the rate equation, which is deterministic and nonlinear!\index{nonlinearity!rate equation}\index{Perron–Frobenius theorem!and the deficiency zero theorem}

In this section, we'll just unfold the concept of `deficiency' so we see what it means.  In the next, we'll show you a slick way to write the rate equation, which is crucial to proving the deficiency zero theorem. Then we'll start the actual proof.

\subsection{Reaction networks revisited}
\label{sec:20_reaction_networks}

Let's recall what a reaction network is, and set up our notation.  In chemistry we consider a finite set of `species' like C, O$_2$, H$_2$O and so on... and then consider reactions like
$$
\text{CH}_4 + 3 \text{O}_2 \longrightarrow \text{CO}_2 + 2 \text{H}_2\text{O}
$$
On each side of this reaction we have a finite linear combination of species, with natural numbers as coefficients.  Chemists call such a thing a \emph{complex}.
\index{complex}

So, given any finite collection of species, say $S$, let's write $\mathbb{N}^S$ to mean the set of finite linear combinations of elements of $S$, with natural numbers as coefficients.   The complexes appearing in our reactions will form a subset of this, say $K$.   

We'll also consider a finite collection of reactions---or as we've been calling them, `transitions'.  Let's call this $T$.  Each transition goes from some complex to some other complex: if we want a reaction to be reversible we'll explicitly include another reaction going the other way.  So, given a transition $\tau \in T$ it will always go from some complex called its {\bf source} $s(\tau)$ to some complex called its {\bf target} $t(\tau)$.

All this data, put together, is a reaction network:

\begin{definition}
\index{reaction network}
A {\bf reaction network} $(S,\; s,t \colon T \to K)$ consists of:
\begin{itemize}
\item  a finite set $S$ of {\bf species}, 
\index{species}

\item  a finite set $T$ of {\bf transitions},
\index{transition}

\item  a finite set $K \subset \mathbb{N}^S$ of {\bf complexes},
\index{complex}

\item  {\bf source} and {\bf target} maps $s,t \colon T \to K$.
\index{source} \index{target}
\end{itemize}
\end{definition}

We can draw a reaction network as a graph with complexes as vertices and transitions as edges:

%<img width = "200" src = "http://math.ucr.edu/home/baez/networks/chemical_reaction_network_part_20_I.png" alt = ""/>
\begin{center}
 \includegraphics[width=40mm]{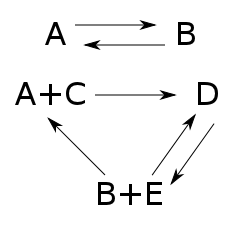}
\end{center}

\noindent
The set of species here is $S = \{A,B,C,D,E\}$, and the set of complexes is $K = \{A,B,D,A+C,B+E\}$.

But to write down the `rate equation' describing our chemical reactions, we need a bit more information: constants $r(\tau)$ saying the rate at which each transition $\tau$ occurs.   So, we define:

\begin{definition}
A {\bf stochastic reaction network} is a reaction network $(S, \; s,t \colon T \to K)$ together with a map $r\colon T \to (0,\infty)$ assigning a {\bf rate constant} to each transition.
\end{definition} 

Let us remind you how the rate equation works.  At any time we have some amount $x_i \in [0,\infty)$ of each species $i$.   These numbers are the components of a vector $x \in \mathbb{R}^S$, which of course depends on time.  The {\bf rate equation} says how this vector changes:
$$   \displaystyle{ \frac{d x_i}{d t} = \sum_{\tau \in T} r(\tau) \left(t_i(\tau) - s_i(\tau)\right)  x^{s(\tau)} } $$
Here we are writing $s_i(\tau)$ for the $i$th component of the vector $s(\tau)$, and similarly for $t_i(\tau)$.   We should remind you what $x^{s(\tau)}$ means, since here we are raising a vector to a vector power, which is a bit unusual.   So, suppose we have any vector $x = (x_1, \dots, x_k)$ and we raise it to the power of $s = (s_1, \dots, s_k)$.  Then by definition we get
$$  \displaystyle{ x^s = x_1^{s_1} \cdots x_k^{s_k} } $$ 

Given this, we hope the rate equation makes intuitive sense!  There's one term for each transition $\tau$.   The factor of $t_i(\tau) - s_i(\tau)$ shows up because our transition destroys $s_i(\tau)$ things of the $i$th species and creates $t_i(\tau)$ of them. The big product 
$$  \displaystyle{ x^{s(\tau)} = x_1^{s(\tau)_1} \cdots x_k^{s(\tau)_k} }$$
shows up because our transition occurs at a rate proportional to the product of the numbers of things it takes as inputs.  The constant of proportionality is the reaction rate $r(\tau)$.  

The deficiency zero theorem says lots of things, but in the next few episodes we'll prove a weak version, like this:

\begin{theorem}[{\bf Deficiency Zero Theorem---Baby Version}]\index{Deficiency Zero Theorem!baby version} 
   Suppose we have a weakly reversible reaction network with deficiency zero.  Then for any choice of rate constants there exists an equilibrium solution $x \in (0,\infty)^S$ of the rate equation.  In other words:
$$   \displaystyle{ \sum_{\tau \in T} r(\tau) \left(t_i(\tau) - s_i(\tau)\right)  x^{s(\tau)} = 0} $$
\end{theorem}

\noindent
An important feature of this result is that all the components of the vector $x$ are \emph{positive}.  In other words, there's actually some chemical of each species present!  

But what do the hypotheses in this theorem mean?

A reaction network is {\bf weakly reversible} if for any transition $\tau \in T$ going from a complex $\kappa$ to a complex $\kappa'$, there is a sequence of transitions going from $\kappa'$ back to $\kappa$.  But what about `deficiency zero'?  As we mentioned, this requires more work to understand.  

So, let's dive in!

\subsection{The concept of deficiency}

In modern math, we like to take all the stuff we're thinking about and compress it into a diagram with a few objects and maps going between these objects.   So, to understand the deficiency zero theorem, we wanted to isolate the crucial maps.  For starters, there's an obvious map
$$ Y \colon K \to \mathbb{N}^S $$
sending each complex to the linear combination of species that it is.  We can think of $K$ as an abstract set equipped with this map saying how each complex is made of species, if we like.  Then all the information in a stochastic reaction network sits in this diagram:
\[
\xymatrix{
(0,\infty) & T \ar[l]_<<<<r 
\ar@<0.5ex>[r]^s 
\ar@<-0.5ex>[r]_t & K \ar[r]^Y & \mathbb{N}^S  \\
}
\]
This is fundamental to everything we'll do from now on, so take a minute to lock it into your brain.

We'll do lots of different things with this diagram.  For example, we often want to use ideas from linear algebra, and then we want our maps to be linear.   For example, $Y$ extends uniquely to a linear map
$$ Y \colon \mathbb{R}^K \to \mathbb{R}^S $$
sending real linear combinations of complexes to real linear combinations of species.   Reusing the name $Y$ here won't cause confusion.  We can also extend $r$, $s$, and $t$ to linear maps in a unique way, getting a little diagram like this:
\[
\xymatrix{
\mathbb{R}  & \mathbb{R}^T \ar[l]_r 
\ar@<0.5ex>[r]^s 
\ar@<-0.5ex>[r]_t & \mathbb{R}^K \ar[r]^Y & \mathbb{R}^S  \\
}
\]
Linear algebra lets us talk about \emph{differences} of complexes.  Each transition $\tau$ gives a vector
$$  \partial \tau = t(\tau) - s(\tau)\in \mathbb{R}^K $$
saying the change in complexes that it causes.  And we can extend $\partial$ uniquely to a linear map
$$  \partial \colon \mathbb{R}^T \to \mathbb{R}^K $$
defined on linear combinations of transitions.  Mathematicians call $\partial$ a {\bf boundary operator}.

So, we have a little sequence of linear maps:
\[
\xymatrix{
 \mathbb{R}^T 
\ar[r]^\partial 
& \mathbb{R}^K 
\ar[r]^Y
& \mathbb{R}^S  
\\
}
\]
This turns a transition into a change in complexes, and then a change in species.

If you know fancy math you'll be wondering if this sequence is a `chain complex', which is a fancy way of saying that $Y \partial = 0$.  The answer is no.   This equation means that every linear combination of reactions leaves the amount of all species unchanged.   Or equivalently: every reaction leaves the amount of all species unchanged.   This only happens in very silly examples.

Nonetheless, it's \emph{possible} for a linear combination of reactions to leave the amount of all species unchanged.

For example, this will happen if we have a linear combination of reactions that leaves the amount of all \emph{complexes} unchanged.   But this sufficient condition is not necessary.  And this leads us to the concept of `deficiency zero':

\begin{definition}
  A reaction network has {\bf deficiency zero} if any linear combination of reactions that leaves the amount of every species unchanged also leaves the amount of every complex unchanged.
\end{definition}

In short, a reaction network has deficiency zero iff
$$   Y (\partial \rho) = 0 \; \Rightarrow \; \partial \rho = 0  $$
for every $\rho \in \mathbb{R}^T$.  In other words---using some basic concepts from linear algebra---a reaction network has deficiency zero iff $Y$ is one-to-one when restricted to the image of $\partial$.  Remember, the \href{http://en.wikipedia.org/wiki/Image_%28mathematics%29}{image} of $\partial$ is 
$$ \mathrm{im} \partial = \{ \partial \rho \; : \; \rho \in \mathbb{R}^T \}  $$
Roughly speaking, this consists of all changes in complexes that can occur due to reactions. 

In still other words, a reaction network has deficiency zero if 0 is the only vector in both the image of $\partial$ and the kernel of $Y$:
$$  \mathrm{im} \partial \cap \mathrm{ker} Y = \{ 0 \}  $$
Remember, the \href{http://en.wikipedia.org/wiki/Kernel_%28mathematics%29}{kernel} of $Y$ is
$$  \mathrm{ker} Y = \{ \phi \in \mathbb{R}^K \; : \; Y \phi = 0 \} $$
Roughly speaking, this consists of all changes in complexes that don't cause changes in species.  So, `deficiency zero' roughly says that if a reaction causes a change in complexes, it causes a change in species.  

(All this `roughly speaking' stuff is because in reality we should be talking about \emph{linear combinations} of transitions, complexes and species.  But it's a bit distracting to keep saying that when we're trying to get across the basic ideas!)

Now we're ready to understand deficiencies other than zero, at least a little.  They're defined like this:

\begin{definition}
 The {\bf deficiency} of a reaction network is the dimension of $\mathrm{im} \partial \cap \mathrm{ker} Y$.  
\end{definition}

\subsection{How to compute the deficiency}

You can compute the deficiency of a reaction network just by looking at it.  However, it takes a little training.  First, remember that a reaction network can be drawn as a graph with complexes as vertices and transitions as edges, like this:

%<img width = "200" src = "http://math.ucr.edu/home/baez/networks/chemical_reaction_network_part_17_II.png" alt = ""/>
\begin{center}
 \includegraphics[width=40mm]{chemical_reaction_network_part_17_II.png}
\end{center}

There are three important numbers we can extract from this graph:
\begin{itemize}
\item  We can count the number of vertices in this graph; let's call that $|K|$, since it's just the number of complexes.  

\item  We can count the number of pieces or `components' of this graph; let's call that $\# \mathrm{components} $ for now.  

\item  We can also count the dimension of the image of $Y \partial$.  This space, $\mathrm{im} Y \partial$, is called the {\bf stoichiometric subspace}: vectors in here are changes in species that can be accomplished by transitions in our reaction network, or linear combinations of transitions.  
\end{itemize}

These three numbers, all rather easy to compute, let us calculate the deficiency:

\begin{theorem}\index{deficiency}
The deficiency of a reaction network equals 
$$|K| - \# \mathrm{components} - \mathrm{dim} ( \mathrm{im} Y \partial) $$
\end{theorem}

\begin{proof}   By definition we have
$$ \mathrm{deficiency} = \dim ( \mathrm{im} \partial \cap \mathrm{ker} Y) $$
but another way to say this is
$$ \mathrm{deficiency} = \mathrm{dim} (\mathrm{ker} Y|_{\mathrm{im} \partial}) $$
where we are restricting $Y$ to the subspace $\mathrm{im} \partial$, and taking the dimension of the kernel of that.

The \href{http://en.wikipedia.org/wiki/Rank\%E2\%80\%93nullity_theorem}{\bf rank-nullity theorem} says that whenever you have a linear map $T\colon V \to W$ between finite-dimensional vector spaces, 
$$    \mathrm{dim} (\mathrm{ker} T) \; = \; \mathrm{dim}(\mathrm{dom} ) \; - \; \mathrm{dim} (\mathrm{im} T) $$  
\index{rank-nullity theorem} where $\mathrm{dom} T$ is the \href{http://en.wikipedia.org/wiki/Domain_of_a_function}{domain} of $T$, namely the vector space $V$.\index{vector space!rank-nullity theorem}  It follows that
$$ \mathrm{dim}(\mathrm{ker} Y|_{\mathrm{im} \partial}) = \mathrm{dim}(\mathrm{dom}  Y|_{\mathrm{im} \partial}) - \mathrm{dim}(\mathrm{im}  Y|_{\mathrm{im} \partial}) $$
The domain of $Y|_{\mathrm{im} \partial}$ is just $\mathrm{im} \partial$, while its image equals $\mathrm{im}(Y \partial)$, so
$$ \mathrm{deficiency} = \mathrm{dim}(\mathrm{im} \partial) - \mathrm{dim}(\mathrm{im} Y \partial) $$ 
The theorem then follows from the following lemma.  
\end{proof}

\begin{lemma}
  $ \mathrm{dim} (\mathrm{im} \partial) = |K| - \# \mathrm{components} $.
  \end{lemma} 

\begin{proof}
In fact this holds whenever we have a finite set of complexes and a finite set of transitions going between them.  We get a diagram
\[
\xymatrix{
 T \ar@<0.5ex>[r]^s \ar@<-0.5ex>[r]_t & K 
}
\]
We can extend the source and target functions to linear maps as usual:
\[
\xymatrix{
 \mathbb{R}^T 
\ar@<0.5ex>[r]^s 
\ar@<-0.5ex>[r]_t & \mathbb{R}^K  \\
}
\]
and then we can define $\partial = t - s$.  We claim that 
$$ \mathrm{dim} (\mathrm{im} \partial) = |K| - \# \mathrm{components} $$
where $\# \mathrm{components}$ is the number of connected components of the graph with $K$ as vertices and $T$ as edges.  

This is easiest to see using an inductive argument where we start by throwing out all the edges of our graph and then put them back in one at a time.  When our graph has no edges, $\partial = 0$ and the number of components is $|K|$, so our claim holds: both sides of the equation above are zero.  Then, each time we put in an edge, there are two choices: either it goes between two different components of the graph we have built so far, or it doesn't.  If it goes between two different components, it increases $\mathrm{dim} (\mathrm{im} \partial)$ by 1 and decreases the number of components by 1, so our equation continues to hold.   If it doesn't, neither of these quantities change, so our equation again continues to hold.    
\end{proof}

\subsection{Examples}

%<img width = "200" src = "http://math.ucr.edu/home/baez/networks/chemical_reaction_network_part_17_I.png" alt = ""/>
\begin{center}
 \includegraphics[width=40mm]{chemical_reaction_network_part_17_I.png}
\end{center}

This reaction network is not weakly reversible, since we can get from $B + E$ and $A + C$ to $D$ but not back.   It becomes weakly reversible if we throw in another transition:

%<img width = "200" src = "http://math.ucr.edu/home/baez/networks/chemical_reaction_network_part_17_II.png" alt = ""/>
\begin{center}
 \includegraphics[width=40mm]{chemical_reaction_network_part_17_II.png}
\end{center}

Taking a reaction network and throwing in the reverse of an existing transition never changes the number of complexes, the number of components, or the dimension of the stoichiometric subspace.  So, the deficiency of the reaction network remains unchanged.  We computed it back in Section~\ref{sec:17_theorem}, but let's do it again.  For either reaction network above:

\begin{itemize}
\item the number of complexes is 5:
$$  |K| = |\{A,B,D, A+C, B+E\}| = 5 $$
\item the number of components is 2:
$$ \# \mathrm{components} = 2 $$
\item the dimension of the stoichometric subspace is 3.  For this we go through each transition, see what change in species it causes, and take the vector space\index{vector!span of}\index{vector space!spanned by} spanned by these changes.  Then we find a basis for this space and count it:
$$ \begin{array}{ccl}  \mathrm{im} Y\partial &=& 
\langle B-A, A-B, D - A - C, D - (B+E) , (B+E)-(A+C) \rangle \\
&=& \langle B -A, D - A - C, D - A - E \rangle \end{array} $$ 
so 
$$  \mathrm{dim} \left(\mathrm{im} Y\partial \right) = 3 $$
\end{itemize}
As a result, the deficiency is zero:
$$  \begin{array}{ccl} \mathrm{deficiency} 
&=& |K| - \# \mathrm{components} - \mathrm{dim} \left( \mathrm{im} Y\partial \right) \\
&=& 5 - 2 - 3 \\
&=&  0 \end{array} $$

Now let's throw in another complex and some more transitions:

%<img width = "200" src = "http://math.ucr.edu/home/baez/networks/chemical_reaction_network_part_20_II.png" alt = ""/>
\begin{center}
 \includegraphics[width=40mm]{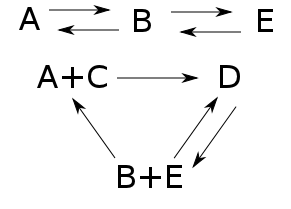}
\end{center}
\noindent
Now:
\begin{itemize}
\item the number of complexes increases by 1:
$$  |K| = |\{A,B,D,E, A+C, B+E\}| = 6 $$
\item the number of components is unchanged:
$$ \# \mathrm{components} = 2 $$
\item the dimension of the stoichometric subspace increases by 1.  We never need to include reverses of transitions when computing this:
$$ \begin{array}{ccl}  \mathrm{im} Y\partial &=& 
\langle B-A, E-B,  D - (A+C), D - (B+E) , (B+E)-(A+C) \rangle \\
&=& \langle B -A, E-B, D - A - C, D - B - E \rangle \end{array} $$ 
so 
$$  \mathrm{dim} \left(\mathrm{im} Y\partial \right) = 4 $$
\end{itemize}
As a result, the deficiency is still zero:
$$  \begin{array}{ccl} \mathrm{deficiency} 
&=& |K| - \# \mathrm{components} - \mathrm{dim} \left( \mathrm{im} Y\partial \right) \\
&=& 6 - 2 - 4 \\
&=&  0 \end{array} $$

Do \emph{all} reaction networks have deficiency zero?   That would be nice. Let's try one more example:

\begin{center}
 \includegraphics[width=40mm]{chemical_reaction_network_part_20_V.png}
\end{center}

\noindent
Now: 
\begin{itemize}
\item
The number of complexes is the same as in our last example:
$$  |K| = |\{A,B,D,E, A+C, B+E\}| = 6 $$
\item the number of components is also the same:
$$ \# \mathrm{components} = 2 $$
\item the dimension of the stoichometric subspace is also the same:
$$ \begin{array}{ccl}  \mathrm{im} Y\partial &=& 
\langle B-A,  D - (A+C), D - (B+E) , (B+E)-(A+C), (B+E) - B \rangle \\
&=& \langle B -A, D - A - C, D - B , E\rangle \end{array} $$ 
so 
$$  \mathrm{dim} \left(\mathrm{im} Y\partial \right) = 4 $$
\end{itemize}
So the deficiency is still zero:
$$  \begin{array}{ccl} \mathrm{deficiency} 
&=& |K| - \# \mathrm{components} - \mathrm{dim} \left( \mathrm{im} Y\partial \right) \\
&=& 6 - 2 - 4 \\
&=&  0 \end{array} $$
It's sure easy to find examples with deficiency zero!  

\begin{problem}\label{prob:35} 
Can you find a reaction network where the deficiency is \emph{not} zero?
\end{problem}

\begin{problem}\label{prob:36} 
If you can't find an example, prove the deficiency is always zero. If you can, find 1) the smallest example and 2) the smallest example that actually arises in chemistry.  
\end{problem}

Note that not all reaction networks can actually arise in chemistry.  For example, the transition $A \to A + A$ would violate conservation of mass.  Nonetheless a reaction network like this might be useful in a very simple model of amoeba reproduction, one that doesn't take limited resources into account.  As you might recall, we talked about amoebas in Section \ref{sec:6}.

\subsection{Different kinds of graphs}

We'll conclude with some technical remarks that only a mathematician could love; feel free to skip them if you're not in the mood.  As you've already seen, it's important that a reaction network can be drawn as a graph:

%<img width = "200" src = "http://math.ucr.edu/home/baez/networks/chemical_reaction_network_part_17_II.png" alt = ""/>
\begin{center}
 \includegraphics[width=40mm]{chemical_reaction_network_part_17_II.png}
\end{center}
\noindent
But there are many \emph{kinds} of graph.  What kind of graph is this, exactly?  As mentioned in Section~\ref{sec:17_reaction_networks}, it's a {\bf directed multigraph}, meaning that the edges have arrows on them, more than one edge can go from one vertex to another, and we can have any number of edges going from a vertex to itself.  Not all those features are present in this example, but they're certainly allowed by our definition of reaction network!  

After all, we've seen that a stochastic reaction network amounts to a little diagram like this:
\[
\xymatrix{
(0,\infty) & T \ar[l]_<<<<r 
\ar@<0.5ex>[r]^s 
\ar@<-0.5ex>[r]_t & K \ar[r]^Y & \mathbb{N}^S  \\
}
\]
If we throw out the rate constants, we're left with a reaction network.  So, a {\bf reaction network} is just a diagram like this:
\[
\xymatrix{
T 
\ar@<0.5ex>[r]^s 
\ar@<-0.5ex>[r]_t & K \ar[r]^Y & \mathbb{N}^S  \\
}
\]
If we throw out the information about how complexes are made of species, we're left with an even smaller diagram:
\[
\xymatrix{
 T \ar@<0.5ex>[r]^s \ar@<-0.5ex>[r]_t & K 
}
\]
And this precisely matches the slick definition of a {\bf directed multigraph}\index{graph theory!directed multigraph}: it's a set $E$ of {\bf edges}, a set $V$ of {\bf vertices}, and functions $s,t \colon E \to V$ saying where each edge starts and where it ends: its {\bf source} and {\bf target}.    \index{graph theory!edge}
\index{graph theory!vertex} \index{edge} \index{vertex}

Why don't we just call this a `graph'?  Why the fancy term `directed multigraph'?\index{graph theory!directed multigraph}  The problem is that there are many kinds of graphs.  While their diversity is a bit annoying at first, we must learn to love it, at least if we're mathematicians and want everything to be completely precise.   

Indeed, there are at least $2^3 = 8$ kinds of graphs, depending on our answers to three questions:

\begin{itemize}
\item  Do the edges have arrows on them?
\item  Can more than one edge go between a specific pair of vertices?
\end{itemize}
and
\begin{itemize}
\item  Can an edge go from a vertex to itself?
\end{itemize}
We get directed multigraphs if we answer YES, YES and YES to these questions.  Since they allow all three features, directed multigraphs are very important.  For example, a `category' is a directed multigraph equipped with some extra structure.\index{graph theory!directed multigraph!category as}\index{category theory!category as directed multigraph}  Also, what mathematicians call a `quiver' is just another name for a directed multigraph.\index{graph theory!directed multigraph}\index{graph theory!quiver}\index{quiver}

We've met two other kinds of graph so far:

\begin{itemize} 

\item  In Section~\ref{sec:14_laplacians} and Section Section~\ref{sec:15} we studied graph Laplacians and circuits made of resistors---or in other words, Dirichlet operators---using `simple graphs'.  We get simple graphs when we answer NO, NO and NO to the three questions above.   The slick definition of a {\bf simple graph} is that it's a set $V$ of vertices together with a subset $E$ of the collection of 2-element subsets of $V$.

\item   In Section~\ref{sec:16_matrices_to_graphs} we studied Markov processes on finite sets---or, in other words, infinitesimal stochastic operators---using `directed graphs'.  We get directed graphs when we answer YES, NO and YES to the three questions.  The slick definition of a {\bf directed graph} is that it's a set $V$ of vertices together with a subset $E$ of the ordered pairs of $V$:\index{graph theory!directed graph} 
$$ E \subseteq V \times V$$
\end{itemize} 

There is a lot to say about this business, but for now we'll just note that you can use directed multigraphs with edges labelled by positive numbers to describe Markov processes, just as we used directed graphs.  You don't get anything more general, though! After all, if we have multiple edges from one vertex to another, we can replace them by a single edge as long as we add up their rate constants.  And an edge from a vertex to itself has no effect at all.

In short, both for Markov processes and reaction networks, we can take `graph' to mean either `directed graph' or `directed multigraph', as long as we make some minor adjustments. In what follows, we'll use directed multigraphs for \emph{both} Markov processes and reaction networks.  One tiny advantage is that you can take a directed multigraph and add a new edge from one vertex to another without first checking to see if there's already one there.  Combining our work with category theory will also go easier if we use directed multigraphs.\index{graph theory!directed multigraph}

\subsection{Answers}

\vskip 1em \noindent {\bf Problem 35.} 
Can you find a reaction network where the deficiency is \emph{not} zero?

\begin{answer}
The reaction network
$$ A \to A + A \to A + A + A $$
has deficiency $1$.  If we go on this way adding multiples of $A$ we can keep on
increasing the deficiency, as only the number of complexes increases.
\end{answer}

\vskip 1em \noindent {\bf Problem 36.} 
Find 1) the smallest example of a reaction network where the deficiency is nonzero and 2) the smallest example that actually arises in chemistry.  

\begin{answer}
The answers here depend on the notion of `smallest'.  However, the reaction network in the answer above has the smallest possible number of components, namely one---and among those, it has the smallest possible number of complexes, namely 3, since a reaction network with just one component and one or two complexes has deficiency zero.

The reaction network above cannot occur in chemistry, except as a simplified model
that omits some of the species involved in the reactions, because it violates conservation
of mass.  Here is a rather small reaction network that can occur in chemistry:
$$  A \to B  \qquad A + C \to B + C $$
This arises when a molecule $A$ can turn into a molecule $B$ either with or without
the assistance of a catalyst $C$.  The number of complexes is $4$, the number of 
components is $2$, and the dimension of the stochiometric subspace is $1$, so the
deficiency is $4 - 2 - 1 = 1$.  
\end{answer}

%%%%%%%%%% SECTION 21 %%%%%%%%%%%%%%%%%

\newpage
\section{Rewriting the rate equation}\label{sec:21} 

Okay, now let's dig deeper into the proof of the deficiency zero theorem.  We're only going to prove a baby version, at first.  Later we can enhance it:

\begin{theorem}[{\bf Deficiency Zero Theorem---Baby Version}]\index{Deficiency Zero Theorem!baby version}
   Suppose we have a weakly reversible reaction network with deficiency zero.  Then for any choice of rate constants there exists an equilibrium solution of the rate equation where all species are present in nonzero amounts.
\end{theorem} 

The first step is to write down the rate equation in a new, more conceptual way.  It's incredibly cool.   We've mentioned Schr\"{o}dinger's equation\index{quantum mechanics!Schr\"{o}dinger's equation}, which describes the motion of a quantum particle:
$$ \displaystyle { \frac{d \psi}{d t} = -i H \psi } $$
We've spent a lot of time on Markov processes, which describe the motion of a `stochastic' particle:
$$ \displaystyle { \frac{d \psi}{d t} = H \psi } $$
A `stochastic' particle is one that's carrying out a random walk\index{random walk!definition of}, and now $\psi$ describes its \emph{probability} to be somewhere, instead of its amplitude.  But now comes a surprise: soon we'll see that the rate equation for a reaction network looks somewhat similar:
$$ \displaystyle { \frac{d x}{d t} =  Y H x^Y } $$
where $Y$ is some matrix, and $x^Y$ is defined using a new thing called `matrix exponentiation', which makes the equation nonlinear!\index{nonlinearity!matrix exponentiation}\index{matrix!exponentiation} 

If you're reading this you probably know how to multiply a vector by a matrix.  But if you're like us, you've never seen anyone take a vector and raise it to the power of some matrix!   We'll explain it, don't worry... right now we're just trying to get you intrigued.  It's not complicated, but it's exciting how this unusual operation shows up naturally in chemistry.  

Since we're looking for an \emph{equilibrium} solution of the rate equation, we actually want to solve
$$ \displaystyle { \frac{d x}{d t} =  0 } $$
or in other words
$$ Y H x^Y = 0 $$
In fact we will do better: we will find a solution of
$$ H x^Y = 0 $$
And we'll do this in two stages:

\begin{itemize}
\item  First we'll find all solutions of
$$ H \psi = 0 $$
This equation is \emph{linear}, so it's easy to understand.

\item  Then, among these solutions $\psi$, we'll find one that also obeys
$$ \psi = x^Y $$
This is a \emph{nonlinear} problem involving matrix exponentiation, but still, we can do it, using a clever trick called `logarithms'.\index{nonlinearity!equation}
\end{itemize} 
Putting the pieces together, we get our solution of 
$$ H x^Y = 0 $$
and thus our equilibrium solution of the rate equation:
$$ \displaystyle { \frac{d x}{d t} = Y H x^Y = 0 } $$

That's a rough outline of the plan.  But now let's get started, because the details are actually fascinating.  In this section we'll just show you how to rewrite the rate equation in this new way.

\subsection{The rate equation and matrix exponentiation}\index{matrix!exponentiation}\index{nonlinearity!matrix exponentiation|(} 
\label{sec:21_rate_equation}

Remember how the rate equation goes.  We start with a {\bf stochastic reaction network}, meaning a little diagram like this:
\[
\xymatrix{
(0,\infty) & T \ar[l]_<<<<r 
\ar@<0.5ex>[r]^s 
\ar@<-0.5ex>[r]_t & K \ar[r]^Y & \mathbb{N}^S  \\
}
\]
This contains quite a bit of information:
\begin{itemize} 
\item a finite set $T$ of {\bf transitions},   
\item a finite set $K$ of {\bf complexes},
\item a finite set $S$ of {\bf species}, 
\item a map $r \colon T \to (0,\infty)$ giving a {\bf rate constant} for each transition,
\item {\bf source} and {\bf target} maps $s,t \colon T \to K$ saying where each transition starts and ends,
\item a one-to-one map $Y \colon K \to \mathbb{N}^S$ saying how each complex is made of species.
\end{itemize} 

Given all this, the rate equation says how the amount of each species changes with time.  We describe these amounts with a vector $x \in [0,\infty)^S$.  So, we want a differential equation filling in the question marks here:
$$ \displaystyle{ \frac{d x}{d t} = ??? } $$
Now in Section~\ref{sec:20_reaction_networks}, we thought of $K$ as a subset of $\mathbb{N}^S$, and thus of the vector space $\mathbb{R}^S$.  Back then, we wrote the rate equation as follows:
$$   \displaystyle{ \frac{d x}{d t} = \sum_{\tau \in T} r(\tau) \left(t(\tau) - s(\tau)\right) x^{s(\tau)} } $$
where vector exponentiation is defined by
$$ x^s = x_1^{s_1} \cdots x_k^{s_k} $$
when $x$ and $s$ are vectors in $\mathbb{R}^S$.

However, we've now switched to thinking of our set of complexes $K$ as a set in its own right that is mapped into $\mathbb{N}^S$ by $Y$.  This is good for lots of reasons, like defining the concept of `deficiency', which we did last time.   But it means the rate equation above doesn't quite parse anymore!   Things like $s(\tau)$ and $t(\tau)$ live in $K$; we need to explicitly convert them into elements of $\mathbb{R}^S$ using $Y$ for our equation to make sense!

So now we have to write the {\bf rate equation} like this:
$$   \displaystyle{ \frac{d x}{d t} = Y \sum_{\tau \in T} r(\tau) \left(t(\tau) - s(\tau)\right) x^{Y s(\tau)} } $$
This looks more ugly, but if you've got even one mathematical bone in your body, you can already see vague glimmerings of how we'll rewrite this the way we want:
$$ \displaystyle { \frac{d x}{d t} =  Y H x^Y } $$
Here's how.

First, we extend our maps $s, t$ and $Y$ to linear maps between vector spaces:
\[
\xymatrix{
 \mathbb{R}^T 
\ar@<0.5ex>[r]^s 
\ar@<-0.5ex>[r]_t & \mathbb{R}^K \ar[r]^Y & \mathbb{R}^S  \\
}
\]

Then, we put an inner product on the vector spaces $\mathbb{R}^T$, $\mathbb{R}^K$ and $\mathbb{R}^S$.  For $\mathbb{R}^K$ we do this in the most obvious way, by letting the complexes be an orthonormal basis.  So, given two complexes $\kappa, \kappa'$, we define their inner product by
$$ \langle \kappa, \kappa' \rangle = \delta_{\kappa, \kappa'} $$
We do the same for $\mathbb{R}^S$.  But for $\mathbb{R}^T$ we define the inner product in a more clever way involving the rate constants.  If $\tau, \tau' \in T$ are two transitions, we define their inner product by:
$$ \langle \tau, \tau' \rangle = \frac{1}{r(\tau)} \delta_{\tau, \tau'} $$

Having put inner products on these three vector spaces, we can take the adjoints of the linear maps between them, to get linear maps going back the other way:
\[
\xymatrix{
 \mathbb{R}^T 
\ar@<0.5ex>[r]^\partial
& \mathbb{R}^K 
\ar@<0.5ex>[l]^{s^\dagger} 
\ar@<0.5 ex>[r]^Y
& \mathbb{R}^S  
\ar@<0.5 ex>[l]^{Y^\dagger}
\\
}
\]
These are defined in the usual way---though we're using daggers here the way physicists do, where mathematicians would prefer to see stars.   For example, $s^\dagger \colon \mathbb{R}^K \to \mathbb{R}^T$ is defined by the relation
$$ \langle s^\dagger \phi, \psi \rangle = \langle \phi, s \psi \rangle $$
and so on.

Next, we set up a random walk\index{random walk!on reaction network}\index{reaction network!random walk on complexes} on the set of complexes.  Remember, our reaction network is a graph with complexes as vertices and transitions as edges, like this:

%<img width = "200" src = "http://math.ucr.edu/home/baez/networks/chemical_reaction_network_part_20_V.png" alt = ""/>
\begin{center}
 \includegraphics[width=40mm]{chemical_reaction_network_part_20_V.png}
\end{center}

Each transition $\tau$ has a number attached to it: the rate constant $r(\tau)$.  So, we can randomly hop from complex to complex along these transitions, with probabilities per unit time described by these numbers.  The probability of being at some particular complex will then be described by a function 
$$ \psi \colon K \to \mathbb{R} $$
which also depends on time, and changes according to the equation
$$ \displaystyle { \frac{d \psi}{d t} = H \psi } $$
for some {\bf Hamiltonian} 
$$ H \colon \mathbb{R}^K \to \mathbb{R}^K $$
This Hamiltonian can be described as a matrix with off-diagonal entries with
$$ i \ne j \quad \Rightarrow \quad H_{i j} =  \sum_{\tau \colon j \to i} r(\tau) $$
where we write $\tau \colon j \to i$ when $\tau$ is a transition with source $j$ and target $i$.  Since $H$ is infinitesimal stochastic, the diagonal entries are determined by the requirement that the columns sum to zero.  

In fact, there's a very slick formula for this Hamiltonian:
$$ H = \partial s^\dagger $$
We'll prove this in Section \ref{sec:22_hamiltonian}.  For now, the main point is that with this Hamiltonian, the rate equation is equivalent to this:
$$ \displaystyle{ \frac{d x}{d t} = Y H x^Y } $$
The only thing not defined yet is the funny exponential $x^Y$.  That's what makes the equation nonlinear.  We're taking a \emph{vector} to the power of a \emph{matrix} and getting a \emph{vector}.  This sounds weird---but it actually makes sense!

It only makes sense because we have chosen bases for our vector spaces.  To understand it, let's number our species $1, \dots, k$ as we've been doing all along, and number our complexes $1, \dots, \ell$.   Our linear map $Y \colon \mathbb{R}^K \to \mathbb{R}^S$ then becomes a $k \times \ell$ matrix of natural numbers.  Its entries say how many times each species shows up in each complex:
$$ Y = \left( \begin{array}{cccc} 
Y_{11} & Y_{12}  & \cdots & Y_{1 \ell} \\
Y_{21} & Y_{22}  & \cdots & Y_{2 \ell} \\
\vdots  & \vdots   & \ddots & \vdots \\
Y_{k1} & Y_{k2}  & \cdots & Y_{k \ell} \end{array} \right) $$
The entry $Y_{i j}$ says how many times the $i$th species shows up in the $j$th complex.  

Now, let's be a bit daring and think of the vector $x \in \mathbb{R}^S$ as a row vector with $k$ entries:
$$ x = \left(x_1 , x_2 ,  \dots ,  x_k \right) $$
Then we can multiply $x$ on the \emph{right} by the matrix $Y$ and get a vector in $\mathbb{R}^K$:
$$ \begin{array}{ccl} x Y &=& (x_1, x_2, \dots, x_k) \;
 \left( \begin{array}{cccc} 
Y_{11} & Y_{12}  & \cdots & Y_{1 \ell} \\
Y_{21} & Y_{22}  & \cdots & Y_{2 \ell} \\
\vdots  & \vdots   & \ddots & \vdots \\
Y_{k1} & Y_{k2}  & \cdots & Y_{k \ell} \end{array} \right) 
\\
&&\\
&=& \left( x_1 Y_{11} + \cdots + x_k Y_{k1}, \; \dots, \; x_1 Y_{1 \ell} + \cdots + x_k Y_{k \ell} \right) 
\end{array}
$$
So far, no big deal.  But now you're ready to see the definition of $x^Y$, which is very similar:
$$ \begin{array}{ccl} x^Y &=& (x_1, x_2, \dots, x_k)^{
 \left( \begin{array}{cccc} 
Y_{11} & Y_{12}  & \cdots & Y_{1 \ell} \\
Y_{21} & Y_{22}  & \cdots & Y_{2 \ell} \\
\vdots & \vdots  & \ddots & \vdots \\
Y_{k1} & Y_{k2}  & \cdots & Y_{k \ell} \end{array} \right)} 
\\
&&\\
&=& ( x_1^{Y_{11}} \cdots  x_k^{Y_{k1}} ,\; \dots, \; x_1^{Y_{1 \ell}} \cdots x_k^{Y_{k \ell}} )
\end{array}
$$
It's exactly the same, but with \emph{multiplication} replacing \emph{addition}, and \emph{exponentiation} replacing \emph{multiplication!}   Apparently our class on matrices stopped too soon: we learned about matrix multiplication, but matrix exponentiation is also worthwhile. 

What's the point of it?  Well, suppose you have a certain number of hydrogen molecules, a certain number of oxygen molecules, a certain number of water molecules, and so on---a certain number of things of each species.  You can list these numbers and get a vector $x \in \mathbb{R}^S$.  Then the components of $x^Y$ describe how many ways you can build up each complex from the things you have.  For example, 
$$  x_1^{Y_{11}} x_2^{Y_{21}} \cdots  x_k^{Y_{k1}} $$
say roughly how many ways you can build complex 1 by picking $Y_{11}$ things of species 1, $Y_{21}$ things of species 2, and so on.  

Why `roughly'?   Well, we're pretending we can pick the same thing twice.  So if we have 4 water molecules and we need to pick 3, this formula gives $4^3$.  The right answer is $4 \times 3 \times 2$.   To get this answer we'd need to use the `falling power' $4^{\underline{3}} = 4 \times 3 \times 2$, as explained in Section~\ref{sec:3_general_rule}.  But the rate equation describes chemistry in the limit where we have lots of things of each species.  In this limit, the ordinary power becomes a good approximation.

\begin{problem}\label{prob:37} 
In this section we've seen a vector raised to a matrix power, which is a vector, and also a vector raised to a vector power, which is a number.  How are they related?
\end{problem} 

There's more to say about this.  But let's get to the punchline:

\begin{theorem} 
\label{theorem_rate_equation}\index{rate equation!in terms of matrix exponentiation} 
The rate equation:
$$   \displaystyle{ \frac{d x}{d t} = Y \sum_{\tau \in T} r(\tau) \; \left(t(\tau) - s(\tau)\right) \; x^{Y s(\tau)} } $$
is equivalent to this equation:
$$ \displaystyle{ \frac{d x}{d t} = Y \partial s^\dagger x^Y } $$
or in other words:
$$ \displaystyle{ \frac{d x}{d t} = Y H x^Y } $$
where $H = \partial s^\dagger$.
\end{theorem} 

\begin{proof}
Since $\partial = t - s$, it is enough to show 
$$  \displaystyle{ (t - s) s^\dagger x^Y = \sum_{\tau \in T} r(\tau) \; \left(t(\tau) - s(\tau)\right) \; x^{Y s(\tau)} } $$
So, we'll compute $(t - s) s^\dagger x^Y$, and think about the meaning of each quantity we get \emph{en route}.  

We start with $x \in \mathbb{R}^S$.  This is a list of numbers saying how many things of each species we have: our raw ingredients, as it were.  Then we compute
$$ x^Y = (x_1^{Y_{11}} \cdots  x_k^{Y_{k1}} ,  \dots,  x_1^{Y_{1 \ell}} \cdots x_k^{Y_{k \ell}} ) $$
This is a vector in $\mathbb{R}^K$.  It's a list of numbers saying how many ways we can build each complex starting from our raw ingredients.  

Alternatively, we can write this vector $x^Y$ as a sum over basis vectors:
$$ x^Y = \sum_{\kappa \in K} x_1^{Y_{1\kappa}} \cdots  x_k^{Y_{k\kappa}} \; \kappa $$
Next let's apply $s^\dagger$ to this.  We claim that
$$ \displaystyle{ s^\dagger \kappa = \sum_{\tau \; : \;  s(\tau) = \kappa} r(\tau) \; \tau } $$
In other words, we claim $s^\dagger \kappa$ is the sum of all the transitions having $\kappa$ as their source, \emph{weighted by their rate constants!}   To prove this claim, it's enough to take the inner product of each side with any transition $\tau'$, and check that we get the same answer.  For the left side we get
$$  \langle s^\dagger \kappa, \tau' \rangle = \langle \kappa, s(\tau') \rangle = \delta_{\kappa, s (\tau') }$$  

To compute the right side, we need to use the cleverly chosen inner product on $\mathbb{R}^T$.  Here we get
$$ \displaystyle{ \left\langle \sum_{\tau \; : \; s(\tau) = \kappa} r(\tau) \tau, \; \tau' \right\rangle =  \sum_{\tau \; : \; s(\tau) = \kappa} \delta_{\tau, \tau'} = \delta_{\kappa, s(\tau')}  } $$
In the first step here, the factor of $1 /r(\tau)$ in the cleverly chosen inner product canceled the visible factor of $r(\tau)$.  For the second step, you just need to think for half a minute---or ten, depending on how much coffee you've had.   

Either way, we conclude that indeed
$$ \displaystyle{ s^\dagger \kappa = \sum_{\tau \; : \; s(\tau) = \kappa} r(\tau) \tau } $$ 

Next let's combine this with our formula for $x^Y$:
$$ x^Y = \sum_{\kappa \in K} x_1^{Y_{1\kappa}} \cdots  x_k^{Y_{k\kappa}} \; \kappa $$
We get this:
$$  s^\dagger x^Y = \sum_{\kappa, \tau \; : \; s(\tau) = \kappa} r(\tau) \; x_1^{Y_{1\kappa}} \cdots  x_k^{Y_{k\kappa}} \; \tau $$
In other words, $s^\dagger x^Y$ is a linear combination of transitions, where each one is weighted both by \emph{the rate it happens} and \emph{how many ways it can happen} starting with our raw ingredients.

Our goal is to compute $(t - s)s^\dagger x^Y$.  We're almost there.  Remember, $s$ says which complex is the input of a given transition, and $t$ says which complex is the output.  So, $(t - s) s^\dagger x^Y$ says the total rate at which complexes are created and/or destroyed starting with the species in $x$ as our raw ingredients.  

That sounds good.   But let's just pedantically check that everything works.  Applying $t - s$ to both sides of our last equation, we get
$$ \displaystyle{ (t - s) s^\dagger x^Y = \sum_{\kappa, \tau \colon s(\tau) = \kappa} r(\tau) \; x_1^{Y_{1\kappa}} \cdots  x_k^{Y_{k\kappa}} \; \left( t(\tau) - s(\tau)\right) } $$
Remember, our goal was to prove that this equals
$$ \displaystyle{ \sum_{\tau \in T} r(\tau) \; \left(t(\tau) - s(\tau)\right) \; x^{Y s(\tau)} } $$
But if you stare at these a while and think, you'll see they're equal.      
\end{proof}

All this seems peculiar at first, but ultimately it all makes sense.  The interesting subtlety is that we use the linear map called `multiplying by $Y$':
$$  \begin{array}{ccc} \mathbb{R}^K &\to& \mathbb{R}^S \\  
                           \psi     &\mapsto&   Y \psi \end{array} $$
to take a bunch of complexes and work out the species they contain, while we use the \emph{nonlinear} map called `raising to the $Y$th power':
$$  \begin{array}{ccc} \mathbb{R}^S &\to& \mathbb{R}^K \\  
                           x     &\mapsto&   x^Y \end{array} $$
to take a bunch of species and work out how many ways we can build each complex from them.  There is much more to say about this: for example, these maps arise from a pair of what category theorists call `adjoint functors'.\index{category theory!adjoint functors}  But we're worn out and you probably are too.

\index{nonlinearity!matrix exponentiation|)}

\subsection{References}

We found this thesis to be the most helpful reference when trying to understand the proof of the deficiency zero theorem:

\begin{enumerate} 
\item[\cite{Gub03}] Jonathan M. Guberman, \href{http://www4.utsouthwestern.edu/altschulerwulab/theses/GubermanJ_Thesis.pdf}{\emph{Mass Action Reaction Networks and the Deficiency Zero Theorem}}, B.A. thesis, Department of Mathematics, Harvard University, 2003.
\end{enumerate} 
\noindent
We urge you to check it out.  In particular, Section 2 and Appendix A discuss matrix exponentiation.

Here's another good modern treatment of the deficiency zero theorem:

\begin{enumerate} 
\item[\cite{Gun03}] Jeremy Gunawardena, \href{http://vcp.med.harvard.edu/papers/crnt.pdf}{Chemical reaction network theory for \emph{in-silico} biologists}, 2003.
\end{enumerate} 

\noindent
Some of the theorem was first proved here:

\begin{enumerate} 
\item[\cite{Fei87}] Martin Feinberg, \href{http://www.seas.upenn.edu/~jadbabai/ESE680/Fei87a.pdf}{Chemical reaction network structure and the stability of complex isothermal reactors: I. The deficiency zero and deficiency one theorems}, \emph{Chemical Engineering Science} {\bf 42} (1987), 2229-2268. 
\end{enumerate} 

\noindent
However, Feinberg's treatment here relies heavily on this paper:

\begin{enumerate} 
\item[\cite{HJ72}] F. Horn and R. Jackson, General mass action kinetics, \href{http://www.springerlink.com/content/p345k578348107tj/}{\emph{Archive for Rational Mechanics and Analysis}} {\bf 47} (1972), 81-116.
\end{enumerate} 

\noindent
These lectures are also very helpful:
\begin{enumerate} 
\item[\cite{Fei79}] Martin Feinberg, \href{http://www.che.eng.ohio-state.edu/~FEINBERG/LecturesOnReactionNetworks/}{\sl Lectures on Reaction Networks}, 1979.
\end{enumerate} 

%%%%%%%%%%%% SECTION 22 %%%%%%%%%%%%%%%%

\newpage
\section[Markov processes]{The rate equation and Markov processes}
\label{sec:22} 

We've been looking at reaction networks, and we're getting ready to
find equilibrium solutions of the equations they give.  To do this,
we'll need to connect them to another kind of network we've studied.
A reaction network is something like this:

%<div align = "center}<img width = "200" src = "http://math.ucr.edu/home/baez/networks/chemical_reaction_network_part_20_III.png" alt = "}</div>
\begin{center}
 \includegraphics[width=40mm]{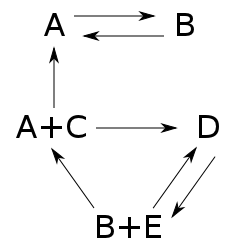}
\end{center}

\noindent
It's a bunch of {\bf complexes}, which are sums of basic building-blocks called {\bf species}, together with arrows called {\bf transitions} going between the complexes.  \index{complex} \index{species} \index{transition}  If we know a number  for each transition describing the rate at which it occurs, we get an equation called the `rate equation'.  This describes how the amount of each species changes with time.   We've been talking about this equation ever since the start of this course!   In Section~\ref{sec:21}, we wrote it down in a new very compact form:
$$  \displaystyle{ \frac{d x}{d t} = Y H x^Y  } $$
\noindent Here $ x$ is a vector whose components are the amounts of each species, while $ H$ and $ Y$ are certain matrices.

But now suppose we forget how each complex is made of species!  Suppose we just think of them as abstract things in their own right, like numbered boxes:

\begin{center}
 \includegraphics[width=50mm]{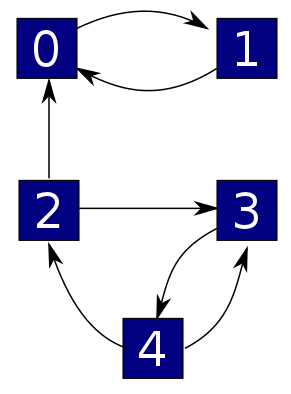}
\end{center}

\noindent
We can use these boxes to describe {\bf states} of some system.  The arrows still describe {\bf transitions}, but now we think of these as ways for the system to hop from one state to another.   Say we know a number for each transition describing the probability per time at which it occurs:

\begin{center}
 \includegraphics[width=50mm]{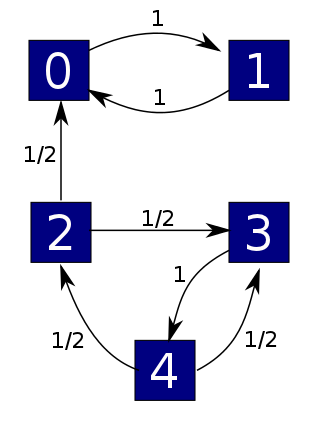}
\end{center}

\noindent
Then we get a `Markov process'---or in other words, a random walk\index{random walk!as Markov processes} where our system hops from one state to another.  If $\psi$ is the probability distribution saying how likely the system is to be in each state, this Markov process is described by this equation:
$$  \displaystyle{ \frac{d \psi}{d t} = H \psi  } $$
This is simpler than the rate equation, because it's linear.  But the matrix $H$ is the same---we'll see that explicitly later on.

What's the point?  Well, our ultimate goal is to prove the deficiency zero theorem, which gives equilibrium solutions of the rate equation.  That means finding $x$ with 
$$    Y H x^Y = 0 $$
Now we'll find all equilibria for the Markov process, meaning all $\psi$ with
$$     H \psi = 0 $$
In the next section we'll show some of these have the form
$$   \psi = x^Y $$
So, we'll get 
$$  H x^Y = 0 $$
and thus
$$  Y H x^Y = 0 $$
as desired!   So, let's get to to work.

\subsection{The Markov process of a weighted graph}
\label{sec:22_weighted_graph}

We've been looking at stochastic reaction networks, which are things like this:
\[
\xymatrix{
(0,\infty) & T \ar[l]_<<<<r 
\ar@<0.5ex>[r]^s 
\ar@<-0.5ex>[r]_t & K \ar[r]^Y & \mathbb{N}^S  \\
}
\]
However, we can build a Markov process starting from just part of this information:
\[
\xymatrix{
(0,\infty)  & T \ar[l]_<<<<r 
\ar@<0.5ex>[r]^s \ar@<-0.5ex>[r]_t & K 
}
\]
Let's call this thing a {\bf weighted graph}.  We've been calling the things in $K$ `complexes', but now we'll think of them as `states'.  So:

\begin{definition}  \label{def:weighted_graph} A {\bf weighted graph} consists of:\index{graph theory!weighted graph}\index{weighted graph} 
\begin{itemize}
\item a finite set of {\bf states} $K$, \index{state} 
\item a finite set of {\bf transitions} $T$, \index{transition} 
\item a map $r \colon T \to (0,\infty)$ giving a {\bf rate constant} for each transition, \index{rate constant} \index{rate constant}
\item {\bf source} and {\bf target} maps $s,t \colon T \to K$ saying where each transition starts and ends.  \index{source} \index{target}
\end{itemize}
\end{definition} 
Graph theorists would call the states {\bf vertices}, the transitions {\bf edges} 
and the rate constants {\bf weights}.  However, graph theorists might let their weights
be other kinds of numbers, not necessarily positive real numbers.

Anyway, starting from a weighted graph, we can get a Markov process describing how a probability distribution $\psi$ on our set of states will change with time.  As usual, this Markov process is described by a master equation:
$$  \displaystyle{ \frac{d \psi}{d t} = H \psi } $$
for some Hamiltonian:
$$ H \colon \mathbb{R}^K \to \mathbb{R}^K $$
What is this Hamiltonian, exactly?  Let's think of it as a matrix where $H_{i j}$ is the probability per time for our system to hop from the state $j$ to the state $i$.   This looks backwards, but don't blame us---blame the guys who invented the usual conventions for matrix algebra.\index{matrix!matrix algebra}  Clearly if $i \ne j$ this probability per time should be the sum of the rate constants of all transitions going from $j$ to $i$:
$$ i \ne j \quad \Rightarrow \quad H_{i j} =  \sum_{\tau \colon j \to i} r(\tau) $$
\noindent where we write $\tau \colon j \to i$ when $\tau$ is a transition with source $j$ and target $i$.

Now, we saw way back in Section~\ref{sec:4_hamiltonians} that for a probability distribution to remain a probability distribution as it evolves in time according to the master equation, we need $H$ to be {\bf infinitesimal stochastic}: its off-diagonal entries must be nonnegative, and the sum of the entries in each column must be zero.\index{operator!infinitesimal stochastic} \index{infinitesimal stochastic operator}

The first condition holds already, and the second one tells us what the diagonal entries must be.  So, we're basically done describing $H$.  But we can summarize it this way:

\begin{problem}\label{prob:38} 
\label{problem_hamiltonian}
\index{Hamiltonian!of weighted graph}
 Think of $\mathbb{R}^K$ as the vector space consisting of finite linear combinations of elements $\kappa \in K$.  Then show
$$   H \kappa = \sum_{s(\tau) = \kappa} r(\tau) (t(\tau) - s(\tau)) $$ 
Check directly that $H$ is infinitesimal stochastic.
\end{problem} 

Not only can we get an infinitesimal stochastic operator and thus a Markov process
from any weighted graph, the process can also be reversed!

\begin{problem}\label{prob:39}
Suppose $K$ is any finite set and suppose $H$ is an infinitesimal stochastic operator
on $\mathbb{R}^K$.  Show that $H$ comes from some weighted graph having
$K$ as its set of states.  Furthermore, show that this graph is unique (up to 
isomorphism) if we require that it has most one transition from any state $i \in K$ 
to any state $j \in K$.
\end{problem}

\subsection{Equilibrium solutions of the master equation}
\label{sec:22_equilibrium}

Now we'll classify {\bf equilibrium solutions} of the master equation, meaning $\psi \in \mathbb{R}^K$ with 
$$  H \psi = 0 $$
We'll do only do this when our weighted graph is `weakly reversible'.  This concept doesn't actually depend on the rate constants, so let's be general and say:

\begin{definition}\index{weakly reversible!graph}\index{graph theory!weakly reversible graph} A graph is {\bf weakly reversible} if for every edge $\tau \colon i \to j$, there is {\bf directed path} going back from $j$ to $i$, meaning that we have edges
$$ \tau_1 \colon j \to j_1 , \quad \tau_2 \colon j_1 \to j_2 , \quad \dots, \quad \tau_n \colon j_{n-1} \to i $$
\end{definition} 

This weighted graph is \emph{not} weakly reversible:

\begin{center}
 \includegraphics[width=47mm]{markov_process_vs_reaction_network_2.png}
\end{center}

\noindent
but this one is:

\begin{center}
 \includegraphics[width=50mm]{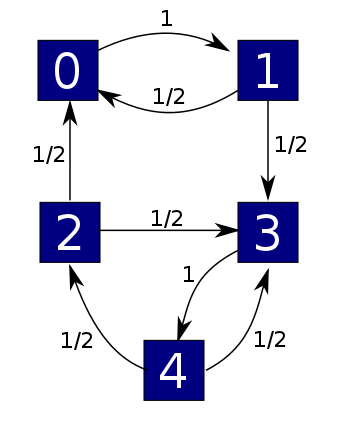}
\end{center}

\noindent
The good thing about the weakly reversible case is that we get one equilibrium solution of the master equation for each component of our graph, and all equilibrium solutions are linear combinations of these.   This is \emph{not} true in general!  For example, this guy is not weakly reversible:

\begin{center}
 \includegraphics[width=70mm]{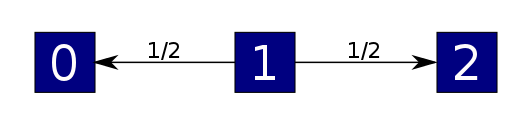}
\end{center}

\noindent
It has only one component, but the master equation has two linearly independent equilibrium solutions: one that vanishes except at the state 0, and one that vanishes except at the state 2.  

\index{connected component|(} \index{graph theory!connected component|(}
The idea of a `component' is supposed to be fairly intuitive---our graph falls apart into pieces called components---but we should make it precise.  As explained in Section~\ref{sec:12}, the graphs we're using here are directed multigraphs, meaning things like
$$ s, t \colon E \to V  $$
where $E$ is the set of edges (our transitions) and $V$ is the set of vertices (our states).    There are actually two famous concepts of `component' for graphs of this sort: `strongly connected' components and `connected' components.   We only need connected components, but let us explain both concepts, in a doubtless futile attempt to slake your insatiable thirst for knowledge. \index{directed multigraph}

Two vertices $i$ and $j$ of a graph lie in the same {\bf strongly connected component} iff you can find a directed path of edges from $i$ to $j$ and also one from $j$ back to $i$.\index{strongly connected component|(}\index{graph theory!strongly connected component|(}    

Remember, a directed path from $i$ to $j$ looks like this:
$$ i \to a \to b \to c \to j $$
Here's a path from $x$ to $y$ that is not directed:
$$ i \to a \leftarrow b \to c \to j $$
and we hope you can write down the obvious but tedious definition of an `undirected path', meaning a path made of edges that don't necessarily point in the correct direction.   Given that, we say two vertices $i$ and $j$ lie in the same {\bf connected component} iff you can find an {\bf undirected} path going from $i$ to $j$.  In this case, there will automatically also be an undirected path going from $j$ to $i$. \index{graph theory!connected component|(} \index{connected component|(}

For example, $i$ and $j$ lie in the same connected component here, but not the same strongly connected component:
$$ i \to a \leftarrow b \to c \to j $$
Here's a graph with one connected component and 3 strongly connected components, which are marked in blue:

\begin{center}
 \includegraphics[width=60mm]{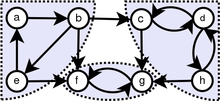}
\end{center}

For the theory we're looking at now, \emph{we only care about connected components, not strongly connected components!}   However:

\begin{problem}\label{prob:40} 
Show that for weakly reversible graphs, the connected components are the same as the strongly connected components.  
\end{problem} 
\index{strongly connected component|)} \index{graph theory!strongly connected component|)} 

With these definitions out of the way, we can state this theorem on equilibria for Markov processes:

\begin{theorem} 
\label{theorem_equilibria}\index{Markov process!equilibrium theorem} 
Suppose $H$ is the Hamiltonian of a weakly reversible weighted graph:
\[
\xymatrix{
(0,\infty)  & T \ar[l]_<<<<r 
\ar@<0.5ex>[r]^s \ar@<-0.5ex>[r]_t & K 
}
\]
Then for each connected component\index{connected component!and probability distributions}  $C \subseteq K$, there exists a unique probability distribution $\psi_C \in \mathbb{R}^K$ that is positive on that component, zero elsewhere, and is an equilibrium solution of the master equation:
$$   H \psi_C = 0 $$
Moreover, these probability distributions $\psi_C$ form a basis for the space of equilibrium solutions of the master equation.  So, the dimension of this space is the number of components of $K$.
\end{theorem} 

\begin{proof} \index{Perron–Frobenius theorem!and the deficiency zero theorem|(}
We start by assuming our graph has one connected component.\index{connected component!single}   We use the Perron--Frobenius theorem, as explained in Section~\ref{sec:16_perron_frobenius}.  This applies to `nonnegative' matrices, meaning those whose entries are all nonnegative.  That is not true of $H$ itself, but only its diagonal entries can be negative, so if we choose a large enough number $c > 0$, $H + c I$ will be nonnegative.  

Since our graph is weakly reversible and has one connected component, it follows straight from the definitions that the operator $H + c I$ will also be `irreducible' in the sense of Section~\ref{sec:16_matrices_to_graphs}.  The Perron--Frobenius theorem then swings into action, and we instantly conclude several things.  \label{irreducible operator}

First, $H + c I$ has a positive real eigenvalue $r$ such that any other eigenvalue, possibly complex, has absolute value $\le r$.  Second, there is an eigenvector $\psi$ with eigenvalue $r$ and all positive components.  Third, any other eigenvector with eigenvalue $r$ is a scalar multiple of $\psi$.  

Subtracting $c$, it follows that $\lambda = r - c$ is the eigenvalue of $H$ with the largest real part.  We have $H \psi = \lambda \psi$, and any other vector with this property is a scalar multiple of $\psi$.   

We can show that in fact $\lambda = 0$.  To do this we copy the argument in Theorem~\ref{theorem_equilibrium}.   First, since $\psi$ is positive we can normalize it to be a probability distribution:
$$ \displaystyle{ \sum_{i \in K} \psi_i = 1 } $$
Since $H$ is infinitesimal stochastic, $\exp(t H)$ sends probability distributions to probability distributions:
$$ \displaystyle{ \sum_{i \in K} (\exp(t H) \psi)_i = 1 } $$
for all $t \ge 0$.  On the other hand,
$$ \sum_{i \in K} (\exp(t H)\psi)_i = \sum_{i \in K} e^{t \lambda} \psi_i = e^{t \lambda} $$
so we must have $\lambda = 0$.  

We conclude that when our graph has one connected component, there is a probability distribution $\psi \in \mathbb{R}^K$ that is positive everywhere and has $H \psi = 0$.  Moreover, any $\phi \in \mathbb{R}^K$ with $H \phi = 0$ is a scalar multiple of $\psi$.

When our graph is weakly reversible but has several components, the matrix $H$ is 
block diagonal, with one block for each component.  So, we can run the above 
argument on each component $C \subseteq K$ and get a probability distribution 
$\psi_C \in \mathbb{R}^K$ that is positive on $C$.  We can then check that $H \psi_C 
= 0$ and that every $\phi \in \mathbb{R}^K$ with $H \phi = 0$ can be expressed as a 
linear combination of these probability distributions $\psi_C$ in a unique way.    
\index{connected component|)} \index{graph theory!connected component|)}
\end{proof} 

By the way: even in the fully general case, where our weighted graph may not be 
weakly reversible, we can completely understand all equilibrium solutions of the master 
equation. We do not need this to prove the deficiency zero theorem.  So, only read the 
rest of this section if you want to become an expert on Markov processes.

Suppose we have any weighted graph.  Write its set $K$ of vertices as a disjoint union of strongly connected components:
$$  K =\bigcup_{i=1}^n C_i. $$
Say $i \le j$ if there is a directed path from some vertex in $C_i$ to some vertex in $C_j$.  (This is equivalent to there being a directed path from \emph{any} vertex in $C_i$ to \emph{any} vertex in $C_j$.)  If $i \le j$, you can show that probability distributions that are nonzero only on vertices in $C_i$ will become nonzero on $C_j$ as you evolve them using the master equation.  In other words, probability will flow 
from the $i$th strongly connected component to the $j$th.  

We say a strongly connected component $C_i$ is {\bf terminal} \index{graph theory!strongly connected component!terminal} \index{strongly connected component!terminal} if $i \le j$ implies $i = j$.  The idea is that probability cannot flow out of a terminal component. Using this idea, you can classify equilibrium solutions of the master equation for any weighted graph:

\begin{problem}\label{prob:41} 
Suppose $H$ is the Hamiltonian of a weighted graph:
\[
\xymatrix{
(0,\infty)  & T \ar[l]_<<<<r 
\ar@<0.5ex>[r]^s \ar@<-0.5ex>[r]_t & K 
}
\]
Prove that for each terminal strongly connected component of this graph, there is a unique probability distribution $\psi \in \mathbb{R}^K$ that is positive on that component, zero elsewhere, and is an equilibrium solution of the master equation:
$$   H \psi = 0. $$ 
Prove that probability distributions of this sort are a basis for the space of equilibrium solutions of the master equation.  So, the dimension of this space is the number of terminal strongly connected components of $K$.
\end{problem}

\subsection{The Hamiltonian, revisited}
\label{sec:22_hamiltonian}

One last small piece of business: in Section~\ref{sec:21} we mentioned very slick formula for the Hamiltonian of a weighted graph:
$$  H = \partial s^\dagger $$
We'd like to prove it agrees with the formula we derived in Problem \ref{problem_hamiltonian}:
$$ H \kappa = \sum_{s(\tau) = \kappa} r(\tau) (t(\tau) - s(\tau)) $$ 
for any complex $\kappa$, regarded as one of th standard basis vectors 
in $\mathbb{R}^K$.  

Recall the setup.  We start with any weighted graph:
\[
\xymatrix{
(0,\infty)  & T \ar[l]_<<<<r 
\ar@<0.5ex>[r]^s \ar@<-0.5ex>[r]_t & K 
}
\]
We extend $s$ and $t$ to linear maps between vector spaces:
\[
\xymatrix{
 \mathbb{R}^T 
\ar@<0.5ex>[r]^s 
\ar@<-0.5ex>[r]_t & \mathbb{R}^K  \\
}
\]
We define the {\bf boundary operator} by:
\index{boundary operator}
$$ \partial = t - s $$
Then we put an inner product on the vector spaces $\mathbb{R}^T$ and $\mathbb{R}^K$.   For $\mathbb{R}^K$ we let the elements of $K$ be an orthonormal basis, but for $\mathbb{R}^T$ we define the inner product in a more clever way involving the rate constants:
$$ \displaystyle{ \langle \tau, \tau' \rangle = \frac{1}{r(\tau)} \delta_{\tau, \tau'} } $$
where $\tau, \tau' \in T$.  These inner products let us define adjoints of the maps $s, t$ and $\partial$, via formulas like this:
$$ \langle s^\dagger \phi, \psi \rangle = \langle \phi, s \psi \rangle $$
Then we claim:

\begin{theorem}
\label{theorem_hamiltonian}
The operator
$$   H = \partial s^\dagger $$
\index{Hamiltonian!of weighted graph} has
$$   H \kappa = \sum_{s(\tau) = \kappa} r(\tau) (t(\tau) - s(\tau)) $$ 
for any $\kappa \in K$.
\end{theorem} 

\begin{proof} 
First we claim that
$$  s^\dagger \kappa = \sum_{\tau\; : \; s(\tau) = \kappa} r(\tau) \, \tau $$
where we write $\tau \in T$ for the corresponding standard basis vector of
$\mathbb{R}^T$.  To prove this, it's enough to check that taking the inner products of either side with any basis vector $\tau'$, we get results that agree.  On the one hand:
$$ \begin{array}{ccl}  \langle \tau' , s^\dagger \kappa \rangle &=& 
\langle s \tau', \kappa \rangle \\
\\
&=& \delta_{s(\tau'), \kappa}  
\end{array} $$
On the other hand, using our clever inner product on $\mathbb{R}^T$ we see that:
$$ \begin{array}{ccl} \displaystyle{ \langle \tau', \sum_{\tau \; : \; s(\tau) = \kappa} r(\tau) \, \tau \rangle } &=& \displaystyle{ \sum_{\tau \; : \; s(\tau) = \kappa} r(\tau) \, \langle \tau', \tau \rangle } \\   \\
&=& \displaystyle{ \sum_{\tau \; : \; s(\tau) = \kappa} \delta_{\tau', \tau} }\\  \\
&=& 
\delta_{s(\tau'), \kappa} 
\end{array} $$
where the factor of $1/r(\tau)$ in the inner product on $\mathbb{R}^T$ cancels the visible factor of $r(\tau)$.    So indeed the results match.

Using this formula for $s^\dagger \kappa $ we now see that
$$  \begin{array}{ccl}  H \kappa &=& \partial s^\dagger \kappa 
\\   \\
&=& \partial \displaystyle{ \sum_{\tau \; : \; s(\tau) = \kappa} r(\tau) \, \tau }  \\  \\
&=& \displaystyle{ \sum_{\tau \; : \; s(\tau) = \kappa} r(\tau) \, (t(\tau) - s(\tau)) }
\end{array} $$
which is precisely what we want.   
\end{proof} 

We hope you see through the formulas to their intuitive meaning.  As usual, the formulas are just a way of precisely saying something that makes plenty of sense.  If $\kappa$ is some state of our Markov process, $ s^\dagger \kappa$ is the sum of all transitions starting at this state, weighted by their rates.  Applying $ \partial$ to a transition tells us what change in state it causes.  So $ \partial s^\dagger \kappa$ tells us the rate at which things change when we start in the state $ \kappa$.  That's why $\partial s^\dagger$ is the Hamiltonian for our Markov process.  After all, the Hamiltonian tells us how things change:
$$ \displaystyle{ \frac{d \psi}{d t} = H \psi } $$
Okay, we've got all the machinery in place.  Next we'll prove the deficiency zero theorem!

%%%%%%%%% SECTION 23 %%%%%%%%%%% 

\newpage
\section[Proof of the deficiency zero theorem]{Proof of the deficiency zero theorem}\label{sec:23}
\index{Deficiency Zero Theorem!proof of|(} 

Now we've reached the climax of our story so far: we're ready to prove the deficiency zero theorem.  First let's talk about it informally a bit.  Then we'll review the notation, and then---hang on to your seat!---we'll give the proof.

The crucial trick is to relate a bunch of chemical reactions, described by a `reaction network' like this:

%<div align = "center}<img width = "200" src = "http://math.ucr.edu/home/baez/networks/chemical_reaction_network_part_20_VII.png" alt = "}</div>
\begin{center}
 \includegraphics[width=40mm]{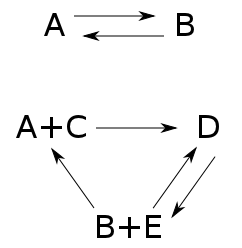}
\end{center}

\noindent
to a simpler problem where a system randomly hops between states arranged in the same pattern:

\begin{center}
 \includegraphics[width=50mm]{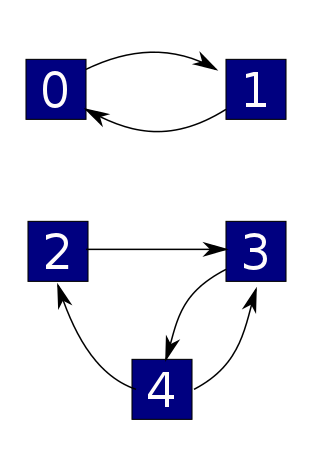}
\end{center}

\noindent
Tihs is sort of amazing, because we've thrown out lots of detail.  It's also amazing because this simpler problem is \emph{linear}.  In the original problem, the chance that a reaction turns a $B + E$ into a $D$ is proportional to the number of $B$'s \emph{times} the number of $E$'s.  That's nonlinear!  But in the simplified problem, the chance that your system hops from state 4 to state 3 is just proportional to the probability that it's in state 4 to begin with.  That's linear.  

The wonderful thing is that, at least under some conditions, we can an find an \emph{equilibrium} for our chemical reactions starting from an equilibrium solution of the simpler problem. 

Let's roughly sketch how it works, and where we are so far.  Our simplified problem is described by an equation like this:
$$ \displaystyle{ \frac{d}{d t} \psi = H \psi } $$
where $\psi$ is a function that the probability of being in each state, and $H$ describes the probability per time of hopping from one state to another.   We can easily understand quite a lot about the equilibrium solutions, where $\psi$ doesn't change at all:
$$   H \psi = 0 $$
because this is a linear equation.  We did this in Theorem \ref{theorem_equilibria}.  Of course, when we say `easily', that's a relative thing: we needed to use the Perron--Frobenius theorem, which we introduced in Section~\ref{sec:16}.  But that's a well-known theorem in linear algebra, and it's easy to apply here.
\index{Perron–Frobenius theorem!and the deficiency zero theorem|)}

In Theorem \ref{theorem_rate_equation}, we saw that the original problem was described by an equation like this, called the `rate equation':
$$  \displaystyle{ \frac{d x}{d t} = Y H x^Y  } $$
Here $x$ is a vector whose entries describe the amount of each kind of chemical: the amount of A's, the amount of B's, and so on.  The matrix $H$ is the same as in the simplified problem, but $Y$ is a matrix that says how many times each chemical shows up in each spot in our reaction network:

%<div align="center}<img width="200" src="http://math.ucr.edu/home/baez/networks/chemical_reaction_network_part_20_VII.png" alt="" /></div>
\begin{center}
 \includegraphics[width=40mm]{chemical_reaction_network_part_20_VII.png}
\end{center}

The key thing to notice is $x^Y$, where we take a vector and raise it to the power of a matrix.   We explained this operation back in Section~\ref{sec:21_rate_equation}.  It's this operation that says how many B + E pairs we have, for example, given the number of B's and the number of E's.  It's this that makes the rate equation nonlinear.\index{nonlinearity!rate equation}  

Now, we're looking for equilibrium solutions of the rate equation, where the rate of change is zero:
$$  Y H x^Y = 0  $$
But in fact we'll do even better!  We'll find solutions of this:
$$ H x^Y = 0$$
And we'll get these by taking our solutions of this:
$$ H \psi = 0 $$
and adjusting them so that 
$$ \psi = x^Y $$
while $\psi$ remains a solution of $H \psi = 0$.  

But: how do we do this `adjusting'?  That's the crux of the whole business!  That's what we'll do now.   Remember, $\psi$ is a function that gives a probability for each `state', or numbered box here:

\begin{center}
 \includegraphics[width=50mm]{markov_process_vs_reaction_network_4.png}
\end{center}

\noindent
The picture here consists of two pieces, called `connected components':\index{connected component!functions on}
 the piece containing boxes 0 and 1, and the piece containing boxes 2, 3 and 4.  It turns out that we can multiply $\psi$ by a function that's constant on each connected component, and if we had $H \psi = 0$ to begin with, that will still be true afterward.  The reason is that there's no way for $\psi$ to `leak across' from one component to another.  It's like having water in separate buckets.  You can increase the amount of water in one bucket, and decrease it in another, and as long as the water's surface remains flat in each bucket, the whole situation remains in equilibrium.

That's sort of obvious.  What's not obvious is that we can adjust $\psi$ this way so as to ensure
$$ \psi = x^Y $$
\noindent for some $x$.  

And indeed, it's not always true!  It's only true if our reaction network obeys a special condition.  It needs to have `deficiency zero'.  We defined this concept back in Section~\ref{sec:20}, but now we'll finally use it for something.  It turns out to be precisely the right condition to guarantee we can tweak any function on our set of states,  multiplying it by a function that's constant on each connected component, and get a new function $\psi$ with
$$ \psi = x^Y $$
When all is said and done, that is the key to the deficiency zero theorem.

\subsection[Review]{Review}

The battle is almost upon us---we've got one last chance to review our notation.  We start with a stochastic reaction network:
\[
\xymatrix{
(0,\infty)  & T \ar[l]_<<<<r 
\ar@<0.5ex>[r]^s \ar@<-0.5ex>[r]_t & K \ar[r]^Y & \mathbb{N}^S  \\
}
\]

This consists of:

\begin{itemize} 
\item finite sets of {\bf transitions} $T$, {\bf complexes} $K$ and {\bf species} $S$,
\index{complex} \index{species} \index{transition}
\item a map $r\colon T \to (0,\infty)$ giving a {\bf rate constant} for each transition,
\index{rate constant}
\item {\bf source} and {\bf target} maps $s,t \colon T \to K$ saying where each transition starts and ends,
\index{source} \index{target}
\item a one-to-one map $Y \colon K \to \mathbb{N}^S$ saying how each complex is made of species.
\end{itemize} 

Then we extend $s, t$ and $Y$ to linear maps:
\[
\xymatrix{
 \mathbb{R}^T 
\ar@<0.5ex>[r]^s 
\ar@<-0.5ex>[r]_t & \mathbb{R}^K \ar[r]^Y & \mathbb{R}^S  \\
}
\]

Then we put inner products on these vector spaces as described in Section~\ref{sec:21_rate_equation}, which lets us `turn around' the maps $s$ and $t$ by taking their adjoints:
$$ s^\dagger, t^\dagger \colon \mathbb{R}^K \to \mathbb{R}^T $$
More surprisingly, we can `turn around' $Y$ and get a \emph{nonlinear} map using `matrix exponentiation':\index{matrix!exponentiation}\index{nonlinearity!matrix exponentiation}
$$  \begin{array}{ccc} \mathbb{R}^S &\to& \mathbb{R}^K \\  
                           x     &\mapsto&   x^Y \end{array} $$
This is most easily understood by thinking of $x$ as a row vector and $Y$ as a matrix:
$$ \begin{array}{ccl} x^Y &=& {\left( \begin{array}{cccc} 
x_1 , & x_2 , & \dots, & x_k \end{array} \right)}^{
 \left( \begin{array}{cccc} 
Y_{11} & Y_{12}  & \cdots & Y_{1 \ell} \\
Y_{21} & Y_{22}  & \cdots & Y_{2 \ell} \\
\vdots & \vdots  & \ddots & \vdots \\
Y_{k1} & Y_{k2}  & \cdots & Y_{k \ell} \end{array} \right)} 
\\
\\ 
&=& \left( x_1^{Y_{11}} \cdots  x_k^{Y_{k1}} ,\; \dots, \; x_1^{Y_{1 \ell}} \cdots x_k^{Y_{k \ell}} \right)
\end{array}
$$
Remember, complexes are made out of species.   The matrix entry $Y_{i j}$ says how many things of the $j$th species there are in a complex of the $i$th kind.   If $\psi \in \mathbb{R}^K$ says how many complexes there are of each kind, $Y \psi \in \mathbb{R}^S$ says how many things there are of each species.  Conversely, if $x \in \mathbb{R}^S$ says how many things there are of each species, $x^Y \in \mathbb{R}^K$ says how many ways we can build each kind of complex from them.

So, we get these maps:
\[
\xymatrix{
 \mathbb{R}^T  \ar@<0.5ex>[rr]^s \ar@<-0.5ex>[rr]_t
&& \mathbb{R}^K \ar@/_0.8pc/@<-1ex>[ll]_{s^\dagger} \ar@/^0.8pc/@<1ex>[ll]^{t^\dagger}
\ar[rr]_Y
&& \mathbb{R}^S
\ar@/_0.8pc/@<-0.5ex>[ll]_{(\cdot)^Y}
 \\
}
\]

Next, the {\bf boundary operator}
$$  \partial \colon \mathbb{R}^T \to \mathbb{R}^K $$ 
describes how each transition causes a change in complexes:
$$ \partial = t - s $$
As we saw in Theorem~\ref{theorem_hamiltonian}, there is a {\bf Hamiltonian}
$$  H  \colon \mathbb{R}^K \to \mathbb{R}^K $$
describing a Markov processes on the set of complexes, given by
$$  H = \partial s^\dagger $$ 
But the star of the show is the rate equation.  This describes how the number of things of each species changes with time.  We write these numbers in a list and get a vector $x \in \mathbb{R}^S$ with nonnegative components.  The {\bf rate equation} says:
$$ \displaystyle{ \frac{d x}{d t} = Y H x^Y } $$
We can read this as follows:

\begin{itemize} 
\item $x$ says how many things of each species we have now.
\item $x^Y$ says how many complexes of each kind we can build from these species.
\item $s^\dagger x^Y$ says how many transitions of each kind can originate starting from these complexes, with each transition weighted by its rate.
\item $ H x^Y = \partial s^\dagger x^Y$ is the rate of change of the number of complexes of each kind, due to these transitions.
\item $Y H x^Y$ is the rate of change of the number of things of each species.
\end{itemize} 

\subsection[The proof]{The proof}

We are looking for {\bf equilibrium solutions} of the rate equation, where the number of things of each species doesn't change at all:
$$   Y H x^Y = 0 $$
In fact we will find {\bf complex balanced} equilibrium solutions, where even the number of complexes of each kind doesn't change:
$$  H x^Y = 0 $$
\index{equilibrium} \index{complex balanced equilibrium} \index{equilibrium!complex
balanced}
More precisely, we have:

\begin{theorem}[{\bf Deficiency Zero Theorem---Child's Version}]\index{Deficiency Zero Theorem!child's version} 
Suppose we have a reaction network obeying these two conditions:
 \begin{itemize} 
 \item It is {\bf weakly reversible}, meaning that whenever there's a transition from one complex $\kappa$ to another $\kappa'$, there's a directed path of transitions going back from $\kappa'$ to $\kappa$.
\item  It has {\bf deficiency zero}, meaning $  \mathrm{im} \partial  \cap \mathrm{ker} Y = \{ 0 \} $.
\end{itemize} 
Then for any choice of rate constants there exists a complex balanced equilibrium solution of the rate equation where all species are present in nonzero amounts.  In other words, there exists $x \in \mathbb{R}^S$ with all components positive and such that:
$$H x^Y = 0$$
\end{theorem} 

\begin{proof}
Because our reaction network is weakly reversible, Theorem \ref{theorem_equilibria} there exists $\psi \in (0,\infty)^K$ with 
$$ H \psi = 0 $$
This $\psi$ may not be of the form $x^Y$, but we shall adjust $\psi$ so that it becomes of this form, while still remaining a solution of $H \psi = 0 $.   To do this, we need a couple of lemmas:

\begin{lemma} 
\label{lemma:decomposition}
$\mathrm{ker} \partial^\dagger + \mathrm{im} Y^\dagger = \mathbb{R}^K$. 
\end{lemma} 

\begin{proof} 
We need to use a few facts from linear algebra.  If $V$ is a finite-dimensional vector space with inner product, the \href{http://en.wikipedia.org/wiki/Orthogonal_complement}{\bf orthogonal complement} $L^\perp$ of a subspace $L \subseteq V$ consists of vectors that are orthogonal to everything in $L$:
$$  L^\perp = \{ v \in V \; : \; \quad \forall w \in L \quad \langle v, w \rangle = 0 \} $$
We have
$$  (L \cap M)^\perp = L^\perp + M^\perp  $$
where $L$ and $M$ are subspaces of $V$ and $+$ denotes the sum of two subspaces: that is, the smallest subspace containing both.  Also, if $T \colon V \to W$ is a linear map between finite-dimensional vector spaces with inner product, we have
$$   (\mathrm{ker} T)^\perp = \mathrm{im} T^\dagger $$
and
$$  (\mathrm{im} T)^\perp = \mathrm{ker} T^\dagger $$
Now, because our reaction network has deficiency zero, we know that
$$ \mathrm{im} \partial \cap \mathrm{ker} Y = \{ 0 \} $$
Taking the orthogonal complement of both sides, we get
$$ (\mathrm{im} \partial \cap \mathrm{ker} Y)^\perp = \mathbb{R}^K $$
and using the rules we mentioned, we obtain
$$  \mathrm{ker} \partial^\dagger + \mathrm{im} Y^\dagger = \mathbb{R}^K $$
as desired.    
\end{proof} 

Now, given a vector $\phi$ in $\mathbb{R}^K$ or $\mathbb{R}^S$ with all positive components, we can define the logarithm of such a vector, componentwise:
$$    (\ln \phi)_i = \ln (\phi_i)   $$
Similarly, for any vector $\phi$ in either of these spaces, we can define its exponential in a componentwise way:
$$    (\exp \phi)_i = \exp(\phi_i)  $$
These operations are inverse to each other.  Moreover:

\begin{lemma} 
\label{lemma:exponential}
The nonlinear operator\index{nonlinearity!relation to linear operator} 
$$  \begin{array}{ccc} \mathbb{R}^S &\to& \mathbb{R}^K \\  
                           x     &\mapsto&   x^Y \end{array} $$
is related to the linear operator
$$  \begin{array}{ccc} \mathbb{R}^S &\to& \mathbb{R}^K \\  
                           x     &\mapsto&   Y^\dagger x \end{array} $$
by the formula 
$$         x^Y = \exp(Y^\dagger \ln x ) $$
which holds for all $x \in (0,\infty)^S$.  
\end{lemma} 

\begin{proof} 
A straightforward calculation.  By the way, this formula would look a bit nicer if we treated $\ln x$ as a row vector and multiplied it on the right by $Y$: then we would have
$$          x^Y = \exp((\ln x) Y) $$
The problem is that we are following the usual convention of multiplying vectors by matrices on the left, yet writing the matrix on the right in $x^Y$.  Taking the transpose $Y^\dagger$ of the matrix $Y$ serves to compensate for this.   

Now, given our vector $\psi \in (0,\infty)^K$ with $H \psi = 0$, we can take its logarithm and get $\ln \psi \in \mathbb{R}^K$.   Lemma \ref{lemma:decomposition} says that
$$ \mathbb{R}^K = \mathrm{ker} \partial^\dagger + \mathrm{im} Y^\dagger $$
so we can write
$$ \ln \psi =  \alpha + Y^\dagger \beta $$
where $\alpha \in \mathrm{ker} \partial^\dagger$ and $\beta \in \mathbb{R}^S$.  Moreover, we can write
$$ \beta = \ln x $$
for some $x \in (0,\infty)^S$, so that
$$ \ln \psi = \alpha + Y^\dagger (\ln x) $$
Exponentiating both sides componentwise, we get
$$   \psi  =   \exp(\alpha) \; \exp(Y^\dagger (\ln x)) $$
where at right we are taking the componentwise product of vectors.  Thanks to Lemma \ref{lemma:exponential}, we conclude that
$$   \psi = \exp(\alpha) x^Y   $$
as desired.  \end{proof}

So, we have taken $\psi$ and \emph{almost} written it in the form $x^Y$---but not quite!  We can adjust $\psi$ to make it be of this form:
$$   \exp(-\alpha) \psi = x^Y $$
Clearly all the components of $\exp(-\alpha) \psi$ are positive, since the same is true for both $\psi$ and $\exp(-\alpha)$.  So, the only remaining task is to check that 
$$  H(\exp(-\alpha) \psi) = 0 $$
We do this using two lemmas:

\begin{lemma} 
 If $H \psi = 0$ and $\alpha \in \mathrm{ker} \partial^\dagger$, then $H(\exp(-\alpha) \psi) = 0$. 
\end{lemma} 
 
 \begin{proof} 
 It is enough to check that multiplication by $\exp(-\alpha)$ commutes with the Hamiltonian $H$, since then
$$ H(\exp(-\alpha) \psi) = \exp(-\alpha) H \psi = 0 $$
Recall from Section~\ref{sec:22_weighted_graph} that $H$ is the Hamiltonian of a Markov process associated to this `weighted graph':
\[
\xymatrix{
(0,\infty)  & T \ar[l]_<<<<r 
\ar@<0.5ex>[r]^s \ar@<-0.5ex>[r]_t & K 
}
\]
As noted here:

\begin{enumerate}
\item[\cite{BF12}]  John Baez and Brendan Fong, A Noether theorem for Markov processes, {\sl J.\ Math.\ Phys.} {\bf 54}, 013301, 2013.  Also available as \href{http://arxiv.org/abs/1203.2035}{arXiv:1203.2035}. 
\end{enumerate}

\noindent
multiplication by some function on $K$ commutes with $H$ if and only if that function is constant on each connected component of this graph.   Such functions are called {\bf conserved quantities}.  

So, it suffices to show that $\exp(-\alpha)$ is constant on each connected component.   For this, it is enough to show that $\alpha$ itself is constant on each connected component.  But this will follow from the next lemma, since $\alpha \in \mathrm{ker} \partial^\dagger$.   
\end{proof} 

\begin{lemma}
A function $\alpha \in \mathbb{R}^K$ is a conserved quantity iff $\partial^\dagger \alpha = 0 $.   In other words, $\alpha$ is constant on each connected component of the graph $s, t \colon T \to K$ iff $\partial^\dagger \alpha = 0 $.\index{connected component!and conserved quantities} 
\end{lemma} 

\begin{proof} 
Suppose $\partial^\dagger \alpha = 0$, or in other words, $\alpha \in \mathrm{ker} \partial^\dagger$, or in still other words, $\alpha \in (\mathrm{im} \partial)^\perp$.  To show that $\alpha$ is constant on each connected component, it suffices to show that whenever we have two complexes connected by a transition, like this:
$$  \tau \colon \kappa \to \kappa' $$
then $\alpha$ takes the same value at both these complexes:
$$ \alpha_\kappa = \alpha_{\kappa'} $$
To see this, note
$$  \partial \tau = t(\tau) - s(\tau) = \kappa' - \kappa $$
and since $\alpha \in (\mathrm{im} \partial)^\perp$, we conclude
$$  \langle \alpha, \kappa' - \kappa \rangle = 0 $$
But calculating this inner product, we see
$$ \alpha_{\kappa'} - \alpha_{\kappa} = 0  $$
as desired.  

For the converse, we simply turn the argument around: if $\alpha$ is constant on each connected component,\index{connected component!constant} we see $\langle \alpha, \kappa' - \kappa \rangle = 0$ whenever there is a transition $\tau \colon \kappa \to \kappa'$.  It follows that $\langle \alpha, \partial \tau \rangle = 0$ for every transition $\tau$, so $\alpha \in (\mathrm{im} \partial)^\perp $.

And thus concludes the proof of the lemma!   
\end{proof} 
\index{Deficiency Zero Theorem!proof of|)}

And thus concludes the proof of the theorem!    \end{proof}

\index{Deficiency Zero Theorem|)}

Of course there is much more to say, but we'll stop here---except for a few
final tidbits.

\newpage

%%%%%%%%%%%%%%%%%%%%%%%%%%%%%%%%%%
\part*{FURTHER DIRECTIONS}

There's a lot we didn't get off our chests yet.  And we thought: perhaps you want to say something too? Because we're so nice, we're going to tell you about a few special tidbits we've been saving.  We think they represent good starting points for further research.  The entire subject of stochastic mechanics is really just getting started!  So please let us know if you take it further in any direction.

In Chapter \ref{sec:24} we present some further thoughts on Noether's theorem
relating symmetries and conserved quantities.  In Chapter \ref{sec:10} we saw that 
Noether's theorem looks a bit more complicated in stochastic mechanics than in
quantum mechanics.   But in Chapters \ref{sec:15} and \ref{sec:16} we saw that
Dirichlet operators can serve as Hamiltonians for \emph{both} quantum and stochastic mechanics.  This suggests that the stochastic version of Noether's theorem may
become more beautiful when the Hamiltonian is a Dirichlet operator.  And indeed
this is true!  Our presentation centers around a result of this sort, proved with the
help of Ville Bergholm. \index{Bergholm, Ville}  We find it tantalizing that this 
result looks so much like the quantum version of Noether's theorem, yet its proof 
is different.  Hopefully this will shed more light on the differences and connections between stochastic and quantum mechanics.  We conclude Chapter \ref{sec:24} 
with a series of problems for further study.

Chapter \ref{sec:25} discusses various connections between Petri nets and computation.  Because of their versatility and simplicity, Petri nets have long been used as models of computation.  In Sections \ref{sec:25_1} and \ref{sec:25_2}
we discuss the `reachability problem': the question of which collections of molecules can turn into which other collections, given a fixed set of chemical reactions.  In Section \ref{sec:25_2} we explain how the reachability problem for Petri nets is really a problem in category theory.   More recently, molecular biologists have used \emph{stochastic} Petri nets as a 
model of `chemical computers'.  Section \ref{sec:25_3}, based on a talk by David Soloveichik \cite{Sol14}, reviews a bit of this recent work.  Finally, Section \ref{sec:25_5} lists some free software for working with Petri nets and
reaction networks.  We thank Jim Stuttard \index{Stuttard, Jim} for making
this possible.
 
\newpage 

%%%%%%%%% SECTION 24 %%%%%%%%%%% 
\section[Noether's theorem for Dirichlet operators]{Noether's theorem for Dirichlet operators}
\label{sec:24}

In 1915, Emmy Noether showed that the symmetries of a physical system give conserved quantities, and vice versa.  Her result has become a staple of modern physics, so it is often known as \emph{Noether’s theorem}, though 
in fact she proved more than one important theorem---even on this subject!

\begin{center}
 \includegraphics[width=.3\textwidth]{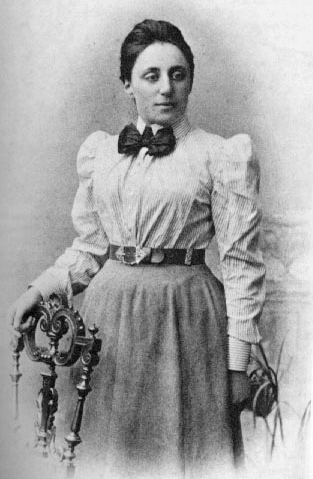}
\end{center}

Though she introduced it in classical mechanics, this theorem and its generalizations
have found particularly wide use in quantum mechanics.  In Chapters \ref{sec:10} and \ref{sec:12} we introduced a version for stochastic mechanics and compared it to
the quantum version.  The stochastic version was different than what would be 
expected from the quantum one, so it is interesting to try to figure out why.

Together with Ville Bergholm,\index{Bergholm, Ville} we've been trying to get to the bottom of these differences.  This chapter is based on an article by Bergholm\index{Bergholm, Ville} on the Azimuth blog\index{Azimuth!blog}. We 
thank him\index{Bergholm, Ville} for writing this article and working with us on the 
result presented here: a Noether's theorem applicable to stochastic mechanics when
the Hamiltonian is a Dirichlet operator.

As you may recall from Chapters \ref{sec:15} and \ref{sec:16}, Dirichlet operators 
lie in the intersection of stochastic and quantum mechanics.  That is, they can be used to generate both Markov semigroups and 1-parameter unitary groups.
So, they're a subclass of the infinitesimal stochastic operators considered in the stochastic version of Noether's theorem.

The extra structure of Dirichlet operators---compared with the wider class of infinitesimal stochastic operators---provides a handle for us to dig a little deeper 
into understanding the intersection of these two theories.  We hope the problems at
the end of this chapter hint at some directions for future work.  

Before we get into the details, let’s recall how conservation laws work in quantum mechanics, and then contrast this with what was shown for stochastic systems in Chapter \ref{sec:10}. (For a more detailed comparison between 
the stochastic and quantum versions of Noether's theorem, see Chapter \ref{sec:12}.)

\subsection{The quantum case} 
\label{sec:24_1}
\index{Noether's theorem!quantum vs stochastic|(}

Let's review the basics of quantum theory, to ensure that everything we need is 
right here in front of us.   In a quantum system with $n$ states, a state
$\psi$ is a unit vector in $\mathbb{C}^n$.  The unitary time evolution of a state $\psi(t) \in \mathbb{C}^n$ is generated by a self-adjoint $n \times n$ matrix $H$
called the Hamiltonian.  In other words, $\psi(t)$ satisfies Schr\"{o}dinger's equation
$$ i \hbar \frac{d}{d t} \psi(t) = H \psi(t). $$
Thus, the state of a system starting off at time zero in the state $\psi_0$ and 
evolving for a time $t$ is given by
$$ \psi(t) = e^{-i t H} \psi_0. $$
The observable properties of a quantum system are associated with self-adjoint operators.  In the state $\psi \in \mathbb{C}^n$, the expected value of the observable associated to a self-adjoint operator $O$ is
$$ \langle \psi , O \psi \rangle $$
This expected value is constant in time for all states if and only if $O$ commutes with the Hamiltonian $H$:
$$ [O, H] = 0 \quad \iff \quad \frac{d}{d t} \langle O \rangle_{\psi(t)} = 0 \quad \forall \psi_0 , \forall t $$
In this case we say $O$ is a conserved quantity. The equivalence of the two
conditions above is exactly the quantum version of Noether’s theorem, as you'll recall from Theorem \ref{theorem:quantum-noether} in Chapter \ref{sec:12}.  
\index{Noether's theorem!quantum version}

\subsection{The stochastic case} 
\label{sec:24_2}

In stochastic mechanics, the story changes. Now a state $\psi$ is a probability distribution: a vector with $n$ nonnegative components that sum to 1. Schr\"{o}dinger's equation gets replaced by the master equation:
$$ \frac{d}{d t} \psi(t) = H \psi(t) $$
If we start with a probability distribution $\psi_0$ at time zero and evolve it according to this equation, at any later time we have
$$ \psi(t) = e^{t H} \psi_0$$
We want this always to be a probability distribution. To ensure this, the Hamiltonian $H$ should be infinitesimal stochastic: that is, a real-valued $n \times n$ matrix where the off-diagonal entries are nonnegative and the entries of each column sum to zero. It no longer needs to be self-adjoint!

When $H$ is infinitesimal stochastic, the operators $e^{t H}$ map the set of probability distributions to itself whenever $t \ge 0$, and we call this family of operators a continuous-time Markov process, or more precisely a Markov semigroup.

In stochastic mechanics, we say an observable $O$ is a diagonal real $n \times n$
matrix.  This lets us make sense of the commutator $[O,H]$, which is crucial in the stochastic version of Noether's theorem.  However, we can also create a vector $\widehat{O} \in \mathbb{R}^n$ whose components are the diagonal entries of
$O$:
$$  \widehat{O}_i = O_{ii} $$
This lets us write the expected value of the observable $I$ in the state $\psi$ as
$$  \langle \widehat{O}, \psi \rangle $$
We saw this formula long ago, at the very end of Section \ref{sec:11_observables}.  Soon we will finally use it to do something interesting!

Here is the stochastic version of Noether's theorem.  We will state it in a slightly different way than before, but it is really just Theorem \ref{theorem:stochastic-noether}:

\begin{theorem}[Noether’s Theorem for Markov Processes]
\label{theorem:stochastic-noether-2} \index{Noether's theorem!stochastic
version}
Suppose $H$ is an infinitesimal stochastic matrix and $O$ is an observable 
(a real diagonal matrix). Then
$$ [O,H] =0 $$
if and only if 
$$ \frac{d}{d t} \langle \hat{O}, \psi(t) \rangle = 0 $$
 and
$$ \frac{d}{d t}\langle \widehat{O^2}, \psi(t) \rangle = 0 $$
for all $t \ge 0$ and all $\psi(t)$ obeying the master equation.   
\end{theorem}
\noindent Here $\widehat{O^2}$ is the vector whose components are the diagonal
entries of $O^2$.

So, just as in quantum mechanics, whenever $[O,H]=0$ the expected value of $O$ will 
be conserved:
$$ \frac{d}{d t} \langle O, \psi(t) \rangle = 0 $$
for any $\psi_0$ and all $t \ge 0$. However, in Chapter \ref{sec:10} we already saw 
that---unlike in quantum mechanics---you need more than just the expectation value of 
the observable $O$ to be constant to obtain the equation $[O,H]=0$. You really need 
both
$$\frac{d}{d t} \langle \hat{O}, \psi(t) \rangle = 0$$
together with
$$\frac{d}{d t} \langle \widehat{O^2}, \psi(t)\rangle = 0$$
for all initial states $\psi_0$ to be sure that $[O,H]=0$.  So, symmetries and conserved quantities have a rather different relationship than they do in quantum theory.
\index{Noether's theorem!quantum vs stochastic|)}

\subsection{A Noether theorem for Dirichlet operators} 
\label{sec:24_3}
\index{Dirichlet operator!and Noether's theorem|(}\index{Noether's theorem!for Dirichlet operators|(} 

Remember in Chapter \ref{sec:15} where we studied when the infinitesimal generator of our Markov semigroup is also self-adjoint? We're going to think about that scenario here too.  In other words, what if $H$ is both an infinitesimal stochastic matrix but also its own transpose.  Then it’s called a Dirichlet operator... and in this case, we get a stochastic version of Noether’s theorem that more closely resembles the usual quantum one.

\begin{theorem}[Noether’s Theorem for Dirichlet Operators]\index{Noether's theorem!for Dirichlet operators} Suppose $H$ is a Dirichlet operator and $O$ is an observable (a real diagonal matrix).  Then
$$ [O, H] = 0 \quad \iff \quad \frac{d}{d t} \langle \widehat{O} , \psi(t) \rangle 
= 0 $$ 
for all $t \ge 0$ and all $\psi(t)$ obeying the master equation.   
\end{theorem} 

\begin{proof}
Reminiscent of the general proof appearing in Chapter \ref{sec:10}, the $\Rightarrow$ direction is straightforward, and it follows from Theorem \ref{theorem:stochastic-noether-2}. The point is to show the $\Leftarrow$ direction. Since $H$ is self-adjoint, we're allowed to use a spectral decomposition:
$$ H = \sum_k E_k p_k$$
where we choose an orthonormal basis of eigenvectors for $H$, say $\phi_k$,
with 
$$  H \phi_k = E_k \phi_k $$
and let $p_k$ be the projection operator onto the eigenvector $\phi_k$:
$$  p_k \psi = \langle \phi_k , \psi \rangle \, \phi_k $$
For any stochastic state $\psi_0$ let 
$$\psi(t) = e^{tH} \psi_0 $$
be the result of evolving it according to the master equation.  Then
$$
\begin{array}{ll} 
\displaystyle{ \frac{d}{d t} \langle \hat{O} , \psi(t) \rangle } = 0  
& \forall  \psi_0, \forall t \ge 0 \\  \\
 \iff \quad \langle \hat{O} , H e^{t H} \psi_0 \rangle = 0  
& \forall  \psi_0, \forall t \ge 0 \\  \\
\iff \quad \langle H e^{t H} \hat{O} , \psi_0 \rangle = 0  
& \forall  \psi_0, \forall t \ge 0 \\  \\
 \iff \quad  H e^{t H} \hat{O} = 0 
& \forall t \ge 0 \\  \\
 \iff \quad \displaystyle{ \sum_k  E_k e^{t E_k} \langle \phi_k, \hat{O} \rangle \phi_k = 0} 
& \forall t \ge 0   \\  \\
\iff \quad  E_k e^{t E_k} \langle \phi_k, \hat{O} \rangle = 0 
& \forall  k, \forall t \ge 0 \\ \\
\iff \quad \hat{O} \in \mathrm{Span}\{\phi_k \; : \; E_k = 0\} = \ker  H 
\end{array} 
$$
where the next to last equivalence is due to the vectors $\phi_k$ being linearly 
independent. 

For any infinitesimal stochastic operator $H$ the corresponding weighted graph consists of $m$ connected components for some number $m$.  Thus, we can reorder (permute) the vertices of this graph such that $H$ becomes block-diagonal with $m$ blocks.  We saw in Theorem \ref{theorem_equilibria} that the kernel of $H$ is spanned by $m$ eigenvectors $\phi_k$, one for each block.  Since $H$ is also symmetric, the components of each such $\phi_k$ can be chosen to be ones within the block and zeros outside it. Consequently
$$ \hat{O} \in \ker  H $$
implies that we can choose the basis of eigenvectors of $O$ to be the vectors 
$\phi_k$, which implies
$$ [O, H] = 0 \qedhere $$
\end{proof}

In summary, by restricting ourselves to the intersection of quantum and stochastic generators, we have found a version of Noether’s theorem for stochastic
mechanics that looks formally just like the quantum version! 

We conclude with a few problems about stochastic processes.

\subsection{Problems} 

In this section, we've considered a stochastic variant of Noether’s theorem for Dirchlet operators.  Dirchlet operators turn out to be a subclass of a wider class of operators, which generate so-called `doubly stochastic processes'.  We introduce these in the 
next four problems.  

We then leave the class of doubly stochastic processes and consider general stochastic process and the equations of motion for possibly time-dependent observables.  From this we arrive at a stochastic variant of {\bf Ehrenfest's theorem} from quantum mechanics.  Readers are asked to prove this stochastic variant by completing two final problems.   

First, recall from Section \ref{sec:22_weighted_graph} how any weighted graph
gives an infinitesimal stochastic operator called its Hamiltonian, $H$.  In Problem \ref{prob:39} you saw that conversely, any infinitesimal stochastic operator comes from a weighted graph in this way. \index{weighted graph} \index{graph theory!weighted graph}

\begin{problem}\label{prob:42}
Fix a weighted graph and let $H$ be its Hamiltonian.\index{Hamiltonian!of weighted graph} Define the {\bf indegree} of a vertex to be the sum of the rate constants of
edges having that vertex as their target.  Define the {\bf outdegree} to be the sum of the rate constants of edges having that vertex as their source.  \index{graph theory!vertex!indegree and outdegree of} \index{indegree} \index{outdegree}
 
Show that if $H$ is a Dirichlet operator, the indegree and outdegree of any vertex of our graph must be equal.  Weighted graphs with this property are called {\bf balanced}.\index{graph theory!balanced graph}\index{balanced graph}    
\end{problem} 

\begin{problem}\label{prob:43}
\index{stochastic operator!doubly stochastic}\index{doubly stochastic operator}\index{operator!doubly stochastic} An $n \times n$ matrix $U$ is
a {\bf doubly stochastic operator} if it is stochastic and its transpose is also 
stochastic.  In other words, each column and each row of $U$ sums to 1:
$$ \sum_i U_{i j} = \sum_j U_{i j} = 1 $$
and all its matrix entries are nonnegative. 

Prove that $U$ is doubly stochastic if and only if is stochastic and $U \psi = \psi$
where $\psi$ is the `flat' probability distribution, with $ \psi_i = \frac{1}{n}$
for all $i$.
\end{problem}

\begin{problem}\label{prob:44} 
\index{semigroup!doubly stochastic} \index{operator!infinitesimal doubly 
stochastic}  \index{doubly stochastic semigroup} \index{infinitesimal doubly
stochastic operator}
Suppose that $U(t) = e^{t H}$ is a {\bf doubly stochastic semigroup},
meaning that $U(t)$ is doubly stochastic for all $t \ge 0$.
Prove that the Hamiltonian $H$ is {\bf infinitesimal doubly stochastic}, 
meaning that
$$ \sum_i H_{i j} = \sum_j H_{i j} = 0 $$
and all the off-diagonal entries of $H$ are nonnegative.  Also prove the converse:
if $H$ is infinitesimal doubly stochastic, $e^{t H}$ is a doubly stochastic 
semigroup.
\end{problem} 

\begin{problem}\label{prob:45} 
Prove that any infinitesimal doubly stochastic operator $H$ is the Hamiltonian of a balanced graph.\index{stochastic operator!doubly stochastic}\index{doubly stochastic operator}\index{semigroup!doubly stochastic}\index{graph theory!balanced graph}\index{balanced graph} \index{infinitesimal doubly stochastic operator}
\end{problem}  

\begin{problem}\label{prob:46}
Show that every Dirichlet operator is doubly infinitesimal stochastic.
\end{problem}

\begin{problem}\label{prob:47}
Find a doubly infinitesimal stochastic operator that is not a Dirichlet operator.
\end{problem}

\begin{problem}\label{prob:48}
Find necessary and sufficient conditions for the Hamiltonian of a weighted
graph to be a Dirichlet operator.   \index{Hamiltonian!of weighted graph}
\end{problem}

\index{Ehrenfest theorem|(}
The next two problems concern the \href{http://en.wikipedia.org/wiki/Ehrenfest_theorem}{Ehrenfest theorem} for time-dependent observables.   We ask the reader to first 
prove the usual quantum version, and then a stochastic analogue:

\begin{problem}\label{prob:46}
Let $O(t)$ be a possibly time-dependent quantum observable, and
let $\psi(t)$ be a quantum state evolving according to Schr\"odinger's equation
with some self-adjoint Hamiltonian $H$.  Prove the {\bf Ehrenfest theorem}:  \index{quantum mechanics!Ehrenfest theorem}
$$ \frac{d}{d t}\langle O(t) \rangle_{\psi(t)} = -i\left\langle 
[O(t), H] \right\rangle_{\psi(t)} + \left\langle \frac{\partial O(t)}{\partial t}
 \right\rangle_{\psi(t)} $$
Here we are using the abbreviation 
$$\langle A \rangle_{\psi} = \langle \psi, A \psi \rangle $$
for the expected value of any operator $A$ in the quantum state $\psi$.
\end{problem}

\begin{problem}\label{prob:47}
 Let $O(t)$ be a possibly time-dependent stochastic observable, and let $\psi(t)$ be a stochastuc state evolving according to the master equation. Show that
$$ \frac{d}{d t}\langle O(t)\rangle_{\psi(t)} = \left\langle [O(t), H] \right\rangle_{\psi(t)} + \left\langle \frac{\partial O(t)}{\partial t} \right\rangle_{\psi(t)} $$
Here we are using the abbreviation 
$$\langle A \rangle_{\psi} = \langle \hat{A}, \psi \rangle $$
for the expected value of any operator $A$ in the stochastic state $\psi$.
\end{problem} 
\index{Ehrenfest theorem|)}

\index{Dirichlet operator!and Noether's theorem|)}
\index{Noether's theorem!for Dirichlet operators|)} 

\newpage  

%%%%%%%%% SECTION 25 %%%%%%%%%%% 

\section{Computation and Petri nets}
\label{sec:25}
\index{computer science|(}

As we've seen, a stochastic Petri net can be used to describe a bunch of chemical reactions with certain reaction rates.  We could try to use these reactions to build a `chemical computer'.\index{chemical computer}   But \textit{how powerful can such a computer be?}

People know quite a lot about the answer to this question.  But before people got interested in \textit{stochastic} Petri nets, computer scientists spent quite some time studying plain old Petri nets, which don't include the information about reaction rates.  They used these as simple models of computation.  And since computer scientists like to know which questions are \href{http://en.wikipedia.org/wiki/Decidability_(logic)}{decidable} by means of an algorithm and which aren't, they proved some interesting theorems about decidability for Petri nets.\index{decidability|(}\index{computer science!decidability|(} 

Let's talk about `reachability':\index{reachability} the question of which collections of molecules can turn into which other collections, given a fixed set of chemical reactions.   For example, suppose you have these chemical reactions:
\[ \text{C} + \text{O}_2 \longrightarrow \text{CO}_2 \]
\[ \text{CO}_2 + \text{NaOH} \longrightarrow \text{NaHCO}_3 \]
\[ \text{NaHCO}_3 + \text{HCl} \longrightarrow \text{H}_2 \text{O} + \text{NaCl} + \text{CO}_2  \]
Can you use these to turn
\[ \text{C} + \text{O}_2 + \text{NaOH} + \text{HCl}  \]
into 
\[ \text{CO}_2 + \text{H}_2\text{O} + \text{NaCl}  \text{?} \]
It's not hard to settle this particular question---we'll do it soon.  But settling \textit{all possible} questions of this type turns out to be very hard. 

\subsection{The reachability problem}
\label{sec:25_1}
\index{reachability|(}  \index{Petri net!reachability|(}  
\index{reachability!problem|(} 
\index{computer science!reachability problem|(}

Remember: 
\begin{definition} A {\bf Petri net}\index{Petri net!definition of} consists of a set $S$ of {\bf species} and a set $T$ of {\bf transitions}, together with a function 
\index{species} \index{transition}
\[ i \colon S \times T \to \mathbb{N} \]
saying how many copies of each state shows up as {\bf input} for each transition, and a function   \index{input}
\[ o\ \colon S \times T \to \mathbb{N} \]
saying how many times it shows up as {\bf output}.   \index{output}
\end{definition}
\noindent
Let's assume both $S$ and $T$ are finite.

We hope you remember that we draw the species as yellow circles and the transitions as aqua boxes, so the chemical reactions mentioned before give this Petri net:

\begin{center}
 \includegraphics[width=80mm]{chemistryNetBasicA.png}
\end{center}

\noindent A `complex' amounts to a way of putting dots in the yellow circles.   In chemistry this says how many molecules we have of each kind.  Here's an example:

\begin{center}
 \includegraphics[width=80mm]{chemistryNetDot1A.png}
\end{center}

This complex happens to have just zero or one dot in each circle, but that's not required: we could have any number of dots in each circle.  So, mathematically, a complex is a finite linear combination of species, with natural numbers as coefficients: an element of $\mathbb{N}^S.$  \index{complex} In this particular example it's
\[ \text{C} + \text{O}_2 + \text{NaOH} + \text{HCl} \]

Given two complexes, we say the second is {\bf reachable}\index{reachability!definition of} from the first if, loosely speaking, we can get from the first from the second by a finite sequence of transitions.  In chemistry, this means that we are using the reactions described by our Petri net to turn one collection of molecules into another, with no byproducts or leftover raw ingredients.  

For example, earlier we asked if you can get from the complex we just mentioned to this one:
\[ \text{CO}_2 + \text{H}_2\text{O} + \text{NaCl}  \]
And the answer is yes!  We can do it with this sequence of transitions:

\begin{center}
 \includegraphics[width=80mm]{chemistryNetDot1A.png}
\end{center}

\begin{center}
 \includegraphics[width=80mm]{chemistryNetDot2A.png}
\end{center}

\begin{center}
 \includegraphics[width=80mm]{chemistryNetDot3.png}
\end{center}

\begin{center}
 \includegraphics[width=80mm]{chemistryNetDot4.png}
\end{center}

Now, the {\bf reachability problem}\index{reachability!problem!definition of} asks: given a Petri net and two complexes, is one reachable from the other?  

If the answer is `yes', of course you can show that by an exhaustive search of all possibilities.  But if the answer is `no', how can you be sure?  It is not obvious, in general.  Back in the 1970's, computer scientists felt this problem should be decidable by some algorithm... but they had a lot of trouble finding such an algorithm.  

In 1976, Roger J.~Lipton showed that if such an algorithm existed, it would need to take at least an exponential amount of memory space and consequently an exponential amount of time to run:
\begin{itemize}
\item [\cite{Lip76}] Roger J.~Lipton, \href{http://www.cs.yale.edu/publications/techreports/tr63.pdf}{The reachability problem requires exponential space}, Technical Report 62, Yale University, 1976.
\end{itemize}
This means that most computer scientists would consider any algorithm to solve the reachability problem `infeasible', since they prefer \href{http://en.wikipedia.org/wiki/Time_complexity\#Polynomial_time}{polynomial time} algorithms.   

On the bright side, it means that Petri nets might be fairly powerful when viewed as computers themselves! After all, for a \href{http://en.wikipedia.org/wiki/Universal_Turing_machine}{universal Turing machine}, the analogue of the reachability problem is undecidable.  So if the reachability problem for Petri nets were decidable, they couldn't serve as universal computers.  But if it were decidable but hard, Petri nets might be fairly powerful---though still not universal---computers.

At an important computer science conference in 1977, two researchers presented a proof that the reachability problem was decidable:
\begin{itemize}
\item [\cite{ST77}] S.\ Sacerdote and R.\ Tenney, The decidability of the reachability problem for vector addition systems, {\sl Conference Record of the Ninth Annual ACM Symposium on Theory of Computing, 2-4 May 1977, Boulder, Colorado, USA}, ACM, 1977, pp.\ 61--76.
\end{itemize}
However, their proof turned out to be flawed!  You can read the story here:
\begin{itemize}
\item [\cite{Pet81}] James L.\ Peterson, \textit{Petri Net Theory and the Modeling of Systems}, Prentice--Hall, New Jersey, 1981.
\end{itemize}
This is a very nice introduction to early work on Petri nets and decidability.   Peterson had an interesting idea, too: he thought there might be a connection between
Petri nets and something called `Presburger arithmetic'.  He gave some evidence, and suggested using this to settle the decidability of the reachability problem.  

Let's explain why this is so interesting. \href{http://en.wikipedia.org/wiki/Presburger_arithmetic}{Presburger arithmetic} is a simple set of axioms for the arithmetic of natural numbers, much weaker than \href{http://en.wikipedia.org/wiki/Peano_arithmetic}{Peano arithmetic}.  
Peano arithmetic is a set of axioms involving $0, 1, \times$ and $+$ and the
principle of mathematical induction.  Presburger arithmetic is similar, but it doesn't mention multiplication.  The only axioms are these:
\begin{itemize}
\item $x + 0 = x$.
\item $x + (y+1) = (x + y) + 1$.
\item There is no $x$ such that $0 = x + 1$.
\item If $x+1 = y+1$ then $x = y$.
\item The principle of mathematical induction.  This is an `axiom schema': for any formula $P$ written in the language of Presburger arithmetic, we have an axiom
saying that $P(0)$ and $\forall x (P(x) \Rightarrow P(x+1))$ imply $\forall x P(x)$.
\end{itemize}
Unlike Peano arithmetic, Presburger arithmetic is {\bf complete}: \index{computer science!completeness} for every statement, either it or its negation is provable.  And unlike Peano arithmetic, Presburger arithmetic is {\bf decidable}: \index{computer science!decidability} you can write an algorithm that decides for any statement in Presburger arithmetic whether it or
its negation is provable.

However, any such algorithm must be very slow!   In 1974, Fischer and Rabin showed that any decision algorithm for Presburger arithmetic has a worst-case runtime of at least 
\[  2^{2^{p n}} \]
for some constant $p$, where $n$ is the length of the statement.   So we say the complexity is at least {\bf doubly exponential}.  That's much worse than exponential!    On the other hand, an algorithm with a worst-case runtime no more than \textit{triply} exponential was found by Oppen in 1978:
\begin{itemize}
\item [\cite{FR74}] M.\ J.\ Fischer and Michael O.\ Rabin, Super-exponential complexity of Presburger arithmetic, {\sl Proceedings of the SIAM-AMS Symposium in Applied Mathematics} {\bf 7} (1974), 27--41.
\item [\cite{Opp78}] Derek C.\ Oppen, A $2^{2^{2^{pn}}}$ upper bound on the complexity of Presburger arithmetic, \textsl{J.\ Comput.\ Syst.\ Sci.\ }{\bf 16} (1978), 323--332.
\end{itemize}

How is all this connected to the reachability problem for Petri nets?  First, provability is a lot like reachability, since in a proof you're trying to reach the conclusion starting from the assumptions using certain rules.  Second, like Presburger arithmetic, Petri nets are all about addition, since they consists of transitions going between linear combinations like this:
\[ 6 \text{CO}_2 + 6 \text{H}_2 \longrightarrow \text{C}_6\text{H}_{12}\text{O}_6 + 6 \textrm{O}_2 \]
That's why the old literature calls Petri nets {\bf vector addition systems}\index{vector addition system}.\index{vector addition system}\index{vector addition system|seealso{Petri net}}\index{Petri net!vector addition system}  And third, the difficulty of deciding provability in Presburger arithmetic smells a bit like the difficulty of deciding reachability in Petri nets.

So, it is interesting to know what happened after Peterson wrote his book.  For starters, in 1981, the very year Peterson's book came out, Ernst Mayr showed that the reachability problem for Petri nets \textit{is} decidable:
\begin{itemize}
\item [\cite{May81}] Ernst Mayr, Persistence of vector replacement systems is decidable, {\sl Acta Informatica} {\bf 15} (1981), 309--318.
\end{itemize}
As you can see from the title, Mayr actually proved some other property was decidable.  However, it follows that reachability is decidable, and Mayr pointed this out in his paper.  In fact the decidability of reachability for Petri nets is equivalent to lots of other interesting questions.  You can see a bunch here:
\begin{itemize}
\item [\cite{EN94}] Javier Esparza and Mogens Nielsen, \href{http://citeseerx.ist.psu.edu/viewdoc/download;jsessionid=5EBA38B02179E955EB84DCA21345710C?doi=10.1.1.2.3965\&rep=rep1\&type=pdf}{Decidability issues for Petri nets---a survey}, \textit{Bulletin of the European Association for Theoretical Computer Science} {\bf 52} (1994), 245--262.
\end{itemize}
\index{decidability|)}\index{computer science!decidability|)} 

\index{primitive recursive function|(} \index{computer science!primitive recursive
function|(}
Though Mayr found an algorithm to decide the reachability problem, this algorithm is complicated.  Worse, it seems to take a hugely long time to run!  The best known upper bound on its runtime is a function that's not even `primitive recursive'.  Here is the basic idea.  We can define multiplication by iterating addition:
\[ n \times m = n + n + n + \cdots + n \]
where add $n$ to itself $m$ times.  Then we can define exponentiation by iterating multiplication:
\[ n \uparrow m = n \times n \times n \times \cdots \times n \]
where we multiply $n$ by itself $m$ times.  Here we're using Knuth's \href{http://en.wikipedia.org/wiki/Knuth's_up-arrow_notation}{up-arrow notation}.  Then we can define \href{http://en.wikipedia.org/wiki/Tetration}{tetration} by iterating exponentiation:
\[  n \uparrow^2 m = n \uparrow (n \uparrow (n \uparrow \cdots \uparrow n))) \]
Then we can define an operation $\uparrow^3$ by iterating tetration, and so on.  All these functions are said to be `\href{http://en.wikipedia.org/wiki/Primitive_recursive_function}{primitive recursive}': very roughly, this means that each one can be computed from a previously defined one using a program with an extra loop.  But the $n$th {\bf Ackermann number} is not a primitive function of $n$; it's defined to be
\[ n \uparrow^n n\]
The difference is the $\uparrow^n$: now the number of nested loops in our program
must itself grow with $n$.    \index{Ackermann numbers}

The Ackermann numbers grow at an \textit{insanely} rapid rate: they eventually surpass any primitive recursive function!  Here are the first three:
\[ \begin{array}{ccl} 1 \uparrow 1 &=& 1 \\
 2 \uparrow^2 2 &=& 4 \\
 3 \uparrow^3 3 &=& 3^{3^{3^{.^{.^{.}}}}} 
\end{array}\]
where we have a stack of $3^{3^3}$ threes---or in other words, $3^{7625597484987}$ threes!  When we get to $4 \uparrow^4 4,$ our minds boggle.  We wish it didn't, but it does.  
\index{primitive recursive function|)} \index{computer science!primitive recursive
function|)}

In 1998 someone claimed to have a faster algorithm for deciding the reachability problem:
\begin{itemize}
\item [\cite{Bou98}] Zakaria Bouziane, \href{http://hal.inria.fr/inria-00073286/PDF/RR-3404.pdf}{A primitive recursive algorithm for the general Petri net reachability problem}, in \textit{39th Annual Symposium on Foundations of Computer Science}, IEEE, 1998, pp.\ 130--136.
\end{itemize}
\noindent
Unfortunately, it seems that Bouzaine made a mistake:
\begin{itemize}
\item [\cite{Jan08}] Petr Jan\v car, Bouziane’s transformation of the Petri net reachability problem and incorrectness of the related algorithm, \textsl{Information and Computation}, {\bf 206} (2008), 1259--1263.
\end{itemize}
\noindent
So at present, if we list a bunch of chemicals and reactions involving these chemicals, you can decide when some combination of these chemicals can turn into another combination.  But it might take a very long time to decide this.  

What about the connection to Presburger arithmetic?  That is explained here:

\begin{itemize}
\item [\cite{Ler08}] J\'er\^ome Leroux, \href{http://hal.inria.fr/docs/00/31/93/93/PDF/main.pdf}{The general vector addition system reachability problem by Presburger inductive separators}, 2008.
\end{itemize}

If we have a Petri net and one complex \emph{is} reachable from another, we can eventually discover this by a brute-force search.  So, the hard 
part of Mayr's algorithm concerns the case when one complex is \emph{not}
reachable from another.  In this case, Mayr's algorithm will eventually prove
it is impossible---thus deciding the reachability problem. 

What Leroux did was
find another kind of proof, using Presburger arithmetic.  Suppose we have two 
complexes $\kappa, \kappa' \in \mathbb{N}^S$ such that $\kappa'$ is not reachable
from $\kappa$.  Then Leroux's algorithm constructs sets $X, X' \subseteq \mathbb{N}^S$, definable using Presburger arithmetic, such that $\kappa \in X$, $\kappa \in X'$, and it is \emph{provable using Presburger arithmetic}
that no complex in $X'$ is reachable from one in $X$.   

In 2015, Leroux and Schmitz found the first explicit upper bound on the runtime of
an algorithm to decide Petri net reachability.  This bound is not primitive recursive, 
and it grows much faster than the Ackermann numbers; it's a `cubic Ackermann function'.

\begin{itemize}
\item [\cite{LS15}]
J\'er\^ome Leroux and Sylvain Schmitz, Demystifying reachability in vector addition systems,
in \textsl{LICS '15: 30th Annual ACM/IEEE Symposium on Logic in Computer Science}, IEEE, 2015, pp.\ 56--67.  Also available as \href{https://arxiv.org/abs/1503.00745}{arXiv:1503.00745}.
\end{itemize}

\noindent
Finally, in 2018 a team of authors including Leroux found a \emph{lower} bound on the runtime for \emph{any} algorithm to decide Petri net reachability.   This lower bound is also not primitive recursive:

\begin{itemize}
\item[\cite{CLLLM18}]
Wojciech Czerwinski, Slawomir Lasota, Ranko Lazic, J\'er\^ome Leroux and Filip Mazowiecki,
The reachability problem for Petri nets is not elementary.  Available as 
\href{https://arxiv.org/abs/1809.07115}{arXiv:1809.07115}.
\end{itemize}

\noindent
In short, problems about what chemical reactions can do lead naturally to
quite profound issues in computation, which are even connected to the foundations of arithmetic!

\index{Presburger arithmetic|)}

\subsection{Symmetric monoidal categories}
\label{sec:25_2}
\index{category theory|(}

Next, let us explain yet another reason why the reachability problem is so fascinating.  This little section is only for people who know about 
\href{http://en.wikipedia.org/wiki/Symmetric\_monoidal\_category}{symmetric monoidal categories}. 

As we mentioned long ago, a Petri net is actually nothing but a presentation of a symmetric monoidal category that is freely generated by some set $S$ of objects and some set $T$ of morphisms going between tensor products of those objects.  The objects in $S$ are our \emph{species}.  The tensor products of these objects are our \emph{complexes}.  The morphisms in $T$ are our \emph{transitions} or \emph{reactions}.  The details are explained here:

\begin{itemize}
\item [\cite{MM90}] Jos\'e Meseguer and Ugo Montanari, Petri nets are monoids, \textsl{Information and Computation} \textbf{88} (1990), 105--155.
\item[\cite{BMMS01}] Roberto Bruni, Jos\'e Meseguer, Ugo Montanari and Vladimiro
Sassone,
\href{http://eprints.soton.ac.uk/264742/1/prenetsIandCOff.pdf}{Functorial models
for Petri nets}, \textsl{Information and Computation} \textbf{170} (2001), 207--236.
\end{itemize}

In chemistry we write the tensor product as $+$, but category theorists usually write it as $\otimes$.  Then the reachability problem consists of questions like this:\index{tensor product}\index{category theory!tensor product}  

\begin{problem}\label{prob:47}
Suppose we have a symmetric monoidal category freely generated by objects $A, B, C, D, E, F, G, H$ and morphisms
$$\begin{array}{rrcl}
f \colon &A \otimes B &\to& C \\ 
g \colon &C \otimes D &\to& E \\
h \colon &E \otimes F &\to& C \otimes G \otimes H \\
\end{array}$$
Is there a morphism from $A \otimes B \otimes D \otimes F$ to $C \otimes G \otimes H$?  
\end{problem}

Problems of this sort are reminiscent of the \href{http://en.wikipedia.org/wiki/Word_problem_for_groups}{word problem for groups}\index{group theory!word problem} and other problems where we are given a presentation of an algebraic structure and have to decide if two elements are equal... but now, instead of asking whether two elements are equal we are asking if there is a morphism from one object to another.  So, it is fascinating that this problem is decidable---unlike the word problem for groups---but still very hard to decide, in general.

Just in case you want to see a more formal statement, let's finish off by giving you that:
\vskip 1em

\noindent{\bf Reachability Problem.}  Given a symmetric monoidal category $C$ freely generated by a finite set of objects and a finite set of morphisms between tensor products of these objects, and given two objects $x,y \in C,$ is there a morphism $f \colon x \to y$?\index{reachability!problem}   

\begin{theorem}[{\bf Czerwinski, Lasota, Lazic, Leroux, Mazowiecki, Schmitz}]  There is an algorithm that decides the reachability problem.  However, for any such algorithm, the worst-case runtime exceeds $f(n)$, where $f$ is some function that grows faster than any primitive recursive function and $n$ is the size of the problem: the sum of the number of generating objects, the number of factors in the sources and targets of all the generating morphisms, and the number of factors in the objects $x,y \in C$ for which the reachability problem is posed.  There is an algorithm whose runtime is bounded by a cubic Ackermann function of $n$.
\end{theorem}

\index{reachability|)} \index{Petri net!reachability|)}
\index{category theory|)}
\index{reachability!problem|)} 
\index{computer science!reachability problem|)}

\subsection{Computation with Petri nets}
\label{sec:25_3}

We've been talking about using computers to figure out what
a Petri net can do.  Now let's turn the question around and talk about 
using a Petri net as a computer.   We could do this for ordinary Petri nets, but
in keeping with the main theme of this book let us consider stochastic Petri nets.  
Remember, these are Petri nets equipped with a rate constant for each reaction.
This allows us to formulate both the \emph{rate equation}, describing how the 
expected number of things of each species changes with time, and the \emph{master 
equation}, describing how the probability that we have a given number of things of each species changes with time.   

\index{Soloveichik, David|(}
We could try to use either the rate equation or the master equation for computation, and people have studied both.   The rate equation lets us build analog
computers:
\begin{itemize}
\item[\cite{CDS14}]
Ho-Lin Chen, David Doty and David Soloveichik, \href{http://www.dna.caltech.edu/~ddoty/papers/riccrn.pdf}{Rate-independent computation in continuous chemical reaction networks}, in \textsl{ITCS 2014: Proceedings of the 5th Innovations in Theoretical Computer Science Conference}, ACM, New York, pp.\ 313--326. 
\end{itemize}
\noindent
In this particular formalism, the authors want the concentration of a particular
species to converge in the limit as $t \to \infty$ to a desired function of the 
`input' concentrations, \emph{independent of the choice of rate constants}.
They show that this is possible precisely for continuous piecewise-linear functions. 

The master equation lets us build digital computers, as discussed here:
\begin{itemize}
\item[\cite{CSWB09}] Matt Cook, David Soloveichik, Erik Winfree and Jehoshua Bruck, \href{http://www.dna.caltech.edu/Papers/programmability_of_CRNs_preprint2008.pdf}{Programmability of chemical reaction networks}, in \textsl{Algorithmic Bioprocesses}, eds.\ Condon, Harel, Kok, Salomaa and Winfree, Springer, Berlin, 2009, pp.\ 543--584. 
\item[\cite{Sol14}] David Soloveichik, \href{http://crn.thachuk.com/images/d/d5/Soloveichik-Banff_Tutorial.pdf}{The computational power of chemical reaction
networks}, lecture at Programming with Chemical Reaction Networks: Mathematical Foundations, Banff, June 2014.
\item[\cite{SCWB08}] David Soloveichik, Matt Cook, Erik Winfree and Jehoshua Bruck, \href{http://dna.caltech.edu/Papers/sCRN_computation_TR2007.pdf}{Computation with finite stochastic chemical reaction networks}, 
\textsl{Natural Computing} \textbf{7} (2008), 615--633. 
\item[\cite{ZC08}] Gianluigi Zavattaro and Luca Cardelli, \href{http://lucacardelli.name/Papers/Termination\%20Problems\%20in\%20Chemical\%20Kinetics.pdf}
{Termination problems in chemical kinetics}, in \textsl{CONCUR 2008--Concurrency Theory}, eds.\ Zavattaro and Cardelli, \textsl{Lecture Notes in Computer Science} \textbf{5201}, Springer, Berlin, 2008, pp.\ 477--491.
\end{itemize}    \index{Soloveichik, David|)}
\noindent
The input to such a computer
will always be a vector of natural numbers, encoded as the numbers of 
items of various species.   There are some choices as to how the computation
proceeds:
\begin{itemize}
\item {\bf decision problems vs more general computational problems}:
in a decision problem the output is either `yes', indicated by the presence of a species we'll call $Y$, or `no', indicated by the presence of a species 
called $N$.   In a more general computational problem the answer could be a natural
number, or a vector of natural numbers.  
\index{decision problem|(} \index{computer science!decision problem|(}
\item {\bf uniform vs non-uniform}: in the uniform case a single stochastic Petri net is supposed to handle all inputs. In the non-uniform case we can add extra transitions for larger inputs.
\item {\bf deterministic vs probabilistic}: in the deterministic case we are guaranteed to
get the correct output.  In the probabilistic case, it is merely likely.
\item {\bf halting vs stabilizing}: does the Petri net `tell us' when it has
 finished, or not?  In the halting case it irreversibly produces some items 
of a given species that signal that the computation is done.  In the stabilizing case 
it eventually stabilizes to the right answer, but we may not know how long to wait.
\end{itemize}

These choices dramatically affect the computational power of our computer.  For example, consider decision problems and uniform computation.  Then there are four cases:
\begin{itemize}
\item {\bf deterministic and halting}: this has finite computational power.  
\item {\bf deterministic and stabilizing}: this can decide any `semilinear predicate',
as explained in the next section.
\item {\bf probabilistic and halting}: this can decide any `recursive predicate'.  In other words, this class of reactions can solve any decision problem that any Turing machine can solve.
\item {\bf probabilistic and stabilizing}: this can decide so-called `$\Delta_2^0$ predicates', which are \href{http://en.wikipedia.org/wiki/Borel_hierarchy#Lightface_hierarchy}{more general than recursive ones}.  This may sound odd, but if we use Turing machines but don't require them to signal when they've halted, the resulting infinitely long computations can decide predicates that are not recursive.
\end{itemize}
\noindent
Instead of discussing all four cases in detail, let's look at one: deterministic 
stabilizing computations.  These give a nice variant on the earlier theme
of `reachability'.  In the reachability problem we asked if starting from some initial complex we \emph{can} reach some final complex.  In deterministic stabilizing
computations we ask if starting from some initial complex we \emph{must eventually}
reach a final complex \emph{in some set}, and then \emph{stay in that set}.   This set encodes a YES answer to our decision problem.   

In deterministic stabilizing computations, the rate constants of the transitions don't
matter. The reason is that we're demanding that an outcome occur with certainty,
and we don't care how long it takes.  So, even though we've been talking about
stochastic Petri nets, the concept of deterministic stabilizing computation makes sense
for mere Petri nets.  

\subsubsection*{Deterministic stabilizing computations}
\index{computer science!deterministic stabilizing computation|(}
\index{deterministic stabilizing computation|(}

Let us be a bit more precise.
Suppose we have a subset $X \subseteq \mathbb{N}^d,$ and we want to answer this question: is the vector $\kappa \in \mathbb{N}^d$ in the set $X$?  This is a decision problem.  What does it mean to solve this problem with a Petri net using a deterministic stabilizing computation?

To do this, we represent our vector $\kappa$ as a bunch of molecules: $\kappa_1$ of the first kind, $\kappa_2$ of the second kind, and so on.    We may also include a fixed collection of additional molecules to help the reactions run.  Mathematically, $\kappa$ is a complex called the {\bf input}, and the extra molecules are some fixed complex $\lambda \in \mathbb{N}^d$.

Then we choose a Petri net and let it run, starting with the complex $\kappa + \lambda$.  The answer to our question will be encoded in some molecules called $Y$ and $N$.  If $\kappa$ is in $X$, we want our Petri net to eventually produce at least one $Y$ molecule.  If it's not, we want it to produce at least one $N$.

To make this more precise, we need to define what counts as an {\bf output}.  If we've got a complex that contains $Y$ but not $N$, then the output is YES.
If it contains $N$ but not $Y$, then the output is NO. If it contains both, then
the output is undefined.

{\bf Output-stable states} are complexes with YES or NO output such that all complexes reachable from them via transitions in our Petri net give the same output.  We say an output-stable-state is {\bf correct} if this output is the correct answer to the question: is $\kappa$ in $X$?  \index{output-stable state}

Our Petri net gives a {\bf deterministic stabilizing computation} 
\index{deterministic stabilizing computation!definition of}
if for any input, and choosing any state reachable from that input, we can do further chemical reactions to reach a correct output-stable state.  In other words: starting from our input, and using \textit{any transitions we want}, we will eventually reach an output-stable state that gives the right answer to the question ``is $\kappa$ in $X$?"

This sounds a bit complicated, but it's really not.   Let's look at two examples.

\subsubsection*{Checking an inequality}

Suppose you want to check two numbers and see if one is greater than or equal to another.  Here 
$$ X = \{(\kappa_1,\kappa_2) \; : \; \kappa_1 \le \kappa_2 \}$$
How can you decide if a pair of numbers $ (\kappa_1,\kappa_2)$ is in this set? 

You start with $\kappa_1$ molecules of type $A,$ $\kappa_2$ molecules of type $ B,$ and one molecule of type $ Y$.  Then you use a Petri net with these transitions:
$$ A + N \to Y$$
$$ B + Y \to N $$
If you let these run, the $Y$ switches to a $N$ each time the transition destroys an $A$.  But the $N$ switches back to a $Y$ each time the transition destroys a $B.$
When no more transitions are possible, you are left with either one $Y$ or one $N$, which is the correct answer to your question!

\subsubsection*{Checking equality}

Suppose you want to check two numbers and see if they are equal.  Here 
$$ S = \{(\kappa_1,\kappa_2) \; : \; \kappa_1 = \kappa_2 \} $$
How can you decide if a pair of numbers $ (\kappa_1,\kappa_2)$ is in here?  This is a bit harder!  As before, you start with $ \kappa_1$ molecules of type $ A,$ $ \kappa_2$ molecules of type $ B,$ and one molecule of type $ Y.$  Then you use a stochastic
Petri net with these reactions:
$$ A + B \to Y$$
$$ Y + N \to Y$$
$$ A + Y \to A + N$$
$$ B + Y \to B + N$$
The first transition lets an $A$ and a $B$ cancel out, producing a $Y$. If you only run this transition, you'll eventually have some $A\mathrm{s}$ left and no $B\mathrm{s}$, or have some $ B\mathrm{s}$ left and no $A\mathrm{s}$, or you'll have nothing left but $Y\mathrm{s}$.  

If you have nothing but $Y\mathrm{s},$ no more transitions are possible and your numbers were equal.  The other reactions deal with the cases where you have some $A\mathrm{s}$ or $ B\mathrm{s}$ left over.   You can check that no matter what order we run the reactions, we'll eventually get the right answer!  In the end, you'll have either $ Y\mathrm{s}$ or $ N\mathrm{s},$ not both, and this will provide the
correct yes-or-no answer to the question of whether $ \kappa_1 = \kappa_2.$

\subsubsection*{The power of deterministic stabilizing computations}

Now you've seen some examples of deterministic stabilizing computations.  The big question is: what kind of questions can they answer?   More precisely, for what subsets $X \subseteq \mathbb{N}^d$ can we build a deterministic stabilizing computation that ends with output YES if the input $\kappa$ lies in $X$ and with output NO otherwise?  

The answer is: the `semilinear' subsets!
\begin{itemize}
\item[\cite{AAE06}] Dana Angluin, James Aspnes and David Eistenstat, \href{http://www.cs.yale.edu/homes/aspnes/papers/podc2006-proceedings.pdf}{Stably computable predicates are semilinear}, in \textsl{Twenty-Fifth ACM Symposium on Principles of Distributed Computing, July 2006}, ACM Press, New York, 2006, 
pp.\ 292--299. 
\end{itemize}
A subset $ S\subseteq \mathbb{N}^d$ is {\bf linear} if it's of the form
$$ \{u_0 + n_1 u_1 + \cdots + n_p u_p \; : \; n_i \in \mathbb{N}  \} $$
for some fixed vectors $u_i \in \mathbb{N}^d.$  A subset $S \subseteq \mathbb{N}^d$ {\bf semilinear} if it's a finite union of linear sets.   \index{linear subset} \index{semilinear subset}

How did Angluin, Aspnes and Eisenstat prove that deterministic stabilizing computations can decide membership in precisely the semilinear subsets?  The easy part is showing that membership in any semilinear subset can be decided by some Petri net.  The hard part is the converse.  This uses a nice fact which is worth knowing for its own sake:

\begin{theorem}[{\bf \href{http://en.wikipedia.org/wiki/Dickson\%27s_lemma}{Dickson's Lemma}}] Any subset of $\mathbb{N}^d$ has a finite set of minimal elements, where we define $ x \le y$ if $ x_i \le y_i$ for all $ i$. \index{Dickson's lemma}
\end{theorem}
\noindent
For example, the region above and to the right of the hyperbola here has five minimal elements:

\begin{center}
\includegraphics[width=40mm]{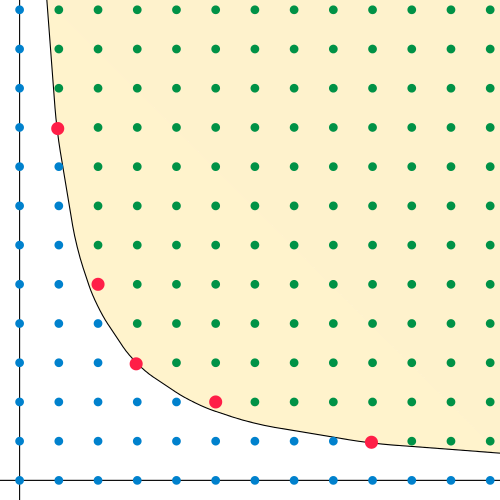}
\end{center}

\noindent It's very neat to see how Dickson's lemma helps!  You can see how it works in Angluin, Aspnes and Eistenstat's paper---but they call it `Higman's lemma'.
\index{computer science|)}
\index{computer science!deterministic stabilizing computation|)}
\index{deterministic stabilizing computation|)}
\index{decision problem|)} \index{computer science!decision problem|)} 

\subsection{Software}
\label{sec:25_5}

Finally, no overview of the connection between Petri nets and computation would
be complete without a word about free software for working with Petri nets,
chemical reaction networks and the like.  Jim Stuttard \index{Stuttard, Jim}
found quite a lot of such software for us, but we will only mention three tools.  
We shall not explain all the terminology 
in what follows: unfortunately, there is a lot about Petri nets that we have not 
discussed!

\subsubsection*{Cain}

\href{http://cain.sourceforge.net/}{Cain} performs stochastic and deterministic simulations of chemical reactions. It can spawn multiple simulation processes to utilize multi-core computers. It stores models, simulation parameters, and simulation results in XML, and it lets you import and export models in Systems Biology Markup Language.   It comes equipped with examples from this nice book:

\begin{itemize}
\item[\cite{Wil06}]
Darren~James Wilkinson,
\textsl{Stochastic Modelling for Systems Biology},
Taylor and Francis, New York, 2006.
\end{itemize}

\subsubsection*{CPN Tools}

\href{http://cpntools.org/}{CPN Tools} 
is probably the most popular tool for editing, simulating, and analyzing colored Petri nets---a generalization of Petri nets where the tokens can have different `colors'. \index{token} \index{Petri net!token} \index{Petri net!colored} \index{colored
Petri net}  Colored Petri nets are explained here:
\begin{itemize}
\item[\cite{AS11}] 
Wil van der Aalst and Christian Stahl, \textsl{Modeling Business Processes: 
A Petri Net-Oriented Approach}, MIT Press, Cambridge, 2011.
\item[\cite{JK09}] Kurt Jensen and Lars M.\ Kristensen, \textsl{Coloured Petri Nets: Modelling and Validation of Concurrent Systems}, Springer, Berlin, 2009.
\end{itemize}
\noindent
Ordinary Petri nets are a special case.  CPN Tools has been developed over the past 25 years.  It features incremental syntax checking and code generation, which take place while a net is being constructed. A fast simulator efficiently handles both untimed and timed nets.  Full and partial state spaces can be generated and analyzed, and a standard state space report contains information, such as boundedness properties and liveness properties.

\subsubsection*{GreatSPN}

\href{http://www.di.unito.it/~greatspn/index.html}{GreatSPN} stands for
`GRaphical Editor and Analyzer for Timed and Stochastic Petri Nets'. 
It is a software package for the modeling, validation, and performance evaluation of distributed systems using generalized stochastic Petri nets.   The tool provides a framework to experiment with timed Petri net based modeling techniques.  It is composed of many separate programs that cooperate in the construction and analysis of Petri net models by sharing files. Using network file system capabilities, different analysis modules can be run on different machines in a distributed computing environment.     

The creators of this software have also written a book explaining generalized stochastic Petri nets:  \index{Petri net!stochastic!generalized} \index{generalized stochastic
Petri net}

\begin{itemize}
\item[\cite{MBCDF95}] 
M.\ Ajmone Marsan, G.\ Balbo, G.\ Conte, S.\ Donatelli and G.\ Franceschinis,
\textsl{\href{http://www.di.unito.it/~greatspn/bookdownloadform.html}{Modelling with Generalized Stochastic Petri Nets}}, Wiley, New York, 1995.
\end{itemize}
\noindent
If you enter some personal data, you can download a gzipped tar file containing
an electronic version of this book.

\subsection{Answers}
\label{sec:25_5}

\vskip 1em \noindent {\bf Problem 51.} 
Suppose we have a symmetric monoidal category freely generated by objects $A, B, C, D, E, F, G, H$ and morphisms
$$\begin{array}{rrcl}
f \colon &A \otimes B &\to& C \\ 
g \colon &C \otimes D &\to& E \\
h \colon &E \otimes F &\to& C \otimes G \otimes H \\
\end{array}$$
Is there a morphism from $A \otimes B \otimes D \otimes F$ to $C \otimes G \otimes H$?  

\vskip 1em
\begin{answer} Yes.  If you draw the Petri net corresponding to this symmetric
monoidal category, you will see it looks exactly like the one we saw in Section
\ref{sec:25_1}:

\begin{center}
 \includegraphics[width=80mm]{chemistryNetBasicA.png}
\end{center}

\noindent
after we rename the objects:
$$  A = \textrm{O}_2, \; B = \textrm{C}, \; C = \textrm{CO}_2 , \; D = \textrm{NaOH} , $$
$$ E = \textrm{NaHCO}_2 ,\; F = \textrm{HCl} ,\;
     G = \textrm{H}_2\textrm{O}, \; H = \textrm{NaCl}  $$
After this renaming the reachability problem here is revealed to be the one we 
already solved n Section
\ref{sec:25_1}, with the answer being `yes'.  What seemed like category theory is
secretly chemistry---or perhaps the other way around!
\end{answer}

\newpage 

\section[Summary table]{Summary table}

\begin{table*}[h]%\caption{{\bf Quantum versus Stochastic Mechanics}}
\renewcommand{\arraystretch}{1.5}
\resizebox{0.8\textwidth}{!}{\begin{minipage}{\textwidth}
\begin{tabular}{ p{3cm}|p{5cm}|p{5cm} }
         & {\bf quantum mechanics} & {\bf stochastic mechanics} \\\hline 
  {\bf state}  
& vector $ \psi \in \mathbb{R}^n$ with
$$ \sum_i |\psi_i|^2 = 1 $$
  &
vector $\psi \in \C^n$ with $$ \sum_i \psi_i = 1 $$ 
and $$ \psi_i \ge 0 $$ 
\\\hline 
{\bf observable} & $n \times n$ matrix $O$ with 
$$  O^\dagger = O$$
where $(O^\dagger)_{i j} := \overline{O}_{j i}$
  & vector $O \in \mathbb{R}^n$  \\\hline 
{\bf expected value} &
$$ \langle \psi , O \psi \rangle := \sum_{i,j} \overline{\psi}_i O_{i j} \psi_j $$
&  $$ \langle O \psi \rangle := \sum_i O_i \psi_i $$ 
\\\hline
  {\bf symmetry}\break (linear map sending states to states) & unitary $n \times n$ matrix: $$ U U^\dagger = U^\dagger U =1 $$
 & stochastic $n \times n$ matrix: $$ \sum_i U_{i j} = 1 , \quad U_{i j} \ge 0 $$
\\\hline  
  {\bf symmetry \break generator} & self-adjoint $n \times n$ matrix: $$H=H^\dagger$$ 
 & infinitesimal stochastic $n \times n$ \break matrix:
$$ \sum_i H_{i j}=0 , \quad  i\neq j \; \Rightarrow \; H_{i j} \ge 0 $$ 
\\\hline
{\bf symmetries from symmetry \break generators} &
$$ U(t) = \exp(-itH) $$ &
$$ U(t) = \exp(tH) $$
\\\hline
{\bf equation of \hfill \break motion} & $$i\frac{d}{dt} \psi(t) = H \psi(t)$$ with solution $$\psi(t) = \exp(-itH)\psi(0)$$ & $$\frac{d}{dt} \psi(t) = H \psi(t)$$ with solution $$\psi(t) = \exp(tH)\psi(0)$$ 
 \end{tabular} 
\end{minipage} }
\end{table*}

\newpage 

\noindent In what follows, the pages that refer to a given citation are indicated by numbers in blue.  Free online references are in red.  These colored items are clickable
links.

\newpage 

\printindex


\begin{thebibliography}{THVL78}

\bibitem[AS11]{AS11} 
Wil van der Aalst and Christian Stahl, \textsl{Modeling Business Processes: 
A Petri Net-Oriented Approach}, MIT Press, Cambridge, 2011.

\bibitem[AAE06]{AAE06} 
Dana Angluin, James Aspnes and David Eistenstat, Stably computable predicates are semilinear, in \textsl{Twenty-Fifth ACM Symposium on Principles of Distributed Computing, July 2006}, ACM Press, New York, 2006, 
pp.\ 292--299.  Also available at \href{http://www.cs.yale.edu/homes/aspnes/papers/podc2006-proceedings.pdf}{http://www.cs.yale.edu/homes/aspnes/papers/podc2006-proceedings.pdf}.

\bibitem[ACK]{ACK08}
David~F. {Anderson}, George {Craciun}, and Thomas~G. {Kurtz},
Product-form stationary distributions for deficiency zero chemical
reaction networks, \textsl{Bull.\ Math.\ Bio.} (2010), 1947--1970.
Also available as \href{http://arxiv.org/abs/0803.3042}{arXiv:0803.3042}.

\bibitem[{And}]{And11}
David~F. {Anderson}, A proof of the Global Attractor Conjecture in the single linkage
  class case.  Available as \href{http://arxiv.org/abs/1101.0761}{arXiv:1101.0761}.

\bibitem[BC07]{BC07}
M.~A. Buice and J.~D. Cowan, Field-theoretic approach to fluctuation effects in neural networks, {\sl Phys. Rev. E} {\bf 75} (2007), 051919.

\bibitem[BC09]{BC09}
M.~A. Buice and J.~D. Cowan, Statistical mechanics of the neocortex, 
\textsl{Prog.\ Biophysics and Mol.\ Bio.} {\bf 99} (2009), 53--86.

\bibitem[BF]{BF12}
John~C. {Baez} and Brendan {Fong},
A Noether theorem for Markov processes,
\textsl{J.\ Math.\ Phys.} {\bf 54} (2013), 013301.  
Also available as \href{http://arxiv.org/abs/1203.2035}{arXiv:1203.2035}.

\bibitem[BFB66]{BFB66}
A.~T. Balaban, D.~F\v{a}rca\c{s}iu, and R.~B\v{a}nic\v{a},
Graphs of multiple 1,2-shifts in carbonium ions and related
systems, {\sl Rev.\ Roum.\ Chim.} (1966), 11.

\bibitem[BR91]{BR91}
Danail Bonchev and D.~H. Rouvray, editors.
\textsl{Chemical Graph Theory: Introduction and Fundamentals},
Taylor and Francis, 1991.

\bibitem[Bou98]{Bou98} 
Zakaria Bouziane, A primitive recursive algorithm for the general Petri net 
reachability problem, in \textit{39th Annual Symposium on Foundations of 
Computer Science}, IEEE, 1998, pp.\ 130--136.  Also available at
\href{http://hal.inria.fr/inria-00073286/PDF/RR-3404.pdf}{http://hal.inria.fr/inria-00073286/PDF/RR-3404.pdf}.

\bibitem[BMMS01]{BMMS01} 
Roberto Bruni, Jos\'e Meseguer, Ugo Montanari and Vladimiro
Sassone, Functorial models for Petri nets,
\textsl{Information and Computation} \textbf{170} (2001), 207--236.
Also available at \href{http://eprints.soton.ac.uk/264742/1/prenetsIandCOff.pdf}
{http://eprints.soton.ac.uk/264742/1/prenetsIandCOff.pdf}.

\bibitem[Car]{Car96}
J.~Cardy,
Renormalisation group approach to reaction-diffusion problems.
Available as \href{http://arxiv.org/abs/cond-mat/9607163}{
arXiv:cond-mat/9607163}.

\bibitem[CDS14]{CDS14}
Ho-Lin Chen, David Doty and David Soloveichik, Rate-independent computation in continuous chemical reaction networks, in \textsl{ITCS 2014: Proceedings of the 5th Innovations in Theoretical Computer Science Conference}, ACM, New York, pp.\ 313--326.   Also available at \href{http://www.dna.caltech.edu/~ddoty/papers/riccrn.pdf}{http://www.dna.caltech.edu/$\sim$ddoty/papers/riccrn.pdf}.

\bibitem[CSWB09]{CSWB09} 
Matt Cook, David Soloveichik, Erik Winfree and Jehoshua Bruck, Programmability of chemical reaction networks, in \textsl{Algorithmic Bioprocesses}, eds.\ Condon, Harel, Kok, Salomaa and Winfree, Springer, Berlin, 2009, pp.\ 543--584.   Also available at \href{http://www.dna.caltech.edu/Papers/programmability_of_CRNs_preprint2008.pdf}{http://www.dna.caltech.edu/Papers/programmability\_of\_} \href{http://www.dna.caltech.edu/Papers/programmability_of_CRNs_preprint2008.pdf}{of\_CRNs\_preprint2008.pdf}.

\bibitem[Cra15]{Cra15}
Gheorghe Craciun, Toric differential inclusions and a proof of the Global Attractor Conjecture.  Available as \href{https://arxiv.org/abs/1501.02860}{arXiv:1501.02860}.

\bibitem[CDSS]{CDSS07}
Gheorghe Craciun, Alicia Dickenstein, Anne Shiu, and Bernd Sturmfels,
Toric dynamical systems.  Available as \href{http://arxiv.org/abs/0708.3431}{arXiv:0708.3431}.

\bibitem[CTF06]{CTF06}
Gheorghe Craciun, Yangzhong Tang, and Martin Feinberg,
Understanding bistability in complex enzyme-driven reaction
  networks, \textsl{Proc.\ Nat.\ Acad.\ Sci.\ USA} {\bf 103} (2006), 8697--8702.
Also available at \href{http://mbi.osu.edu/files/3813/3528/0136/techreport_50_1.pdf}{http://mbi.osu.edu/files/3813/3528/0136/techreport\_50\_1.pdf}.

\bibitem[CLLLM18]{CLLLM18}
Wojciech Czerwinski, Slawomir Lasota, Ranko Lazic, J\'er\^ome Leroux and Filip Mazowiecki,
The reachability problem for Petri nets is not elementary.  Available as 
\href{https://arxiv.org/abs/1809.07115}{arXiv:1809.07115}.

\bibitem[DF09]{DF09}
Peter~J. Dodd and Neil~M. Ferguson,
A many-body field theory approach to stochastic models in population
  biology, \textsl{PLoS ONE} (2009), 4.  Available at
\href{http://www.plosone.org/article/info\%3Adoi\%2F10.1371\%2Fjournal.pone.0006855}
{http://www.plosone.org/article/info\%3Adoi\%2F10.1371\%2Fjournal.} 
\href{http://www.plosone.org/article/info\%3Adoi\%2F10.1371\%2Fjournal.pone.0006855}
{pone.0006855}.

\bibitem[Doi76a]{Doi76a}
M.~Doi, Second-quantization representation for classical many-particle
systems, \textsl{J.\ Phys.\ A} \textbf{9} (1976), 1465--1477.

\bibitem[Doi76b]{Doi76b}
M.~Doi, Stochastic theory of diffusion-controlled reactions,
{\sl J.\ Phys.\ A} \textbf{9} (1976), 1479--1495.

\bibitem[DSle]{DS84}
P.~G. Doyle and J.~L. Snell.
\textsl{Random Walks and Electrical Circuits},
Mathematical Association of America, Washington DC, 1984.
Also available at \href{http://www.math.dartmouth.edu/~doyle}{http://www.math.dartmouth.edu/$\sim$doyle}.

\bibitem[EN94]{EN94} 
Javier Esparza and Mogens Nielsen, Decidability issues for Petri nets---a survey, \textit{Bulletin of the European Association for Theoretical Computer Science} 
{\bf 52} (1994), 245--262.  Also available at 
\href{http://citeseerx.ist.psu.edu/viewdoc/summary?doi=10.1.1.2.3965}{http://citeseerx.ist.psu.edu/viewdoc/summary?doi=10.1.1.2.3965}.

\bibitem[Fei79]{Fei79}
Martin Feinberg, {\sl Lectures On Reaction Networks}, 1979.  Available at \href{http://www.che.eng.ohio-state.edu/~FEINBERG/LecturesOnReactionNetworks/}{http://www.che.eng.ohio-state.edu/$\sim$FEINBERG/LecturesOnReactionNetworks/}.

\bibitem[Fei87]{Fei87}
Martin Feinberg, Chemical reaction network structure and the stability of complex
  isothermal reactors: I. The deficiency zero and deficiency one theorems,
\textsl{Chemical Engineering Science} \textbf{42} (1987), 2229--2268.

\bibitem[FR74]{FR74} 
M.\ J.\ Fischer and Michael O.\ Rabin, Super-exponential complexity of Presburger arithmetic, \textsl{Proceedings of the SIAM-AMS Symposium in Applied Mathematics} 
{\bf 7} (1974), 27--41.

\bibitem[For11]{For11}
Daniel~B. Forger, Signal processing in cellular clocks,
\textsl{Proc.\ Nat.\ Acad.\ Sci.\ USA} \textbf{108} (2011), 4281--4285.
Available at \href{http://www.pnas.org/content/108/11/4281.full.pdf}{http://www.pnas.org/content/108/11/4281.full.pdf}.

\bibitem[Fro12]{Fro12}
Georg Frobenius, \"{U}ber Matrizen aus nicht negativen Elemente,
\textsl{S.-B.\ Preuss Acad.\ Wiss.\ Berlin} (1912), 456--477.

\bibitem[Fuk80]{Fuk80}
M.~Fukushima, \textsl{Dirichlet Forms and Markov Processes},
North-Holland, Amsterdam, 1980.

\bibitem[GP98]{GP98}
Peter J.~E. Goss and Jean Peccoud,
Quantitative modeling of stochastic systems in molecular biology by
  using stochastic Petri nets,
\textsl{Proc.\ Natl.\ Acad.\ Sci.\ USA} \textbf{95} (1998), 6750--6755.

\bibitem[Gub03]{Gub03}
Jonathan~M. Guberman,
\textsl{Mass Action Reaction Networks and the Deficiency Zero Theorem},
B.A.\ Thesis, Department of Mathematics, Harvard University, 2003. 
Available at  \href{http://www4.utsouthwestern.edu/altschulerwulab/theses/GubermanJ_Thesis.pdf}
  {http://www4.utsouthwestern.edu/altschulerwulab/theses/
GubermanJ$\underline{\;\;}$Thesis.pdf}.

\bibitem[Gun03]{Gun03}
Jeremy Gunawardena,
Chemical reaction network theory for \textit{in-silico} biologists,
2003.  Available at \href{http://vcp.med.harvard.edu/papers/crnt.pdf}{
  http://vcp.med.harvard.edu/papers/crnt.pdf}.

\bibitem[Haa02]{Haa02}
Peter~J. Haas,
\textsl{Stochastic Petri Nets: Modelling, Stability, Simulation},
Springer, Berlin, 2002.

\bibitem[Hor74]{Hor74} 
Fritz Horn, The dynamics of open reaction systems, in \textsl{Mathematical Aspects of 
Chemical and Biochemical Problems and Quantum Chemistry}, ed.\ Donald S.\ Cohen, 
\textsl{SIAM--AMS Proceedings} \textbf{8}, American Mathematical Society, 
Providence, R.I., 1974, pp.\ 125--137.

\bibitem[HJ72]{HJ72}
F.~Horn and Roy Jackson, General mass action kinetics,
\textsl{Archive for Rational Mechanics and Analysis} \textbf{47} (1972), 81--116.

\bibitem[Jan08]{Jan08} Petr Jan\v car, Bouziane’s transformation of the Petri 
net reachability problem and incorrectness of the related algorithm, 
\textsl{Information and Computation}, {\bf 206} (2008), 1259--1263.

\bibitem[JK09]{JK09} Kurt Jensen and Lars M.\ Kristensen, \textsl{Coloured Petri Nets: Modelling and Validation of Concurrent Systems}, Springer, Berlin, 2009.

\bibitem[KM27]{KM27}
W.~O. Kermack and A.~G. McKendrick,
A contribution to the mathematical theory of epidemics,
\textsl{Proc.\  Roy.\ Soc.\ Lond. A} \textbf{772} (1927), 700--721.  Available
at \href{http://rspa.royalsocietypublishing.org/content/115/772/700.full.pdf}{http://rspa.royalsocietypublishing.org/content/115/772/700.full.pdf}.

\bibitem[Kin93]{Kin93}
R.~Bruce King,
\textsl{Applications of Graph Theory and Topology in Inorganic Cluster
  Coordination Chemistry}, CRC Press, 1993.

\bibitem[Koc10]{K10}
Ina Koch, Petri nets---a mathematical formalism to analyze chemical reaction
  networks, \textsl{Molecular Informatics} \textbf{29} (2010), 838--843.

\bibitem[Ler08]{Ler08}
J\'er\^ome Leroux, The general vector addition system reachability problem by Presburger inductive separators, 2008.  Available at
\href{http://hal.inria.fr/docs/00/31/93/93/PDF/main.pdf}{http://hal.inria.fr/docs/00/31/93/93/PDF/main.pdf}.

\bibitem[LS15]{LS15}
J\'er\^ome Leroux and Sylvain Schmitz, Demystifying reachability in vector addition systems,
in \textsl{LICS '15: 30th Annual ACM/IEEE Symposium on Logic in Computer Science}, IEEE, 2015, pp.\ 56--67.  Also available as \href{https://arxiv.org/abs/1503.00745}{arXiv:1503.00745}.

\bibitem[Lip76]{Lip76} Roger J.~Lipton, The reachability problem requires exponential space, Technical Report 62, Yale University, 1976.  Available at
\href{http://www.cs.yale.edu/publications/techreports/tr63.pdf}{http://www.cs.yale.edu/publications/techreports/tr63.pdf}.

\bibitem[Log96]{Log96}
S.~R. Logan, \textsl{Chemical Reaction Kinetics}, Longman, Essex, 1996.

\bibitem[LR68]{LR68}
Paul~C. Lauterbur and Fausto Ramirez,
Pseudorotation in trigonal-bipyramidal molecules,
\textsl{J.\ Am.\ Chem.\ Soc.} \textbf{90} (1968), 6722--6726.

\bibitem[Man06]{Man06}
Marc Mangel,
\textsl{The Theoretical Biologist's Toolbox: Quantitative Methods for
  Ecology and Evolutionary Biology}, Cambridge U.\ Press, Cambridge, 2006.

\bibitem[MG98]{MG98}
D.\ C.\ Mattis and M.\ L.\ Glasser, The uses of quantum field theory in 
diffusion-limited reactions, \textsl{Rev.\ Mod.\ Phys.} \textbf{70} (1998), 979--1001.

\bibitem[MR92]{MR92}
Zhi-Ming Ma and Michael R\"{o}ckner,
\textsl{Introduction to the Theory of (Non-Symmetric) Dirichlet
  Forms}, Springer, Berlin, 1992.

\bibitem[MBCDF95]{MBCDF95}
M.\ Ajmone Marsan, G.\ Balbo, G.\ Conte, S.\ Donatelli and G.\ Franceschinis,
\textsl{Modelling with Generalized Stochastic Petri Nets}, Wiley, New York, 1995.
Also available at \href{http://www.di.unito.it/~greatspn/bookdownloadform.html}
{http://www.di.unito.it/$\sim$greatspn/bookdownloadform.html}.

\bibitem[May81]{May81}
Ernst Mayr, 
Persistence of vector replacement systems is decidable, \textsl{Acta Informatica} {\bf 15} (1981), 309--318.

\bibitem[MM90]{MM90}
Jos\'e Meseguer and Ugo Montanari, 
Petri nets are monoids, \textsl{Information and Computation} \textbf{88} (1990), 105--155.  Available at \href{http://www.sciencedirect.com/science/article/pii/0890540190900138}{http://www.sciencedirect.com/science/article/pii/0890540190900138}.

\bibitem[New]{New00}
Michael~William Newman, 
\textsl{The Laplacian Spectrum of Graphs},
Masters Thesis, Department of Mathematics, University of Manitoba, 2000.  
Available at
  \href{http://www.seas.upenn.edu/~jadbabai/ESE680/Laplacian_Thesis.pdf}{http:/$\!$/www.seas.upenn.edu/$\sim$jadbabai/ESE680/Laplacian\_Thesis.pdf}.

\bibitem[Opp78]{Opp78} 
Derek C.\ Oppen, A $2^{2^{2^{pn}}}$ upper bound on the complexity of 
Presburger arithmetic, \textsl{J.\ Comput.\ Syst.\ Sci.\ }{\bf 16} (1978), 323--332.

\bibitem[Pel85]{Pel85}
L.~Peliti,
Path integral approach to birth-death processes on a lattice,
\textsl{J.\ Phys.\ France} \textbf{46} (1985), 1469--1483.

\bibitem[Per07]{Per07}
Oskar Perron, 
Zur Theorie der Matrizen,
\textsl{Math.\ Ann.} \textbf{64} (1907), 248--263.

\bibitem[Pet81]{Pet81}
James L.\ Peterson, \textit{Petri Net Theory and the Modeling of Systems}, 
Prentice--Hall, New Jersey, 1981.

\bibitem[Ran97]{Ran97}
Milan Randic,
Symmetry properties of graphs of interest in chemistry: II:
Desargues--Levi graph,
\textsl{Int.\ Jour.\ Quantum Chem.} \textbf{15} (1997), 663--682.

\bibitem[RRTK]{RPT02}
Katarzyna~Pich\`{o}r Ryszard~Rudnicki and Marta Tyran-Kam\`{i}nska,
Markov semigroups and their applications, 2002.
  \href{http://www.impan.pl/~rams/r48-ladek.pdf}{Available at
  http://www.impan.pl/$\sim$rams/r48-ladek.pdf}.

\bibitem[ST77]{ST77}
S.\ Sacerdote and R.\ Tenney, 
The decidability of the reachability problem for vector addition systems, {\sl Conference Record of the Ninth Annual ACM Symposium on Theory of Computing, 2-4 May 1977, Boulder, Colorado, USA}, ACM, 1977, pp.\ 61--76.

\bibitem[Sol14]{Sol14} 
David Soloveichik, The computational power of chemical reaction
networks, lecture at Programming with Chemical Reaction Networks: Mathematical Foundations, Banff, June 2014.  Available at 
\href{http://crn.thachuk.com/images/d/d5/Soloveichik-Banff_Tutorial.pdf}{http://crn.thachuk.com/images/d/d5/Soloveichik-Banff\_Tutorial.pdf}.

\bibitem[SCWB08]{SCWB08} 
David Soloveichik, Matt Cook, Erik Winfree and Jehoshua Bruck, Computation with finite stochastic chemical reaction networks, 
\textsl{Natural Computing} \textbf{7} (2008), 615--633. 
Also available at \href{http://dna.caltech.edu/Papers/sCRN_computation_TR2007.pdf}
{http://dna.caltech.edu/Papers/sCRN\_computation\_TR2007.pdf}.

\bibitem[THVL78]{Ta05}
Uwe T\"auber, Martin Howard, and Benjamin~P. Vollmayr-Lee,
Applications of field-theoretic renormalization group methods to
reaction-diffusion problems,
\textsl{J.\ Phys.\ A} \textbf{38} (2005), R79.  Also available as
  \href{http://arxiv.org/abs/cond-mat/0501678}{arXiv:cond-mat/0501678}.

\bibitem[Tri92]{Tr92}
Nenad Trinajstic,
\textsl{Chemical Graph Theory},
CRC Press, Boca Raton, Florida, 1992.

\bibitem[Wil06]{Wil06}
Darren~James Wilkinson,
\textsl{Stochastic Modelling for Systems Biology},
Taylor and Francis, New York, 2006.

\bibitem[Van07]{Van07}
Nico Van Kampen,
\textsl{Stochastic Processes in Physics and Chemistry},
North--Holland, New York, 2007.

\bibitem[ZC08]{ZC08} Gianluigi Zavattaro and Luca Cardelli, 
Termination problems in chemical kinetics, in \textsl{CONCUR 2008--Concurrency Theory}, eds.\ Zavattaro and Cardelli, \textsl{Lecture Notes in Computer Science} \textbf{5201}, Springer, Berlin, 2008, pp.\ 477--491.  Also available at
\href{http://lucacardelli.name/Papers/Termination\%20Problems\%20in\%20Chemical\%20Kinetics.pdf}{http://lucacardelli.name/Papers/Termination Problems in Chemical Kinetics.pdf}

\end{thebibliography}
\end{document}